\documentclass[12pt,english]{article}
\usepackage[T1]{fontenc}
\usepackage[latin9]{inputenc}
\usepackage{geometry}
\usepackage{babel}
\usepackage{varioref}
\usepackage{url}
\usepackage{amsmath}
\usepackage{amsthm}
\usepackage{amssymb}
\usepackage{graphicx}
\usepackage{setspace}
\usepackage[authoryear]{natbib}
\PassOptionsToPackage{normalem}{ulem}
\usepackage{ulem}
\usepackage[unicode=true]
 {hyperref}

\makeatletter
\theoremstyle{definition}
\newtheorem{defn}{\protect\definitionname}
\theoremstyle{plain}
\newtheorem{thm}{\protect\theoremname}
\theoremstyle{plain}
\newtheorem{assumption}{\protect\assumptionname}
\theoremstyle{plain}
\newtheorem{lem}{\protect\lemmaname}

\usepackage{setspace}
\usepackage{pdflscape}
\usepackage{graphicx}
\usepackage{epstopdf}
\usepackage{ifpdf}
\ifpdf
  \DeclareGraphicsExtensions{.pdf,.png,.jpg}
\else
  \DeclareGraphicsExtensions{.eps}
\fi
\usepackage{url}

\makeatother

\providecommand{\assumptionname}{Assumption}
\providecommand{\definitionname}{Definition}
\providecommand{\lemmaname}{Lemma}
\providecommand{\theoremname}{Theorem}

\begin{document}
\begin{singlespacing}
\title{Network Data}

\maketitle
\medskip{}

\begin{center}
{\large{}(prepared for the }\emph{\large{}Handbook of Econometrics}{\large{},
Volume 7A)}{\large\par}
\par\end{center}

\medskip{}

\begin{center}
{\large{}Bryan S. Graham}\footnote{{\footnotesize{}Department of Economics, University of California
- Berkeley, 530 Evans Hall \#3380, Berkeley, CA 94720-3880 and National
Bureau of Economic Research, }{\footnotesize{}\uline{e-mail:}}{\footnotesize{}
\href{http://bgraham@econ.berkeley.edu}{bgraham@econ.berkeley.edu},
}{\footnotesize{}\uline{web:}}{\footnotesize{} \url{http://bryangraham.github.io/econometrics/}.
Financial support from NSF grants SES \#1357499 and SES \#1851647
is gratefully acknowledged. I am grateful for comments provided by
the co-editors and other participants at a conference held at the
University of Chicago in August of 2017. Portions of the material
presented below benefited from conversations with Peter Bickel, Michael
Jansson and Jim Powell. I am especially grateful to Eric Auerbach,
Seongjoo Min, Chris Muris, Fengshi Niu and Konrad Menzel, as well
as an anonymous referee, for written feedback which greatly improved
the chapter. All the usual disclaimers apply.}}
\par\end{center}

\begin{center}
\medskip{}
\textsc{\large{}Initial Draft: June 2017, This Draft: September 2019}{\large\par}
\par\end{center}
\begin{abstract}
Many economic activities are embedded in \emph{networks}: sets of
agents and the (often) rivalrous relationships connecting them to
one another. Input sourcing by firms, interbank lending, scientific
research, and job search are four examples, among many, of networked
economic activities. Motivated by the premise that networks' structures
are consequential, this chapter describes econometric methods for
analyzing them. I emphasize (i) dyadic regression analysis incorporating
unobserved agent-specific heterogeneity and supporting causal inference,
(ii) techniques for estimating, and conducting inference on, summary
network parameters (e.g., the degree distribution or transitivity
index); and (iii) empirical models of strategic network formation
admitting interdependencies in preferences. Current research challenges
and open questions are also discussed.
\end{abstract}
\end{singlespacing}

\thispagestyle{empty} 

\pagebreak{}

\setcounter{page}{1}

\tableofcontents{}

\pagebreak{}

\section{Introduction and summary}

Many economic activities are embedded in \emph{networks}: sets of
agents and the (often) rivalrous relationships connecting them to
one another. Firms generally buy and sell inputs not in anonymous
markets, but via bilateral contracts \citep{Kranton_Minehart_AER01}.
In addition to public listings, individuals gather information about
job opportunities from friends and acquaintances \citep{Granovetter_AJS73}.
We similarly poll friends for information about new products, books,
movies and so on \citep[e.g., ][]{Jackson_Rogers_BE07,Banerjee_et_al_Sci13,Kim_et_al_Lancet2015}.
Banks generally meet reserve requirements through peer-to-peer interbank
lending. The structure of this interbank lending network has profound
implications for the vulnerability of the financial system to large
negative shocks \citep{Bech_Atalay_PhyA10,Gofman_JFE17}. Additional
examples abound \citep[cf., ][]{Jackson_et_al_JEL17}.

\begin{figure}
\caption{\label{fig: World-Trade-Network-1928}World Trade Network in 1928}

\begin{centering}
\includegraphics[scale=0.35]{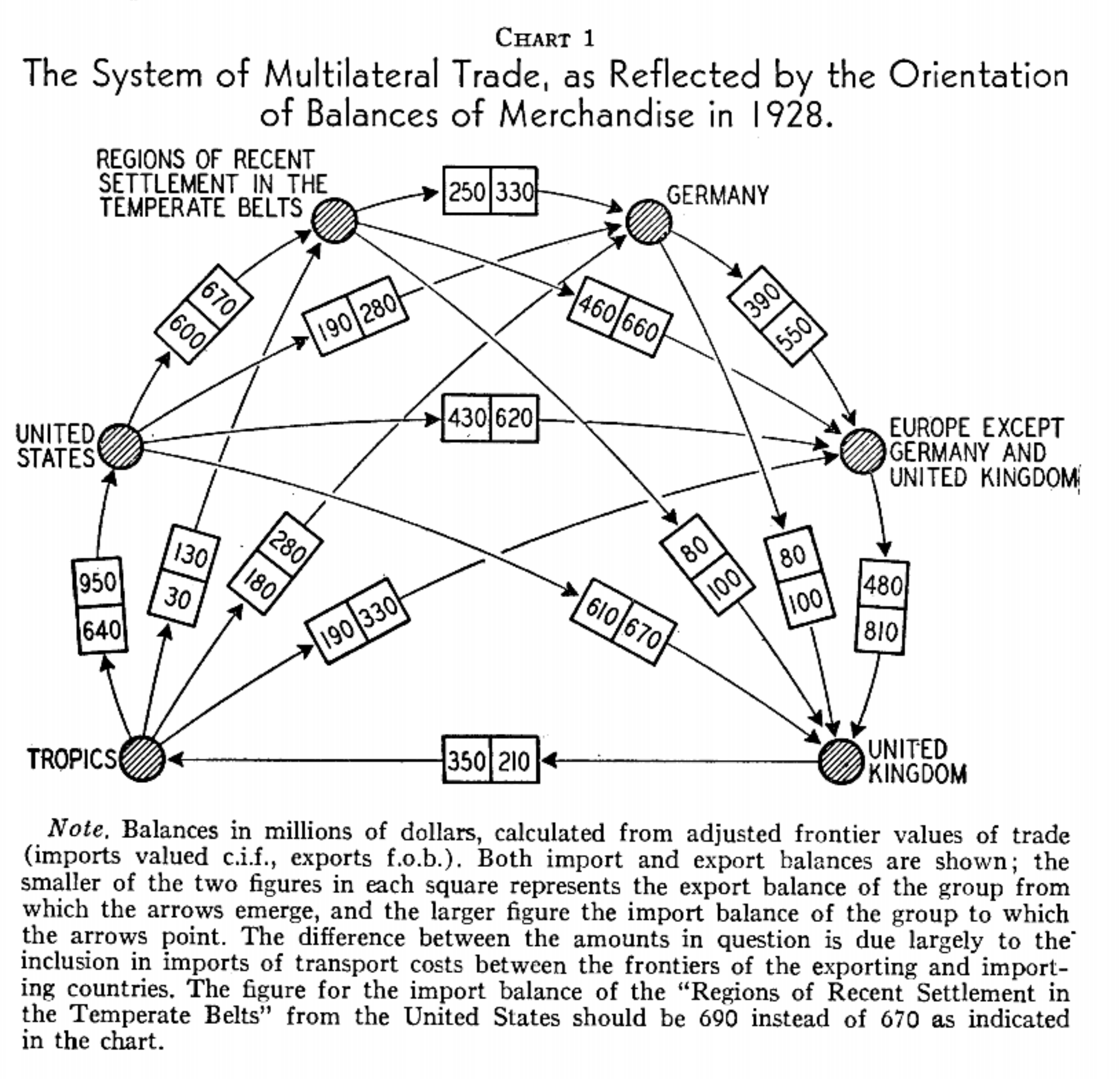}
\par\end{centering}
\uline{Notes:} This figured appeared in Folke Hilgerdt's 1943 \emph{American
Economic Review }article ``The case for multilateral trade''. The
figure shows aggregate trade balances between selected large countries
and different regions of the world. The paper includes a narrative
discussion of how the patterns of trade depicted in weighted digraph
drawn in the figure developed historically.

\uline{Source:} Reproduced from \citet[Chart 1]{Hilgerdt_AER43}.
\end{figure}

Although important exceptions exists, some highlighted below, economists
historically avoided the study networks (see Figure \ref{fig: World-Trade-Network-1928}).\footnote{In contrast our colleagues in sociology studied networks from the
outset of their discipline in its modern form. The monograph by \citet{Wasserman_Faust_Bk94}
provides a somewhat dated introduction to this literature. See also
\citet{Granovetter_AJS85}.} This is now changing, very quickly, and for several reasons. First,
starting in the 1990s economic theorists applied the tools of game
theory to formally study network formation \citep[e.g., ][]{Jackson_Wolinsky_JET96}.
In the resulting models agents add, maintain, and subtract links in
order to maximize utility, with the realized network satisfying a
pairwise stability equilibrium condition.\footnote{Other equilibrium concepts have been explored as well \citep[cf., ][]{Bloch_Jackson_IJGT06}.}
Second, in parallel to this theoretical work, a lively empirical and
methodological literature on peer group and neighborhood effects also
arose \citep[e.g., ][]{Manski_ReStud93,Brock_Durlauf_RES01,Graham_EM08,Angrist_LE14}.
Finally, largely driven by questions in empirical industrial organization,
econometricians made substantial progress on the econometric analysis
of games \citep[cf., ][]{Bajari_et_al_WC13,dePaula_ARE13}. Each of
these literatures serve as foundations for material introduced below.

Outside of economics, two key initiators have been (i) the increasing
availability of datasets with natural graph theoretic structure (see
below for examples) and (ii) innovations in applied probability and
theoretical statistics pertaining to random graph models \citep[e.g., ][]{Diaconis_Janson_RM08}.
These innovations provide a foundation upon which recent work in statistics
and machine learning on networks is largely based.

A consequence of these developments is the emergence of a small methodological
literature on the econometrics of networks. Empirical applications
with substantial network content, spurred largely by access to new
datasets, arose more quickly \citep[e.g., ][]{Fafchamps_Minten_OEP02,deWeerdt_IAP04,Conley_Udry_AER10,Atalay_et_al_PNAS11,Acemoglu_etal_EM12,Banerjee_et_al_Sci13,Barrot_Sauvagnat_QJE16}.
Furthermore, these applications now span the major fields of our discipline.
Nevertheless many open questions in the econometrics of networks remain.
In this chapter I attempt to provide an account of recent progress
as well as make suggestions for future research. My audience is both
econometricians and empirical researchers.

I divide my discussion into five parts. The discussion draws from
recent contributions to the analysis of networks made in probability,
econometrics, and statistics (including machine learning); approximately
in that order. After an initial outline of recent empirical research
with a network dimension in economics, Section \ref{sec: Basic-probability-tools}
introduces some basic probability tools that will prove useful for
what follows. Several of these tools are of quite recent origin. Next,
in Sections \ref{sec: dyadic_regression} to \ref{sec: heterogeneity}
I turn to the analysis of dyadic regression models. Such models go
back, at least, to the pioneering work of \citet[Appendix VI]{Tinbergen_SWE62}
on gravity trade models. Although dyadic regression is a core empirical
method in international trade, as well as in certain areas of political
science and development economics, a coherent inferential foundation
for empirical practice is only now emerging. My discussion, in addition
to covering methods of inference, discusses how to incorporate unobserved
heterogeneity into dyadic regression models (Section \ref{sec: heterogeneity}).
Here I appropriate and extend insights from panel data \citep{Chamberlain_ReStud80,Chamberlain_HBE84,Chamberlain_LALMD85,Hahn_Newey_EM04,Arellano_Hahn_WC07}.
This section also sketches out how to answer causal questions in dyadic
settings.

Section \ref{sec: Statistics} turns to the large network properties
of several common network statistics. I focus on so-called \emph{network
moments}, or the frequencies with which certain low order subgraph
configurations (e.g., triangles $\vcenter{\hbox{\includegraphics[scale=0.125]{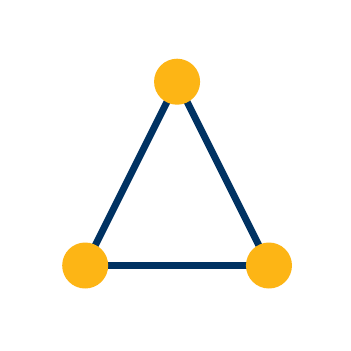}}}$)
occur within a network. Subgraph counts, in the form of the triad
census, were introduced by \citet{Holland_Leinhardt_AJS70} almost
a half-century ago. Recent developments in probability and statistics
have substantially improved our understanding of these counts \citep[e.g., ][]{Diaconis_Janson_RM08,Bickel_et_al_AS11}. 

Subgraph counts may be of direct interest, but also serve as the building
blocks of several popular network statistics, such as transitivity
or moments of the degree distribution. \citet{Jackson_et_al_JEL17}
survey the mapping between different network statistics and economic
phenomena and questions. My interest in network moments also stems
from their value as inputs into structural model estimation in a manner
akin to the way sample moments are paired with model moments in the
simulated method of moments \citep[e.g., ][]{Gourieroux_et_al_JAE93}.
This idea is developed in Section \ref{sec: Strategic-models}.

The discussion of dyadic regression in Sections \ref{sec: dyadic_regression}
to \ref{sec: heterogeneity} rules out interdependencies in link formation.
In dyadic models the utility two agents generate by forming a link
is invariant to the presence or absence of links elsewhere in the
network. Beginning with the seminal work of \citet{Jackson_Wolinsky_JET96},
the relaxation of this assumption is a central preoccupation of both
theoretical and econometric researchers. When link formation decisions
are interdependent, inefficient network structures may occur in equilibrium,
making policy analysis interesting. Empirical network formation models
allowing for interdependencies are also challenging to study. In a
typical model many equilibrium network configurations can arise for
any given parameter value; such models are incomplete \citep[e.g., ][]{Tamer_ReStud03}.
In principle, standard tools developed in the context of economic
games between a small number of agents apply. Practically speaking
such methods are computationally infeasible in the many agent context
of networks. Recent research proposes a variety of ways of getting
around this conundrum.

Economists' interest in networks stems from the belief that their
structure is consequential. For example, \citet{Loury_ARI02} argues
that differences in social networks across Blacks and Whites drives,
in part, racial inequality \citep[cf., ][]{Graham_JEL18}. \citet{Acemoglu_etal_EM12}
argue that the Leontief input-output structure of the economy shapes
technology shock propagation. \citet{Alatas_et_al_AER16} show that
network structure influences the flow and aggregation of information
within rural villages. Theorists also study the interplay between
network structure and agent behavior \emph{on} that structure \citep{Jackson_Yariv_HSE11,Jackson_Zenou_HBGT15}.
Methodological research relating network structure to economic outcomes
builds-upon the line of peer effects research initiated by \citet{Manski_ReStud93}.
The paper by \citet{Bramoulle_et_al_JOE09} is a nice, and influential,
example of recent work along such lines.

This survey, however, does not review methods for the empirical analysis
of behavior on networks. Instead I focus on modeling their \emph{formation}.
My motivation for this emphasis is two-fold. First, \citet{Blume_et_al_HSE2011}
already survey work at the intersection of peer group effect identification
and networks \citep[cf., ][]{Blume_et_al_JPE15,dePaula_WC17}. Second,
the current state of research in this area suggests that a better
understanding of how networks form is a prerequisite for more credible
research on their consequences. 

Current research on the effects of network structure on outcomes largely
treats it as exogenously given (although this is not always made explicit).
This decision is one reason why research on peer effects and networks
remains controversial a quarter century after Manski's foundational
paper.\footnote{For example, \citet[p. 81]{Jackson_et_al_JEL17} argue that endogenous
network formation, the tendency for the unobserved drivers of link
formation and the behavior of interest to the econometrician to covary,
poses a key challenge to ``accurately estimating interactive effects
in networked settings''.} The focus maintained here, on \emph{formation}, therefore seems to
be a natural one. Ultimately, of course, the goal is to study the
formation of networks and their consequences jointly, but such an
integrated treatment remains largely aspirational at this stage. Although,
\citet{Goldsmith-Pinkham_Imbens_JBES13} provide one recent ``proof
of possibilities'' example of such an integrated approach. \citet{Qu_Lee_JOE15},
\citet{Auerbach_JMP16}, \citet{Badev_arXiv17}, and \citet{Johnson_Moon_INET17}
represent other steps in this direction.

\section{Examples, questions and notation}

The analysis of datasets with natural graph theoretic structure has
a long history in the other social sciences \citep[e.g.,][]{Moreno_Book34},
and more recently emerged as an area of focus within the statistics
and machine learning community \citep[e.g., ][]{Goldenberg_etal_FTML09,Kolaczyk_NetBook09}.
Although we were late adopters, interest in these types of datasets
now also extends across virtually all fields of economics. Nevertheless,
as already noted, appropriate methods for the analysis of network
data are not widely available. Ad hoc and/or heuristically motivated
approaches to estimation and inference abound in empirical work. Networks
are characterized by complex dependencies across agents, as well as
other difficult modeling, estimation and inferential challenges. These
challenges are just starting be understood and solved. Before discussing
methods for the analysis of network data, I briefly introduce some
recent examples of empirical network research in economics. These
examples also serve to introduce some basic notation.

\subsection{Empirical analysis of trade flows}

Figure \ref{fig: World-banana-trade} visually depicts international
trade in bananas, a widely-eaten tropical fruit, in 2015. Each dot
or \emph{node} in the figure corresponds to a country. If, for example,
Honduras, exports at least 50,000 tons of bananas to the United States,
then there exists a \emph{directed edge} $\vcenter{\hbox{\includegraphics[scale=0.125]{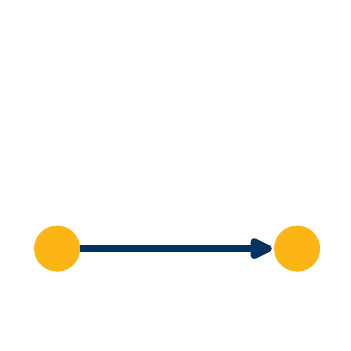}}}$
from Honduras to the United States.\footnote{In constructing this network, I binarized the underlying trade flow
data to determine edge placement.} The exporting country (left node) is called the \emph{tail} of the
edge, while the importing country (right node) is its \emph{head.
}The set of all such exporter-importer relationships forms $G\left(\mathcal{V},\mathcal{E}\right)$,
a directed network or \emph{digraph} defined on $N=\left|\mathcal{V}\right|$
vertices or agents (here countries). The set $\mathcal{V}=\left\{ 1,\ldots,N\right\} $
includes all agents (countries) in the network and $\mathcal{E}\subseteq\mathcal{V}\times\mathcal{V}$
the set of all directed links (exporter-importer relationships of
50,000 tons or greater) among them.\footnote{Here $\mathcal{U}\times\mathcal{V}$ denotes the Cartesian product
of the set $\mathcal{U}$ and $\mathcal{V}$ (i.e, $\mathcal{U}\times\mathcal{V}=\left\{ \left(u,v\right)\thinspace:\thinspace u\in\mathcal{U},\thinspace v\in\mathcal{V}\right\} $).} Let $N$ be the \emph{order} of the digraph and $\left|\mathcal{E}\right|$
its \emph{size}. In what follows nodes may be equivalently referred
to as vertices, agents, individuals, countries and so on depending
on the context. Likewise edges may be called links, friendships, ties,
arcs, relationships and so on.

There are $N=220$ countries in the banana network and hence up to
$2\tbinom{220}{2}=48,180$ directed trading relationships among them.
How might an econometrician model the presence or absence of a trading
relationship from country $i$ to $j$? Over fifty years ago \citet[Appendix VI]{Tinbergen_SWE62}
introduced gravity models, suitable for data of the type shown in
Figure \ref{fig: World-banana-trade}. In a gravity model trade between
two countries, a \emph{dyad} in network parlance, is modeled as a
function of exporter and importer attributes (e.g., their gross domestic
products), as well as dyad-specific covariates (e.g., physical distance
between them). Generalizations of Tinbergen's approach are workhorses
of modern empirical trade research \citep[e.g., ][]{SantosSilva_Tenreyro_RESTAT06,Helpman_et_al_QJE08,Anderson_AR11}.

Their ubiquity notwithstanding, serious open questions remain about
how to estimate, and conduct inference on, the parameters of gravity
trade models. Questions of particular interest here include how to
account for the dependence across dyads sharing a country in common,
how to incorporate country-specific (correlated) unobserved heterogeneity,
and how to formalize causal policy effects in dyadic settings. As
an example of the latter challenge, consider the effects of participation
in multi-lateral trading agreements, such as the General Agreement
on Tariffs and Trade (GATT) or its successor, the World Trade Organization
(WTO), on trade flows. Does trade increase across participating countries
\citep{Rose_AER04,Helpman_et_al_QJE08}? While a mature literature
on program evaluation suitable for single agent settings now exists
\citep[cf., ][]{Heckman_Vytlacil_HBE07,Imbens_Wooldridge_JEL09},
a networked counterpart has yet to emerge.

\subsection{Corporate governance}

Next consider the affiliation network of (corporate board) directors
and firms. This bipartite network $B\left(\mathcal{\mathcal{U}},\mathcal{V},\mathcal{E}\right)$
consists of two sets of agents, the set of possible directors, $\mathcal{U}$,
and the set of firms, $\mathcal{V}$. Edges, $\mathcal{E}$, match
directors to firms (i.e., corporate boards), and hence may only run
between $\mathcal{V}$ and $\mathcal{U}$. A longstanding interest
among corporate governance researchers centers on the implications
of so-called board interlocks. When a single director sits on multiple
corporate boards, then these corporations have interlocking directorates
\citep{Dooley_AER69}. Interlocking directorships may facilitate collusion
and other anti-competitive activities as well as, perhaps more positively,
the diffusion of innovations in corporate governance \citep{Davis_ASM91,Davis_CG96}. 

Figure \ref{fig: Board-Interlocks} plots the one-mode projection
of the directors-to-firms bipartite network for S\&P 1,500 firms in
2016. This projection generates an \emph{undirected} network $G\left(\mathcal{V},\mathcal{E}\right)$
on the set of all firms, with an edge between any two firms sharing
at least one director in common (i.e., with interlocking corporate
boards). Large firms in United States are inter-connected via overlapping
corporate board membership. On average firms share at least one board
member in common with four other firms and over 80 percent of S\&P
1,500 firms form a giant connected component of board interlocks.
The board interlock network is also highly \emph{transitive}: two
firms are much more likely to share a director in common, if they
also share one in common with a third firm.

\citet{Chu_Davis_AJS16} and \citet{Gualdani_JOE19} provide recent
analyses of board interlocks as well as references to earlier work.

\subsection{Production networks}

\citet{Atalay_et_al_PNAS11} study the production network of the United
States economy. The sale and purchase of intermediate inputs between
firms joins virtually all publicly traded corporations in the United
States into one giant buyer-supplier network. 

\citet{Serpa_Krishnan_MS17} present evidence of productivity spillovers
across firms linked together via supply chain relationships \citep[cf., ][]{Acemoglu_et_al_NBERMacro16}.
\citet{Acemoglu_etal_EM12} study the effect of the Leontief input-output
structure of the US economy on shock propagation. Their analysis suggests
that idiosyncratic technology shocks to critical input suppliers may
have macro-level effects.

\citet{Bernard_et_al_JPE18}, using detailed supply-chain data from
Japan, show how lowering supplier search costs allows firms to source
inputs more efficiently, in turn lowering marginal production costs.
The rich supply-chain data underlying the analysis of \citet{Bernard_et_al_JPE18}
is emblematic of the increasing availability of detailed supply chain
network data from different countries \citep[e.g., ][]{Dhyne_et_al_WP15}.
These datasets have the potential to dramatically improve our understanding
of, for example the sources of heterogeneity in productivity across
firms \citep[e.g., ][]{Atalay_et_Al_AER14} and the upstream and downstream
implications of (horizontal) mergers \citep[e.g., ][]{Fee_Thomas_JFE04,Bhattacharya_Nain_JFE11,Ahern_Harford_JoF14},
among many other areas of industrial organization and regulation policy.

\pagebreak{}

\begin{landscape}

\begin{figure}
\caption{World Trade in Bananas, 2015 \label{fig: World-banana-trade}}

\begin{centering}
\includegraphics{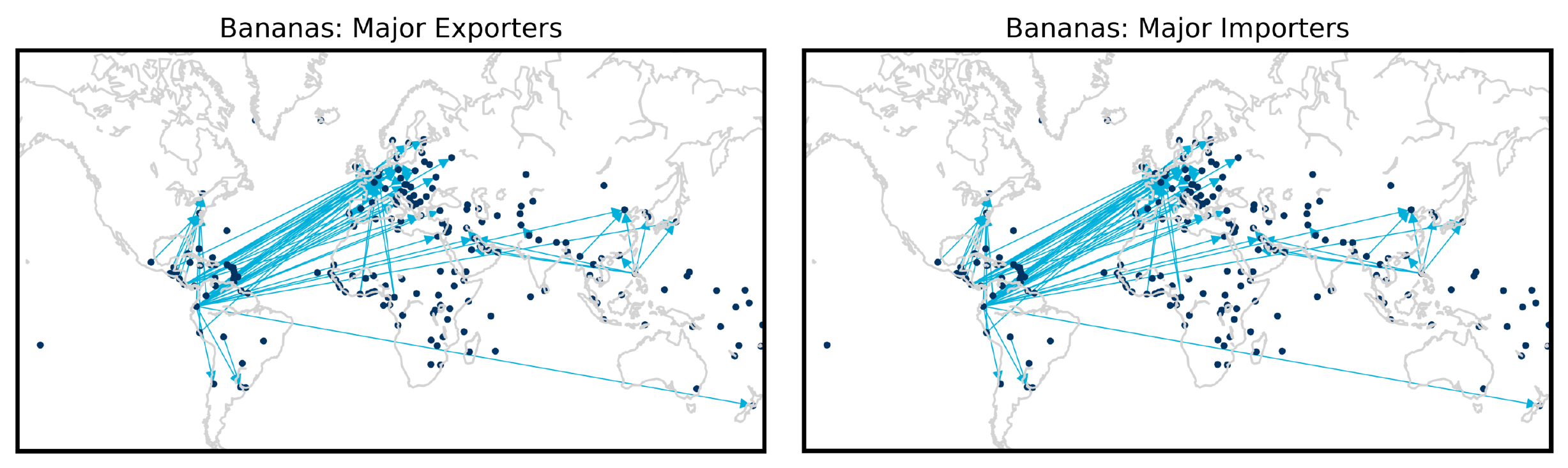}
\par\end{centering}
\textsc{\small{}\uline{Source:}}{\small{} BACI-CEPII International
Trade Database \citep[cf., ][]{Gaulier_et_al_CEPII10,DeBenedictis_GE14}
and author's calculations.}{\small\par}

\textsc{\small{}\uline{Notes:}}{\small{} International trade of
bananas in 2015 (HS6 code 080390). Each node in the figure represents
a country (nodes are positioned at capital cities) and an edge between
two nodes indicates the presence of at least 50,000 tons of directed
banana flows (the head of each directed edge corresponds to the importing
nation). In the left-hand panel node size is proportional to the total
exports of bananas by the relevant nation, while in the right it is
proportional to its total imports.}{\small\par}
\end{figure}

\end{landscape}

\pagebreak{}

\subsection{Research collaboration}

\citet{Jaffe_AER86}, in a classic study, presented evidence of research
and development (R\&D) spillovers across technologically adjacent
firms \citep{Bloom_et_al_EM13,Acemoglu_Akcigit_Kerr_PNAS16}. Such
spillovers provide a motivation for firms to undertake collaborative
R\&D, a tendency which has increased over time \citep{Hagedoorn_RP02,Tomasello_et_al_ICC17}.
\citet{Konig_et_al_RESTAT19} model the formation of R\&D partnerships
across firms theoretically and empirically, exploring the implications
of network structure for optimal R\&D subsidy policies. The structure
of spillovers across firms, as well as the mechanisms whereby they
form R\&D partnerships, determines optimal policies.

\citet{Ductor_et_al_RESTAT14} study collaboration and research output
among and across economists. \citet{Newman_PNAS01} explores collaboration
networks in the various sciences.

\subsection{Risk-sharing across households}

A classic question in development economics is whether households
efficiently share risk through informal agreements \citep{Townsend_EM94,Udry_ReStud94}.
Recently economists have directly collected information on risk-sharing
relationships across households. For example, \citet{deWeerdt_IAP04}
collected data on risk-sharing links across households in a village
in Tanzania and empirically modeled the determinants of these links
\citep[cf., ][]{Fafchamps_Lund_JDE03,Fafchamp_Gubert_JDE07}. \citet{Ambrus_et_al_AER14}
investigate how the precise structure of links across households determines
the amount of risk that can be insured, as well as the form of second
best, more local, network structures.

Network structure now informs many other areas of development economics,
including research on technology adoption and program take-up in rural
settings \citep[e.g., ][]{Banerjee_et_al_Sci13,Kim_et_al_Lancet2015},
the productivity of small traders and firms \citep[e.g., ][]{Fafchamps_Minten_OEP02},
and post-migration employment outcomes \citep{Beaman_ReStud11,Munshi_QJE03},
among other examples.

\pagebreak{}

\begin{landscape}

\begin{figure}
\caption{United States Corporate Board Interlocks, 2016\label{fig: Board-Interlocks}}

\begin{centering}
\includegraphics{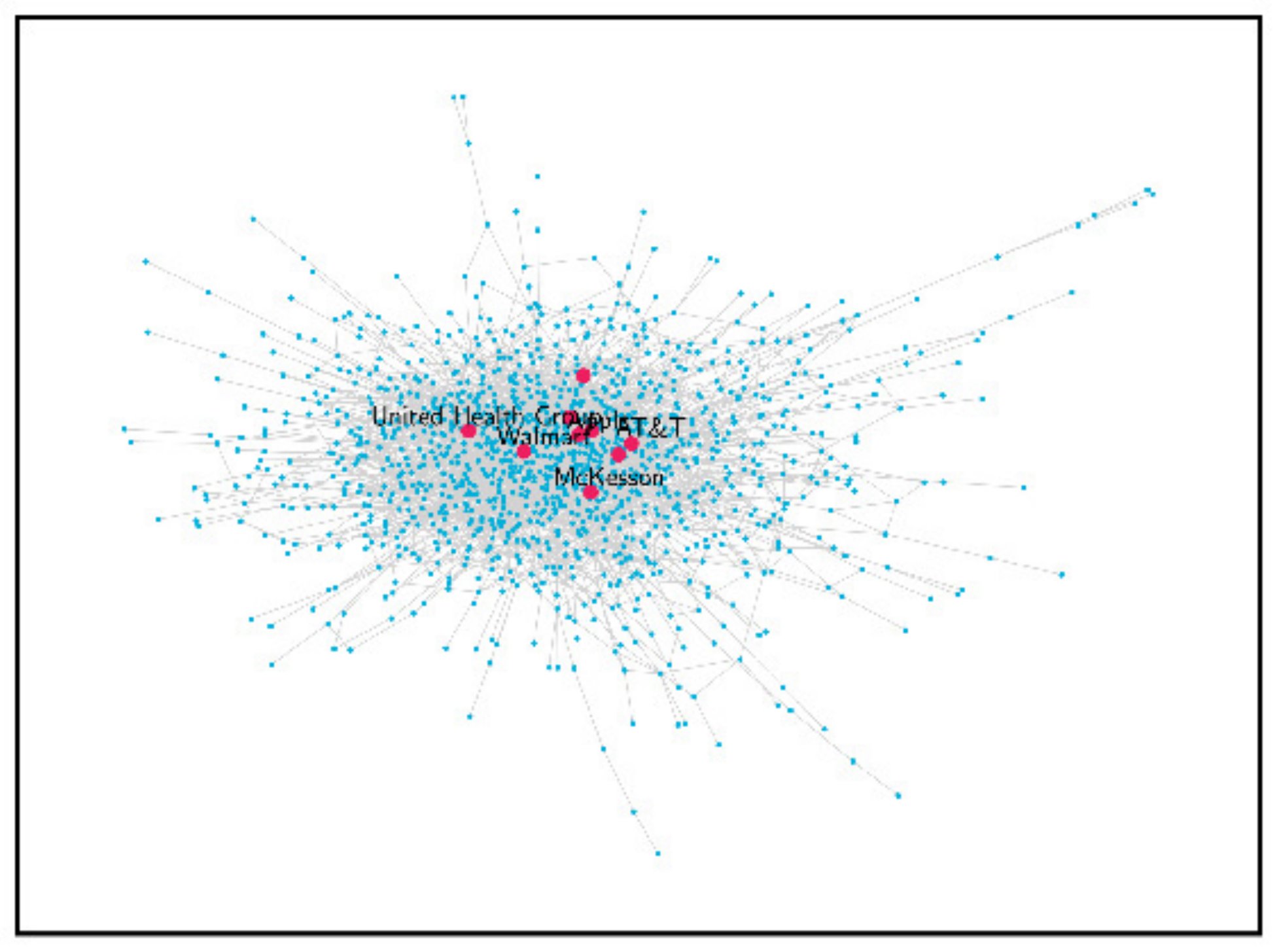}
\par\end{centering}
\textsc{\small{}\uline{Source:}}{\small{} Wharton Research Data
Services (WRDS) - Institutional Shareholder Services (ISS) Directors
dataset and author's calculations \citep[cf., ][]{Chu_Davis_AJS16}.}{\small\par}

\textsc{\small{}\uline{Notes:}}{\small{} The figure plots the largest
connected component of the corporate board interlock network in 2016
among S\&P 1,500 firms. The top 10 Fortune 500 firms in 2016 are the
larger `Rose Garden' colored nodes. A total of 1,216 firms belong
to the largest connected component. See \citet[p. 124 - 127]{Newman_NetBook10}
for details on how to construct one-mode projections of bipartite
graphs.}{\small\par}
\end{figure}

\end{landscape}

\pagebreak{}

\subsection{Insurer-provider and referral networks for healthcare}

Many features of the health care market naturally map into graphs.
For example, physicians may have admitting privileges across multiple
hospitals, insurers typically offer preferential terms to selected
networks of providers, and doctors vary in the intensity with which
they refer patients to one another.\footnote{\citet{Barnett_et_al_HSR11} and \citet{An_et_al_SIM18} use patient
referral patterns to map out relationships among physicians.}

The welfare and economic implications of these networks are likely
immense, given the magnitude of the health care sector in the United
States economy. \citet{Ho_AER09} represents one attempt to grapple
with the network structure of healthcare markets. 

\subsection{Employment search}

\citet{Loury_Ioannides_JEL04} survey the substantial literature on
the interplay between social networks and job acquisition, a topic
that has fascinated both sociologists and economists at least since
\citet{Granovetter_AJS73}. The growing availability of longitudinal
register data from various countries provides an opportunity to study
the interface between networks and inequality in the labor market
more carefully. 

For example, \citet{Saygin_et_al_IZA14} use the Austrian Social Security
Database to construct a co-worker network for middle aged workers
in Austria. A co-worker is anyone who an individual has ever worked
with previously. They find that the structure of these co-worker networks
predict the ease with which workers find employment after establishment
closures (i.e., mass layoffs). This paper provides a nice example
of how new data may facilitate the re-visiting of a classic networks
question \citep[cf.,][]{Hensvik_Skans_JOLE16}. 

\subsection{Questions}

The examples outlined above represent only a small sample of recent
appearances of network data in empirical economic research.\footnote{\citet{dePaula_WC17} and \citet{Jackson_et_al_JEL17} provide additional
references.} What do we hope to learn from this growing body of research? As noted
in the introduction, empirical research on networks can usefully be
divided between that which studies the \emph{consequences} of networks
and that which studies their \emph{formation}. The premise of this
chapter is that network linkages across agents are consequential.
That is, I take as given that networks are important venues for shock
propagation, information diffusion, learning and various types of
peer interactions. Maintaining this premise justifies my focus on
the econometric modeling of network formation.

An analogy with the development of single agent models of discrete
choice is useful. \citet{McFadden_FinE74}, in a pioneering paper,
initiated a research program on identifying and estimating random
utility models of discrete choice. Empirical application, computation,
semiparametric identification and estimation, the inclusion of unobserved
choice attributes, and allowing for strategic behavior, all have been
important accomplishments of this research program. These econometric
models are, in turn, routinely used in virtually all areas of economics. 

The goal here is analogous. Relational data are ubiquitous in economics,
but econometric models for such data are not. The goal, therefore,
is to develop models for these data, preferably with (i) strong microeconomic
foundations, (ii) that allow for unobserved agent-level heterogeneity,
and (iii) incorporate interdependencies in preferences over links.
Also required are feasible methods of estimation and inference (and
in this area interesting and challenging questions are abundant).
The availability of econometric methods for network analysis will,
in turn, allow for counterfactual policy and welfare analysis. How
would a particular horizontal merger affect upstream supply chain
structure? What is the effect on trade flows of Eurozone membership?
Could a school principal increase friendships across races, or raise
average achievement, by structuring classrooms under her purview differently?

Some readers may wish to skip Section \ref{sec: Basic-probability-tools}
initially and instead start with Sections \ref{sec: dyadic_regression}
to \ref{sec: heterogeneity}. They could then return to Section \ref{sec: Basic-probability-tools}
before tackling Sections \ref{sec: Statistics} and \ref{sec: Strategic-models}.
Graph theoretic concepts and notation appears throughout the chapter.
While many terms and definitions are formally defined, others are
not. Missing definitions can be found in any basic graph theory textbook.

\section{\label{sec: Basic-probability-tools}Basic probability tools: random
graphs, graphons, graph limits and sampling}

This section provides an informal introduction to key ideas from the
applied probability literature on exchangeable random graphs. The
main concepts are (i) exchangeable random graphs and their representation,
(ii) subgraph densities or network moments, (iii) limits of sequences
of exchangeable random graphs, and (iv) sampling. These ideas underlie
a substantial share of recent research on the statistics of networks
\citep[e.g., ][]{Airoldi_et_al_JMLR08,Diaconis_Holmes_Janson_IM08,Bickel_Chen_PNAS09,Bickel_et_al_AS11,Bhamidi_et_al_AAP11,Chatterjee_et_al_AAP11,Olhede_Wolfe_PNAS14,Orbanz_Roy_IEEE15,Gao_et_al_AS15}. 

Much of this statistics work has been motivated by research questions
in computational biology and neuroscience \citep[e.g., ][]{Picard_et_al_JCB08}.
Link formation in these settings is not driven by purposeful agents.
Consequently this research may initially appear rather distant from
the concerns of econometricians. Nevertheless my view is that recent
developments in probability and statistics have much to offer econometricians
interested in networks (and also vice-versa, although making this
second argument this is not on my agenda here). 

The basic concepts introduced in this section appear frequently in
later portions of the chapter.

\subsection{Notation}

Let $G\left(\mathcal{V},\mathcal{E}\right)$ be a finite undirected
network or graph defined on $N=\left|\mathcal{V}\left(G\right)\right|$
\emph{vertices} or agents; here $\mathcal{V}\left(G\right)=\left\{ 1,\ldots,N\right\} $
denotes the set of all agents in the network.\footnote{If $\mathbb{X}$ is a set, then $\left|\mathbb{X}\right|$ denotes
the cardinality of that set. If $\mathbf{X}$ is a matrix of reals,
then $\left|\mathbf{X}\right|$ equals its (element-wise) absolute
value.} Any two agents may be connected or not. The set of such links is
recorded in the \emph{edge} list $\mathcal{E}\left(G\right)=\left\{ \left(i,j\right),\left(k,l\right),\ldots\right\} $,
consisting of the (unordered) indices of all connected agent pairs.
Call $N$ the \emph{order} of the network and $\left|\mathcal{E}\left(G\right)\right|$
its \emph{size. }We can represent $G\left(\mathcal{V},\mathcal{E}\right)$
by the $N\times N$ adjacency matrix $\mathbf{D}=\left[D_{ij}\right]_{i,j\in\mathcal{V}\left(G\right)}$
with $ij^{th}$ element
\[
D_{ij}=\left\{ \begin{array}{cl}
1, & \left(i,j\right)\in\mathcal{E}\left(G\right)\\
0, & \text{otherwise}
\end{array}\right..
\]

For an undirected network, with self-ties or loops ruled out, such
that $D_{ii}=0$ for $i\in\mathcal{V}\left(G\right)$, $\mathbf{D}$
is a symmetric binary matrix with a diagonal of structural zeros.
I focus on undirected networks initially, but also present some results
for directed networks and bipartite networks. Specific notation for
these special cases will be introduced as needed.

In settings where it is useful to emphasize the order of $G$, I use
the notation $G_{N}$. This is especially useful when considering
sequences of graphs. Let $\left(i,j\right)\in\mathcal{E}\left(G\right)$
be an edge in $G$; sometimes I will abbreviate $\left(i,j\right)$
as $ij$. The complete graph on $p$ vertices is denoted by $K_{p}$.

Following \citet{Jackson_NetBook08}, let $G-ij$ denote the network
obtained by deleting edge $ij$ from $G$ (if present), and $G+ij$
the network one gets after adding this link. Let $\mathbf{D}\pm ij$
denote the adjacency matrix associated with the network obtained by
adding/deleting edge $\left(i,j\right)$ from $G$. Let $\mathbb{D}_{N}$
denote the set of all $2^{\binom{N}{2}}$ possible adjacency matrices
and $\mathbb{I}_{N}$ the set of all possible $N$-dimensional binary
vectors.

Let $N\left(i\right)=\left\{ j\in\mathcal{V}\thinspace:\thinspace ij\in\mathcal{E}\right\} $
be the set of agent $i$'s \emph{neighbors: }agents to which she is
directly linked. The \emph{degree} of agent $i$ is given by the cardinality
of this set. Equivalently agent $i$'s degree may be computed by summing
the elements of the $i^{th}$ row of the adjacency matrix. Let $\iota_{N}$
be an $N\times1$ vector of ones. The vector $\mathbf{D}_{+}=\mathbf{D}\iota_{N}$
is called the \emph{degree sequence} of the network (typically we
re-arrange the order of agents such that the elements of this vector
are in ascending order).

I informally call a network dense\emph{ }if its size, or number of
edges, is ``close to'' $N^{2}$ and sparse if its size is ``close
to'' $N$. More precisely a sequence of graphs is \emph{sparse} in
the limit if the number of edges in it grows linearly with $N$, \emph{dense}
if this growth is quadratic. 

There are $n\overset{def}{\equiv}\tbinom{N}{2}=\frac{1}{2}N\left(N-1\right)$
pairs of agents, or \emph{dyads}, in a network consisting of $N$
agents. Triples, quadruples and quintuples of agents are call \emph{triads},
\emph{tetrads} and \emph{pentads} respectively. A tuple of 17 agents,
which arises rather rarely in everyday empirical work, is evidently
called a \emph{septendecuple}. Not having formally studied Latin,
I offer the reader no guidance on pronunciation.

Let $\sum_{i<j}$ be shorthand for $\sum_{i=1}^{N-1}\sum_{j=i+1}^{N}$
with $\sum_{i<j<k}$ similarly defined. The \emph{density} of a network,
\[
P_{N}\left(\vcenter{\hbox{\includegraphics[scale=0.1250]{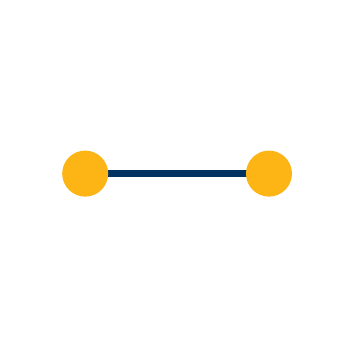}}}\right)\overset{def}{\equiv}\hat{\rho}_{N}\overset{def}{\equiv}\frac{2}{N\left(N-1\right)}\sum_{i<j}D_{ij},
\]
equals the proportion of connected dyads. Let $D_{i+}$ be the $i^{th}$
element of the degree sequence. \emph{Average degree,} 
\[
\hat{\lambda}_{N}\overset{def}{\equiv}\left(N-1\right)\hat{\rho}_{N}\overset{def}{\equiv}\frac{1}{N}\sum_{i=1}^{N}D_{i+},
\]
equals the average number of links per agent in the network.

In what follows random variables are (generally) denoted by capital
Roman letters, specific realizations by lower case Roman letters and
their support by blackboard bold Roman letters. That is $Y$, $y$
and $\mathbb{Y}$ respectively denote a generic random draw of, a
specific value of, and the support of, $Y$. The abbreviations i.i.d.,
CLT , LLN and GGP stand for, respectively, ``independent and identically
distributed'', ``central limit theorem'', ``law of large numbers''
and ``graph generating process''. For the vector $\mathbf{b}$,
$\left\Vert \mathbf{b}\right\Vert _{2}$ denotes the Euclidean norm;
for the matrix $\mathbf{B}$, $\left\Vert \mathbf{B}\right\Vert _{F}$
denotes the Frobenius norm. I use $I_{N}$ to denote the $N\times N$
identity matrix. $\mathbb{N}$ denotes the set of natural numbers
and $\left[Y_{ij}\right]_{i,j\in\mathbb{N}}$ an infinite two-dimensional
array with $ij^{th}$ element $Y_{ij}$.

I use the big-Omega notation $X_{N}=\Omega\left(Y_{N}\right)$ to
denote that $X_{N}=O\left(Y_{N}\right)$ \emph{and} $Y_{N}=O\left(X_{N}\right)$.
The notation $\overset{D}{=}$ denotes equality in distribution, $\overset{def}{\equiv}$
a mathematical definition. Let $\theta$ be some parameter value in
the space $\Theta$. Let $S_{N}\left(\theta\right)$ be some statistic
indexed by this parameter with population value $\theta_{0}$. I let
$S_{N}=S_{N}\left(\theta_{0}\right)$ denote the statistic evaluated
at $\theta=\theta_{0}$. To economize on space I sometimes abbreviate
$\Pr\left(\left.Y=y\right|X=x\right)$ as $\Pr\left(\left.Y=y\right|x\right)$
or $\Pr\left(\left.y\right|x\right)$ and similarly for $\mathbb{E}\left[\left.Y\right|x\right]$,
$\mathbb{V}\left(\left.Y\right|x\right)$ etc.

\subsection{Exchangeable random graphs}

Initially assume the unavailability of agent-specific covariates,
making it natural to assume that agents are exchangeable (models with
covariates, and a correspondingly weaker notion of exchangeability,
feature in Sections \ref{sec: dyadic_regression}, \ref{sec: policy_analysis},
\ref{sec: heterogeneity} and \ref{sec: Strategic-models}). Let $\pi:\left\{ 1,\ldots,N\right\} \mapsto\left\{ 1,\ldots,N\right\} $
be a permutation of the node labels of $G\left(\mathcal{V},\mathcal{E}\right)$
and $\Pi$ the set of all such permutations. The random graph $G$
is \emph{jointly exchangeable} if
\begin{equation}
\left[D_{ij}\right]\overset{D}{=}\left[D_{\pi\left(i\right)\pi\left(j\right)}\right]\label{eq: finite_joint_exchangeability}
\end{equation}
for every permutation $\pi\in\Pi$.

In settings where node labels have no meaning, exchangeability is
an implication of \emph{a priori} researcher belief (and hence a natural
modeling assumption). Consider a researcher analyzing the adjacency
matrix associated with a set of friendship links among adolescents
in a high school \citep[e.g.,][]{Currarini_et_al_EM09}, in the absence
of node-specific covariates, there is no reason to change one's modeling
approach after simultaneously applying a particular reshuffling of
agents to \emph{both} the rows and columns of $\mathbf{D}$ \citep[cf.,][]{Rubin_JES81}.
Put differently, when node labels have no meaning, the probability
attached to any isomorphism of $G$ should be the same as that attached
to $G$ itself.

There are many interesting statistics of $\mathbf{D}$ which are invariant
to simultaneous row and column permutations. Examples include a network's
density, diameter and triangle $(\vcenter{\hbox{\includegraphics[scale=0.125]{triangle}}})$
count. A family of such statistics, network moments, is introduced
below. Exchangeability suggests that a statistical model should attach
different probabilities to networks with different values of such
(permutation invariant) statistics, but the same probability to two
networks which are isomorphic (which will share common values of any
permutation invariant statistic).

\subsubsection*{An exchangeable model with strategic interaction}

Most extant models of network formation satisfy condition (\ref{eq: finite_joint_exchangeability}).
As an example, which will help to fix some ideas, consider the model
of strategic network formation with bilateral transfers studied by
\citet{Graham_Pelican_BookCh2020}. Let $\nu_{i}\thinspace:\thinspace\mathbb{D}_{N}\rightarrow\mathbb{R}$
be a utility function for agent $i$, which maps networks into utility.
Define the marginal utility of edge $ij$ for agent $i$ as
\begin{equation}
MU_{ij}\left(\mathbf{D}\right)=\left\{ \begin{array}{cc}
\nu_{i}\left(\mathbf{D}\right)-\nu_{i}\left(\mathbf{D}-ij\right) & \text{if}\thinspace D_{ij}=1\\
\nu_{i}\left(\mathbf{D}+ij\right)-\nu_{i}\left(\mathbf{D}\right) & \text{if}\thinspace D_{ij}=0
\end{array}\right..\label{eq: marginal_utility}
\end{equation}

From \citet{Bloch_Jackson_IJGT06}, a network is \emph{pairwise stable
with transfers} if the following condition holds.
\begin{defn}
\label{def: Pairwise-stability}(\textsc{Pairwise stability with Transfers)
}The network $G\left(\mathcal{V},\mathcal{E}\right)$ is pairwise
stable with transfers if \\
(i) $\forall\left(i,j\right)\in\mathcal{E}\left(G\right),\thinspace MU_{ij}\left(\mathbf{D}\right)+MU_{ji}\left(\mathbf{D}\right)\geq0$\\
(ii) $\forall\left(i,j\right)\notin\mathcal{E}\left(G\right),\thinspace MU_{ij}\left(\mathbf{D}\right)+MU_{ji}\left(\mathbf{D}\right)<0$
\end{defn}
If the network in hand is a pairwise stable one, then any links actually
present generate (weakly) positive utility (on net for the two agents
on each side of a link). Unobserved links, in contrast, would not
generate net positive utility if present.

\citet{Graham_Pelican_BookCh2020} focus on a general family of parametric
utility functions which includes, among others, the specification
\begin{equation}
\nu_{i}\left(\mathbf{\left.d\right|}\mathbf{A},\mathbf{B},\mathbf{V^{*}};\gamma_{0}\right)=\sum_{j}d_{ij}\left[A_{i}+B_{j}+\gamma_{0}\left(\sum_{k}d_{ik}d_{jk}\right)-V_{ij}^{*}\right]\label{eq: utility_function}
\end{equation}
with $\mathbf{V^{*}}=\left[V_{ij}^{*}\right]$ , $\mathbf{A}=\left[A_{i}\right]$
and $\mathbf{B}=\left[B_{i}\right]$. Under (\ref{eq: utility_function}),
assuming $\gamma_{0}>0$, dyad $\left\{ i,j\right\} $ will generate
more utility when forming a link if they already share many links
or ``friends'' in common (i.e., if $\sum_{k}d_{ik}d_{jk}$ is large).
Here $A_{i}$ and $B_{j}$ are agent-specific ``extroversion'' and
``popularity'' parameters, the effect of which is to generate degree
heterogeneity \citep[cf.,][]{Graham_EM17}. The term $V_{ij}^{*}$
is an idiosyncratic dyad-specific utility shifter. \citet{Graham_Pelican_BookCh2020}
leave the joint distribution of $\mathbf{A}$ and $\mathbf{B}$ unrestricted,
but here I will assume that $\left\{ \left(A_{i},B_{i}\right)\right\} _{i=1}^{N}$
is an i.i.d. sequence which is independent of $\left\{ \left(V_{ij}^{*},V_{ji}^{*}\right)\right\} _{i,j\in\left\{ 1,\ldots,N\right\} ,i<j}$,
also assumed i.i.d.

When the utility function is of the form given in (\ref{eq: utility_function})
the marginal utility agent $i$ gets from a link with $j$ is
\[
MU_{ij}\left(\left.\mathbf{d}\right|\mathbf{A},\mathbf{B},\mathbf{V^{*}};\gamma_{0}\right)=A_{i}+B_{j}+\gamma_{0}\left(\sum_{k}d_{ik}d_{jk}\right)-V_{ij}^{*}.
\]
Pairwise stability then implies, conditional on the realizations of
$\mathbf{A}$, $\mathbf{B}$, $\mathbf{V^{*}},$ and the value of
externality parameter, $\gamma_{0}$, that the observed network must
satisfy, for $i=1,\ldots,N-1$ and $j=i+1,\ldots,N$
\begin{equation}
D_{ij}=\mathbf{1}\left(U_{i}+U_{j}+2\gamma_{0}\left(\sum_{k}D_{ik}D_{jk}\right)\geq V_{ij}\right)\label{eq: link_rule_incomplete}
\end{equation}
with $U_{i}=A_{i}+B_{i}$ and $V_{ij}=V_{ij}^{*}+V_{ji}^{*}$. Equation
(\ref{eq: link_rule_incomplete}) defines a system of $\tbinom{N}{2}=\frac{1}{2}N\left(N-1\right)$
nonlinear simultaneous equations. Any solution to this system --
and there will typically be multiple ones -- constitutes a pairwise
stable (with transfers) network.\footnote{Note that in this example existence of an equilibrium is easy to show
using Tarski's \citeyearpar{Tarski_PJM55} fixed point theorem.}

As written, model (\ref{eq: link_rule_incomplete}) is incomplete
\citep[cf.,][]{dePaula_ARE13}. Even if we assume that the observed
network is a pairwise stable one, we have not specified a mechanism
for selecting, when there are multiple ones, a specific equilibrium
configuration. To complete the model, following the more careful development
in \citet{Pelican_Graham_WP2019}, let $\mathcal{N}_{\mathbf{d}}\left(\mathbf{V};\mathbf{U},\gamma\right)$
equal the probability that configuration $\mathbf{D}=\mathbf{d}$
is selected. If $\mathbf{d}$ is not an equilibrium -- given $\mathbf{U}$,
$\mathbf{V}$ and $\gamma$ -- then $\mathcal{N}_{\mathbf{d}}\left(\mathbf{V};\mathbf{U},\gamma\right)=0$.
If $\mathbf{d}$ is the unique equilibrium then $\mathcal{N}_{\mathbf{d}}\left(\mathbf{V};\mathbf{U},\gamma\right)=1$.
If $\mathbf{d}$ is one of several equilibria, then $0\leq\mathcal{N}_{\mathbf{d}}\left(\mathbf{V};\mathbf{U},\gamma\right)\leq1$
etc.

For $\mathbb{D}_{N}$ the net of all $N\times N$ undirected adjacency
matrices, we have that $\sum_{\mathbf{d}\in\mathbb{D}_{N}}\mathcal{N}_{\mathbf{d}}\left(\mathbf{V};\mathbf{U},\gamma\right)=1$.
The conditional likelihood of observing network wiring $\mathbf{D}=\mathbf{d}$
is therefore
\[
\Pr\left(\left.\mathbf{D}=\mathbf{d}\right|\mathbf{U};\gamma\right)=\int_{\mathbf{v}\in\mathbb{R}^{n}}\mathcal{N}_{\mathbf{d}}\left(\mathbf{v};\mathbf{U},\gamma\right)f_{\mathbf{V}}\left(\mathbf{v}\right)\mathrm{d}\mathbf{v}.
\]
The $\tbinom{N}{2}$ equilibrium conditions (\ref{eq: link_rule_incomplete})
indicate that if $\mathbf{d}=\left[d_{ij}\right]$ is an equilibrium,
then so is $\mathbf{d}_{\pi}\overset{def}{\equiv}\left[d_{\pi\left(i\right)\pi\left(j\right)}\right]$.
Hence as long as the equilibrium selection mechanism is also invariant
to index permutations, as is natural to require, condition (\ref{eq: finite_joint_exchangeability})
holds.

Under the null of no strategic interaction, $\gamma=0$, the likelihood
simplifies to
\begin{equation}
\Pr\left(\left.\mathbf{D}=\mathbf{d}\right|\mathbf{U};0\right)=\int_{\mathbf{v}\in\mathbb{R}^{n}}\mathcal{N}_{\mathbf{d}}\left(\mathbf{v};\mathbf{U},0\right)f_{\mathbf{V}}\left(\mathbf{v}\right)\mathrm{d}\mathbf{v}\label{eq: graham_pelican_null_likelihood}
\end{equation}
with
\begin{align*}
\mathcal{N}_{\mathbf{d}}\left(\mathbf{v};\mathbf{U},0\right)= & \prod_{i=1}^{N-1}\prod_{j=i+1}^{N}\mathbf{1}\left(U_{i}+U_{j}\geq v_{ij}\right)^{d_{ij}}\\
 & \times\mathbf{1}\left(U_{i}+U_{j}<v_{ij}\right)^{1-d_{ij}}.
\end{align*}
Since $\left\{ \left(V_{ij}\right)\right\} _{i,j\in\left\{ 1,\ldots,N\right\} ,i<j}$
is i.i.d., if we further assume that $f_{V_{12}}\left(v\right)=e^{v}/\left[1+e^{v}\right]^{2}$,
the logistic density, explicitly evaluating the integral in (\ref{eq: graham_pelican_null_likelihood})
yields
\begin{equation}
\Pr\left(\left.\mathbf{D}=\mathbf{d}\right|\mathbf{U};0\right)=\prod_{i=1}^{N-1}\prod_{j=i+1}^{N}\left[\frac{\exp\left(U_{i}+U_{j}\right)}{1+\exp\left(U_{i}+U_{j}\right)}\right]^{d_{ij}}\left[\frac{1}{1+\exp\left(U_{i}+U_{j}\right)}\right]^{1-d_{ij}},\label{eq: beta_model_likelihood}
\end{equation}
which is the likelihood associated with the so-called $\beta$-model
of \citet{Frank_MSH97} and \citet{Chatterjee_et_al_AAP11}.

A feature of the $\beta$-model is that links form independently\emph{
conditional} on the latent agent-specific effects $\left\{ U_{i}\right\} _{i=1}^{N}$.
Equation (\ref{eq: beta_model_likelihood}) consists of a product
of $\tbinom{N}{2}$ conditionally independent likelihood contributions.

Evidently, this conditional independence structure is not typically
a feature of the model when $\gamma>0$, such that strategic interaction
is present. To see why by means of a simple example, consider a network
consisting of just three homogenous agents (i.e., $U_{1}=U_{2}=U_{3}=0$).
Initially assume that both $V_{12}$ and $V_{13}$ are less then zero,
but that $0<V_{23}\leq2\gamma_{0}$. This corresponds to edges $\left(1,2\right)$
and $\left(1,3\right)$ generating so much intrinsic utility that
they will form irrespective of what other edges may or may not be
present in the network. In contrast, the intrinsic utility attached
to edge $\left(2,3\right)$ falls in an intermediate range: the edge
forms if edges $\left(1,2\right)$ and $\left(1,3\right)$ are present
-- such that agents $2$ and $3$ share agent $1$ as a friend in
common -- and does not form if they are absent. This configuration
of utility shocks is depicted in the left-hand panel of Figure \ref{fig: CID_counter_example}.
The unique equilibrium outcome in this case is a triangle $\left(\vcenter{\hbox{\includegraphics[scale=0.125]{triangle}}}\right)$
network.

If, instead, $V_{12}$ and $V_{13}$ are both greater than $2\gamma$,
such that the $\left(1,2\right)$ and $\left(1,3\right)$ edges never
form because of their low intrinsic utility (again irrespective of
what other edges may or may not be present in the network), then the
$\left(2,3\right)$ edge will not form either. This scenario is depicted
in the right-hand panel of Figure \ref{fig: CID_counter_example}.
The unique equilibrium outcome in this case is an empty $\left(\vcenter{\hbox{\includegraphics[scale=0.125]{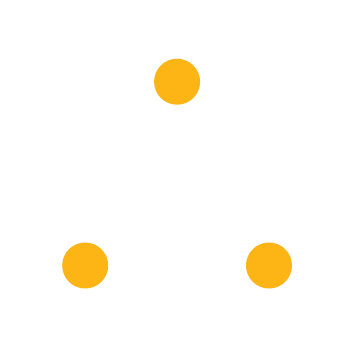}}}\right)$
network.

This simple example shows that $D_{23}$ need not vary independently
of $D_{12}$ and $D_{13}$ conditional on $\left(U_{1},U_{2},U_{3}\right)$
in the presence of strategic interaction ($\gamma>0$) . Such conditional
independence \emph{is} a feature of the $\beta$-model ($\gamma=0$).
While the model is exchangeable both when $\gamma>0$ and when $\gamma=0$,
the conditional independence of edges only obtains under the no strategic
interaction null.

\begin{figure}
\caption{Dependent link formation\label{fig: CID_counter_example}}

\begin{centering}
\includegraphics{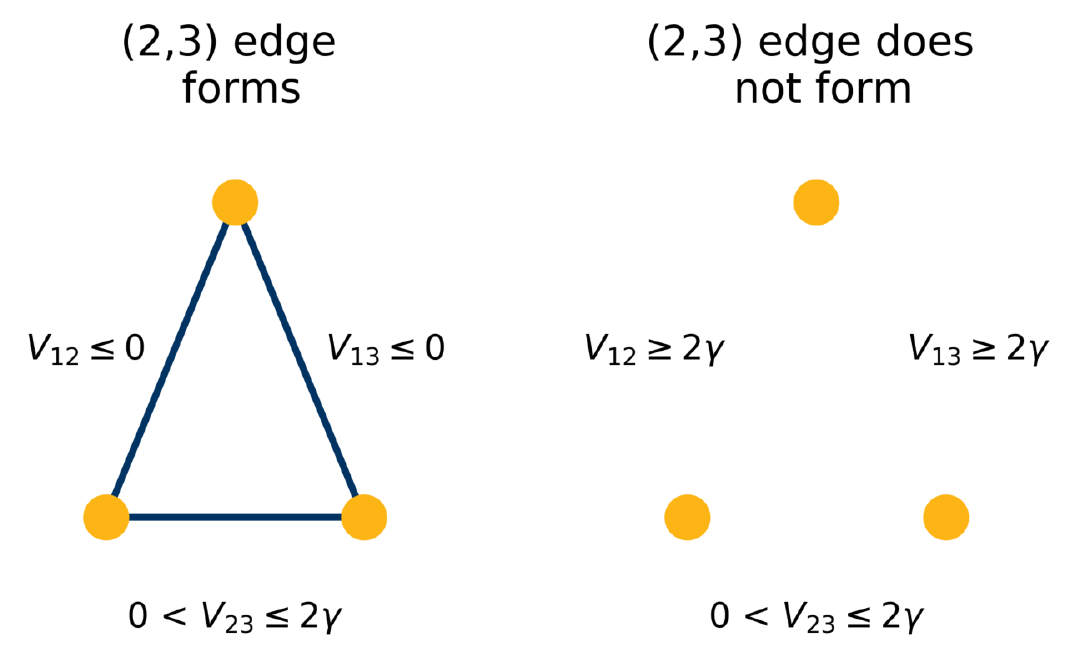}
\par\end{centering}
\uline{Notes:} Both panels depict the unique pairwise stable equilibrium
associated with the shown triple of dyad-level utility shifters $V_{12}$,
$V_{13}$ and $V_{23}$ and agent-level heterogeneity parameters $U_{1}$,
$U_{2}$ and $U_{3}$ identically equal to zero. In both panels the
realized value of $V_{23}$ is the same, but whether $D_{23}=1$ or
$0$ varies with the realized values of $V_{12}$ and $V_{13}$. If
$V_{12}$ and $V_{13}$ are sufficiently low, then $D_{23}=1$; if
they are sufficiently high, then $D_{23}=0$. Links are not conditionally
independent given $\left\{ U_{i}\right\} _{i=1,2,3}$.
\end{figure}

\subsection{Conditionally independent dyad (CID) models and the graphon}

Having established that a network probability model should satisfy
the joint exchangeability condition (\ref{eq: finite_joint_exchangeability}),
it is important to articulate classes of models that do so. One such
family of models, suggested by the last example, are conditionally
independent dyad (CID) models \citep{Chandrasekhar_Book15,Shalizi_LN16}.
In these models each agent is characterized by an unobserved latent
attribute, $U_{i}$. The $N$ agents in the network in hand are viewed
as independent random draws from some population, such that the $\left\{ U_{i}\right\} _{i=1}^{N}$
are independently and identically distributed. Conditional on the
agent-specific latent variables $\mathbf{U}=\left(U_{1},\ldots,U_{N}\right)'$
edges form independently with
\[
\left.D_{ij}\right|U_{i},U_{j}\sim\mathrm{Bernoulli}\left(h\left(U_{i},U_{j}\right)\right),
\]
for every dyad $\left\{ i,j\right\} $ with $i<j$. Here $h\left(u,v\right)=h\left(v,u\right)$
for all $\left(u,v\right)\in\mathbb{U}\times\mathbb{U}$ is a symmetric
edge probability function. In anticipation of results to come, call
this function a \emph{graphon}: short for \textbf{\uline{graph}}
functi\textbf{\uline{on}}.

Conditional on the latent agent-specific effects the likelihood of
the network is
\[
\Pr\left(\left.\mathbf{D}=\mathbf{d}\right|\mathbf{U}=\mathbf{u}\right)=\prod_{i<j}h\left(u_{i},u_{j}\right)^{d_{ij}}\left[1-h\left(u_{i,}u_{j}\right)\right]^{1-d_{ij}}.
\]
Unconditional on $\mathbf{U}$, the likelihood equals
\begin{equation}
\Pr\left(\mathbf{D}=\mathbf{d}\right)=\int\cdots\int\left\{ \prod_{i<j}h\left(u_{i},u_{j}\right)^{d_{ij}}\left[1-h\left(u_{i,}u_{j}\right)\right]^{1-d_{ij}}\right\} \prod_{i=1}^{N}f_{U}\left(u_{i}\right)\mathrm{d}u_{i},\label{eq: CEI_integrated}
\end{equation}
where $f_{U}\left(u\right)$ is the density of $U$. Importantly (\ref{eq: CEI_integrated})
allows for dependence across dyads which share agents in common. Independence
holds only conditional on the latent agent attributes \citep{Graham_EM17}.
Similar independence restrictions play a prominent role in the econometrics
of panel data \citep{Chamberlain_HBE84,Arellano_Honore_HBE01}.

It is an easy exercise to show that (\ref{eq: CEI_integrated}) is
compatible with the finite joint exchangeability restriction (\ref{eq: finite_joint_exchangeability}).

The $\beta$-model, introduced above, belongs to the family of CID
models with a graphon of
\[
h\left(u,v\right)=\frac{\exp\left(u+v\right)}{1+\exp\left(u+v\right)}.
\]
Random threshold graphs \citep[e.g., ][]{Diaconis_Holmes_Janson_IM08}
are also members of this family with graphon
\[
h\left(u,v\right)=\mathbf{1}\left(F_{U}\left(u\right)+F_{U}\left(v\right)\geq\alpha\right),
\]
and $F_{U}\left(u\right)$ the CDF of $U$.

It is important to realize that CID models constitute only a subset
of all jointly exchangeable random graph models when $N$ -- the
number of agents in the network -- is finite. As shown by means of
the example introduced above, strategic interaction in link formation
can induce dependence across elements of the adjacency matrix that
evidently cannot be eliminated by conditioning (see Figure \ref{fig: CID_counter_example}
above). Although not all exchangeable models are CID ones, this family
of models plays an outsized role in extant large sample theory for
networks.

\subsection{Aldous-Hoover representation theorem and the graphon}

Joint exchangeability imposes more structure on the network probability
distribution when there are an\emph{ infinite} number of agents. Specifically,
if we strengthen (\ref{eq: finite_joint_exchangeability}) to hold
for any permutation of a finite number of the indices of the infinite
sequence $\mathbb{N}=\left\{ 1,2,3,\ldots\right\} $, we have a generalization
of \citet{deFinetti_AN1931} type exchangeability of an infinite sequence,
appropriate for infinite random graphs. In independent work \citet{Aldous_JMA81}
and \citet{Hoover_WP79} showed the following representation result
for infinite random adjacency matrices \citep[cf., ][]{Kallenberg_PSIP05}.
\begin{thm}
\textsc{\label{thm: Aldous-Hoover}(Aldous-Hoover)} A random adjacency
matrix $\left[D_{ij}\right]_{i,j\in\mathbb{N}}$ is jointly exchangeable
if and only if there is a measurable function $g:\left[0,1\right]^{4}\rightarrow\left\{ 0,1\right\} $
such that
\[
\left[D_{ij}\right]\overset{D}{=}\left[g\left(\alpha,U_{i},U_{j},V_{ij}\right)\right]
\]
for $\alpha$, $\left\{ U_{i}\right\} _{i\in\mathbb{N}}$ , and $\left\{ V_{ij}\right\} _{i,j\in\mathbb{N},i<j}$
independently and identically distributed $\mathcal{U}\left[0,1\right]$
random variables with $V_{ij}=V_{ji}$.
\end{thm}
Here $\alpha$ is a mixing parameter, analogous to the one appearing
in de Finetti's \citeyearpar{deFinetti_AN1931} classic representation
theorem for exchangeable binary sequences.\footnote{To make the connection with \citet{deFinetti_AN1931} transparent
\citet[Lemma 1.5]{Aldous_JMA81} also shows that an infinite sequence
$\left\{ Y_{i}\right\} _{i=1}^{\infty}$ is exchangeable if and only
if there exists a measurable function $f$ such that $\left[Y_{i}\right]\overset{D}{=}\left[f\left(\alpha,U_{i}\right)\right]$.} Theorem \ref{thm: Aldous-Hoover} implies that if network agents
are exchangeable for all $N$, then we can proceed `as if' edges formed
according to a CID model or a mixture of such models.

Exploiting the fact that the elements of $\mathbf{D}$ are binary,
we can simplify Theorem \ref{thm: Aldous-Hoover} as follows. Averaging
over $V_{ij}$ yields
\begin{align*}
h\left(\alpha,u_{i},u_{j}\right) & \overset{def}{\equiv}\int_{0}^{1}g\left(\alpha,u_{i},u_{j},v\right)\mathrm{d}v
\end{align*}
from which we get the more convenient representation, for $i<j$,
\begin{equation}
\left[D_{ij}\right]\overset{D}{=}\left[\mathbf{1}\left(V_{ij}\leq h\left(\alpha,U_{i},U_{j}\right)\right)\right].\label{eq: Aldous-Hoover-DGP}
\end{equation}

This is, of course, just a conditional edge independence model (or,
more precisely, a mixture of such models). In what follows I focus
on inference which conditions on the empirical distribution of the
data; consequently $\alpha$ can often safely be ignored. When this
is the case I suppress the $\alpha$ argument in the graphon, writing
$h\left(U_{i},U_{j}\right)$. See \citet{Bickel_Chen_PNAS09} and
\citet{Menzel_arXiv17} for additional discussion.

Theorem \ref{thm: Aldous-Hoover} motivates an approach to nonparametric
modeling of\emph{ large} networks that proceeds `as if' links form
independently conditional on the agent-specific latent variables $\mathbf{U}=\left(U_{1},\ldots,U_{N}\right)'$.
This is convenient because CID models induce a very particular dependence
structure across the rows and columns of the network adjacency matrix.

Consider, without loss of generality, agents $1$, $2$ and $3$.
In a CID model $D_{12}$ and $D_{13}$ may covary; the dyads $\left\{ 1,2\right\} $
and $\left\{ 1,3\right\} $ share the agent $1$ in common and hence
both links form, in part, based on the value of $U_{1}$. However
$D_{12}$ and $D_{13}$ vary independently conditional on $U_{1}$,
$U_{2}$ and $U_{3}$ (hence the conditionally independent dyad nomenclature).
Links involving pairs of dyads which share no agents in common, for
example $D_{12}$ and $D_{34}$, form independently.

The structured pattern of dependence, independence and conditional
independence associated with CID models facilitates the development
of LLNs and CLTs that can be applied to statistics of the adjacency
matrix. A group of statistics for which some large network distribution
theory is available are network moments.

\subsection{Network moments\label{subsec: Network-moments}}

Almost fifty years ago \citet{Holland_Leinhardt_AJS70} suggested
that a network's architecture could be usefully summarized by its
average local structure. Agent exchangeability, in conjunction with
Theorem \ref{thm: Aldous-Hoover}, also motivates an approach to network
modeling based on the frequency of low order subgraph configurations
(i.e., the number of edges, two stars, triangles, squares, k-stars
etc). 

Consider, for example, the set of all $\tbinom{N}{3}$ \emph{triads}
-- unordered triples of agents -- in a network; what fraction of
these triads take two-star $\vcenter{\hbox{\includegraphics[scale=0.125]{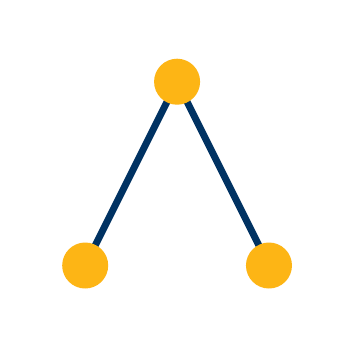}}}$
or triangle $\vcenter{\hbox{\includegraphics[scale=0.125]{triangle}}}$
configurations? These frequencies, called \emph{network moments} by
\citet{Bickel_et_al_AS11}, feature prominently in research by sociologists
\citep*[e.g.,][]{Granovetter_AJS73,Coleman_AJS88,Gould_Fernandez_SM89}
and computational biologists \citep*[e.g.,][]{Milo_el_al_Sci02,Przulj_et_al_BI04};
albeit in the context of two largely independent and desynchronized
literatures.

In economics, network moments play an increasingly important role
in empirical research as well. Examples include \citet{Jackson_et_al_AER12},
who explore, theoretically and empirically, how different triad configurations
can support infrequent favor exchange between agents; \citet{Atalay_et_al_PNAS11},
who calibrate a model of buyer-seller networks to the US economy by
modeling its degree distribution\footnote{Below I show that network moments and moments of the degree distribution
are closely connected.}; and \citet*{dePaula_et_al_EM18}, who present conditions under which
(a variant of) network moments (partially) identify preferences in
a structural model of strategic network formation.

Network moments, in addition to being important summary statistics
for graphs, play an important role in (i) the distribution theory
for dyadic regression discussed in Sections \ref{sec: dyadic_regression}
and \ref{sec: policy_analysis}, (ii) understanding the degree distribution
and (iii) structural model estimation. The material which follows
is dense.

\subsubsection*{Subgraphs and isomorphisms}

The exact sense in which a network is summarized by its moments can
be made precise using the graphon, as introduced above, and the notion
of a graph limit, which will be introduced below \citep{Diaconis_Janson_RM08,Lovasz_AMS12}.
First we require a formal definition of a subgraph. There are two
definitions used by empirical network researchers.
\begin{defn}
\label{def: Partial-Subgraph}(\textsc{\uline{Partial Subgraph}})
Let $\mathcal{V}\left(S\right)\subseteq\mathcal{V}\left(G\right)$
be any subset of the vertices of $G$ and $\mathcal{E}\left(S\right)\subseteq\mathcal{E}\left(G\right)\cap\mathcal{V}\left(S\right)\times\mathcal{V}\left(S\right)$,
then $S=\left(\mathcal{V}\left(S\right),\mathcal{E}\left(S\right)\right)$
is a \emph{partial subgraph} of $G$.
\end{defn}
A partial subgraph $S$ of $G$ consists of a subset of agents in
$G$ and a \emph{subset} of all edges among $\mathcal{V}\left(S\right)$
also appearing in $G$. Counts of partial subgraphs are often referred
to as \emph{network motif }counts \citep[e.g.,][]{Milo_el_al_Sci02},
although this terminology is not used consistently. The two star motif
$S=\vcenter{\hbox{\includegraphics[scale=0.125]{twostar}}}$ is
a partial subgraph of $G=\vcenter{\hbox{\includegraphics[scale=0.125]{triangle}}}$.
Note that in this example $S$ does not include the edge between agents,
numbered clockwise from the top, $2$ and $3$.
\begin{defn}
\label{def: Induced-Subgraph}(\textsc{\uline{Induced Subgraph}})
Let $\mathcal{V}\left(S\right)\subseteq\mathcal{V}\left(G\right)$
be any subset of the vertices of $G$ and $\mathcal{E}\left(S\right)=\mathcal{E}\left(G\right)\cap\mathcal{V}\left(S\right)\times\mathcal{V}\left(S\right)$,
then $S=\left(\mathcal{V}\left(S\right),\mathcal{E}\left(S\right)\right)$
is an \emph{induced subgraph} of $G$.
\end{defn}
An induced subgraph $S$ includes \emph{all} edges in $G$ connecting
any two agents in $\mathcal{V}\left(S\right)$. Although $S=\vcenter{\hbox{\includegraphics[scale=0.125]{twostar}}}$
is a partial subgraph of $G=\vcenter{\hbox{\includegraphics[scale=0.125]{triangle}}}$,
it is not an induced one. Counts of induced subgraphs are often referred
to as \emph{graphlet }counts \citep[e.g.,][]{Przulj_et_al_BI04},
although again not consistently so.

Consider two graphs, $R$ and $S$, of the same order. Let $\varphi\thinspace:\mathcal{\thinspace V}\left(R\right)\rightarrow\mathcal{V}\left(S\right)$
be a bijection from the nodes of $R$ to those of $S$. The bijection
$\varphi\thinspace:\thinspace\mathcal{V}\left(R\right)\rightarrow\mathcal{V}\left(S\right)$
\emph{maintains adjacency} if for every dyad $i,j\in\mathcal{V}\left(R\right)$
if $\left(i,j\right)\in\mathcal{E}\left(R\right)$, then $\left(\varphi\left(i\right),\varphi\left(j\right)\right)\in\mathcal{E}\left(S\right)$;
it \emph{maintains non-adjacency} if for every dyad $i,j\in\mathcal{V}\left(R\right)$
if $\left(i,j\right)\notin\mathcal{E}\left(R\right)$, then $\left(\varphi\left(i\right),\varphi\left(j\right)\right)\notin\mathcal{E}\left(S\right)$.
If the bijection maintains both adjacency and non-adjacency we say
it \emph{maintains structure}.
\begin{defn}
\label{def: Graph-Isomorphism}(\textsc{\uline{Graph Isomorphism}})
The graphs $R$ and $S$ are \emph{isomorphic} if there exists a structure-maintaining
bijection $\varphi\thinspace:\thinspace\mathcal{V}\left(R\right)\rightarrow\mathcal{V}\left(S\right)$.
\end{defn}
In what follows I use the notation $R\cong S$ to denote that ``$R$
is isomorphic to $S$.''

Two special families of motifs/graphlets will play a prominent role
in the analysis of network summary statistics presented in Section
\ref{sec: Statistics} below. First, a $p$\emph{-cycle} is $p^{th}$
order graphlet with nodes labeled (or relabeled) such that its edges
form a cycle:
\[
\mathcal{E}\left(S\right)=\left\{ \left(i_{1},i_{2}\right),\left(i_{2},i_{3}\right),\ldots,\left(i_{p},i_{1}\right)\right\} .
\]
A $p$-cycle is a connected graphlet with $p$ edges on $p$ nodes.
As one transverses a $p$-cycle graphlet no vertex is crossed more
than once except for the first/last one. Important examples of $p$-cycles
are triangles ($S=\vcenter{\hbox{\includegraphics[scale=0.1250]{triangle}}}$)
and 4-cycles ($S=\vcenter{\hbox{\includegraphics[scale=0.1250]{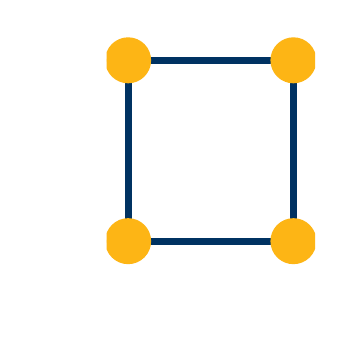}}}$).

Second, a \emph{tree} is a connected graph with no cycles. The number
of edges on a $p^{th}$ order tree is $p-1$; a feature which will
prove highly convenient. Important examples of trees are $p$-star
graphlets, such as two-stars ($S=\vcenter{\hbox{\includegraphics[scale=0.1250]{twostar}}}$)
and three-stars ($S=\vcenter{\hbox{\includegraphics[scale=0.1250]{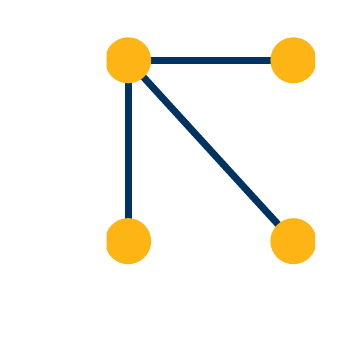}}}$).
Trees will feature in the analysis of the degree distribution given
below. Trees are also called connected acyclic graphs.

\subsubsection*{Induced subgraph density}

Using Definitions \ref{def: Induced-Subgraph} and \ref{def: Graph-Isomorphism}
we can formally introduce the induced subgraph density. This will
be our first measure of the frequency with which a specific low-order
local configuration of links appears within a network. Let $S$ be
a $p^{th}$-order graphlet of interest (e.g., $S=\vcenter{\hbox{\includegraphics[scale=0.1250]{onethreewheel}}}$
or $S=\vcenter{\hbox{\includegraphics[scale=0.1250]{triangle}}}$),
$\mathrm{iso}\left(S\right)$ the group of isomorphisms of $S$, and
$\left|\mathrm{iso}\left(S\right)\right|$ its cardinality. It is
helpful to observe that $\left|\mathrm{iso}\left(S\right)\right|$
equals the number of (partial) subgraphs of $K_{p}$ that are isomorphic
to $S$. For example, $\left|\mathrm{iso}\left(\vcenter{\hbox{\includegraphics[scale=0.1250]{twostar}}}\right)\right|=3$
since there are three ways to draw a two-star configuration on three
vertices. $G_{N}$ is the real world network under study.

Let $\mathbf{i}_{p}\subseteq\left\{ 1,2,\ldots,N\right\} $ be a set
of $p$ integers. If we require that $i_{1}<i_{2}<\cdots<i_{p}$,
then there are $\tbinom{N}{p}$ such integer sets; denote this set
of integer sets by $\mathcal{C}_{p,N}$. If all that is required is
that $i_{k}\neq i_{l}$ for $k\neq l$, then there are $\frac{N!}{\left(N-K\right)!}$
such integer sets; denote this set of integer sets by $\mathcal{A}_{p,N}.$

Let the vertex set of $S$ be $\left\{ 1,\ldots,p\right\} $. Let
$G\left[\mathbf{i}_{p}\right]$ denote the induced subgraph of $G$
associated with vertex set $\mathbf{i}_{p}$. Since we wish to compare
$S$ and $G\left[\mathbf{i}_{p}\right]$ it will be convenient to
relabel the latter. Let $\tilde{G}\left[\mathbf{i}_{p}\right]$ be
a relabelling of $G\left[\mathbf{i}_{p}\right]$ such that $i_{1}=1$,
$i_{2}=2,\ldots,i_{p}=p$ so that $kl\in\mathcal{E}\left(\tilde{G}\left[\mathbf{i}_{p}\right]\right)$
if $i_{k}i_{l}\in\mathcal{E}\left(G\left[\mathbf{i}_{p}\right]\right)$.
Let $\mathbf{i}_{p}\sim\text{Uniform\ensuremath{\left(\mathcal{A}_{p,N}\right)}}$;
the frequency with which $\tilde{G}_{N}\left[\mathbf{i}_{p}\right]$
equals $S$ is then
\begin{equation}
P_{N}\left(S\right)\overset{def}{\equiv}\Pr\left(S=\tilde{G}_{N}\left[\mathbf{i}_{p}\right]\right),\thinspace\thinspace\mathbf{i}_{p}\sim\text{Uniform\ensuremath{\left(\mathcal{A}_{p,N}\right)}}.\label{eq: induced_subgraph_density_1}
\end{equation}
Call (\ref{eq: induced_subgraph_density_1}) the \emph{induced subgraph
density} of $S$ in $G_{N}$. Alternatively we can write
\begin{equation}
P_{N}\left(S\right)=\frac{\Pr\left(S\cong G_{N}\left[\mathbf{i}_{p}\right]\right)}{\left|\mathrm{iso}\left(S\right)\right|},\thinspace\thinspace\mathbf{i}_{p}\sim\text{Uniform\ensuremath{\left(\mathcal{C}_{p,N}\right)}}\label{eq: induced_subgraph_density_2}
\end{equation}

The induced subgraph frequency of $S$ in $G_{N}$ equals the fraction
of injective mappings $\varphi\thinspace:\thinspace\mathcal{V}\left(S\right)\rightarrow\mathcal{V}\left(G_{N}\right)$
that preserve both edge adjacency \emph{and} non-adjacency. Direct
computation of this fraction yields the equalities
\begin{align}
P_{N}\left(S\right)= & \frac{N!}{\left(N-p\right)!}\sum_{\mathbf{i}_{p}\in A_{p,N}}\mathbf{1}\left(S=\tilde{G}_{N}\left[\mathbf{i}_{p}\right]\right)\label{eq: induced_subgraph_density}\\
= & \frac{1}{\tbinom{N}{p}\left|\mathrm{iso}\left(S\right)\right|}\sum_{\mathbf{i}_{p}\in\mathcal{C}_{p,N}}\mathbf{1}\left(S\cong G_{N}\left[\mathbf{i}_{p}\right]\right)\nonumber \\
\overset{def}{\equiv} & t_{\mathrm{ind}}\left(S,G_{N}\right)\nonumber 
\end{align}

In order to understand the mechanics of computing (\ref{eq: induced_subgraph_density})
it is useful to reformulate, one again, its definition. Let $\mathbf{D}_{\left[\mathbf{i}_{p},\mathbf{i}_{p}\right]}$
be the $p\times p$ sub-adjacency matrix constructed by removing all
rows and columns of $\mathbf{D}$ except those in $\mathbf{i}_{p}=\left\{ i_{1},\ldots,i_{p}\right\} .$
We can check for whether $G\left[\mathbf{i}_{p}\right]$ is an isomorphism
of $S$ by inspecting the elements of the $\mathbf{D}_{\left[\mathbf{i}_{p},\mathbf{i}_{p}\right]}$
sub-adjacency matrix. 

Consider the two star triad $S=\vcenter{\hbox{\includegraphics[scale=0.125]{twostar}}}$
, we can express $\mathbf{1}\left(S\cong G_{N}\left[\mathbf{i}_{p}\right]\right)$
in terms of $\mathbf{D}_{\left[\mathbf{i}_{p},\mathbf{i}_{p}\right]}$
as
\begin{equation}
\mathbf{1}\left(\vcenter{\hbox{\includegraphics[scale=0.125]{twostar}}}\cong G_{N}\left[\mathbf{i}_{p}\right]\right)=D_{i_{1}i_{2}}D_{i_{1}i_{3}}\left(1-D_{i_{2}i_{3}}\right)+D_{i_{1}i_{2}}\left(1-D_{i_{1}i_{3}}\right)D_{i_{2}i_{3}}+\left(1-D_{i_{1}i_{2}}\right)D_{i_{1}i_{3}}D_{i_{2}i_{3}}.\label{eq: two_star_indicator}
\end{equation}
We have $\left|\mathrm{iso}\left(\vcenter{\hbox{\includegraphics[scale=0.125]{twostar}}}\right)\right|=3$
with the three terms to the right of the equality in (\ref{eq: two_star_indicator})
equal to indicators for these three possible isomorphisms (on triad/vertex
set $\left\{ i_{1},i_{2},i_{3}\right\} $). In general $\mathbf{1}\left(S\cong G_{N}\left[\mathbf{i}_{p}\right]\right)$
may be defined in terms of $\mathbf{D}_{\left[\mathbf{i}_{p},\mathbf{i}_{p}\right]}$
with the number of components equal to the number of possible isomorphisms
of $S$. There is only one isomorphism of the $\vcenter{\hbox{\includegraphics[scale=0.125]{triangle}}}$
configuration, yielding a second example of 
\[
\mathbf{1}\left(\vcenter{\hbox{\includegraphics[scale=0.125]{triangle}}}\cong G_{N}\left[\mathbf{i}_{p}\right]\right)=D_{i_{1}i_{2}}D_{i_{1}i_{3}}D_{i_{2}i_{3}}.
\]

Recognizing that $t_{\mathrm{ind}}\left(S,G_{N}\right)$ is a functional
of the adjacency matrix of $G_{N}$ allows us to easily compute its
expectation when edges form according to the conditional edge independence
model (\ref{eq: Aldous-Hoover-DGP}). Once again consider the two
star configuration; iterated expectations and conditional independence
of edges given $\mathbf{U}=\left(U_{1},\ldots,U_{N}\right)'$ yield
\begin{align*}
\mathbb{E}\left[D_{i_{1}i_{2}}D_{i_{1}i_{3}}\left(1-D_{i_{2}i_{3}}\right)\right] & =\mathbb{E}\left[\mathbb{E}\left[\left.D_{i_{1}i_{2}}D_{i_{1}i_{3}}\left(1-D_{i_{2}i_{3}}\right)\right|\mathbf{U}\right]\right]\\
 & =\mathbb{E}\left[h\left(U_{i_{1}},U_{i_{2}}\right)h\left(U_{i_{1}},U_{i_{3}}\right)\left[1-h\left(U_{i_{2}},U_{i_{3}}\right)\right]\right]\\
 & =\int\int\int h\left(t,u\right)h\left(t,v\right)\left[1-h\left(u,v\right)\right]\mathrm{d}t\mathrm{d}u\mathrm{d}v
\end{align*}
(and also that the value of $\mathbb{E}\left[D_{i_{1}i_{2}}D_{i_{1}i_{3}}\left(1-D_{i_{2}i_{3}}\right)\right]$
is invariant to permutations of its indices). Finally we have, recalling
that $\left|\mathrm{iso}\left(\vcenter{\hbox{\includegraphics[scale=0.125]{twostar}}}\right)\right|=3$,
\[
\mathbb{E}\left[\mathbf{1}\left(\vcenter{\hbox{\includegraphics[scale=0.125]{twostar}}}\cong G_{N}\left[\mathbf{i}_{p}\right]\right)\right]=3\cdot\int\int\int h\left(t,u\right)h\left(t,v\right)\left[1-h\left(u,v\right)\right]\mathrm{d}t\mathrm{d}u\mathrm{d}v,
\]
for $\mathbf{i}_{p}\sim\text{Uniform\ensuremath{\left(\mathcal{C}_{p,N}\right)}}.$
For a generic graphlet configuration we have
\begin{align}
\mathbb{E}\left[t_{\mathrm{ind}}\left(S,G_{N}\right)\right] & =\left|\mathrm{iso}\left(S\right)\right|^{-1}\mathbb{E}\left[\mathbf{1}\left(S\cong G_{N}\left[\mathbf{i}_{p}\right]\right)\right]\label{eq: induced_subgraph_density_population}\\
 & =\mathbb{E}\left[\prod_{\left\{ i,j\right\} \in\mathcal{E}\left(S\right)}h\left(U_{i},U_{j}\right)\prod_{\left\{ i,j\right\} \in\mathcal{E}\left(\bar{S}\right)}\left[1-h\left(U_{i},U_{j}\right)\right]\right]\nonumber \\
 & \overset{def}{\equiv}P\left(S\right)\nonumber 
\end{align}
where $\bar{G}$ denotes the complement of the graph $G$: the graph
defined on the same nodes as $G$ with an edge present if, and only
if, it is not present in $G$. The graph sum of $G$ and $\bar{G}$
therefore coincides with the complete graph $K_{\left|\mathcal{V}\left(G\right)\right|}$.

Call the \emph{expectation} of $t_{\mathrm{ind}}\left(S,G_{N}\right)$
the induced subgraph density of $S$ in the \emph{graphon} $h\left(\cdot\right)$
and write it as, in an abuse of notation, $\mathbb{E}\left[t_{\mathrm{ind}}\left(S,G_{N}\right)\right]=t_{\mathrm{ind}}\left(S,h\right)=P\left(S\right)$.
Clearly $P_{N}\left(S\right)$ is an unbiased estimate of $t_{\mathrm{ind}}\left(S,h\right)=P\left(S\right)$
\emph{when the true network generating process is of the CID type}.
Notice how the graphon provides a language for connecting \emph{empirical}
graphlet counts, first studied by \citet{Holland_Leinhardt_AJS70},
with well-defined probabilistic objects. This connection will prove
useful for developing a procedure for conducting inference on $P\left(S\right)$
using the sample graph $G_{N}$. Since $P\left(S\right)$ generally
varies with the graphon $h\left(u,v\right)$, the idea is that by
identifying $P\left(S\right)$ for enough specific configurations
(e.g., $S=\vcenter{\hbox{\includegraphics[scale=0.125]{twostar}}},\vcenter{\hbox{\includegraphics[scale=0.125]{triangle}}},\vcenter{\hbox{\includegraphics[scale=0.125]{fourcycle}}},\vcenter{\hbox{\includegraphics[scale=0.125]{onethreewheel}}}$
etc.), we may be able to identify $h\left(u,v\right)$ itself \citep[cf.,][]{Bickel_et_al_AS11}.

\subsubsection*{Injective homomorphism density}

A second notion of subgraph density also appears in some of the results
which follow. Let $S\subseteq G$ denote that $S$ is a partial subgraph
of $G$. Using Definitions \ref{def: Partial-Subgraph} and \ref{def: Graph-Isomorphism},
we can also define what I will call, following \citet{Lovasz_AMS12},
the injective homomorphism density.\footnote{The \citet{Lovasz_AMS12} monograph presents several different notions
of a subgraph density. The two introduced here were chosen for their
connection to actual empirical practice. See also \citet{Diaconis_Janson_RM08}.} The homomorphism density gives the probability that \emph{a (partial)
subgraph of} $G_{N}\left[\mathbf{i}_{p}\right]$, for $\mathbf{i}_{p}$
chosen uniformly at random from $\mathcal{A}_{p,N}$, is equal to
$S$. Alternatively the homomorphism density equals the fraction of
injective mappings $\varphi\thinspace:\thinspace\mathcal{V}\left(S\right)\rightarrow\mathcal{V}\left(G_{N}\right)$
that preserve edge adjacency. These mappings do not need to preserve
non-adjacency.\footnote{In contrast the induced subgraph density requires preservation of
both adjacency and non-adjacency.} The \emph{injective homomorphism density} of $S$ in $G_{N}$ equals
\begin{align}
Q_{N}\left(S\right) & =\frac{1}{\tbinom{N}{p}\left|\mathrm{iso}\left(S\right)\right|}\sum_{R\subseteq K_{N},R\cong S}\mathbf{1}\left(R\subseteq G_{N}\right)\label{eq: homomorphism_density}\\
 & =\frac{1}{\tbinom{N}{p}\left|\mathrm{iso}\left(S\right)\right|}\sum_{R\subseteq K_{N},\left|V\left(R\right)\right|=p}\mathbf{1}\left(R\cong S\right)\prod_{\left\{ i,j\right\} \in\mathcal{E}\left(R\right)}D_{ij}\nonumber \\
 & \overset{def}{\equiv}t_{\mathrm{inj}}\left(S,G_{N}\right)\nonumber 
\end{align}
The two equivalent definitions are given to develop familiarity with
notation. To understand expression (\ref{eq: homomorphism_density})
it is helpful to calculate the injective homomorphism density of $S=\vcenter{\hbox{\includegraphics[scale=0.125]{twostar}}}$
in $G_{N}=\vcenter{\hbox{\includegraphics[scale=0.125]{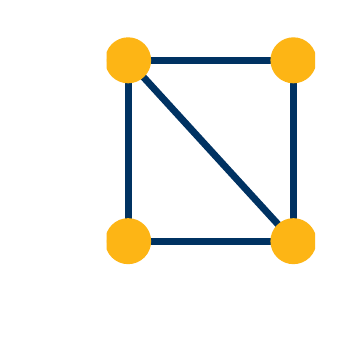}}}.$
There are three isomorphisms of the two star configuration such that
$\tbinom{4}{3}\left|\mathrm{iso}\left(\vcenter{\hbox{\includegraphics[scale=0.125]{twostar}}}\right)\right|=4\cdot3=12$.
Next consider the summation in the first line of (\ref{eq: homomorphism_density}).
This summation is over all $3^{rd}$ order partial subgraphs of $K_{4}$
which are isomorphic to $S=\vcenter{\hbox{\includegraphics[scale=0.125]{twostar}}}$.
There are exactly $12$ two star partial subgraphs in $K_{4}$ (three
for each of its four triads), a total of 8 of these configurations
are subgraphs of $G_{N}$ such that $t_{\mathrm{inj}}(\vcenter{\hbox{\includegraphics[scale=0.125]{twostar}}},\vcenter{\hbox{\includegraphics[scale=0.125]{chordalcycle}}})=\frac{8}{12}$.
Note that the induced subgraph density of $S=\vcenter{\hbox{\includegraphics[scale=0.125]{twostar}}}$
in $G_{N}=\vcenter{\hbox{\includegraphics[scale=0.125]{chordalcycle}}}$
is just $\frac{2}{12}$.

Under an Aldous-Hoover GGP we have
\begin{align*}
\mathbb{E}\left[t_{\mathrm{inj}}\left(S,G_{N}\right)\right] & =\frac{1}{\tbinom{N}{p}\left|\mathrm{iso}\left(S\right)\right|}\sum_{R\subseteq K_{N},\left|V\left(R\right)\right|=p}\mathbf{1}\left(R\cong S\right)\mathbb{E}\left[\mathbb{E}\left[\left.\prod_{\left\{ i,j\right\} \in\mathcal{E}\left(R\right)}D_{ij}\right|U_{1},\ldots,U_{N}\right]\right]\\
 & =\mathbb{E}\left[\prod_{\left\{ i,j\right\} \in\mathcal{E}\left(S\right)}h\left(U_{i},U_{j}\right)\right]\\
 & \overset{def}{\equiv}Q\left(S\right).
\end{align*}
Call the expectation of $t_{\mathrm{inj}}\left(S,G_{N}\right)$ the
injective homomorphism density of $S$ in the \emph{graphon} $h\left(\cdot\right)$
and write it as $\mathbb{E}\left[t_{\mathrm{inj}}\left(S,G_{N}\right)\right]=t_{\mathrm{inj}}\left(S,h\right)=Q\left(S\right)$.

\subsection{\label{subsec: Graph-limits}Graph limits}

Let $G_{N}$ be a finite exchangeable graph with adjacency matrix
$\mathbf{D}$. Let 
\[
h_{G_{N}}\left(u,v\right)=\begin{cases}
\begin{array}{c}
1\\
0
\end{array} & \begin{array}{c}
\text{if}\thinspace\left(\left\lceil uN\right\rceil ,\left\lceil vN\right\rceil \right)\in\mathcal{E}\left(G_{N}\right)\\
\text{otherwise}
\end{array}\end{cases}.
\]
Observe that $h_{G_{N}}\left(u,v\right)$ is a valid graphon and further
that
\[
t_{\mathrm{ind}}\left(S,G_{N}\right)=t_{\mathrm{ind}}\left(S,h_{G_{N}}\right)
\]
for any $S$ of order $K\leq N$ \citep[p. 28]{Chatterjee_LNM17}.
This equality connects the definition of the induced subgraph frequency
of $S$ in $G_{N}$, denoted by $P_{N}\left(S\right)$ in equation
(\ref{eq: induced_subgraph_density}), with its ``population'' counterpart
-- equation (\ref{eq: induced_subgraph_density_population}). It
also motivates the idea of the graphon as the appropriate limit object
for a sequence of graphs, $G_{N}$. If the subgraph frequency
\[
\underset{N\rightarrow\infty}{\lim}t_{\mathrm{ind}}\left(S,h_{G_{N}}\right)
\]
converges to some limit for \emph{all} fixed subgraphs $S$, then
we say that $G_{N}$ has a limit. \citet{Lovasz_Szegedy_JCT06} showed
the natural limiting object is a graphon (i.e, heuristically, $h_{G_{N}}\rightarrow h$
as $N\rightarrow\infty$). \citet{Diaconis_Janson_RM08} connect this
finding with the Aldous-Hoover representation theorem. Collectively
these results motivate an approach to summarizing a network by the
frequency of different low order subgraph configurations within it;
by its \emph{average} \emph{local structure}. \citet{Lovasz_AMS12}
provides a rigorous and compresensive introduction to theory of graph
limits.

\subsection{Sampling}

In this chapter I will adopt two perspectives on ``sampling''. In
the first we view the network in hand as the one induced by a random
sample of agents from some large (i.e., infinite) population. Let
$G_{\infty}$ be an (infinite) exchangeable random graph. Let $\mathcal{V}$
be a random sample of agents of size $N$ from $G_{\infty}$. We assume
that the observed network, $G_{N}$, coincides with the subgraph induced
by this random sample of vertices:
\begin{equation}
G_{N}=G_{\infty}\left[\mathcal{V}\right].\label{eq: random_induced_subgraph}
\end{equation}

Let $\mathbf{D}_{\infty}=\left[D_{ij}\right]$ with $i,j\geq1$ be
the adjacency matrix of $G_{\infty}$. Exchangeability implies the
characterization
\begin{equation}
D_{ij}=\mathbf{1}\left(h\left(\alpha,U_{i},U_{j}\right)\geq V_{ij}\right)\label{eq: aldous_hoover_DGP}
\end{equation}
with $\alpha$, $U_{i}$ and $V_{ij}=V_{ji}$ independent $\mathcal{U}\left[0,1\right]$
random variables \citep[cf.,][]{Aldous_JMA81,Hoover_WP79}. Here $h:\left[0,1\right]^{3}\rightarrow\left[0,1\right]$
is symmetric in its second and third arguments.

Under (\ref{eq: random_induced_subgraph}) the elements of $\mathbf{D}$,
the adjacency matrix for the network in hand, also obey the characterization
(\ref{eq: aldous_hoover_DGP}). The ``sampling distribution'' of
some statistic of $\mathbf{D}$, say $t_{N}\left(\mathbf{D}\right)$,
is simply the one induced by repeated random sampling from the underlying
infinite population. We calculate limit distributions by studying
the sampling distribution of $t_{N}\left(\mathbf{D}\right)$ as $N\rightarrow\infty$.

An advantage of this first perspective it that is allows the econometrician
to fully exploit the independence/dependence structure associated
with the Aldous-Hoover Theorem. If the graph in hand is the one induced
by a random sample of agents from some infinite exchangeable population,
then we can proceed ``as if''
\begin{equation}
\left.D_{ij}\right|U_{i},U_{j}\sim\mathrm{Bernoulli}\left(h\left(U_{i},U_{j}\right)\right)\label{eq: nonparametric_dgp}
\end{equation}

for $i=1,\ldots,N-1$ and $j=i+1,\ldots,N$. Although (\ref{eq: nonparametric_dgp})
is a nonparametric data generating process, it is a structured one.
We can use this structure to our advantage.

An unattractive feature of this perspective is that if the density
of the population graph is very low, then that of the sampled graph
may be zero with high probability. To see this point heuristically
assume that the population consists of $N^{*}$ agents, with $N^{*}$
very large. Assume that average degree, $\lambda$, is some small
positive constant that does not dependent on $N^{*}$. The probability
of observing an edge between the two independent random draws from
the population is thus
\[
\Pr\left(D_{12}=1\right)=\frac{\frac{1}{2}\lambda N^{*}}{\tbinom{N^{*}}{2}}\approx\frac{\lambda}{N^{*}}.
\]

Boole's inequality then gives a probability of observing at least
one edge in our sampled network no greater than $\tbinom{N}{2}\lambda/N^{*}$,
which will be close to zero when $N<<\sqrt{N^{*}}$. When the population
graph is ``sparse'', it is quite likely that the subgraph induced
by a random sample of agents from it will be empty and hence completely
uninformative. See \citet[Chapter 3]{Crane_PFSNA18} for more discussion
and examples.

This example raises two questions. First, how does one sample from
a large sparse graph in practice? I ignore this question here, but
flag it as an interesting one which merits thought. The monograph
by \citet{Crane_PFSNA18} surveys extant work in this area. Second,
if the sampling is fictitious (i.e., analysis is based upon the full
graph), what mistakes might be made by proceeding ``as if'' we had
randomly sampled from some (now entirely hypothetical) large graph?

To answer the second question is useful to return to an empirical
example. \citet{Atalay_et_al_PNAS11} study the supply chain network
of large publicly traded firms in the United States. Their network
is not sampled, but rather constructed from Securities and Exchange
Commission (SEC) reports filed by the entire universe of publically
trade firms. If the model of network formation of interest is a conditional
independent dyad (CID) one, then we are free to proceed ``as if''
the observed network were generated according to (\ref{eq: nonparametric_dgp}).
If, instead, we view the network in hand as, for example, an equilibrium
of a finite $N$-player supply chain formation game, then it may be
difficult to justify (\ref{eq: nonparametric_dgp}); strategic interaction
may induce dependence across links that cannot be conditioned away.
We cannot appeal directly to the Aldous-Hoover Theorem.

As in \citet{deFinetti_AN1931}, the Aldous-Hoover Theorem requires
that the agent indices constitute an infinite sequence. However, just
as the de Finetti result fails for finite sequences \citep[e.g., ][]{Diaconis_Syn77},
but approximately holds when the sequence is large enough \citep[e.g., ][]{Diaconis_Freedman_AP80},
the hope is that in \emph{large} (but finite) networks Theorem \ref{thm: Aldous-Hoover}
remains useful \citep[cf.,][]{Volfosky_Airoldi_SPL16}.

One possibility would be to assume that $N$ is large enough such
that a representation like (\ref{eq: nonparametric_dgp}) ``approximately''
holds. One could then conduct inference on model parameters by comparing
observed network moments with model generated ones. The sampling distribution
of the observed network moments would be calculated assuming an Aldous-Hoover
DGP (which is appropriate for $N$ large enough). I sketch this idea
in a bit more detail in Section \ref{sec: Strategic-models} below.
Many gaps in this discussion remain. Alternatively we could proceed
along the lines of \citet{Menzel_WP16}. In this approach we would
approximate our finite player network formation game, with a limit
game which is easier to deal with (see Section \ref{sec: Strategic-models}).

\subsection{Adding sparsity: the \citet{Bickel_Chen_PNAS09} model}

For any finite network of unlabelled agents, exchangeability is a
natural, indeed unavoidable, modeling assumption. Unfortunately its
extension to infinite exchangeability, as needed for Theorem \ref{thm: Aldous-Hoover},
has the unattractive implication that the network is either empty
or dense in the limit. Specifically a (random) agent will either never
form links or do so infinitely often as $N\rightarrow\infty$. Denseness
and sparseness are limit properties of infinite sequences of graphs.
Any empirical network is neither ``dense'' nor ``sparse'', it
just is what it is. However, in most real world networks the numbers
of agents and links are of similar magnitudes. This suggests that
approximation results based on sequences of graphs that are sparse
in the limit may be more useful than those with dense limits. Whether
this is, in fact, the case remains an open question \citep{Green_Shalizi_arXiv17}.

One way to model sequences of graphs with sparse limits, while still
preserving the analytic convenience of conditional independence across
edges, was proposed by \citet{Bickel_Chen_PNAS09}. The Bickel-Chen
model is the default one in the nonparametric statistics and machine
learning literatures on random graphs.

Let $G_{N}$ be a random network of order $N$ generated according
to (\ref{eq: Aldous-Hoover-DGP}). The expected average number of
links an agent has in this network, that is \emph{average degree},
equals
\begin{equation}
\lambda_{N}=\left(N-1\right)\rho_{\alpha}\label{eq: average_degree}
\end{equation}
for $\rho_{\alpha}=\int h\left(\alpha,u,v\right)\mathrm{d}u\mathrm{d}v$.
Average degree (\ref{eq: average_degree}) either tends toward infinity
or is zero, depending on whether $\rho_{\alpha}$ is greater than
or equal to zero.

To extend model (\ref{eq: Aldous-Hoover-DGP}) so that it can accommodate
sparse graph sequences \citet{Bickel_Chen_PNAS09} define the conditional
density
\begin{align*}
w_{\alpha}\left(u,v\right) & =f_{\left.U_{i},U_{j}\right|D_{ij},\alpha}\left(\left.u,v\right|D_{ij}=1,\alpha\right).
\end{align*}

Next observe that since $f_{\left.U_{i},U_{j}\right|\alpha}\left(\left.u,v\right|\alpha\right)=1$
on $\left[0,1\right]^{2}$ we get can decompose the graphon as
\begin{equation}
h\left(\alpha,u,v\right)=\rho_{\alpha}w_{\alpha}\left(u,v\right).\label{eq: Bickel_Chen_decomp}
\end{equation}
With this parameterization, \citet{Bickel_Chen_PNAS09} and \citet{Bickel_et_al_AS11}
argue that it is natural to let $\rho_{\alpha}=\rho_{\alpha,N}$,
but retain independence of $w_{\alpha}\left(u,v\right)$ from $N$.
Suppressing the $\alpha$ argument (it is never identifiable), they
write
\begin{equation}
\Pr\left(\left.D_{ij}=1\right|U_{i}=u,U_{j}=v\right)=h_{N}\left(u,v\right)=\rho_{N}w\left(u,v\right).\label{eq: Bickel_Chen_Model}
\end{equation}
The rate at which $\rho_{N}\rightarrow0$ then controls the rate of
average degree growth as $N$ grows large. If $\lambda_{N}=\left(N-1\right)\rho_{N}\rightarrow\lambda$
with $0<\lambda<\infty$ as $N\rightarrow\infty$, then the graph
is \emph{sparse. }If $\lambda_{N}=\Omega\left(N\right)$ we say the
graph is \emph{dense}, $\lambda_{N}=\Omega\left(\ln N\right)$ \emph{semi-dense}
etc. Many of the results presented below require that $\lambda_{N}=\Omega\left(N^{\alpha}\right)$
for some $0<\alpha\leq1$, despite the fact that $\lambda_{N}=\Omega\left(1\right)$
might best describe real world networks (where average degree is generally
low even when $N$ is very large). In what follows I will try to highlight
those few known results which can accommodate sparse graph sequences.

\subsection{Further reading}

\citet{Orbanz_Roy_IEEE15} provide a non-technical introduction to
the probability literature on exchangeable random arrays; the monograph
by \citet{Kallenberg_PSIP05} a more complete development. \citet{Crane_PFSNA18}
also surveys this material, at a fairly accessible level, and with
a somewhat contrarian point of view.

\citet{Lovasz_AMS12} provides an overview of the theory of graph
limits. \citet{Diaconis_Janson_RM08} connect much of this theory
to the older literature on exchangeable random arrays.

\section{\label{sec: dyadic_regression}Dyadic regression}

Jan Tinbergen's 1962 report \emph{Shaping the World Economy}, commissioned
by Twenty Century Fund, featured, along with its sculptural title,
a remarkable empirical analysis of trade flows \citep{Tinbergen_SWE62}.
Table VI-1 in that report presented the results of a least squares
fit of the logarithm of exports \emph{from} country $i$ \emph{to}
country $j$ onto a constant, the (log) Gross National Product (GNP)
of both countries $i$ and $j$, the (log) distance between $i$ and
$j$, and a variety of other covariates capturing different relationships
between $i$ and $j$. Tinbergen's \citeyearpar{Tinbergen_SWE62}
analysis was based upon a sample of $N=18$ countries, or $N\left(N-1\right)=306$
directed trading relationships.\footnote{A second analysis, based upon a larger sample of countries, was also
reported upon in Table VI-4 of the report.}

Table VI-1 of \citet{Tinbergen_SWE62} presents the results of what
I will call a dyadic regression analysis. This particular analysis
continues to serve as prototype for a substantial body of empirical
work in international trade \citep{Anderson_AR11}. Dyadic regression
analyses also appear in other areas of social science research. They
have been used, to give just a few recent examples, to study the onset
of war among nation states \citep[e.g.,][]{Russett_Oneal_Book01},
risk-sharing across households \citep[e.g.,][]{deWeerdt_IAP04,Fafchamp_Gubert_JDE07,Attanasio_AEJ12},
supply chain linkages across firms \citep[e.g.,][Table S3]{Atalay_et_al_PNAS11},
the formation of commercial R\&D collaborations \citep[Table 4]{Konig_et_al_RESTAT19},
and co-camping behavior among hunter-gathers \citep[Tables S2 to S49]{Apicella_et_al_Nat12}.

Familiar methods of econometric analysis appropriate for single agent
models, typically utilizing a random sample from the population of
interest, are ill-suited for dyadic settings \citep[cf.,][]{Cameron_Golotvina_WP05}.
Consequently, considerable confusion and controversy is associated
with dyadic analyses in practice \citep[e.g., ][]{Erikson_Pinto_Rader_PA14}.
It is remarkable that, over a half-century after Tinbergen's \citeyearpar{Tinbergen_SWE62}
pioneering analysis of trade flows across countries, and also given
the considerable empirical work that has followed, a textbook treatment
of estimation and inference methods for gravity and other dyadic regression
models remains unavailable.

\subsection{Population and sampling framework}

Let $i\in\mathbb{N}$ index agents in an infinite population of interest.
Associated with each agent is the observable attribute $X_{i}\in\mathbb{X}=\left\{ x_{1},\ldots,x_{L}\right\} $.
This attribute partitions the population into $L=\left|\mathbb{X}\right|$
subpopulations which I will refer to as ``types''. Let $\mathbb{N}\left(x\right)=\left\{ i\thinspace:\thinspace X_{i}=x_{l}\right\} $
be the index set for type $l$ agents. Although $L$ may be very large,
I assume that the size of each subpopulation, $\left|\mathbb{N}\left(x\right)\right|$,
is infinite with positive frequency (i.e., $\Pr\left(X_{i}=x_{l}\right)>0$
for $l=1,\ldots,L$).

When all observable agent attributes are discretely-valued, then $\mathbb{X}$
simply enumerates all distinct combinations of these attributes (e.g.,
$X=x_{l}$ might correspond to a Hispanic female, living in the Florida,
with 12 years of schooling and two college-educated parents). More
heuristically we can think of $\mathbb{X}$ as consisting of the support
points of a multinomial approximation to the support of a bundle of
attributes, some of which might be continuously-valued. The finite
support restriction is made in order to invoke a representation result
due to \citet{Crane_Towsner_JSL18}; I do not think it is essential.

Associated with each ordered pair of agents is the scalar directed
outcome $Y_{ij}\in\mathbb{Y}\subseteq\mathbb{R}$. I will refer to
agent $i$ as the ``ego'' of the directed dyad and agent $j$ as
its ``alter''. In the context of the trade example the ego agent
is the exporting country, the alter the importing one. The \emph{adjacency
matrix }$\left[Y_{ij}\right]_{i,j\in\mathbb{N}}$ collects all such
outcomes into an infinite random array.

From the standpoint of the econometrician, the indexing of agents
within subpopulations homogenous in $X_{i}$ is arbitrary: agents
of the same type are exchangeable. Exchangeability of agents within
subpopulations homogenous in $X_{i}$ induces a particular form of
exchangeability on the adjacency matrix. This form of exchangeability,
in turn, induces a particular form of dependence across the rows and
columns of $\left[Y_{ij}\right]_{i,j\in\mathbb{N}}$. The structure
of this dependence allows for the formulation of LLNs and CLTs.

Let $\sigma_{x}:\mathbb{N}\rightarrow\mathbb{N}$ be any permutation
of a finite number of the agent indices which satisfies the restriction
\begin{equation}
\left[X_{\sigma_{x}\left(i\right)}\right]_{i\in\mathbb{N}}=\left[X_{i}\right]_{i\in\mathbb{N}}.\label{eq: within_type_permutation}
\end{equation}
Condition (\ref{eq: within_type_permutation}) constrains index permutations
to occur among agents of the same type (i.e., we may permute the indices
in $\mathbb{N}\left(x\right)$, but not those within, for example,
$\mathbb{N}\left(x\right)\cup\mathbb{N}\left(x'\right)$). \citet{Crane_Towsner_JSL18}
call a network \emph{relatively exchangeable} with respect to $X$
(or $X$-exchangeable) if
\begin{equation}
\left[Y_{\sigma_{x}\left(i\right)\sigma_{x}\left(j\right)}\right]_{i,j\in\mathbb{N}}\overset{D}{=}\left[Y_{ij}\right]_{i,j\in\mathbb{N}}\label{eq: X-exchangeability}
\end{equation}
for all permutations $\sigma_{x}$ satisfying (\ref{eq: within_type_permutation}).
$X$-exchangeablility is a natural generalization of joint exchangeability,
as introduced in the context of the \citet{Aldous_JMA81} and \citet{Hoover_WP79}
Theorem earlier.

A insightful way to think about condition (\ref{eq: X-exchangeability})
is in terms of vertex colored graphs. Associate $X_{i}$ with the
color of a vertex; condition (\ref{eq: X-exchangeability}) states
that all colored graph isomorphisms are equally probable. Since, when
vertices of the same color are exchangeable, there is no reason to
attach more or less probability to particular isomorphisms of a given
vertex colored graph, any probability model for $\left[Y_{ij}\right]_{i,j\in\mathbb{N}}$
should be consistent with condition (\ref{eq: X-exchangeability}).
As long as $X_{i}$ encodes all the vertex-specific information available
to the econometrician, then $X$-exchangeability is a nature \emph{a
priori }modeling restriction.

Let $\alpha$, $\left\{ \left(U_{i},X_{i}\right)\right\} _{i\geq1}$
and $\left\{ \left(V_{ij},V_{ji}\right)\right\} _{i\geq1,j\geq1}$
be (sequences of ) i.i.d. random variables, additionally independent
of one another, and consider the random array $\left[Y_{ij}^{*}\right]_{i,j\in\mathbb{N}}$
generated according to the rule
\begin{equation}
Y_{ij}^{*}=\tilde{h}\left(\alpha,X_{i},X_{j},U_{i},U_{j},V_{ij}\right)\label{eq: dyadic_regression_DGP}
\end{equation}
with $\tilde{h}:\left[0,1\right]\times\mathbb{X}\times\mathbb{X}\times\left[0,1\right]^{3}\rightarrow\mathbb{Y}$
a measurable function (we normalize $\alpha$, $U_{i}$ and $V_{ij}$
to have support on the unit interval without loss of generality).
Clearly a graph generated according to (\ref{eq: dyadic_regression_DGP})
is $X$-exchangeable \citep[cf.,][Chapter 8]{Crane_PFSNA18}.

Here $\alpha$ is a mixing parameter analogous to the one appearing
in de Finetti's \citeyearpar{deFinetti_AN1931} original representation
theorem. Since this parameter is unidentified, and the focus here
is upon inference conditional on the realized data distribution, I
will depress the dependence of $\tilde{h}$ on $\alpha$, defining
the notation $h\left(X_{i},X_{j},U_{i},U_{j},V_{ij}\right)\overset{def}{\equiv}\tilde{h}\left(\alpha,X_{i},X_{j},U_{i},U_{j},V_{ij}\right)$.
Consistent with earlier terminology, the function $h:\mathbb{X}\times\mathbb{X}\times\left[0,1\right]^{3}\rightarrow\mathbb{Y}$
will be referred to as a graphon.

Because doing so is convenient for the discussion of causal inference
in dyadic settings which follows, (\ref{eq: dyadic_regression_DGP})
makes no presumption of independence between $X_{i}$ and $U_{i}$.
Of course we can always write 
\begin{align*}
Y_{ij}^{*} & =h\left(X_{i},X_{j},F_{\left.U_{1}\right|X_{1}}\left(\left.U_{i}\right|X_{i}\right),F_{\left.U_{1}\right|X_{1}}\left(\left.U_{j}\right|X_{j}\right),V_{ij}\right)\\
 & \overset{def}{\equiv}h^{*}\left(X_{i},X_{j},U_{i}^{*},U_{j}^{*},V_{ij}\right)
\end{align*}
with $U_{i}^{*}=F_{\left.U_{1}\right|X_{1}}\left(\left.U_{i}\right|X_{i}\right)$
equal to unit $i$'s rank among all those units of her type. The resulting
$\left\{ U_{i}^{*}\right\} _{i\geq1}$ sequence of 0-to-1 uniform
random variables is independent of $\left\{ X_{i}\right\} _{i\geq1}$
by construction \citep[cf.,][]{Graham_Imbens_Ridder_NBER10}.

Depending on the context, it is fine to work with either $h$ or $h^{*}$,
but, as explained below, the former is more useful for making causal
arguments; hence I allow for dependence between the \emph{observed}
covariate vector $X_{i}$ and the \emph{unobserved} unit-specific
effect $U_{i}$ in what follows (akin to a correlated random effects
panel data analysis). The nuances involved will become clear as we
proceed.

Networks generated by (\ref{eq: dyadic_regression_DGP}) exhibit a
very particular pattern of dependence across the rows and columns
of $\left[Y_{ij}\right]_{i,j\in\mathbb{N}}$. Consider, without loss
of generality, agents $1$, $2$, $3$ and $4$. The outcomes $Y_{12}$
and $Y_{34}$ are independent of one another; the outcomes $Y_{12}$
and $Y_{13}$ are, however, dependent. These two outcomes share agent
$1$ in common; the value of $X_{1}$ and $U_{1}$ influences both
$Y_{12}$ and $Y_{13}$, inducing dependence. But conditional on $\left(X_{1},X_{2},X_{3}\right)$
and $\left(U_{1},U_{2},U_{3}\right)$, $Y_{12}$ and $Y_{13}$ are
independent; if we condition on the observed covariates $\left(X_{1},X_{2},X_{2}\right)$
alone, however, they remain dependent. Finally $Y_{12}$ and $Y_{21}$
are dependent, this dependence holds even conditional on $\left(X_{1},X_{2}\right)$
and $\left(U_{1},U_{2}\right)$ because $V_{12}$ and $V_{21}$ may
covary.

In words we have independence across dyads sharing no agents in common
(exports from Japan to the United States and from Turkey to Germany),
dependence across those sharing at least one agent in common (exports
from Japan to the United States and from Japan to the United Kingdom),
and ``even more'' dependence across dyads sharing both agents in
common (e.g., exports from Japan to the United States and vice-versa).

Models with this type of dependence structure, as already noted, are
called conditionally independent dyad (CID) models. The ``conditionally
independent'' terminology reflects the fact that the outcomes $Y_{12}$
and $Y_{13}$, associated with a pair of dyads sharing one agent in
common, can be rendered independent of one another by conditioning
on the observed covariates $\left(X_{1},X_{2},X_{2}\right)$ \emph{as
well as} the unobserved latent attributes $\left(U_{1},U_{2},U_{3}\right)$.

\citet{Crane_Towsner_JSL18}, in an extension of the Aldous-Hoover
representation result described earlier, show that for any $X$-exchangeable
random array $\left[Y_{ij}\right]_{i,j\in\mathbb{N}}$ there exists
another array $\left[Y_{ij}^{*}\right]_{i,j\in\mathbb{N}}$ generated
according to (\ref{eq: dyadic_regression_DGP}) such that the two
arrays have the same distribution:
\begin{equation}
\left[Y_{ij}\right]_{i,j\in\mathbb{N}}\overset{D}{=}\left[Y_{ij}^{*}\right]_{i,j\in\mathbb{N}}.\label{eq: Crane_Towsner_Theorem}
\end{equation}
We can therefore use (\ref{eq: dyadic_regression_DGP}) as an `as
if' non-parametric data generating process for $\left[Y_{ij}\right]_{i,j\in\mathbb{N}}$;
this will facilitate a variety of probabilistic calculations (e.g.,
computing conditional expectations, variances and, especially, covariances).

Let $i=1,\ldots,N$ index a simple random sample from the target population.
For each of the $N$ sampled units the econometrician observes $X_{i}$
and for each of the $\tbinom{N}{2}$ sampled dyads she observes $\left(Y_{ij},Y_{ji}\right)$.
From hereon I will assume that $Y_{ii}$ is undefined (normalized
to zero for convenience). Adapting what follows to accommodate self-loops
is straightforward.

\subsection{Composite likelihood}

Let $\left\{ f_{\left.Y_{12}\right|X_{1},X_{2}}\left(\left.Y_{12}\right|X_{1},X_{2};\theta\right):\theta\in\Theta\subseteq\mathbb{R}^{\dim\left(\theta\right)}\right\} $
be a parametric family of distributions for the conditional distribution
of $Y_{12}$ given $X_{1}$ and $X_{2}$. For example, \citet{SantosSilva_Tenreyro_RESTAT06}
model trade from exporter $i$ to importer $j$ given covariates as
a Poisson random variable:
\begin{equation}
f_{\left.Y_{12}\right|X_{1},X_{2}}\left(\left.y_{ij}\right|X_{i},X_{j};\theta\right)=\exp\left[-\exp\left[W_{ij}'\theta\right]\right]\frac{\left\{ \exp\left[W_{ij}'\theta\right]\right\} ^{y_{ij}}}{y_{ij}!}\label{eq: poisson_marginal}
\end{equation}
with $y_{ij}=0,1,2,\ldots$ and $W_{ij}\overset{def}{\equiv}w\left(X_{i},X_{j}\right)$
a known $J\times1$ vector of functions of $X_{i}$ and $X_{j}$.
As an example, if $X_{i}=\left(\ln\mathtt{GDP_{i}},\mathtt{LAT_{i}},\mathtt{LONG_{i}}\right)'$,
then setting
\[
W_{ij}=\left(\begin{array}{c}
\ln\mathtt{GDP_{i}}\\
\ln\mathtt{GDP_{j}}\\
\ln\left[\left(\mathtt{LAT_{i}}-\mathtt{LAT_{j}}\right)^{2}+\left(\mathtt{LONG_{i}}-\mathtt{LONG_{j}}\right)^{2}\right]^{1/2}
\end{array}\right)
\]
results in a basic gravity trade model specification.\footnote{In practice distance is measured using the so-called great circle
formula; which accounts for the curvature of the Earth's surface.}

Similar to \citet{Russett_Oneal_Book01}, a researcher might model
the conditional probability that country $i$ attacks country $j$
using logistic regression such that
\begin{equation}
f_{\left.Y_{12}\right|X_{1},X_{2}}\left(\left.y_{ij}\right|X_{i},X_{j};\theta\right)=\left[F\left(W_{ij}'\theta\right)\right]^{y_{ij}}\left[1-F\left(W_{ij}'\theta\right)\right]^{1-y_{ij}}\label{eq: logit_marginal}
\end{equation}
with $y_{ij}=0,1$ and $F\left(W_{ij}'\theta\right)=\exp\left(W_{ij}'\theta\right)/\left[1+\exp\left(W_{ij}'\theta\right)\right]$.

An important feature of both (\ref{eq: poisson_marginal}) and (\ref{eq: logit_marginal})
is that they only specify the marginal distribution of $Y_{ij}$ given
$X_{i}$ and $X_{j}$. The econometrician is not asserting, for example,
that
\[
f_{\left.Y_{12},Y_{13}\right|X_{1},X_{2},X_{3}}\left(\left.y_{12},y_{13}\right|X_{1},X_{2},X_{3};\theta\right)=f_{\left.Y_{12}\right|X_{1},X_{2}}\left(\left.y_{12}\right|X_{1},X_{2};\theta\right)f_{\left.Y_{12}\right|X_{1},X_{2}}\left(\left.y_{13}\right|X_{1},X_{3};\theta\right),
\]
since doing so would imply independence of $Y_{12}$ and $Y_{13}$
given covariates; but such dependence is precisely the complication
under consideration. Formulating a conditional likelihood for the
entire adjacency matrix $\mathbf{Y}\overset{def}{\equiv}\left[Y_{ij}\right]_{1\leq i,j\leq N,i\neq j}$
given $\mathbf{X}\overset{def}{\equiv}\left[X_{i}\right]_{1\leq i\leq N}$
would require an explicit specification of the dependence structure
across dyads sharing agents in common. In contrast $f_{\left.Y_{12}\right|X_{1},X_{2}}\left(\left.Y_{12}\right|X_{1},X_{2};\theta\right)$,
which is a model for the marginal distribution of $Y_{12}$ alone,
does not require modeling this dependence.

Let $l_{ij}\left(\theta\right)=\ln f_{\left.Y_{12}\right|X_{1},X_{2}}\left(\left.Y_{ij}\right|X_{i},X_{j};\theta\right)$
and consider the estimator which chooses $\hat{\theta}$ to maximize:
\begin{equation}
\hat{\theta}=\arg\underset{\theta\in\Theta}{\max}\frac{1}{N\left(N-1\right)}\sum_{i=1}^{N}\sum_{j\neq i}l_{ij}\left(\theta\right).\label{eq: composite_loglikelihood}
\end{equation}
Because its summands are not independent of one another -- at least
those sharing indices in common are not -- (\ref{eq: composite_loglikelihood})
does not correspond to a log-likelihood function for $\mathbf{Y}$
given $\mathbf{X}$. Instead it corresponds to what is sometimes called
a composite log-likelihood \citep[e.g.,][]{Lindsay_CM88,Cox_Reid_BM04}.
A composite likelihood ``is an inference function derived by multiplying
a collection of component likelihoods'' \citep[p. 5]{Varin_et_al_SS11}.
Although (\ref{eq: composite_loglikelihood}) fails to correctly represent
the dependence structure across the elements of the adjacency matrix,
if it is based upon a correctly specified marginal density, $\hat{\theta}$
generally will be consistent for $\theta_{0}$. This follows because
the derivative of composite log-likelihood is mean zero at $\theta=\theta_{0}$
under correct specification of its components. While an appropriately
specified composite log-likelihood typically delivers a valid estimating
equation, accurate inference is more challenging, since the unmodeled
dependence structure in the data does need to be explicitly taken
into account at the inference stage.

\subsection{Limit distribution}

Consider a mean value expansion of the first order condition associated
with the maximizer of (\ref{eq: composite_loglikelihood}).\footnote{A general result on consistency of $\hat{\theta}$ could be constructed
by adapting the results of \citet{Honore_Powell_JOE94}.} Such an expansion yields, after some re-arrangement,
\[
\sqrt{N}\left(\hat{\theta}-\theta_{0}\right)=\left[-H_{N}\left(\bar{\theta}\right)\right]^{+}\sqrt{N}S_{N}\left(\theta_{0}\right)
\]
with $\bar{\theta}$ a mean value between $\hat{\theta}$ and $\theta_{0}$
which may vary from row to row and the $+$ superscript denoting a
Moore-Penrose inverse. Here $S_{N}\left(\theta_{0}\right)$ is the
``score'' vector
\begin{equation}
S_{N}\left(\theta\right)=\frac{1}{N}\frac{1}{N-1}\sum_{i}\sum_{j\neq i}s_{ij}\left(Z_{ij},\theta\right)\label{eq: S_N}
\end{equation}
with $s\left(Z_{ij},\theta\right)=\partial l_{ij}\left(\theta\right)/\partial\theta$
for $Z_{ij}=\left(Y_{ij},X_{i}',X_{j}'\right)'$ and $H_{N}\left(\theta\right)=\frac{1}{N}\frac{1}{N-1}\sum_{i}\sum_{j\neq i}\frac{\partial^{2}l_{ij}\left(\theta\right)}{\partial\theta\partial\theta'}$.
If the Hessian matrix $H_{N}\left(\bar{\theta}\right)$ converges
in probability to the invertible matrix $\Gamma_{0}$, as I will assume,
then
\[
\sqrt{N}\left(\hat{\theta}-\theta_{0}\right)=-\Gamma_{0}^{-1}\sqrt{N}S_{N}\left(\theta_{0}\right)+o_{p}\left(1\right)
\]
so that the asymptotic sampling properties of $\sqrt{N}\left(\hat{\theta}_{\mathrm{}}-\theta_{0}\right)$
will be driven by the behavior of $\sqrt{N}S_{N}\left(\theta_{0}\right)$.

As with the composite log-likelihood criterion function, the summands
of $\sqrt{N}S_{N}\left(\theta_{0}\right)$ are not independent of
one another \citep[cf.,][]{Cameron_Golotvina_WP05,Fafchamp_Gubert_JDE07}.
A standard central limit theorem cannot be used to demonstrate asymptotic
normality of $\sqrt{N}S_{N}\left(\theta_{0}\right)$. Fortunately
$S_{N}\left(\theta_{0}\right)$, although not a U-Statistic, has a
dependence structure similar to one. This insight can be used to derive
the limit properties of $\sqrt{N}\left(\hat{\theta}-\theta_{0}\right)$.

Begin by re-writing $S_{N}\left(\theta_{0}\right)$ as
\begin{equation}
S_{N}=\tbinom{N}{2}^{-1}\sum_{i=1}^{N-1}\sum_{j=i+1}^{N}\frac{s_{ij}+s_{ji}}{2},\label{eq: S_N_U_statistic_like}
\end{equation}
where $s_{ij}\overset{def}{\equiv}s\left(Z_{ij},\theta_{0}\right)$
and $S_{N}\overset{def}{\equiv}S_{N}\left(\theta_{0}\right).$ While
(\ref{eq: S_N_U_statistic_like}) has the cursory appearance of a
U-Statistic it is, in fact not one: $Y_{ij}$, which enters $s_{ij}$,
varies at the dyad level, hence $S_{N}$ is not a function of $N$
i.i.d. random variables.

Let $\mathbf{U}=\left[U_{i}\right]_{1\leq i\leq N}$; the projection
of $S_{N}$ onto the \emph{observed} covariate matrix $\mathbf{X}$
and the \emph{unobserved} vector of unit-specific effects $\mathbf{U}$
equals:
\begin{equation}
V_{N}\overset{def}{\equiv}\mathbb{E}\left[\left.S_{N}\right|\mathbf{X},\mathbf{U}\right]=\tbinom{N}{2}^{-1}\sum_{i<j}\frac{\bar{s}_{ij}+\bar{s}_{ji}}{2}\label{eq: first_projection}
\end{equation}
with $\bar{s}_{ij}\overset{def}{\equiv}\bar{s}\left(X_{i},U_{i},X_{j},U_{j}\right)$
and $\bar{s}\left(X_{i},U_{i},X_{j},U_{j}\right)\overset{def}{\equiv}\mathbb{E}\left[\left.s\left(Z_{ij},\theta_{0}\right)\right|X_{i},U_{i},X_{j},U_{j}\right].$
The expression to the right of the equality in (\ref{eq: first_projection})
follows from the \citet{Crane_Towsner_JSL18} representation (\ref{eq: dyadic_regression_DGP})
and independence of $V_{ij}$ from $\left(\mathbf{X},\mathbf{U}\right)$.

An important observation is that the projection (\ref{eq: first_projection})
\emph{is} a U-statistic of order two: specifically it is a summation
over all $\tbinom{N}{2}$ dyads that can be formed from the i.i.d.
sample $\left\{ \left(X_{i},U_{i}\right)\right\} _{1\leq i\leq N}$.
Unusually our U-statistic is defined in terms of a combination of
both \emph{observed} $\left\{ X_{i}\right\} _{1\leq i\leq N}$ and
\emph{unobserved} $\left\{ U_{i}\right\} _{1\leq i\leq N}$ random
variables.

The projection error $T_{N}=S_{N}-V_{N}$ consists of a summation
of $\tbinom{N}{2}$ conditionally uncorrelated summands; hence $\mathbb{V}\left(T_{N}\right)=\tbinom{N}{2}^{-1}\mathbb{E}\left(\mathbb{V}\left(\left.\frac{s_{12}+s_{21}}{2}\right|X_{1},U_{1},X_{2},U_{2}\right)\right)=O\left(N^{-2}\right)$
(as long as $\mathbb{V}\left(\left.\frac{s_{12}+s_{21}}{2}\right|X_{1},U_{1},X_{2},U_{2}\right)$
does not change as $N\rightarrow\infty$). We also have that $T_{N}$
and $V_{N}$ are uncorrelated by construction.

Although we cannot numerically compute $V_{N}$ -- even if $\theta_{0}$
is known -- because the $\left\{ U_{i}\right\} _{1\leq i\leq N}$
are unobserved, we can use the theory of U-statistics to characterize
its sampling properties as $N\rightarrow\infty$. Decomposing $V_{N}$
into a Hájek projection and a second remainder term yields \citep[e.g.,][]{Lehmann_LSTBook99,vanderVaart_ASBook00}:
\[
V_{N}=V_{1N}+V_{2N}
\]
where, \label{def: ego_alter_projections}defining $\bar{s}^{e}\left(x,u\right)=\mathbb{E}\left[\bar{s}\left(x,u,X_{1},U_{1}\right)\right]$
and $\bar{s}^{a}\left(x,u\right)=\mathbb{E}\left[\bar{s}\left(X_{1},U_{1},x,u\right)\right]$,
\begin{align}
V_{1N}= & \frac{2}{N}\sum_{i=1}^{N}\left\{ \frac{\bar{s}_{1}^{e}\left(X_{i},U_{i}\right)+\bar{s}_{1}^{a}\left(X_{i},U_{i}\right)}{2}\right\} \label{eq: hajek_projection}\\
V_{2N}= & \tbinom{N}{2}^{-1}\sum_{i<j}\left\{ \frac{\bar{s}_{ij}+\bar{s}_{ji}}{2}\right.\nonumber \\
 & \left.-\frac{\bar{s}_{1}^{e}\left(X_{i},U_{i}\right)+\bar{s}_{1}^{a}\left(X_{i},U_{i}\right)}{2}-\frac{\bar{s}_{1}^{e}\left(X_{j},U_{j}\right)+\bar{s}_{1}^{a}\left(X_{j},U_{j}\right)}{2}\right\} .\label{eq: hajek_projection_error}
\end{align}
The superscript `e' denotes `ego' since it is the ego unit's attributes
which are being held fixed in the average used to compute $\bar{s}^{e}\left(x,u\right)$.
Similarly the `a' denotes `alter', since it is that unit's attributes
which are held fixed when defining $\bar{s}^{a}\left(x,u\right)$.
Conveniently $V_{1N}$ is a sum of i.i.d. random variables to which,
after scaling by $\sqrt{N}$, a CLT may be applied. Furthermore it
can be shown that $\mathbb{V}\left(V_{2N}\right)=O\left(N^{-2}\right)$.

Putting these results together yields the asymptotically linear representation
\begin{align*}
\sqrt{N}\left(\hat{\theta}-\theta_{0}\right)= & -\Gamma_{0}^{-1}\sqrt{N}\left(V_{1N}+V_{2N}+T_{N}\right)+o_{p}\left(1\right)\\
= & -\Gamma_{0}^{-1}\sqrt{N}V_{1N}+o_{p}\left(1\right)\\
= & -\Gamma_{0}^{-1}\frac{2}{\sqrt{N}}\sum_{i=1}^{N}\left\{ \frac{\bar{s}_{1}^{e}\left(X_{i},U_{i}\right)+\bar{s}_{1}^{a}\left(X_{i},U_{i}\right)}{2}\right\} +o_{p}\left(1\right),
\end{align*}
and hence a limit distribution for $\sqrt{N}\left(\hat{\theta}-\theta_{0}\right)$
of
\begin{equation}
\sqrt{N}\left(\hat{\theta}-\theta_{0}\right)\overset{D}{\rightarrow}\mathcal{N}\left(0,4\left(\Gamma_{0}'\Sigma_{1}^{-1}\Gamma_{0}\right)^{-1}\right)\label{eq: dyadic_regression_limit_distribution}
\end{equation}
where $\Sigma_{1}=\mathbb{V}\left(\frac{\bar{s}_{1}^{e}\left(X_{1},U_{1}\right)+\bar{s}_{1}^{a}\left(X_{1},U_{1}\right)}{2}\right).$
Although $S_{N}$ is not a U-statistic, under the assumptions maintained
here, its limit distribution coincides with that of $V_{N}$ (which
is a U-statistic).

Before turning to practicalities of variance estimation I will present
a useful property of the kernel entering the Hájek Projection, $V_{1N}$
above. By the usual conditional mean zero property of the score function
we have that $\mathbb{E}\left[\left.s\left(Z_{12};\theta_{0}\right)\right|X_{1}=x_{1},X_{2}=x_{2}\right]=0$
as long as marginal density of $Y_{12}$ given $X_{1}$ and $X_{2}$
is correctly specified. This property can be used to show that the
averages, $\bar{s}_{1}^{e}\left(X_{1},U_{1}\right)$ and $\bar{s}_{1}^{a}\left(X_{1},U_{1}\right)$,
are also conditionally mean zero. Taking the former we have that
\begin{align*}
\mathbb{E}\left[\left.\bar{s}_{1}^{e}\left(X_{1},U_{1}\right)\right|X_{1}=x_{1}\right] & =\int\left[\int\int\bar{s}\left(x_{1},u_{1},x_{2},u_{2}\right)f_{X_{1},U_{1}}\left(x_{2},u_{2}\right)\mathrm{d}x_{2}\mathrm{d}u_{2}\right]f_{\left.U_{1}\right|X_{1}}\left(\left.u_{1}\right|x_{1}\right)\mathrm{d}u_{1}\\
 & =\int\left[\int\int\bar{s}\left(x_{1},u_{1},x_{2},u_{2}\right)f_{\left.U_{1}\right|X_{1}}\left(\left.u_{2}\right|x_{2}\right)f_{X_{1}}\left(x_{2}\right)\mathrm{d}x_{2}\mathrm{d}u_{2}\right]f_{\left.U_{1}\right|X_{1}}\left(\left.u_{1}\right|x_{1}\right)\mathrm{d}u_{1}\\
 & =\int\left[\int\int\bar{s}\left(x_{1},u_{1},x_{2},u_{2}\right)f_{\left.U_{1}\right|X_{1}}\left(\left.u_{1}\right|x_{1}\right)\mathrm{d}u_{1}f_{\left.U_{1}\right|X_{1}}\left(\left.u_{2}\right|x_{2}\right)\mathrm{d}u_{2}\right]f_{X_{1}}\left(x_{2}\right)\mathrm{d}x_{2}\\
 & =\int\mathbb{E}\left[\left.\bar{s}\left(X_{1},U_{1},X_{2},U_{2}\right)\right|X_{1}=x_{1},X_{2}=x_{2}\right]f_{X_{1}}\left(x_{2}\right)\mathrm{d}x_{2}\\
 & =\int\mathbb{E}\left[\left.s\left(Z_{12};\theta_{0}\right)\right|X_{1}=x_{1},X_{2}=x_{2}\right]f_{X_{1}}\left(x_{2}\right)\mathrm{d}x_{2}\\
 & =0
\end{align*}
where the first equality follows from the definition of $\bar{s}_{1}^{e}\left(X_{1},U_{1}\right)$,
the third from a change in the order of integration, and the second
to last from iterated expectations. These calculations imply that
\[
\mathbb{E}\left[\left.\bar{s}_{1}^{e}\left(X_{1},U_{1}\right)\right|X_{1},X_{2}\right]=\mathbb{E}\left[\left.\bar{s}_{1}^{a}\left(X_{1},U_{1}\right)\right|X_{1},X_{2}\right]=0
\]
and hence that $\left[\bar{s}_{1}^{e}\left(X_{1},U_{1}\right)+\bar{s}_{1}^{a}\left(X_{1},U_{1}\right)\right]/2$
is conditionally mean-zero given $X_{1}$ and $X_{2}$. This property
will be helpful for understanding the asymptotic precision of estimates
of various causal parameters introduced below.

\subsection{Variance estimation}

In order to conduct inference on the components of $\theta_{0}$,
an estimate of the variance of $\sqrt{N}\left(\hat{\theta}-\theta_{0}\right)$
is required. Although the distribution theory outlined above is novel\footnote{See \citet{Tabord-Meehan_JBES18}, \citet{Davezies_et_al_arXiv2019}
and, especially, \citet{Menzel_arXiv17} for related independent work.}, the history of variance estimation for ``dyadic statistics'' goes
back at least to \citet{Holland_Leinhardt_SM76}. In economics, a
variance estimator first proposed by \citet{Fafchamp_Gubert_JDE07},
is widely -- although not universally -- used for dyadic regression
analysis. In order to understand extant approaches to variance estimation,
as well as to propose new ones, it is helpful to examine the structure
of $S_{N}$'s variance in detail.

The arguments used to derive the limit distribution of $\sqrt{N}\left(\hat{\theta}-\theta_{0}\right)$
above suggest that it may be insightful to think about the variance
of $S_{N}$ in terms of the ANOVA decomposition
\begin{align}
\mathbb{V}\left(S_{N}\right) & =\mathbb{V}\left(\mathbb{E}\left[\left.S_{N}\right|\mathbf{X},\mathbf{U}\right]\right)+\mathbb{E}\left[\mathbb{V}\left(\left.S_{N}\right|\mathbf{X},\mathbf{U}\right)\right]\nonumber \\
 & =\mathbb{V}\left(V_{N}\right)+\mathbb{V}\left(T_{N}\right)\nonumber \\
 & =\mathbb{V}\left(V_{1N}\right)+\mathbb{V}\left(V_{2N}\right)+\mathbb{V}\left(T_{N}\right),\label{eq: Var_S_N}
\end{align}
where the second and third equalities follow from the decomposition
for $S_{N}$ developed in the previous subsection.

Let $p=1,2$ equal the number of agents dyads $\left\{ i_{1},i_{2}\right\} $
and $\left\{ j_{1},j_{2}\right\} $ share common and define the matrix
$\Sigma_{p}$ as
\begin{align}
\Sigma_{p} & \overset{def}{\equiv}\mathbb{C}\left(\frac{\bar{s}\left(X_{i_{1}},U_{i_{1}},X_{i_{2}},U_{i_{2}}\right)+\bar{s}\left(X_{i_{2}},U_{i_{2}},X_{i_{1}},U_{i_{1}}\right)}{2}\right.,\label{eq: SIGMAp}\\
 & \left.\frac{\bar{s}\left(X_{j_{1}},U_{j_{1}},X_{j_{2}},U_{j_{2}}\right)'+\bar{s}\left(X_{j_{2}},U_{j_{2}},X_{j_{1}},U_{j_{1}}\right)'}{2}\right)\nonumber \\
 & =\mathbb{C}\left(\mathbb{E}\left[\left.\frac{s_{i_{1}i_{2}}+s_{i_{2}i_{1}}}{2}\right|X_{i_{1}},U_{i_{1}},X_{i_{2}},U_{i_{2}}\right],\mathbb{E}\left[\left.\frac{s_{j_{1}j_{2}}+s_{j_{2}j_{1}}}{2}\right|X_{j_{1}},U_{j_{1}},X_{j_{2}},U_{j_{2}}\right]'\right).\nonumber 
\end{align}
When $p=1$ we have
\begin{align*}
\Sigma_{1} & =\mathbb{C}\left(\mathbb{E}\left[\left.\frac{s_{12}+s_{21}}{2}\right|X_{1},U_{1},X_{2},U_{2}\right],\mathbb{E}\left[\left.\frac{s_{13}+s_{31}}{2}\right|X_{1},U_{1},X_{3},U_{3}\right]'\right)\\
 & =\mathbb{C}\left(\frac{s_{12}+s_{21}}{2},\frac{s_{13}'+s_{31}'}{2}\right),
\end{align*}
with the second equality an implication of conditional independence
of $\frac{s_{12}+s_{21}}{2}$ and $\frac{s_{13}+s_{31}}{2}$ given
$\left(X_{1},X_{2},X_{3}\right)$ and $\left(U_{1},U_{2},U_{3}\right)$.
Hence $\Sigma_{1}$ equals the covariance between any pair of summands
in $S_{N}$ -- see equation (\ref{eq: S_N}) above -- sharing an
index in common. There are many such pairs of summands in $S_{N}$
(actually $\frac{1}{2}N\left(N-1\right)\left(N-2\right)$ such pairs-of-dyads).
It is the preponderance of these non-zero covariances that drives
their importance for understanding the sampling distribution of $\sqrt{N}\left(\hat{\theta}-\theta_{0}\right)$.

In a slight abuse of notation, additionally define the matrix
\begin{equation}
\Sigma_{3}\overset{def}{\equiv}\mathbb{E}\left[\mathbb{V}\left(\left.\frac{s_{12}+s_{21}}{2}\right|X_{1},U_{1},X_{2},U_{2}\right)\right].\label{eq: SIGMA3}
\end{equation}
Calculations analogous to those use in variance analyses for U-statistics
\citep[e.g.,][]{Hoeffding_AMS48,Lehmann_LSTBook99} yield
\begin{align}
\mathbb{V}\left(V_{1N}\right) & =\frac{4\Sigma_{1}}{N}\label{eq: Var_V1_N}\\
\mathbb{V}\left(V_{2N}\right) & =\frac{2}{N\left(N-1\right)}\left(\Sigma_{2}-2\Sigma_{1}\right)\label{eq: Var_V2_N}\\
\mathbb{V}\left(T_{N}\right) & =\frac{2}{N\left(N-1\right)}\Sigma_{3},\label{eq: Var_T_N}
\end{align}
such that, defining the notation $\Omega\overset{def}{\equiv}\mathbb{V}\left(\sqrt{N}S_{N}\right),$
from (\ref{eq: Var_S_N}), (\ref{eq: Var_V1_N}), (\ref{eq: Var_V2_N})
and (\ref{eq: Var_T_N}):
\begin{equation}
\Omega=4\Sigma_{1}+\frac{2}{N-1}\left(\Sigma_{2}+\Sigma_{3}-2\Sigma_{1}\right).\label{eq: OMEGA}
\end{equation}

Consistent with the form of the limit distribution given in (\ref{eq: dyadic_regression_limit_distribution}),
the variances of $V_{2N}$ and $T_{N}$ are of smaller order. Although
the contribution of these terms to the variance of $\sqrt{N}S_{N}$
is asymptotically negligible, their contribution for finite $N$ need
not be. As alluded to earlier, the appearance of the covariance $\Sigma_{1}$
as the leading term in (\ref{eq: OMEGA}) reflects the large number
of non-zero covariance terms that arise when the variance operator
is applied to the sum $S_{N}=\tbinom{N}{2}^{-1}\sum_{i=1}^{N-1}\sum_{j=i+1}^{N}\frac{s_{ij}+s_{ji}}{2}$.
In practice, especially if $h\left(x_{1},x_{2},u_{1},u_{2},v_{12}\right)$
is nearly constant in $u_{1}$ and $u_{2}$, $\Sigma_{1}$ may be
small in magnitude. In such cases it may be that $4\Sigma_{1}$ and
$\frac{2}{N-1}\left(\Sigma_{2}+\Sigma_{3}-2\Sigma_{1}\right)$ are
of comparable magnitude for modest $N$. Using a variance estimator
which includes estimates of both these terms may therefore result
in tests with better size and power properties \citep[cf.,][]{Hoeffding_AMS48,Graham_Imbens_Ridder_QE14,Cattaneo_et_al_ET14}.
To construct such an estimator I propose using analog estimates of
the terms appearing to the right of the equality in (\ref{eq: OMEGA}).

\subsubsection*{A benchmark analog variance estimate}

A natural analog estimate of $\Sigma_{1}$, the leading variance term,
is
\begin{align}
\hat{\Sigma}_{1}= & \tbinom{N}{3}^{-1}\sum_{i=1}^{N-2}\sum_{j=i+1}^{N-1}\sum_{k=j+1}^{N}\frac{1}{3}\left\{ \left(\frac{\hat{s}_{ij}+\hat{s}_{ji}}{2}\right)\left(\frac{\hat{s}_{ik}+\hat{s}_{ki}}{2}\right)'\right.\nonumber \\
 & \left.\left(\frac{\hat{s}_{ij}+\hat{s}_{ji}}{2}\right)\left(\frac{\hat{s}_{jk}+\hat{s}_{kj}}{2}\right)'+\left(\frac{\hat{s}_{ik}+\hat{s}_{ki}}{2}\right)\left(\frac{\hat{s}_{jk}+\hat{s}_{kj}}{2}\right)'\right\} ,\label{eq: SIGMA1_hat}
\end{align}
with $\hat{s}_{ij}\overset{def}{\equiv}s\left(Z_{ij},\hat{\theta}\right)$.
Equation (\ref{eq: SIGMA1_hat}) is a summation over all $\tbinom{N}{3}=\frac{1}{6}N\left(N-1\right)\left(N-3\right)$
triads in the dataset. Each triad $ijk$ can be further divided into
three pairs of dyads, $\left\{ ij,ik\right\} $, $\left\{ ij,jk\right\} $
and $\left\{ ik,jk\right\} $, with each such pair sharing exactly
one agent in common. Equation (\ref{eq: SIGMA1_hat}) corresponds
to the sample covariance of $\left(\hat{s}_{ij}+\hat{s}_{ji}\right)/2$
and $\left(\hat{s}_{ik}+\hat{s}_{ki}\right)/2$ across these $3\tbinom{N}{3}$
pairs of dyads.

To construct an estimate of $\mathbb{V}\left(\sqrt{N}S_{N}\right)$
separate estimates of $\Sigma_{2}$ and $\Sigma_{3}$ are not required,
only their sum is needed. Using an ANOVA decomposition we can express
this sum as
\begin{align*}
\Sigma_{2}+\Sigma_{3} & =\mathbb{V}\left(\mathbb{E}\left[\left.\frac{s_{12}+s_{21}}{2}\right|X_{1},U_{1},X_{2},U_{2}\right]\right)+\mathbb{E}\left[\mathbb{V}\left(\left.\frac{s_{12}+s_{21}}{2}\right|X_{1},U_{1},X_{2},U_{2}\right)\right]\\
 & =\mathbb{V}\left(\frac{s_{12}+s_{21}}{2}\right).
\end{align*}
This suggests the analog estimate
\begin{equation}
\widehat{\Sigma_{2}+\Sigma_{3}}=\tbinom{N}{2}^{-1}\sum_{i=1}^{N-1}\sum_{j=i+1}^{N}\left(\frac{\hat{s}_{ij}+\hat{s}_{ji}}{2}\right)\left(\frac{\hat{s}_{ij}+\hat{s}_{ji}}{2}\right)'.\label{eq: SIGMA2_SIGMA3_hat}
\end{equation}

From (\ref{eq: OMEGA}), (\ref{eq: SIGMA1_hat}) and (\ref{eq: SIGMA2_SIGMA3_hat})
we get the variance estimate
\begin{equation}
\mathbb{\hat{V}}\left(\sqrt{N}\left(\hat{\theta}-\theta_{0}\right)\right)=\left(\hat{\Gamma}'\hat{\Omega}^{-1}\hat{\Gamma}\right)^{-1}\label{eq: Var_theta_hat}
\end{equation}
where
\begin{align}
\hat{\Gamma}= & H_{N}\left(\hat{\theta}\right)\label{eq: GAMMA_hat}\\
\hat{\Omega}= & 4\hat{\Sigma}_{1}+\frac{2}{N-1}\left(\widehat{\Sigma_{2}+\Sigma_{3}}-2\hat{\Sigma}_{1}\right).\label{eq: OMEGA_hat}
\end{align}

\subsubsection*{\citet{Fafchamp_Gubert_JDE07} variance estimate}

Just over a decade ago, \citet{Fafchamp_Gubert_JDE07} presented a
variance-covariance estimator for $\hat{\theta}_{\mathrm{}}$ that
they informally argued leads to asymptotically correct inference.\footnote{This estimator has been further explored by \citet{Cameron_Miller_WP14},
\citet{Aronow_et_al_PA17} and \citet{Tabord-Meehan_JBES18}.} They proposed estimating the variance of $\sqrt{N}S_{N}$ by
\begin{align}
\hat{\Omega}_{\mathrm{FG}}= & \frac{1}{N\left(N-1\right)^{2}}\sum_{i_{1}}\sum_{i_{2}\neq i_{1}}\sum_{j_{1}}\sum_{j_{2}\neq j_{1}}C_{i_{1}i_{2}j_{1}j_{2}}\hat{s}_{i_{1}i_{2}}\hat{s}_{j_{1}j_{2}}',\label{eq: AVar_Fafchampfs_Gubert}
\end{align}
where $C_{i_{1}i_{2}j_{1}j_{2}}=1$ if $i_{1}=j_{1}$, $i_{2}=j_{2}$,
$i_{1}=j_{2}$ or $i_{2}=j_{1}$ and zero otherwise (here the `or'
is inclusive).\footnote{My definition of $\hat{\Omega}_{\mathrm{FG}}$ actually differs slightly
from the one given by \citet{Fafchamp_Gubert_JDE07}, due to a finite
sample correction term introduced in the latter. Their expression
also appears to include a notational inconsistency with $N$ (apparently)
denoting both the number of agents as well as the number of dyads
(here $n=\frac{1}{2}N\left(N-1\right)$) in different components of
the expression. Once these typos are corrected (\ref{eq: AVar_Fafchampfs_Gubert})
agrees with their expression up to a finite sample correction.} Equation (\ref{eq: AVar_Fafchampfs_Gubert}) is a summation across
$\tbinom{N}{2}\times\tbinom{N}{2}$ ``pairs-of-pairs'' or pairs
of dyads. As noted above, there are $3\tbinom{N}{3}=\frac{1}{2}N\left(N-1\right)\left(N-2\right)$
unique pairs of dyads sharing one agent in common; but each of these
pairs of dyads is counted eight different times in (\ref{eq: AVar_Fafchampfs_Gubert}).
Likewise there are $\tbinom{N}{2}=\frac{1}{2}N\left(N-1\right)$ pairs
of dyads sharing both agents in common (i.e., straight up dyads) and
each of these pairs is counted four different times in (\ref{eq: AVar_Fafchampfs_Gubert}).
From this observation we have that
\begin{align*}
\hat{\Omega}_{\mathrm{FG}} & =\frac{1}{N\left(N-1\right)^{2}}\left[8\times\frac{1}{2}N\left(N-1\right)\left(N-2\right)\hat{\Sigma}_{1}+4\times\frac{1}{2}N\left(N-1\right)\widehat{\Sigma_{2}+\Sigma_{3}}\right]\\
 & =4\hat{\Sigma}_{1}+\frac{2}{N-1}\left(\widehat{\Sigma_{2}+\Sigma_{3}}-2\hat{\Sigma}_{1}\right),
\end{align*}
which exactly coincides with expression (\ref{eq: OMEGA_hat}) above.
Not only does $\hat{\Omega}_{\mathrm{FG}}$ include a consistent estimate
of the leading term in $\mathbb{V}\left(\sqrt{N}S_{N}\right)$, but
it also includes an estimate of the asymptotically negligible higher
order component.

\citet{Fafchamp_Gubert_JDE07} is widely-cited in the context of covariance
estimation by empirical researchers, with a STATA implementation for
linear and logistic dyadic regression freely available \citep[cf., ][]{Cameron_Miller_WP14}.
Consequently considerable practical experience and Monte Carlo evidence
on the properties of standard error estimates based on (\ref{eq: AVar_Fafchampfs_Gubert})
exists. Among empirical researchers, the consensus is that such standard
errors are often much larger than those based on the (possibly erroneous)
assumption of independence across dyads.

\subsubsection*{\citet{Snijders_Borgatti_C99} jackknife variance estimate}

\citet{Snijders_Borgatti_C99}, inspired by the prior work of \citet{Frank_Snijder_JOS94},
suggest\footnote{They actually propose a jackknife estimate for the variance of $S_{N}$.
I have multiplied their expression by $N$ to get the corresponding
expression for the variance of $\sqrt{N}S_{N}$ (see Equation (2)
of \citet{Snijders_Borgatti_C99}).} a jackknife variance estimate for $\mathbb{V}\left(\sqrt{N}S_{N}\right)$
of
\begin{equation}
\hat{\Omega}_{\mathrm{SB}}=\left(\frac{N-2}{2}\right)\sum_{i}\left[S_{N,-i}\left(\hat{\theta}\right)-S_{N}\left(\hat{\theta}\right)\right]\left[S_{N,-i}\left(\hat{\theta}\right)-S_{N}\left(\hat{\theta}\right)\right]',\label{eq: AVar_Snijders_Borgatti}
\end{equation}

where $S_{N,-i}\left(\theta\right)$ is the average of the dyadic
scores over the $\tbinom{N-1}{2}$ dyads which do not include agent
$i$:
\begin{align*}
S_{N,-i}\left(\theta\right)\overset{def}{\equiv} & \tbinom{N-1}{2}^{-1}\left[\sum_{j=1}^{N-1}\sum_{k=j+1}^{N}\frac{s\left(Z_{jk};\theta\right)+s\left(Z_{kj};\theta\right)}{2}-\sum_{l\neq i}\frac{s\left(Z_{il};\theta\right)+s\left(Z_{li};\theta\right)}{2}\right].
\end{align*}
The \citet{Snijders_Borgatti_C99} proposal, the basis of which they
acknowledge was primarily intuitive, does not provide a consistent
estimate of $\mathbb{V}\left(\sqrt{N}S_{N}\right)$, but, as I will
now show, a slight modification of their proposal does.

With some manipulation we can write, defining $\hat{\bar{s}}_{1i}\left(\theta\right)\overset{def}{\equiv}\frac{1}{N-1}\sum_{j\neq i}\frac{s\left(Z_{ij};\theta\right)+s\left(Z_{ji};\theta\right)}{2}$
(in a slight abuse of notation),
\begin{align}
S_{N,-i}\left(\theta\right)-S_{N}\left(\theta\right) & =\tbinom{N-1}{2}^{-1}\left[\tbinom{N}{2}S_{N}\left(\theta\right)-\left(N-1\right)\hat{\bar{s}}_{1i}\left(\theta\right)\right]-S_{N}\left(\theta\right)\nonumber \\
 & =-\frac{2}{N-2}\left[\hat{\bar{s}}_{1i}\left(\theta\right)-S_{N}\left(\theta\right)\right].\label{eq: jackknifed_score}
\end{align}

Observe that $\hat{\bar{s}}_{1i}\overset{def}{\equiv}\hat{\bar{s}}_{1i}\left(\hat{\theta}\right)$
would be the usual estimate of the the $i^{th}$ summand in the Hájek
projection given in (\ref{eq: hajek_projection}) above (see, for
example, \citet{Callaert_Veeraverbeke_AS81} or \citet{Cattaneo_et_al_ET14}
and the references therein). Indeed, on the basis of the limit theory
outlined above, a natural estimate of $\Sigma_{1}$ would be
\begin{equation}
\tilde{\Sigma}_{1}\overset{def}{\equiv}\frac{1}{N}\sum_{i=1}^{N}\hat{\bar{s}}_{1i}\hat{\bar{s}}_{1i}'.\label{eq: SIGMA1_jackknife}
\end{equation}

After some tedious manipulation it is possible to show that
\begin{align}
\tilde{\Sigma}_{1} & =\hat{\Sigma}_{1}+\frac{\widehat{\Sigma_{2}+\Sigma_{3}}-\hat{\Sigma}_{1}}{N-1}\label{eq: SIGMA1_tilde}\\
 & =\hat{\Sigma}_{1}+O_{p}\left(N^{-1}\right)\nonumber 
\end{align}
with $\hat{\Sigma}_{1}$ and $\widehat{\Sigma_{2}+\Sigma_{3}}$ as
defined in (\ref{eq: SIGMA1_hat}) and (\ref{eq: SIGMA2_SIGMA3_hat})
above.

Equations (\ref{eq: jackknifed_score}), (\ref{eq: SIGMA1_tilde})
and the observation that $S_{N}\left(\hat{\theta}\right)=0$ implies
that the jackknife variance estimate
\begin{align}
\hat{\Omega}_{\mathrm{JK}} & \overset{def}{\equiv}\frac{\left(N-2\right)}{N}^{2}\sum_{i}\left[S_{N,-i}\left(\hat{\theta}\right)-S_{N}\left(\hat{\theta}\right)\right]\left[S_{N,-i}\left(\hat{\theta}\right)-S_{N}\left(\hat{\theta}\right)\right]'\label{eq: Corrected_Jackknife}\\
 & =4\tilde{\Sigma}_{1}\nonumber \\
 & =4\hat{\Sigma}_{1}+\frac{4\left(\widehat{\Sigma_{2}+\Sigma_{3}}-\hat{\Sigma}_{1}\right)}{N-1},\nonumber 
\end{align}
provides a consistent estimate of the asymptotic variance $\sqrt{N}S_{N}$.

Furthermore, inspired by \citet{Efron_Stein_AS81} and, especially,
\citet{Cattaneo_et_al_ET14}, we can bias correct (\ref{eq: Corrected_Jackknife}):
\begin{equation}
\hat{\Omega}_{\mathrm{JK-BC}}\overset{def}{\equiv}\hat{\Omega}_{\mathrm{JK}}-\frac{2}{N-1}\left(\widehat{\Sigma_{2}+\Sigma_{3}}\right)=\hat{\Omega}_{\mathrm{FG}}\label{eq: awesomeness}
\end{equation}
with $\widehat{\Sigma_{2}+\Sigma_{3}}$ as defined by (\ref{eq: SIGMA2_SIGMA3_hat})
and the equality an implication of (\ref{eq: SIGMA1_tilde}). Equation
(\ref{eq: awesomeness}) implies that the \citet{Fafchamp_Gubert_JDE07}
variance estimator, or equivalently the analog estimator proposed
above, coincides with a bias corrected jackknife variance estimate.
This is awesome.

\subsection{Bootstrap inference}

Relative to analytic variance estimation, the theory of the bootstrap
for dyadic regression is comparatively less well-understood. Rewriting
our dyadic regression coefficient estimate in pseudo-U-Process form
yields
\[
\hat{\theta}=\arg\underset{\theta\in\Theta}{\max}\frac{2}{N\left(N-1\right)}\sum_{i=1}^{N-1}\sum_{j=i+1}^{N}\left\{ \frac{l_{ij}\left(\theta\right)+l_{ji}\left(\theta\right)}{2}\right\} .
\]

Next let $\left\{ V_{i}^{b}\right\} _{i=1}^{N}$ be a sequence of
i.i.d. mean one random weights independent of the data. One such sequence
is drawn for each of $b=1,\ldots,B$ bootstrap replications. In the
$b^{th}$ such replication we compute
\begin{align*}
\hat{\theta}_{b} & =\arg\underset{\theta\in\Theta}{\max}\frac{2}{N\left(N-1\right)}\sum_{i=1}^{N-1}\sum_{j=i+1}^{N}V_{i}^{b}V_{j}^{b}\left\{ \frac{l_{ij}\left(\theta\right)+l_{ji}\left(\theta\right)}{2}\right\} .
\end{align*}
The bootstrap distribution $\left\{ \hat{\theta}_{b}\right\} _{b=1}^{B}$
can then be used to approximate the sampling distribution of $\hat{\theta}$.
Letting $V_{i}^{b}$ be an exponential random variable with rate parameter
$1$ results in a Bayesian bootstrap which is, of course, preferred.
The above algorithm was proposed in the context of U-statistics by
\citet{Janssen_JSPI94}. If we let $V_{i}^{b}$ equal the number of
times agent $i$ is sampled from the set $\left\{ 1,\ldots,N\right\} $
across $N$ draws with replacement, we get the proposal of \citet{Davezies_et_al_arXiv2019},
who show -- under certain assumptions -- validity for the dyadic
regression case considered here.

\citet{Snijders_Borgatti_C99} proposed a bootstrap procedure for
jointly exchangeable random arrays which is very close to the proposal
of \citet{Davezies_et_al_arXiv2019}. As with their jackknife variance
estimator, their development was intuitive and informal. For simplicity
consider the application of their proposal for inference on the dyadic
mean $\bar{Y}=\frac{2}{N\left(N-1\right)}\sum_{i=1}^{N-1}\sum_{j=i+1}^{N}\left\{ \frac{Y_{ij}+Y_{ji}}{2}\right\} $.
Let $i_{1}^{b},\ldots,i_{N}^{b}$ be $N$ indices drawn uniformly
at random (with replacement) from $\left\{ 1,\ldots,N\right\} $.
Let $\mathbf{Y}^{b}$ be the adjacency matrix induced by $\left\{ i_{1}^{b},\ldots,i_{N}^{b}\right\} $.
If agent $j$ is sampled twice, say $i_{1}^{b}=j$ and $i_{2}^{b}=j$
we face the practical problem that the outcome $Y_{i_{1}^{b}i_{2}^{b}}=Y_{jj}$
is undefined. \citet{Snijders_Borgatti_C99} propose filling in such
cells with independent random draws from $\left\{ Y_{12},Y_{21},\ldots,Y_{N-1N},Y_{NN-1}\right\} $;
they note that the expected fraction of bootstrap dyads constructed
from a single underlying agent in the original sample will be vanishingly
small as $N\rightarrow\infty$ (suggesting that this problem may not
matter for asymptotic properties). Snjiders and Borgatti's \citeyearpar{Snijders_Borgatti_C99}
algorithm essentially coincides with the pigeon-hole bootstrap proposed
by \citet{Owen_AAS07} for separately exchangeable random arrays (in
which the problem of ``zero diagonals'' does not arise).

A final bootstrap procedure is proposed by \citet{Menzel_arXiv17}.
He is particularly concerned with formulating a procedure that adaptively
handles the possibility that there is, in fact, no dyadic correlation
in the data (i.e., $\Sigma_{1}=0$). Degeneracy of this type occurs,
in our regression setting, when the graphon $h\left(x_{1},x_{2},u_{1},u_{2},v_{12}\right)$
is constant in both $u_{1}$ and $u_{2}$ (but also in more exotic
situations where there is dyadic dependence in higher order moments,
but no correlation). The arguments in \citet{Menzel_arXiv17} suggest
that the weighted bootstrap of \citet{Janssen_JSPI94} and \citet{Davezies_et_al_arXiv2019}
will be inconsistent under degeneracy.

\citet{Menzel_arXiv17} proposes several different bootstraps; what
I sketch here is a simplified version of his `BS-N' procedure (adapted
to the dyadic regression case). Let
\[
\hat{s}_{i}^{e}=\frac{1}{N-1}\sum_{j=1}^{N}\hat{s}_{ij},\thinspace\thinspace\hat{s}_{j}^{a}=\frac{1}{N-1}\sum_{i=1}^{N}\hat{s}_{ij}
\]
be estimates of the average dyadic score for `ego' $i$ and `alter'
$j$. Let
\[
\hat{e}_{ij}=\hat{s}_{ij}-\hat{s}_{i}^{e}-\hat{s}_{j}^{a}
\]
equal the residual for $\hat{s}_{ij}$ after subtracting off these
ego and alter means. \citet{Menzel_arXiv17} actually suggests subtracting
off rescaled versions of $\hat{s}_{i}^{e}$ and $\hat{s}_{j}^{a}$
when forming $\hat{e}_{ij}$. Rescaling improves the accuracy of his
procedure when dyadic correlation is, in fact, absent. I omit this
detail since describing it requires introducing substantial additional
notation. The stylized version sketched here will be conservative
under degeneracy (similar to the pigeonhole bootstrap).

Let $\left\{ V_{i}^{b}\right\} _{i=1}^{N}$ be a sequence of i.i.d.
mean \emph{zero} random weights with unit variance (and unit third
moment as well). Let $i_{1}^{b},\ldots,i_{N}^{b}$ be $N$ indices
drawn uniformly at random (with replacement) from $\left\{ 1,\ldots,N\right\} $.
For all $\tbinom{N}{2}$ dyads induced by the $b^{th}$ such bootstrap
sample construct the scores

\[
\hat{s}_{i_{j}^{b}i_{k}^{b}}=\hat{s}_{i_{j}^{b}}^{e}+\hat{s}_{i_{k}^{b}}^{a}+V_{i_{j}^{b}}^{b}V_{i_{k}^{b}}^{b}\hat{e}_{i_{j}^{b}i_{k}^{b}},\ j=1,\ldots,N-1\ \&\ k=j+1,\ldots,N
\]
and compute their mean
\[
\hat{S}_{N}^{b}=\frac{2}{N\left(N-1\right)}\sum_{j=1}^{N-1}\sum_{k=j+1}^{N}\left(\frac{\hat{s}_{i_{j}^{b}i_{k}^{b}}+\hat{s}_{i_{k}^{b}i_{j}^{b}}}{2}\right).
\]
The variance of $\sqrt{N}\hat{S}_{N}^{b}$ across the $b=1,\ldots,B$
bootstrap replications can be used to construct an estimate of $\Omega=\mathbb{V}\left(\sqrt{N}S_{N}\right)$.
Menzel's \citeyearpar{Menzel_arXiv17} preferred procedures involve
an additional ``model selection'' step, not described here, as well
as pivotizing using $\hat{\Omega}_{\mathrm{FG}}=\hat{\Omega}_{\mathrm{JK-BC}}$.

\subsection{Further reading and open questions}

A special case of the \citet{Fafchamp_Gubert_JDE07} variance estimator
was first proposed by \citet{Holland_Leinhardt_SM76} in the context
of inference on network density; the equivalent of the dyadic mean
$\mu_{Y_{12}}=\mathbb{E}\left[Y_{12}\right]$ here (estimated by $\bar{Y}=\frac{1}{N\left(N-1\right)}\sum_{i=1}^{N}\sum_{j\neq i}Y_{ij}$).
The \citet{Holland_Leinhardt_SM76} variance estimate was used with
some regularity in empirical social network analysis in the 1980s
and 1990s \citep[cf.,][]{Wasserman_Faust_Bk94}. The reference distribution
was assumed to be normal, but no proof for this was available. \citet{Bickel_et_al_AS11}
appear to be the first to have shown asymptotic normality of $\sqrt{N}\left(\bar{Y}-\mu_{Y_{12}}\right)$
under dyadic dependence. The double projection argument used to produce
the $S_{N}=V_{1N}+V_{2N}+T_{N}$ decomposition used above is implicit
in their work. A similar decomposition was used by \citet{Graham_EM17}
to show asymptotic normality of the Tetrad Logit estimator, which
is described further below. The bootstrap procedure of \citet{Menzel_arXiv17}
is also based upon this decomposition. \citet{Tabord-Meehan_JBES18}
demonstrates asymptotic normality of dyadic regression coefficients
estimated by ordinary least squares. His method of proof is very different
from the argument outlined here.

\citet{Cameron_Miller_WP14}, \citet{Aronow_et_al_PA17} and \citet{Tabord-Meehan_JBES18}
provide further results on variance estimation for dyadic regression;
each building upon the proposal of \citet{Fafchamp_Gubert_JDE07}.

\citet{Menzel_arXiv17} and \citet{Davezies_et_al_arXiv2019} provide
large sample theory in some generality -- including for cases not
covered here. Both these papers provide formal results on inference
using the bootstrap as well. The presentation here is based upon \citet{Graham_DyadLectureNotes18},
a revised and expanded version of which appears as a chapter in \citet{Graham_dePaula_Bk20}.

When dyadic correlation is weak limit theory can be non-standard.
\citet{Menzel_arXiv17} provides examples and discussion. Related
issues arise in \citet{Graham_Niu_Powell_WP2019}, who study nonparametric
density and regression estimation with dyadic data. Developing inference
procedures with good properties across a range of (dyadic) data generating
processes remains largely open.

Open research problems include extending the material summarized here
to accommodate regressor endogeneity and settings where the number
of regressors is comparable to, or even exceeds, the number of agents
(or dyads).

\section{\label{sec: policy_analysis}Policy analysis}

One motivation for Tinbergen's \citeyearpar{Tinbergen_SWE62} dyadic
regression analysis was to evaluate the effect of preferential trade
agreements on export flows. \citet{Rose_AER04} explores the related
question of whether membership in the General Agreement on Trade and
Tariffs (GATT) or its successor, the World Trade Organization (WTO),
promoted trade (see also \citet{Rose_RIE05}). \citet{Baldwin_Taglioni_JEI07}
and \citet{SantosSilva_Tenreyro_AR10} use gravity models to assess
whether common currency zones, such as the Eurozone, promote trade.
As with conventional regression analysis, a desire to assess different
programs or policies underlies many dyadic regression analyses.\footnote{Other examples of recent attempts to reason about causal questions
with dyadic data include Schwartz and Sommers' \citeyearpar{Schwartz_Sommers_HA14}
and Goodman's \citeyearpar{Goodman_JPAM17} analyses of whether Medcaid
expansion states experienced in-migration from neighboring states
which chose to forgo the Affordable Care Act\textquoteright s expansion
of Medicaid and Mayda's \citeyearpar{Mayda_JPopE10} and Oretega and
Peri's \citeyearpar{Ortega_Peri_MS13} studies of the relationship
between immigration entry tightness and cross-country migration.}

While the logic and mechanics of program evaluation are well understood
in the context of single agent models \citep[cf.,][]{Heckman_Vytlacil_HBE07,Imbens_Wooldridge_JEL09},
a comparable framework for causal reasoning is not, to my knowledge,
available in the dyadic setting considered here. In this section I
make a start at formulating such a framework. In doing so I will attempt
to follow the notation and language of the standard single agent causal
inference framework reviewed in, for example, \citet{Imbens_Wooldridge_JEL09}.
What follows are some initial ideas and results; much work remains
to be done.

\subsection{Dyadic potential response}

Let $W_{i}\in\mathbb{W}=\left\{ w_{1},\ldots,w_{K}\right\} $ and
$X_{i}\in\mathbb{X}=\left\{ x_{1},\ldots,x_{L}\right\} $ be a finite
set of \emph{ego} and \emph{alter} policies. For example $\mathbb{W}$
might enumerate different export promotion policies (e.g., tax subsidies
or preferential credit schemes for exporting firms), while $\mathbb{X}$
might enumerate different combinations of protectionist policies (e.g.,
tariff levels). The goal is to understand how different counterfactual
combinations of ego and alter policy pairs map into (distributions
of) outcomes.

I begin with an assumption about the form of the potential response
function for (directed) dyad $ij$.
\begin{assumption}
\textsc{(Dyadic Potential Response Function)}\label{ass: dyadic_potential_response}
For any ego-alter pair $i,j\in\mathbb{N}$ with $i\neq j$, the potential
(directed) outcome associated with adopting the pair of policies $W_{i}=w$
and $X_{j}=x$ is given by
\begin{equation}
Y_{ij}\left(w,x\right)=h\left(w,x,A_{i},B_{j},V_{ij}\right),\thinspace x\in\mathbb{X},\thinspace w\in\mathbb{W}\label{eq: dyadic_potential_outcome}
\end{equation}
with $\left\{ \left(A_{i},B_{i}\right)\right\} _{i\in\mathbb{N}}$
and $\left\{ \left(V_{ij},V_{ji}\right)\right\} _{i,j\in\mathbb{N},i<j}$
both i.i.d. sequences additionally independent of each other and $h:\mathbb{W}\times\mathbb{X}\times\mathbb{A}\times\mathbb{B}\times\mathbb{V}\rightarrow\mathbb{Y}$
a measurable function.
\end{assumption}
The ego and alter effects, respectively $A_{i}$ and $B_{i}$, induce
dependence across any pair of potential outcomes whose corresponding
dyads share at least one agent in common. This implies a type structured
``interference'' between units, and hence a violation of SUTVA \citep[cf.,][]{Rosenbaum_JASA07}.

Since assignment to treatment is at the `ego' or `alter' level, setting
$X_{i}=x$ and $W_{j}=w$ shapes not just the realized outcome for
dyad $ij$, but also those of all other dyads which include either
agent $i$ or agent $j$. It is because of its implications for dependence
across the rows and columns of $\left[Y_{ij}\left(w,x\right)\right]_{i,j\in\mathbb{N}}$
that I label Assumption \ref{ass: dyadic_potential_response} an ``assumption''.
More than just notation is involved.

It is possible that Assumption \ref{ass: dyadic_potential_response}
could be derived from a more primitive exchangeability type restriction;
for example by viewing $\left[Y_{ij}\left(w,x\right)\right]_{i,j\in\mathbb{N}}$
as a jointly exchangeable random array and appealing to the Aldous-Hoover
Theorem. There may be some deep subtleties involved in such an approach,
so I prefer to maintain (\ref{eq: dyadic_potential_outcome}) as an
explicit assumption in this initial exploration.

I could have also written $Y_{ij}\left(w,x\right)=h\left(w,x,\left(1,0\right)U_{i},\left(0,1\right)U_{j},V_{ij}\right)=h^{*}\left(w,x,U_{i},U_{j},V_{ij}\right)$
with $U_{i}=\left(A_{i},B_{i}\right)'$. Explicitly separating out
an `ego' and `alter' effect, however, is conceptually useful and also
facilitates, as will be demonstrated by example below, parametric
modeling.

In some cases of interest the support of the ego and alter policies
will coincide (i.e., $\mathbb{W=X}$). Following \citet{SantosSilva_Tenreyro_AR10},
for example, both $X_{i}$ and $W_{j}$ might be indicators for Eurozone
membership. This example implies the additional restriction that $X_{i}=W_{i}$
for all $i\in\mathbb{N}$, since a country belongs to the Eurozone
in both their exporter (ego) and importer (alter) role. These special
cases can be deduced from the more general results which follow.

\subsection{Average structural function (ASF)}

Dyad-level treatment effects are defined in the usual way. The effect
on $ij$'s outcome of adopting policy pair $(w',x')$ vs. $(w,x$)
is
\[
Y_{ij}(w',x')-Y_{ij}(w,x).
\]
As in the standard case, identification of such effects at the dyad-level
is infeasible. This is because the econometrician only observes the
outcome associated with the policy pair actually adopted. Specifically,
for each of $N$ randomly sampled units she observes the assigned
or chosen ego and alter policies, $\left\{ \left(W_{i},X_{i}\right)\right\} _{i=1}^{N}$
and the $N\left(N-1\right)$ realized (directed) outcomes $\left\{ \left(Y_{ij},Y_{ji}\right)\right\} _{i<j}$,
where 
\begin{equation}
Y_{ij}\overset{def}{\equiv}Y_{ij}\left(W_{i},X_{j}\right)\label{eq: dyadic_realized_outcome}
\end{equation}
equals (directed) dyad $ij$'s realized outcome. No counterfactual
outcomes are observed.

Although dyad-level treatment effects are not identified, averages
of such effects over agents and/or dyads are (under certain assumptions).
Here I will focus on identifying average treatment effect (ATE) type
parameters. Consider the following thought experiment: (i) draw an
ego unit at random from the target population and exogenously assign
it policy $W_{i}=w$, (ii) independently draw an alter unit at random
and assign it policy $X_{j}=x$. The (ex ante) expected outcome associated
with this directed dyad, so configured, is
\begin{align}
m^{\mathrm{ASF}}\left(w,x\right) & \overset{def}{\equiv}\mathbb{E}\left[Y_{12}\left(w,x\right)\right]\label{eq: ASF_definition}\\
 & =\int\int\int h\left(w,x,a,b,v\right)f_{A_{1}}\left(a\right)f_{B_{2}}\left(b\right)f_{V_{12}}\left(v\right)\mathrm{d}a\mathrm{d}b\mathrm{d}v\nonumber \\
 & \overset{def}{\equiv}\int\int\int\bar{h}\left(w,x,a,b\right)f_{A_{1}}\left(a\right)f_{B_{2}}\left(b\right)\mathrm{d}a\mathrm{d}b,\nonumber 
\end{align}
where the second `$\overset{def}{\equiv}$' in (\ref{eq: ASF_definition})
follows from defining $\bar{h}\left(w,x,a,b\right)\overset{def}{\equiv}\mathbb{E}\left[h\left(w,x,a,b,V_{12}\right)\right]\overset{def}{\equiv}\bar{Y}_{ij}\left(w,x\right)$.
Note also that $\mathbb{E}\left[\left.h\left(w,x,a,b,V_{12}\right)\right|A_{1}=a,B_{2}=b\right]=\mathbb{E}\left[h\left(w,x,a,b,V_{12}\right)\right]$
by independence of $A_{1}$, $B_{2}$ and $V_{12}$ (Assumption \ref{ass: dyadic_potential_response}).

Differences of the form $m^{\mathrm{ASF}}\left(w',x'\right)-m^{\mathrm{ASF}}\left(w,x\right)$
measure the expected effects of different combinations of policies
on the directed dyadic outcome. If $W_{i}\in\left\{ 0,1\right\} $
and $X_{i}\in\left\{ 0,1\right\} $ are both binary indicators for
GATT/WHO membership, as in \citet{Rose_AER04}, then the contrast
\begin{equation}
m^{\mathrm{ASF}}\left(1,1\right)-m^{\mathrm{ASF}}\left(0,0\right)\label{eq: ATE}
\end{equation}
gives differences in export flows between a random pair of countries
in the GATT/WHO vs. non-GATT/WHO states of the world. This is an average
treatment effect (ATE) type parameter, adapted to the dyadic setting. 

The dyadic setting also raises new questions. For example the double
difference
\begin{equation}
m^{\mathrm{ASF}}\left(1,1\right)-m^{\mathrm{ASF}}\left(0,1\right)-\left[m^{\mathrm{ASF}}\left(1,0\right)-m^{\mathrm{ASF}}\left(0,0\right)\right]\label{eq: complementarity}
\end{equation}
measures complementarity in a binary policy/treatment across the two
agents in the dyad.

Other estimands beside the ASF may be of interest. The difference
of sample means
\[
\frac{1}{N-1}\sum_{j\neq i}\left[Y_{ij}\left(1,X_{j}\right)-Y_{ij}\left(0,X_{j}\right)\right]
\]
measures the average effect -- for unit $i$ alone -- of adopting
ego policy $W_{i}=1$ versus $W_{i}=0$; the average is over the status
quo distribution of alter polices. Additionally averaging over ego
units gives
\[
\frac{1}{N}\frac{1}{N-1}\sum_{i}\sum_{j\neq i}\left[Y_{ij}\left(1,X_{j}\right)-Y_{ij}\left(0,X_{j}\right)\right].
\]
This equals the average effect, across all units in the sample, of
adopting ego policy $W_{i}=1$ versus $W_{i}=0$, again given the
status quo distribution of alter policies. The population counterparts
of these two sample averages may also be of interest.

For the purposes of illustration, assume that $\mathbb{W}=\mathbb{X}=\left\{ 0,1\right\} $.
A parametric form for $Y_{ij}\left(w,x\right)$ that will be helpful
for both understanding extant empirical work and interpreting some
of the assumptions which follow is:
\begin{equation}
Y_{ij}\left(w,x\right)=\alpha+w\beta+x\gamma+wx\delta+A_{i}+B_{j}+V_{ij}.\label{eq: linear_dyadic_potential_outcome}
\end{equation}
Response (\ref{eq: linear_dyadic_potential_outcome}) implies that
treatment effects are constant across units, for example,
\[
Y_{ij}\left(1,0\right)-Y_{ij}\left(0,0\right)=\beta,
\]
which is constant in $i\in\mathbb{N}$. Under (\ref{eq: linear_dyadic_potential_outcome})
we also have estimand (\ref{eq: ATE}) equaling $\beta+\gamma+\delta$
and (\ref{eq: complementarity}) equal to $\delta$.

The average structural function (ASF) estimand is a leading case and
will be emphasized here. However, as I hope the brief sketch above
makes clear, other estimands merit exploration and, I conjecture,
will involve interesting identification, estimation and inference
issues.

\subsection{Identification under exogeneity}

In order to identify the ASF I will assert the existence of the observable
proxy variables, $R_{i}$ and $S_{i}$, respectively for the ego and
alter effects $A_{i}$ and $B_{i}$. These proxy variables will satisfy
two key restrictions, the first of which is:
\begin{assumption}
\textsc{(Redundancy)}\label{ass: redundancy} For $R_{i}\in\mathcal{R}\subseteq\mathbb{R}^{\dim\left(R\right)}$
a proxy variable for $A_{i}$, and $S_{i}\in\mathcal{S}\subseteq\mathbb{R}^{\dim\left(S\right)}$
a proxy variable for $B_{i}$, we have that
\[
\mathbb{E}\left[\left.Y_{ij}\left(w,x\right)\right|W_{i},X_{j},A_{i},B_{j},R_{i},S_{j}\right]=\mathbb{E}\left[\left.Y_{ij}\left(w,x\right)\right|W_{i},X_{j},A_{i},B_{j}\right],
\]
for any $w\in\mathbb{W}$ and $x\in\mathbb{X}$.
\end{assumption}
Assumption \ref{ass: redundancy} is a redundancy assumption of the
type introduced by \citet{Wooldridge_IIEM05}; it simply asserts that
$R_{i}$ and $S_{j}$ have no predictive power (in the conditional
mean sense) for the dyadic potential outcome $Y_{ij}(w,x)$ conditional
on the latent ego and alter attributes $A_{i}$ and $B_{j}$. Adapting
Wooldridge's \citeyearpar{Wooldridge_IIEM05} example, it asserts
that ego and alter Armed Forces Qualification Test (AFQT) scores,
$R_{i}$ and $S_{j}$, do not predict $Y_{ij}$ conditional on the
unobserved cognitive abilities, $A_{i}$ and $B_{j}$. Assumption
\ref{ass: redundancy} is a weak requirement since we are free to
conceptualize the latent attributes, $A_{i}$ and $B_{j}$, such that
$R_{i}$ and $S_{j}$ are clearly redundant.
\begin{assumption}
\textsc{(Strict Exogeneity) }\label{ass: strict_exogeneity}The $ij$
ego-alter treatment assignment $\left(W_{i},X_{j}\right)$ is independent
of $V_{ij}$ conditional on the latent ego $A_{i}$ and alter $B_{j}$
effects:
\begin{equation}
\left.V_{ij}\perp\left(W_{i},X_{j}\right)\right|A_{i}=a,B_{i}=b,\thinspace a\in\mathbb{A},\thinspace b\in\mathbb{B}.\label{eq: strict_exogeneity}
\end{equation}
\end{assumption}
While conditional independence assumptions feature prominently in
the causal inference literature, Assumption \ref{ass: strict_exogeneity},
which involves conditioning on \emph{unobservables}, has no clear
analog in the standard program evaluation model. The closest analog
of this assumption I can think of is Chamberlain's \citeyearpar{Chamberlain_HBE84}
definition of strict exogeneity of a time-varying regressor conditional
on a latent (time-invariant) unit-specific effect in the context of
panel data. To see the parallel return to parametric potential response
function (\ref{eq: linear_dyadic_potential_outcome}) and note that
(\ref{eq: dyadic_realized_outcome}) and (\ref{eq: strict_exogeneity})
imply that
\begin{equation}
\mathbb{E}\left[\left.Y_{ij}\right|W_{i},X_{j},A_{i},B_{j}\right]=\alpha+W_{i}\beta+X_{j}\gamma+W_{i}X_{j}\delta+A_{i}+B_{j}\label{eq: stict_exogeneity_analogy}
\end{equation}
since Assumption \ref{ass: strict_exogeneity} gives $\mathbb{E}\left[\left.V_{ij}\right|W_{i},X_{j},A_{i},B_{j}\right]=\mathbb{E}\left[\left.V_{ij}\right|A_{i},B_{j}\right]$
and $\mathbb{E}\left[\left.V_{ij}\right|A_{i},B_{j}\right]=\mathbb{E}\left[V_{ij}\right]$
by independence of $\left\{ \left(A_{i},B_{i}\right)\right\} _{i=1}^{N}$
and $\left\{ \left(V_{ij},V_{ji}\right)\right\} _{i<j}$ (setting
$\mathbb{E}\left[V_{ij}\right]=0$ is a normalization). Equation (\ref{eq: stict_exogeneity_analogy})
looks a lot like the definition of strict exogeneity in \citet[Equation 1.2 on p. 1248]{Chamberlain_HBE84}.
Equation (\ref{eq: stict_exogeneity_analogy}) implies, for example,
that
\begin{align*}
\mathbb{E}\left[\left.Y_{ij}-Y_{il}-\left(Y_{kj}-Y_{kl}\right)\right|W_{i},X_{j},A_{i},B_{j}\right]= & \left(W_{i}-W_{k}\right)\left(X_{j}-X_{l}\right)\delta,
\end{align*}
such that ``within-tetrad'' variation identifies $\delta$. Similar
to how within-group variation in a strictly exogenous regressor identifies
its corresponding coefficient in the panel context.

Under Assumption \ref{ass: strict_exogeneity} we have the density
factorization
\begin{align*}
f_{V_{12},A_{1},W_{1},B_{2},X_{2}}\left(v_{12},a_{1},w_{1},b_{2},x_{2}\right)= & f_{\left.V_{12}\right|A_{1},W_{1},B_{2},X_{2}}\left(\left.v_{12}\right|a_{1},w_{1},b_{2},x_{2}\right)\\
 & \times f_{A_{1},W_{1}}\left(a_{1},w_{1}\right)f_{B_{2},X_{2}}\left(b_{2},x_{2}\right)\\
= & f_{\left.V_{12}\right|A_{1},B_{2}}\left(\left.v_{12}\right|a_{1},b_{2}\right)f_{A_{1},W_{1}}\left(a_{1},w_{1}\right)f_{B_{2},X_{2}}\left(b_{2},x_{2}\right)\\
= & f_{V_{12}}\left(v_{12}\right)f_{A_{1},W_{1}}\left(a_{1},w_{1}\right)f_{B_{2},X_{2}}\left(b_{2},x_{2}\right)
\end{align*}
with the first equality an implication of units $1$ and $2$ being
independent random draws, the second equality following from Assumption
\ref{ass: strict_exogeneity}, and the third from independence of
$\left\{ \left(A_{i},B_{i}\right)\right\} _{i=1}^{N}$ and $\left\{ \left(V_{ij},V_{ji}\right)\right\} _{i<j}$
(i.e., Assumption \ref{ass: dyadic_potential_response}).

This factorization clarifies that the effect of Assumption \ref{ass: strict_exogeneity}
is to ensure that all ``endogeneity'' in treatment choice is reflected
in dependence between $W_{i}$ and $A_{i}$ and/or $B_{j}$ and $X_{j}$.
Conditional on these two latent variables, variation in treatment
is ``idiosyncratic'' or exogenous.

To deal with dependence between $W_{i}$ and $A_{i}$, and $B_{j}$
and $X_{j}$, I make a familiar selection of observables type assumption.
\begin{assumption}
\textsc{(Conditional Independence) }\label{ass: selecton_on_observables}An
ego's (alter's) treatment choice varies independently of their latent
effect $A_{i}$ ($B_{j}$) given the observed proxy $R_{i}$ ($S_{j}$):
\begin{align}
\left.A_{i}\perp W_{i}\right|R_{i} & =r,\thinspace r\in\mathcal{R}\subseteq\mathbb{R}^{\dim\left(R\right)}\label{eq: selection_on_observables_a}\\
\left.B_{i}\perp X_{i}\right|S_{i} & =s,\thinspace s\in\mathcal{S}\subseteq\mathbb{R}^{\dim\left(S\right)}.\label{eq: selection_on_observables_b}
\end{align}
\end{assumption}
Assumption \ref{ass: selecton_on_observables} is a standard one in
the context of single agent program evaluation problems, asserting
-- for example -- that $A_{i}$ and $W_{i}$ vary independently
within subpopulations homogenous in the proxy variable $R_{i}$. Extensive
discussions of selection-on-observables type assumptions like these,
including assessments of their appropriateness in different settings
of interest to empirical researchers, can be found in \citet{Blundel_Powell_WC03},
\citet{Heckman_Vytlacil_HBE07}, \citet{Imbens_Wooldridge_JEL09}
and \citet{Imbens_Rubin_CIBook15}. Their invocation here can raise
new issues, but, for the most part familiar approaches to reasoning
apply; see \citet{Graham_Imbens_Ridder_JBES18} for a related discussion.

Assumptions \ref{ass: dyadic_potential_response} to \ref{ass: selecton_on_observables},
plus an additional support condition described below, are sufficient
to show identification of the ASF. To develop the argument first let
\begin{align}
q\left(w,x,r,s\right) & =\mathbb{E}\left[\left.Y_{ij}\right|W_{i}=w,X_{j}=x,R_{i}=s,S_{j}=s\right]\label{eq: dyadic_proxy_regression}
\end{align}
be the dyadic proxy variable regression (PVR). Under Assumptions \ref{ass: dyadic_potential_response}
through \ref{ass: selecton_on_observables} the PVR relates to $\bar{Y}_{12}\left(w,x\right)=\bar{h}\left(w,x,A_{1},B_{2}\right)$
as follows:
\begin{align}
q\left(w,x,r,s\right)= & \mathbb{E}\left[\left.h\left(W_{i},X_{j},A_{i},B_{j},V_{ij}\right)\right|W_{i}=w,X_{j}=x,R_{i}=r,S_{j}=s\right]\nonumber \\
= & \mathbb{E}\left[\mathbb{E}\left[\left.h\left(W_{i},X_{j},A_{i},B_{j},V_{ij}\right)\right|W_{i}=w,X_{j}=x,A_{i},B_{j},R_{i}=r,S_{j}=s\right]\right.\nonumber \\
 & \left.\left|W_{i}=w,X_{j}=x,R_{i}=r,S_{j}=s\right.\right]\nonumber \\
= & \mathbb{E}\left[\mathbb{E}\left[\left.h\left(W_{i},X_{j},A_{i},B_{j},V_{ij}\right)\right|W_{i}=w,X_{j}=x,A_{i},B_{j}\right]\right.\nonumber \\
 & \left.\left|W_{i}=w,X_{j}=x,R_{i}=r,S_{j}=s\right.\right]\nonumber \\
= & \mathbb{E}\left[\left.\bar{h}\left(w,x,A_{i},B_{j}\right)\right|W_{i}=w,X_{j}=x,R_{i}=r,S_{j}=s\right]\nonumber \\
= & \int_{a}\int_{b}\bar{h}\left(w,x,a,b\right)f_{\left.A\right|R}\left(\left.a\right|r\right)f_{\left.B\right|S}\left(\left.b\right|s\right)\mathrm{d}a\mathrm{d}b\nonumber \\
= & \mathbb{E}\left[\left.\bar{Y}_{12}\left(w,x\right)\right|R_{1}=r,S_{2}=s\right].\label{eq: proxy_to_asf_key1}
\end{align}
where the first equality follows from Assumption \ref{ass: dyadic_potential_response}
and equation (\ref{eq: dyadic_realized_outcome}), the second from
iterated expectations, the third from the redundancy condition (Assumption
\ref{ass: redundancy}), the fourth from Assumption \ref{ass: strict_exogeneity},
independence of $\left\{ \left(A_{i},B_{i}\right)\right\} _{i=1}^{N}$
and $\left\{ \left(V_{ij},V_{ji}\right)\right\} _{i<j}$ and the definition
of $\bar{h}$, and the fifth from selection on observables (Assumption
\ref{ass: selecton_on_observables}).

Equation (\ref{eq: proxy_to_asf_key1}) gives the identification result
\begin{align}
\mathbb{E}_{R}\left[\mathbb{E}_{S}\left[q\left(w,x,R_{i},S_{j}\right)\right]\right]= & \int_{r}\int_{s}\left[\int_{a}\int_{b}\bar{h}\left(w,x,a,b\right)f_{\left.A\right|R}\left(\left.a\right|r\right)f_{\left.B\right|S}\left(\left.b\right|s\right)\mathrm{d}a\mathrm{d}b\right]\label{eq: proxy_to_asf_key2}\\
 & \times f_{R}\left(r\right)f_{S}\left(s\right)\mathrm{d}r\mathrm{d}s\nonumber \\
= & \int_{a}\int_{b}\bar{h}\left(w,x,a,b\right)f_{A}\left(a\right)f_{B}\left(b\right)\mathrm{d}a\mathrm{d}b\nonumber \\
= & \mathbb{E}\left[\bar{Y}_{12}\left(w,x\right)\right]\nonumber \\
= & m^{\mathrm{ASF}}\left(w,x\right).\nonumber 
\end{align}
Since $q\left(w,x,r,s\right)$ is only identified at those points
where $f_{\left.R\right|W}\left(\left.r\right|x\right)f_{\left.S\right|X}\left(\left.s\right|x\right)>0$,
while the integration in (\ref{eq: proxy_to_asf_key2}) is over $\mathcal{R}\times\mathcal{S}$,
we require a formal support condition:
\begin{equation}
\mathbb{S}\left(w,x\right)\overset{def}{\equiv}\left\{ r,s\thinspace:\thinspace f_{\left.R\right|W}\left(\left.r\right|w\right)f_{\left.S\right|X}\left(\left.s\right|x\right)>0\right\} =\mathcal{R}\times\mathcal{S}.\label{eq: overlap_v1}
\end{equation}
When $W_{i}$ and $X_{j}$ are discretely-valued, with a finite number
of support points, as assumed here, (\ref{eq: overlap_v1}) can be
expressed in a form similar to the overlap condition familiar from
the program evaluation literature \citep[e.g., ][]{Heckman_Smith_Clements_ReStud97,Imbens_Wooldridge_JEL09}.
\begin{assumption}
\textsc{(Overlap)}\label{ass: overlap} For $\left(w,x\right)$ the
ego-alter treatment combination of interest
\[
p_{w}\left(r\right)p_{x}\left(s\right)\geq\kappa>0\thinspace\thinspace\text{for all \ensuremath{\left(r,s\right)\in\mathcal{R}\times\mathcal{S}}}
\]
where $p_{w}\left(r\right)\overset{def}{\equiv}\Pr\left(\left.W_{i}=w\right|R_{i}=r\right)$
and $p_{x}\left(s\right)\overset{def}{\equiv}\Pr\left(\left.X_{i}=x\right|S_{i}=s\right).$
\end{assumption}
We have shown.
\begin{thm}
\label{thm: dyadic_asf}Under Assumptions \ref{ass: dyadic_potential_response}
through \ref{ass: overlap} the ASF is identified by
\begin{equation}
m^{\mathrm{ASF}}\left(w,x\right)=\int\int q\left(w,x,r,s\right)f_{R}\left(r\right)f_{S}\left(s\right)\mathrm{d}r\mathrm{d}s.\label{eq: dyadic_asf_identification}
\end{equation}
\end{thm}
Theorem \ref{thm: dyadic_asf} shows that the ASF is identified by
double marginal integration over the dyadic proxy variable regression
function. Double marginal integration also features in \citet{Graham_Imbens_Ridder_JBES18},
in the context of identifying an average match function (AMF), and
\citet{Brown_Newey_EM98}, in their discussion of efficient expectation
estimation under independence restrictions. However the random array
structure present here is absent in both these examples, which accounts
for many of the differences in underlying arguments.

\subsection{Estimation of the average structural function}

Let $q\left(w,x,r,s;\gamma\right)$ be a (flexibly) parametric model
for the dyadic proxy variable regression function. For example, if
the outcome of interest is export flows, we might specify that
\[
q\left(w,x,r,s;\gamma\right)=\exp\left(t\left(Q_{i}\right)'\gamma\right),
\]

with $Q_{i}=\left(W_{i}',X_{i}',R_{i}',S_{i}'\right)'$ and $t\left(Q_{i}\right)$
a finite (and pre-specified) set of basis functions (preferably including
interactions of terms in the treatment variables -- $W,X$ -- and
proxy variables -- $R,S$). We can estimate $\gamma$ use the Poisson
dyadic regression estimator described in Section \ref{sec: dyadic_regression}.
Proceeding in this way delivers an asymptotically linear representation
for $\sqrt{N}\left(\hat{\gamma}-\gamma_{0}\right)$ of
\begin{equation}
\sqrt{N}\left(\hat{\gamma}-\gamma_{0}\right)=-\Gamma_{0}^{-1}\frac{2}{\sqrt{N}}\sum_{i=1}^{N}\left\{ \frac{\bar{s}_{1}^{e}\left(Q_{i},U_{i};\gamma_{0}\right)+\bar{s}_{1}^{a}\left(Q_{i},U_{i};\gamma_{0}\right)}{2}\right\} +o_{p}\left(1\right)\label{eq: PVR_asym_lin}
\end{equation}
with $U_{i}=\left(A_{i},B_{i}\right)'$, $\Gamma_{0}$ the probability
limit of the Hessian matrix associated with the dyadic Poisson composite
log-likelihood, and $\bar{s}_{1}^{e}\left(Q_{i},U_{i};\gamma_{0}\right)$
and $\bar{s}_{1}^{a}\left(Q_{i},U_{i};\gamma_{0}\right)$ as defined
\vpageref{def: ego_alter_projections} (with $Q_{i}$ playing the
role of $X_{i}$).

With an estimate of $\gamma$ in hand, form the fitted values $\left\{ q\left(w,x,R_{i},S_{j};\hat{\gamma}\right)\right\} _{i<j}$
and, invoking Theorem \ref{thm: dyadic_asf}, compute the analog estimate
\begin{equation}
\hat{m}^{\mathrm{ASF}}\left(w,x;\hat{\gamma}\right)=\tbinom{N}{2}^{-1}\sum_{i=1}^{N-1}\sum_{j=i+1}^{N}\frac{q\left(w,x,R_{i},S_{j};\hat{\gamma}\right)+q\left(w,x,R_{j},S_{i};\hat{\gamma}\right)}{2}.\label{eq: ASF_hat}
\end{equation}
To present the limit distribution of $\hat{m}^{\mathrm{ASF}}\left(w,x;\hat{\gamma}\right)$
I impose a regularity condition on the proxy variable regression function:
\begin{assumption}
\label{ass: PVR_regularity_condition} (i) $\gamma\in\mathbb{C}\subseteq\mathbb{R}^{\dim\left(\gamma\right)}$
with $\mathbb{C}$ compact, (ii) $q\left(w,x,r,s;\gamma\right)$ is
twice continuously differentiable in $\gamma$, and (iii) the expectations
$\mathbb{E}\left[\left|q\left(w,x,R_{1},S_{2};\gamma\right)+q\left(w,x,R_{2},S_{1};\gamma\right)\right|\right]$,
$\mathbb{E}\left[\left\Vert \frac{\partial q\left(w,x,R_{1},S_{2};\gamma\right)}{\partial\gamma}+\frac{\partial q\left(w,x,R_{2},S_{1};\gamma\right)}{\partial\gamma}\right\Vert _{2}\right]$
and $\mathbb{E}\left[\left\Vert \frac{\partial^{2}q\left(w,x,R_{1},S_{2};\gamma\right)}{\partial\gamma\partial\gamma'}+\frac{\partial^{2}q\left(w,x,R_{2},S_{1};\gamma\right)}{\partial\gamma\partial\gamma'}\right\Vert _{F}\right]$
are finite.
\end{assumption}
Under this assumption we have the following Lemma.
\begin{lem}
\textsc{(ASF Estimation)}\label{lem: ASF_with_estimated_PVR} Under
Assumption \ref{ass: PVR_regularity_condition}, with $\hat{\gamma}$
a $\sqrt{N}$ consistent estimate of $\gamma_{0}$, we have that
\begin{align}
\sqrt{N}\left(\hat{m}^{\mathrm{ASF}}\left(w,x;\hat{\gamma}\right)-m^{\mathrm{ASF}}\left(w,x;\gamma_{0}\right)\right)= & \frac{2}{\sqrt{N}}\sum_{i=1}^{N}\psi_{0}\left(w,x,R_{i},S_{i};\gamma_{0}\right)\label{eq: U_statistic_with_nuisance_parameter_expansion}\\
 & +M_{0}\left(w,x\right)\sqrt{N}\left(\hat{\gamma}-\gamma_{0}\right)+o_{p}\left(1\right)\nonumber 
\end{align}
where 
\begin{align*}
\psi_{0}\left(w,x,R_{1},S_{1};\gamma\right) & =\frac{q^{e}\left(w,x,R_{1};\gamma\right)+q^{a}\left(w,x,S_{1};\gamma_{0}\right)}{2}-m^{\mathrm{ASF}}\left(w,x;\gamma\right)\\
M_{0}\left(w,x\right)= & \frac{1}{2}\mathbb{E}\left[\frac{\partial q\left(w,x,R_{1},S_{2};\gamma_{0}\right)}{\partial\gamma'}+\frac{\partial q\left(w,x,R_{2},S_{1};\gamma_{0}\right)}{\partial\gamma'}\right]
\end{align*}
with
\begin{align*}
q^{e}\left(w,x,r;\gamma\right) & =\mathbb{E}_{S}\left[q\left(w,x,r,S;\gamma\right)\right]\\
q^{a}\left(w,x,s;\gamma\right) & =\mathbb{E}_{R}\left[q\left(w,x,R,s;\gamma\right)\right].
\end{align*}
\end{lem}
\begin{proof}
The result follows from Assumption \ref{ass: PVR_regularity_condition}
and an application of Lemma \ref{lem: ASF_with_estimated_PVR} in
Appendix \ref{app: proofs_and_derivations}.
\end{proof}
Lemma \ref{lem: ASF_with_estimated_PVR} and equation (\ref{eq: PVR_asym_lin})
yields an asymptotically linear representation for $\sqrt{N}\left(\hat{m}^{\mathrm{ASF}}\left(w,x;\hat{\gamma}\right)-m^{\mathrm{ASF}}\left(w,x;\gamma_{0}\right)\right)$
of
\begin{align}
\sqrt{N}\left(\hat{m}^{\mathrm{ASF}}\left(w,x;\hat{\gamma}\right)-m^{\mathrm{ASF}}\left(w,x;\gamma_{0}\right)\right)= & \frac{2}{\sqrt{N}}\sum_{i=1}^{N}\psi_{0}\left(w,x,R_{1},S_{1};\gamma_{0}\right)\nonumber \\
 & -M_{0}\left(w,x\right)\Gamma_{0}^{-1}\nonumber \\
 & \times\frac{2}{\sqrt{N}}\sum_{i=1}^{N}\left\{ \frac{\bar{s}_{1}^{e}\left(Q_{i},U_{i};\gamma_{0}\right)+\bar{s}_{1}^{a}\left(Q_{i},U_{i};\gamma_{0}\right)}{2}\right\} \nonumber \\
 & +o_{p}\left(1\right).\label{eq: asf_asym_linear_rep}
\end{align}
Under correct (enough) specification of the composite likelihood,
which will typically follow if the parametric form of the the PVR
function is itself correctly specified, both $\bar{s}_{1}^{e}\left(Q_{1},U_{1};\gamma_{0}\right)$
and $\bar{s}_{1}^{a}\left(Q_{1},U_{1};\gamma_{0}\right)$ will be
conditional mean zero given $Q_{1}$, hence the first and second terms
in (\ref{eq: asf_asym_linear_rep}) will be uncorrelated with each
other such that a CLT will imply a limit distribution of
\[
\sqrt{N}\left(\hat{m}^{\mathrm{ASF}}\left(w,x;\hat{\gamma}\right)-m^{\mathrm{ASF}}\left(w,x;\gamma_{0}\right)\right)\overset{D}{\rightarrow}\mathcal{N}\left(0,4\Xi_{0}\left(w,x\right)+4M_{0}\left(w,x\right)\left(\Gamma_{0}'\Sigma_{1}^{-1}\Gamma_{0}\right)^{-1}M_{0}\left(w,x\right)'\right)
\]
with $\Xi_{0}\left(w,x\right)=\mathbb{V}\left(\psi_{0}\left(w,x,R_{1},S_{1};\gamma_{0}\right)\right)$
and $\Sigma_{1}=\mathbb{V}\left(\frac{\bar{s}_{1}^{e}\left(Q_{i},U_{i};\gamma_{0}\right)+\bar{s}_{1}^{a}\left(Q_{i},U_{i};\gamma_{0}\right)}{2}\right)$.

The first term in the asymptotic variance reflects the econometrician's
imperfect knowledge of the distribution of the proxy variables $\left(R_{i}',S_{i}'\right)'$.
The second term reflects the asymptotic penalty associated with not
knowing the conditional distribution of $Y_{12}$ given $W_{1},X_{2},R_{1},S_{2}$.
See \citet{Graham_EM11} and \citet{Graham_Imbens_Ridder_JBES18}
for more expansive discussions in related contexts (see also \citet{Chamberlain_EM92}).

In order to conduct inference an asymptotic variance estimate is required.
Estimation of covariance matrix $\mathbb{V}\left(\sqrt{N}\left(\hat{\gamma}-\gamma_{0}\right)\right)=\left(\Gamma_{0}'\Sigma_{1}^{-1}\Gamma_{0}\right)^{-1}$
can proceed using one of the methods described in Section \ref{sec: dyadic_regression}.
The $\Xi_{0}\left(w,x\right)$ term may be estimated by 
\[
\hat{\Xi}\left(w,x\right)=\frac{1}{N}\sum_{i=1}^{N}\hat{\psi}\left(w,x,R_{i},S_{i};\hat{\gamma}\right)\hat{\psi}\left(w,x,R_{i},S_{i};\hat{\gamma}\right)'
\]
where $\hat{\psi}\left(w,x,R_{i},S_{i};\hat{\gamma}\right)=\frac{1}{N-1}\sum_{j\neq i}\frac{q\left(w,x,R_{i},S_{j};\hat{\gamma}\right)+q\left(w,x,R_{j},S_{i};\hat{\gamma}\right)}{2}-\hat{m}^{\mathrm{ASF}}\left(w,x;\hat{\gamma}\right)$.
The Jacobian, $M_{0}\left(w,x\right)$, is naturally estimated by
\[
M_{0}\left(w,x\right)=\frac{2}{N\left(N-1\right)}\sum_{i=1}^{N-1}\sum_{j=i+1}^{N}\frac{1}{2}\left[\frac{\partial q\left(w,x,R_{i},S_{j};\hat{\gamma}\right)}{\partial\gamma'}+\frac{\partial q\left(w,x,R_{j},S_{i};\hat{\gamma}\right)}{\partial\gamma'}\right].
\]
In practice, for reasons analogous to those discussed in Section \ref{sec: dyadic_regression},
it may be preferable to replace the estimate of $\Sigma_{1}$ with
one for $\Omega$ (as defined in equation (\ref{eq: OMEGA})) and
use a ``Fafchampfs and Gubert'' type estimate of $\mathbb{V}\left(\sqrt{N}\hat{m}^{\mathrm{ASF}}\left(w,x;\gamma_{0}\right)\right)$
in place of $\hat{\Xi}\left(w,x\right)$.

\subsection{Further reading and open questions}

I am aware of no extant work on causal inference in the setting considered
here. There is a large, and rapidly growing, literature on causal
inference and interference, some of which makes connections to networks
\citep[e.g.,][]{Athey_Eckles_Imbens_JASA18}; \citet{VanderWeele_An_HCASR13}
provide a review of some relevant research.

The approach to estimation outlined above builds upon the dyadic regression
material already introduced. A natural extension would to replace
the parametric proxy variable regression function estimate with a
non-parametric one (perhaps estimated using machine learning procedures).
Inverse probability weighting (IPW) type estimators are also easily
constructed \citep[cf.,][]{Graham_Imbens_Ridder_JBES18}. I conjecture
that augmented inverse probability weighting estimators (AIPW), exhibiting
double robustness type properties, could also be constructed. The
maximal asymptotic precision with which $m^{\mathrm{ASF}}\left(w,x;\gamma_{0}\right)$
may be estimated under Assumptions \ref{ass: dyadic_potential_response}
through \ref{ass: overlap} is also unknown. This semiparametric efficiency
bound calculation, as in other network problems with likelihoods that
don't easily factor into independent components, does not appear to
be straightforward.

\section{\label{sec: heterogeneity}Incorporating unobserved heterogeneity}

In its natural to associate the agent-specific $U_{i}$ and $U_{j}$
terms appearing in the \citet{Crane_Towsner_JSL18} representation
result for $X$-exchangeable networks with unobserved correlated heterogeneity.
In Section \ref{sec: dyadic_regression} I introduced methods for
parametric estimation of the dyadic regression function $q\left(x,x'\right)\overset{def}{\equiv}\mathbb{E}\left[\left.Y_{ij}\right|X_{i}=x,X_{j}=x'\right]$.
The relationship between $q\left(x,x'\right)$ and the graphon $h\left(x,x',u,u',v\right)$
depends, of course, on the dependence structure between $X_{i}$ and
$U_{i}$. Assumptions about this dependence structure played a prominent
role in identifying the average structural function (ASF) in Section
\ref{sec: policy_analysis}. In both Sections \ref{sec: dyadic_regression}
and \ref{sec: policy_analysis}, however, the focus was on direct
modeling of the conditional mean of $Y_{ij}$ given observed covariates.

In this section I wish to explore the advantages of a modeling approach
which directly specifies a parametric form for the graphon. This idea,
at least implicitly, goes back to the work of \citet{Holland_Leinhardt_JASA81}
and \citet{vanDuijn_et_al_SN04}.

The analysis in Sections \ref{sec: dyadic_regression} and \ref{sec: policy_analysis}
requires that the researcher directly specify the correct parametric
form of the dyadic regression function. In contrast, the exact structure
of (conditional) dependence across dyads sharing agents in common
was left unspecified. To understand how such dependence might arise,
it is useful to specify a structural \emph{correlated} random effects
model, analogous to those familiar from single-agent discrete choice
panel data settings \citep[e.g., ][]{Chamberlain_ReStud80,Chamberlain_HBE84}.

\subsection{A parametric dyadic potential response function}

For the purposes of illustration, I will focus on modeling a directed
binary outcome variable. The generalization to non-binary outcomes
is straightforward. Refer to the dyadic potential response function
introduced in Assumption \ref{ass: dyadic_potential_response}. Consider
the following parametric form for this function
\begin{align}
Y_{12}\left(w_{1},x_{2}\right) & =\mathbf{1}\left(t^{e}\left(w_{1}\right)'\beta_{0}^{e}+t^{a}\left(x_{2}\right)'\beta_{0}^{a}+\omega\left(w_{1},x_{2}\right)'\gamma_{0}+A_{1}+B_{2}+V_{12}>0\right)\label{eq: dyadic_degree_heterogeneity}\\
 & =h\left(w_{1},x_{2},A_{1},B_{2},V_{12}\right)\nonumber 
\end{align}
with 
\begin{equation}
\left.\left(V_{12},V_{21}\right)\right|Q_{1},Q_{2},A_{1},B_{1},A_{2},B_{2}\sim\mathcal{N}\left(\left(\begin{array}{c}
0\\
0
\end{array}\right),\left(\begin{array}{cc}
1 & \zeta\\
\zeta & 1
\end{array}\right)\right)\label{eq: dyadic_shock}
\end{equation}
and independently distributed across dyads. As in Section \ref{sec: policy_analysis},
$X_{i}$ and $W_{j}$ correspond to the chosen ego and alter treatments;
$A_{i}$ and $B_{j}$ are unobserved ego and alter heterogeneity,
which may be correlated with these treatment choices, and $R_{i}$
and $S_{j}$ are proxy variables (recall that $Q_{i}=\left(W_{i}',X_{i}',R_{i}',S_{i}'\right)'$).
The vectors $t^{e}\left(w_{1}\right)$, $t^{a}\left(x_{2}\right)$
and $\omega\left(w_{1},x_{2}\right)$ consist of known basis functions
in the underlying treatment variables. In the case where both $W_{i}$
and $X_{j}$ are binary we would set $t^{e}\left(w_{1}\right)=w_{1}$,
$t^{a}\left(x_{2}\right)=x_{2}$ and $\omega\left(w_{1},x_{2}\right)=w_{1}x_{2}$.

Next posit the correlated random effects specification for the joint
distribution of the ego and alter heterogeneity
\begin{align}
\left.\begin{array}{c}
A_{i}\\
B_{i}
\end{array}\right|W_{i},X_{i},R_{i},S_{i} & \sim\mathcal{N}\left(\left(\begin{array}{c}
\alpha_{0}^{e}+k^{e}\left(R_{i}\right)'\delta_{0}^{e}\\
\alpha_{0}^{e}+k^{a}\left(S_{i}\right)'\delta_{0}^{a}
\end{array}\right),\left(\begin{array}{cc}
\sigma_{A}^{2} & \rho\sigma_{A}\sigma_{B}\\
\rho\sigma_{A}\sigma_{B} & \sigma_{B}^{2}
\end{array}\right)\right),\label{eq: dyadic_correlated_random_effects}
\end{align}
with $k^{e}\left(R_{i}\right)$ and $k^{a}\left(S_{i}\right)$ vectors
of known functions of the proxy variables. Note that (\ref{eq: dyadic_degree_heterogeneity})
and (\ref{eq: dyadic_correlated_random_effects}) jointly imply the
selection on observables, Assumption \ref{ass: selecton_on_observables}
introduced earlier. Redundancy and strict exogeneity, respectively
Assumptions \ref{ass: redundancy} and \ref{ass: strict_exogeneity},
also hold in this set-up.

Averaging over $A_{i}$ and $B_{j}$ gives a dyadic proxy variable
regression function of
\begin{equation}
q\left(W_{i},X_{j},R_{i},S_{j};\eta_{0}\right)=\Phi\left(T_{ij}'\eta_{0}\right)\label{eq: dyadic_re_probit}
\end{equation}
for $\eta_{0}=\left(1+\sigma_{A}^{2}+\sigma_{B}^{2}\right)^{-1/2}\left(\alpha_{0}^{e}+\alpha_{0}^{a},\left(\beta_{0}^{e}\right)',\left(\beta_{0}^{a}\right)',\gamma_{0}',\left(\delta_{0}^{e}\right)',\left(\delta_{0}^{a}\right)'\right)'$
and 
\[
T_{ij}=\left(1,t^{e}\left(W_{i}\right),t^{a}\left(X_{j}\right),\omega\left(W_{i},X_{j}\right)',k^{e}\left(R_{i}\right),k^{a}\left(S_{j}\right)\right)'.
\]

It is possible to estimate $\eta_{0}$ along the lines outlined in
Section \ref{sec: dyadic_regression} above. Alternatively one could
attempt to directly maximize the integrated likelihood implied by
(\ref{eq: dyadic_degree_heterogeneity}), (\ref{eq: dyadic_shock})
and (\ref{eq: dyadic_correlated_random_effects}). This would be computationally
non-trivial since the integral does not easily factor. \citet{vanDuijn_et_al_SN04}
and \citet{Zijlstra_et_al_BJMSP09} develop this approach using Markov
Chain Monte Carlo (MCMC) methods.

\subsection{Triad probit: a correlated random effects estimator}

An intermediate approach, which is more efficient than the basic dyadic
regression estimator introduced earlier, and additionally recovers
more features of the graph generation process, is what I will call
\emph{triad probit}. Triad probit is also a composite likelihood estimator.
Instead of modeling the dyadic outcome, $\left.Y_{12}\right|Q_{1},Q_{2}$,
marginally however, it is composed of component likelihoods for the
joint outcome $\left.\left(Y_{12},Y_{21},Y_{13},Y_{31}\right)\right|Q_{1},Q_{2},Q_{3}$.
That is I model the outcome configuration associated with a \emph{pair-of-dyads}
sharing one agent in common. An overall criterion function is constructed
by summing over the component log-likelihoods, so constructed, for
all $3\tbinom{N}{3}$ pairs-of-dyads sharing one agent in common.\footnote{This approach is related to the pairwise likelihood estimator for
models with crossed random effects discussed by \citet{Belio_Varin_SM05}
and \citet{Cattelan_Varin_ATAS13}.}

The probability of the event $Y_{12}=y_{12},Y_{21}=y_{21},Y_{13}=y_{13},Y_{31}=y_{31}$
given the parameters and regressors is

\begin{align}
\Pr\left(\left.Y_{12}=y_{12},Y_{21}=y_{21},Y_{13}=y_{13},Y_{31}=y_{31}\right|Q_{1},Q_{2},Q_{3}\right)= & \int_{\mathcal{A}_{12}}\int_{\mathcal{A}_{21}}\int_{\mathcal{A}_{13}}\int_{\mathcal{A}_{31}}\phi_{4}\left(\left.\mathbf{t}\right|\Sigma\right)\mathrm{d}\mathbf{t}\label{eq: triad_probit_contribution}
\end{align}
with $\phi_{4}\left(\left.t\right|\Sigma\right)$ the density of a
tetra-variate normal distribution with mean zero and covariance matrix
$\Sigma$. The intervals of integration are given by
\[
\mathcal{A}_{ij}=\left\{ \begin{array}{cc}
\left(-\infty,T_{ij}'\eta_{0}\right) & \text{if \ensuremath{y_{ij}=1}}\\
\left[T_{ij}'\eta_{0},\infty\right) & \text{if }\ensuremath{y_{ij}=0}
\end{array}\right.,
\]
with the covariance matrix, which is in correlation form (a scale
normalization), taking the form
\[
\Sigma=\Sigma\left(\zeta,\sigma_{A},\sigma_{B},\rho\right)=\left(\begin{array}{cccc}
1 & \frac{\zeta+2\rho\sigma_{A}\sigma_{B}}{1+\sigma_{A}^{2}+\sigma_{B}^{2}} & \frac{\sigma_{A}^{2}}{1+\sigma_{A}^{2}+\sigma_{B}^{2}} & \frac{\rho\sigma_{A}\sigma_{B}}{1+\sigma_{A}^{2}+\sigma_{B}^{2}}\\
\frac{\zeta+2\rho\sigma_{A}\sigma_{B}}{1+\sigma_{A}^{2}+\sigma_{B}^{2}} & 1 & \frac{\rho\sigma_{A}\sigma_{B}}{1+\sigma_{A}^{2}+\sigma_{B}^{2}} & \frac{\sigma_{B}^{2}}{1+\sigma_{A}^{2}+\sigma_{B}^{2}}\\
\frac{\sigma_{A}^{2}}{1+\sigma_{A}^{2}+\sigma_{B}^{2}} & \frac{\rho\sigma_{A}\sigma_{B}}{1+\sigma_{A}^{2}+\sigma_{B}^{2}} & 1 & \frac{\zeta+2\rho\sigma_{A}\sigma_{B}}{1+\sigma_{A}^{2}+\sigma_{B}^{2}}\\
\frac{\rho\sigma_{A}\sigma_{B}}{1+\sigma_{A}^{2}+\sigma_{B}^{2}} & \frac{\sigma_{B}^{2}}{1+\sigma_{A}^{2}+\sigma_{B}^{2}} & \frac{\zeta+2\rho\sigma_{A}\sigma_{B}}{1+\sigma_{A}^{2}+\sigma_{B}^{2}} & 1
\end{array}\right).
\]

The integral (\ref{eq: triad_probit_contribution}) does not have
a closed form expression. Fortunately a large econometrics and statistics
literature suggest various methods for its numerical evaluation; see,
for example, \citet{Keane_EM94} and \citet{Chib_Greenberg_BM98}.

Let $l_{123}^{*}\left(\theta\right)$ equal the logarithm of (\ref{eq: triad_probit_contribution})
with $\theta=\left(\eta',\zeta,\sigma_{A},\sigma_{B},\rho\right)'$.
To induce symmetry in the criterion function summands I form the average
\[
l_{ijk}\left(\theta\right)=\frac{1}{3}\left[l_{ijk}^{*}\left(\theta\right)+l_{jik}^{*}\left(\theta\right)+l_{kij}^{*}\left(\theta\right)\right].
\]
The triad probit estimate $\hat{\theta}_{\mathrm{TP}}$ of $\theta_{0}$
is the maximizer of the sum of the $l_{ijk}\left(\theta\right)$ kernels
over all $\tbinom{N}{3}$ triads in the network:
\begin{equation}
L_{N}\left(\theta\right)=\tbinom{N}{3}^{-1}\sum_{i<j<k}l_{ijk}\left(\theta\right).\label{eq: triad_probit}
\end{equation}
Note that (\ref{ass: selecton_on_observables}) sums over all $3\tbinom{N}{3}$
pairs-of-dyads sharing one agent in common. It does this by summing
over all $\binom{N}{3}$ triads in the network and, for each such
triad, summing over the three pairs-of-dyads sharing an agent in common
that can be constructed from it.

The criterion (\ref{eq: triad_probit}) is not a U-process-minimizer,
although, as in the other contexts introduced above, it shares similarities
with one. The results of \citet{Honore_Powell_JOE94} do not immediately
characterize the asymptotic sampling properties of $\hat{\theta}_{\mathrm{TP}}$.
Nevertheless arguments similar to those outlined in Sections \ref{sec: dyadic_regression}
and \ref{sec: policy_analysis} above can be applied to also analyze
$\hat{\theta}_{\mathrm{TP}}$. 

A quick outline of these arguments goes as follows. Let $S_{N}\left(\theta\right)=\tbinom{N}{3}^{-1}\sum_{i<j<k}s_{ijk}\left(\theta\right)$
with $s_{ijk}\left(\theta\right)=\frac{\partial l_{ijk}\left(\theta\right)}{\partial\theta}$.
Also define $\Gamma_{0}=\mathbb{E}\left[\frac{\partial^{2}l_{ijk}\left(\theta\right)}{\partial\theta\partial\theta'}\right]$
and, as earlier, $\Sigma_{q}=\mathbb{E}\left[s_{i_{1}i_{2}i_{3}}s_{j_{1}j_{2}j_{3}}'\right]$
to be the covariance of $s_{i_{1}i_{2}i_{3}}$ and $s_{j_{1}j_{2}j_{3}}$
when they share $q=0,1,2,3$ indices in common.

Calculation then gives
\begin{equation}
\mathbb{V}\left(\sqrt{N}S_{N}\left(\theta\right)\right)=9\Sigma_{1}+\frac{18}{N-1}\left(\Sigma_{2}-2\Sigma_{1}\right)+\frac{6}{\left(N-1\right)\left(N-2\right)}\left(\Sigma_{3}+3\Sigma_{1}\right)\label{eq: triad_probit_variance_of_score}
\end{equation}
which suggests, under regularity conditions, the limiting distribution
\begin{equation}
\sqrt{N}\left(\hat{\theta}_{\mathrm{TP}}-\theta_{0}\right)\overset{D}{\rightarrow}N\left(0,9\Gamma_{0}^{-1}\Sigma_{1}\Gamma_{0}^{-1}\right).\label{eq: triad_probit_asym_dis}
\end{equation}
Associated with the triad probit is a proxy variable regression function
estimate of
\[
q\left(W_{i},X_{j},R_{i},S_{j};\hat{\eta}_{\mathrm{TP}}\right)=\Phi\left(T_{ij}'\hat{\eta}_{\mathrm{TP}}\right)
\]
from which an estimate of the ASF (or differences thereof) can be
directly constructed according to equation (\ref{eq: ASF_hat}). This
corresponds (essentially) to a dyadic generalization of the average
partial effect (APE) estimator introduced by \citet{Chamberlain_HBE84}
in the context of a correlated random effects probit panel data model.

\subsection{Fixed effects approaches}

The models introduced above, while allowing for dependence in outcomes
across dyads sharing agents in common, restrict its structure. In
contrast, \citet{Graham_EM17} provides a fixed effects analysis of
a model where a undirected binary dyadic outcome is determined according
to
\begin{equation}
Y_{ij}=\mathbf{1}\left(\left[t\left(X_{i}\right)+t\left(X_{j}\right)\right]'\beta_{0}+\omega\left(X_{i},X_{j}\right)'\gamma_{0}+A_{i}+A_{j}-V_{ij}\leq0\right),\label{eq: graham_fe_model}
\end{equation}
with $V_{ij}$ standard logistic and independent across dyads. Specifically
he studies identification and estimation of $\gamma_{0}$, leaving
the joint distribution of $X_{i}$ and $A_{i}$ unrestricted (without
restrictions on this distribution $\beta_{0}$ is unidentified \citep[cf., ][]{Hausman_Taylor_EM81,Arellano_Bover_JE95}.
The parameter of interest, $\gamma_{0}$, indexes the strength of
any homophilous sorting on the observables agent attributes in $X_{i}$,
while $\left\{ A_{i}\right\} _{i=1}^{N}$ indexes unobserved\emph{
degree-heterogeneity}. Since real world network degree distributions
often have high variance (and in particular fat right tails), incorporating
degree heterogeneity may be important in practice \citep{Barabasi_Albert_Sc99,Barabasi_Bonabau_SciAmer03}.
\citet{Graham_EM17} shows how failing to accommodate degree heterogeneity
may attenuate measured homophily (i.e., bias estimates of $\gamma_{0}$).

Conditional on $\mathbf{X}=\left(X_{1},\ldots,X_{N}\right)'$ and
$\mathbf{A}=\left(A_{1},\ldots,A_{N}\right)'$, the likelihood for
the adjacency matrix $\mathbf{D}$ factors into $\binom{N}{2}$ conditionally
independent components. Absorbing $t\left(X_{i}\right)'\beta_{0}$
into the individual effect $A_{i}$, the model consists of the finite
dimensional parameter of interest, $\gamma_{0}$, and the $N$ incidental
heterogeneity parameters, $\mathbf{A}_{0}$. Let $K=\dim\left(\gamma_{0}\right)$;
in this model the number of parameters, $K+N$, is a function of the
order of the network. Since this number grows with $N$, the model
is non-standard \citep[cf., ][]{Holland_Leinhardt_JASA81,Chatterjee_et_al_AAP11}.

\citet{Graham_EM17} analyzes the large network properties of two
estimates of $\gamma_{0}$. The first estimate, leveraging the implicit
``large-N, large-T'' structure of dense networks, is the joint maximum-likelihood
one, which also simultaneously estimates the incidental parameters
$\mathbf{A}_{0}=\left(A_{01},\ldots,A_{0N}\right)'$. The second exploits
the exponential family structure of the model and conditions on a
sufficient statistic for $\mathbf{A}_{0}$. Both estimates have antecedents
in the literature on panel data.

\subsubsection*{Joint estimators}

Let $T_{ij}$ be an $N\times1$ vector with a one in the $i^{th}$
and $j^{th}$ elements and zeros elsewhere. The joint-MLE coincides
with the logit fit of $Y_{ij}$ onto $\omega\left(X_{i},X_{j}\right)$
and $T_{ij}$ for all $i<j$.\footnote{\citet{Graham_EM17} outlines a more convenient nested-fixed-point
approach to estimation based upon an insight due to \citet{Chatterjee_et_al_AAP11}.} Although this estimator involves $K+N$ parameters, it is based upon
a criterion function with $\tbinom{N}{2}=O\left(N^{2}\right)$ summands.
This feature is similar to joint maximum likelihood estimation in
a panel data setting where both $N$ and $T$ are allowed to grow.
Here each of $N$ agents make $N-1$ linking decisions; the latter
is analogous to ``$T$'' in the ``large-N, large-T'' panel data
setting. As the number of agents in the network grows, so to does
the number of link decisions observed for each of them. This feature
of the model allows for consistent estimation of both $\gamma_{0}$
and $\mathbf{A}_{0}$, although, as in the panel data case, there
is a bias in the limit distribution of $\hat{\gamma}$ which must
be corrected in order to undertake asymptotically valid inference
\citep{Hahn_Newey_EM04,Arellano_Hahn_WC07}.\footnote{A technical difficultly involving the inverse Hessian arises in the
network setting. A similar challenge is also present in panel data
models with time effects \citep{FernandezVal_Weidner_JOE16}.} 

Graham's \citeyearpar{Graham_EM17} assumptions imply that the limiting
network will be dense. \citet{Yan_et_al_JASA18} show that it is possible
to weaken his assumptions somewhat, but it appears impossible to accommodate
asymptotic sequences with sparse limits. In Monte Carlo experiments
the joint MLE works poorly in networks with low density. Researchers
are advised to be cautious when applying this estimator to low density
networks.

\citet{Dzemski_RESTAT18} and \citet{Yan_et_al_JASA18} study joint
estimation of a directed version of (\ref{eq: graham_fe_model}).
The former paper presents a method of testing for reciprocity in links
as well as for neglected transitivity.

\subsubsection*{Conditional estimators}

Under the logistic assumption, the likelihood associated with (\ref{eq: graham_fe_model})
is a member of the exponential family. It turns out that the degree
sequence of the network is a sufficient statistic for $\mathbf{A}_{0}$
\citep{Snijders_JSS02}. A conditional maximum likelihood estimator
could be constructed, however, unlike in the panel case considered
by \citet{Chamberlain_ReStud80}, the likelihood does not nicely factor
into independent components. It would also be non-trivial to evaluate
and maximize the conditional likelihood \citep[cf., ][]{Blitzstein_Diaconis_IM11}.

\citet{Graham_EM17} instead builds a criterion involving tetrads
-- quadruples of agents. A tetrad is the smallest subgraph that is
not completely determined by its degree sequence. For example, there
are three isomorphisms of the two edge graphlet $\vcenter{\hbox{\includegraphics[scale=0.125]{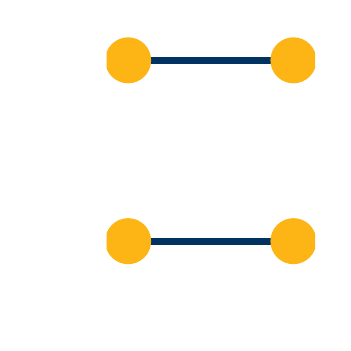}}}$
on four vertices, each with an identical subgraph degree sequence
of $\left(1,1,1,1\right)'$. If $\gamma_{0}=0$, then conditional
on the event that a randomly sampled tetrad takes one of these three
forms, any one of them occurs with an equal probability of one third.
Deviations from this benchmark are possible when $\gamma_{0}\neq0$,
depending on the configuration of covariates across agents in the
sampled tetrad. Graham's \citeyearpar{Graham_EM17} conditional estimator,
which he calls tetrad logit, is based upon this insight.

The large network properties of the tetrad logit estimate of $\gamma_{0}$
may be derived in a way roughly analogous to that of the dyadic regression
estimators introduced above. The analysis in \citet{Graham_EM17},
however, allows for sequences of graphs which are sparse in the limit.
This affects the rate-of-convergence of the tetrad logit estimate.
Conveniently its limit distribution remains normal under both dense
and sparse sequences.

\citet{Jochmans_JBES18} provides a conditional analysis, including
several worked empirical examples, of a directed analog of tetrad
logit. \citet{Nadler_WP15} proposes a related estimator for bipartite
networks and presents an empirical application.

\subsection{Further reading and open questions}

\citet{Varin_et_al_SS11} survey the statistics literature on composite
likelihoods. A standard reference on U-Process minimizers is \citet{Honore_Powell_JOE94}.
Many of the results presented in this section, as well as the previous
ones, utilize ideas coming from the theory of composite likelihood
and U-Process minimizers. Connections to panel data have also featured
prominently; here I recommend \citet{Chamberlain_ReStud80}, \citet{Chamberlain_HBE84},
\citet{Arellano_Honore_HBE01}, and \citet{Arellano_Hahn_WC07}.

The triad probit estimator introduced above has a rate of convergence
equal to $\sqrt{N}$. In the simplest setup the the tetrad logit estimator
has a faster $\sqrt{\tbinom{N}{2}}$ rate of convergence. This is
peculiar because, invoking intuitions familiar from panel data, one
would generally expect an estimate based upon an integrated/random
effects likelihood to be more efficient than one based upon a conditional/fixed
effects likelihood. Here the two estimators have different rates of
convergence with, perhaps, a ranking reverse of what one might expect
\emph{a priori}.

\citet{vanDuijn_et_al_SN04} use MCMC methods to (essentially) maximize
the network likelihood implied by (\ref{eq: dyadic_degree_heterogeneity}),
(\ref{eq: dyadic_shock}) and (\ref{eq: dyadic_correlated_random_effects}).
Their approach to inference is Bayesian; it would be interesting to
formally study the maximum integrated likelihood estimator proper
(as opposed to the triad probit composite likelihood estimator introduced
here). What is the rate of convergence associated with the true random
effects maximum likelihood estimator (MLE)? Likewise, tetrad logit,
while inspired by conditional likelihood ideas, is not a conditional
MLE (it is akin to a conditional composite MLE). \citet{Graham_EM17}
describes the conditional MLE, but does not formally analyze it. Such
a formal analysis could be insightful. More generally we know very
little about efficiency in even the simplest of network problems.

The introduction of heterogeneity in this section is restrictive in
nature. It allows for what \citet{Graham_EM17} calls degree heterogeneity.
Methods for incorporating assortative matching on latent agent-specific
attributes would also be useful. For inspiration see, for example,
\citet{Krivitsky_et_al_SN09}. Recent ideas from panel data may be
useful here too; especially the work on discrete heterogeneity done
by \citet{Bonhomme_Manresa_EM15}. Ideas from the stochastic block
literature -- which is not surveyed in this chapter -- might also
be useful for incorporating richer heterogeneity structure into econometric
models for dyadic outcomes.

\section{\label{sec: Statistics}Asymptotic distribution theory for network
statistics}

\citet{Wasserman_Faust_Bk94} exposit a large post World War II literature
on the computation and interpretation of different statistics of the
adjacency matrix. Researchers routinely report statistics like reciprocity,
transitivity, moments of the degree sequence, and diameter when presenting
real world network data. Measures of statistical uncertainty almost
never accompany these reports. The leading approach to assessing whether
a reported network statistic is unusual is to informally compare it
with its expected value under an Erdös-Renyi null or, alternatively,
a reference sample of real world networks \citep[e.g., ][]{Milo_el_al_Sci02,Newman_NetBook10,Graham_AR15}.\footnote{\citet{Blitzstein_Diaconis_IM11} present an elegant approach based
on comparing statistics of the network in hand to those of the reference
set of all graphs with the same degree sequence (i.e., a $\beta$-model
null).} Informal simulation-based approaches to ``inference'' abound. 

Large network approaches to hypothesis testing only recently emerged
\citep[e.g., ][]{Bobollas_et_al_RSA07,Picard_et_al_JCB08,Bickel_et_al_AS11}.
This is currently an active research area \citep[e.g., ][]{Gao_Lafferty_arXiv17,Green_Shalizi_arXiv17,Menzel_arXiv17},
with many open questions. To be fair, work on the distributional properties
of network statistics under specific graph generation processes, generally
the Erdös-Renyi one or close variants, was undertaken earlier. This
work arose largely in response to the seminal papers by \citet{Holland_Leinhardt_AJS70,Holland_Leinhardt_SM76}.
Examples include the work of Frank \citeyearpar{Frank_ANYAS79,Frank_JMS80,Frank_DM88},
\citet{Wasserman_JMS77} and \citet{Nowicki_SN91}. The last reference
is a useful survey of such analyses.

This section presents results on the large network distribution of
induced subgraph frequencies (and various statistics constructed from
them). I begin, in Subsection \ref{subsec: transitivity}, with a
detailed analysis of triad counts and their application to inference
on the transitivity index or global clustering coefficient \citep[e.g., ][p. 96]{Kolaczyk_NetBook09}.
This is a classic, practically important, and pedagogically valuable,
example. Results on counts of trees and cycles of any order are available
in the Appendix. In Subsection \ref{subsec: the-degree-sequences},
I turn to moments of the degree distribution, an area of intense focus
in applied work \citep[e.g., ][]{Barabasi_Bonabau_SciAmer03,Atalay_et_al_PNAS11,Acemoglu_etal_EM12}.

Not all common network statistics are covered by the results presented
in this section. Statistics such as diameter and average path length,
for example, have, to my knowledge, unknown sampling properties. Subsection
\ref{subsec: open-questions} discusses open questions.

The work surveyed in this section dates to the papers by \citet{Holland_Leinhardt_AJS70,Holland_Leinhardt_SM76}.
More recent contributions, generally by statisticians, were often
motivated by examples from computational biology \citep[e.g., ][]{Picard_et_al_JCB08}.
An especially important contribution is the paper by \citet{Bickel_et_al_AS11}.
This section draws heavily from the work by Bickel and coauthors.
Related ideas were used in the discussion of dyadic regression in
Section \ref{sec: dyadic_regression}. Recent work on strategic models
of network formation, where econometricians play the leading role,
arose separately. However, in Section \ref{sec: Strategic-models}
I argue that ideas from research on subgraph counts could be valuable
there as well. Specifically for structural estimation of strategic
network formation models.

The results in this section are based on the following hypothetical
repeated sampling experiment. Let $G_{\infty,N}$ be an infinite exchangeable
random graph of interest. The network in hand, $G_{N}$, is the one
induced by a random sample of $N$ vertices from $G_{\infty,N}$.
Let $h_{N}\left(u,v\right)$ denote the Aldous-Hoover graphon characterizing
the infinite graph $G_{\infty,N}$ from which the econometrician samples
$N$ agents independently at random. Note I suppress dependence of
this graphon on the mixing parameter, $\alpha$, since I seek to conduct
inference conditional on it (i.e., conditional on the empirical distribution
of $\left[D_{ij}\right]_{i,j\in\mathbb{N},i<j}$).

Using the observed network, $G_{N}$, we construct the statistic $t_{N}\left(G_{N}\right)$.
The sampling distribution of this statistic is the one induced by
repeated sampling of $N$ agents from the underlying infinite graph
$G_{\infty,N}$. To derive a limit distribution I assume there is
a sequence of infinite random graphs $\left\{ G_{\infty,N}\right\} $
-- indexed by $N$ -- such that
\[
h_{N}\left(u,v\right)=\rho_{N}w\left(u,v\right)
\]
with $\rho_{N}$ (possibly) approaching zero as $N\rightarrow\infty$.
In this way I pair a sequence of increasingly larger ``sampled''
networks with a corresponding sequence of infinite networks that are
allowed to become increasingly sparser. With this set-up we can study
the distribution of $t_{N}\left(G_{N}\right)$, appropriately scaled,
as $N\rightarrow\infty$.

As noted earlier, the above thought experiment does not mirror how
empirical networks are constructed in practice. Typically one of two
cases obtain. In the first, the network under study really is a very
large graph (e.g., the Facebook graph) and the econometrician really
does sample from it. However, due to spareness, sampling is rarely
conducted as described above. Instead snowball sampling, edge sampling,
path sampling etc. are typically used \citep{Crane_PFSNA18}. Understanding
how to consistently estimate network statistics and their sampling
distributions under these more exotic data collection schemes is an
interesting topic for future research. In the second case the econometrician
works with the complete graph on some finite population of vertices.
In this cases the idea of sampling from an infinite graph is a thought
experiment used to get results that are hopefully useful in practice.
It is this latter, rather commonplace case, which I have in mind here.

There is a subtlety in this second case, already touched upon in Section
\ref{sec: Basic-probability-tools} in the context of my discussion
of the Aldous-Hoover Theorem. A jointly exchangeable random graph
with a \emph{finite} number of agents need not have a probability
law with a conditionally independent dyad (CID) structure. The pattern
of dependence across links in such a network may be more complicated
than that implied by the Aldous-Hoover representation. I conjecture,
by speculative extrapolation based upon the example introduced in
Section \ref{sec: Basic-probability-tools}, that this is especially
the case when agents form links strategically. We know, however, that,
for $N$ large enough, joint exchangeability will deliver a probability
law for the network that is of the Aldous-Hoover form. This suggests
that, to derive limit theory, it is reasonable to proceed in the way
I do here; but there are missing steps in the argument. \citet{Menzel_WP16}
represents the only attempt I am aware of to struggle with these issues
in a disciplined way. A more rigorous pairing of the game theoretic
models of network formation of interest to many economists, with the
theory of graph limits would be a high priority topic for future research.

\subsection{\label{subsec: transitivity}Large network estimation of the transitivity
index}

In the social sciences, hypothesis formulation often involves graphlet
counts \citep[e.g., ][]{Holland_Leinhardt_AJS70,Bearman_et_al_AJS04,Choi_Wu_JSCM09,Jackson_et_al_AER12,Isakov_et_al_SS19}.\footnote{In practice it is easier to derive results for homomorphism frequencies
and, not coincidentally, the theory of graph limits generally works
with homomorphisms.} Graphlet counts are also used to construct important network statistics
like the transitivity index. It is this last statistic that is studied
in this subsection. 

After introducing some notation and definitions, I apply the basic
approach outlined by \citet[Proposition 6]{Bhattacharya_Bickel_AS15}
to calculate variance expressions for induced subgraph counts of two-stars
($\vcenter{\hbox{\includegraphics[scale=0.125]{twostar}}}$) and
triangles ($\vcenter{\hbox{\includegraphics[scale=0.125]{triangle}}}$).
While this is a relatively straightforward extension, it does require
some carefully constructed notation.\footnote{One could even argue that these expressions are already implicit in
\citet{Holland_Leinhardt_SM76}, although they did not explore the
properties of their expressions under sparse versus dense graph sequences,
nor did they analyze rates of convergence. Indeed, \citet[p. 580]{Wasserman_Faust_Bk94},
referring to the covariance calculations of \citet{Holland_Leinhardt_SM76},
comment that they ``can be time-consuming to calculate (and maybe
even difficult to comprehend)''.} Asymptotic normality of these counts, appropriately scaled, follows
from their results. An analysis of transitivity in the Nyakatoke risk-sharing
network studied by \citet{deWeerdt_IAP04} illustrates the practical
application of these ideas.

A special case of a CID model is the Erdös-Renyi graph generation
process (i.e., $h\left(u,v\right)=\rho$ for some $0<\rho<1$ and
all $\left(u,v\right)\in\left[0,1\right]^{2}$). The behavior of subgraph
counts under this GGP were studied by Nowicki and co-authors in the
late 1980s and early 1990s \citep{Nowicki_Wierman_DM88,Janson_Nowicki_PTRF91,Nowicki_SN91}.
It turns out that this case exhibits a form of degeneracy. Specifically,
the leading terms in the variance expressions presented below are
identically zero under the Erdös-Renyi graph generation process. Subgraph
frequencies remain asymptotically normal in this case, but with a
faster rate of convergence. A separate treatment of this case is provided
below.

\subsubsection*{Notation and estimation}

Recall from Section \ref{sec: Basic-probability-tools} that the induced
subgraph frequency of $S$ in $G_{N}$ is

\begin{equation}
P_{N}\left(S\right)=\frac{1}{\tbinom{N}{p}\left|\mathrm{iso}\left(S\right)\right|}\sum_{\mathbf{i}_{p}\in\mathcal{C}_{p,N}}\mathbf{1}\left(S\cong G_{N}\left[\mathbf{i}_{p}\right]\right).\label{eq: P(S)_hat}
\end{equation}
Under the maintained sampling scheme it is easy to see that (\ref{eq: P(S)_hat})
is an unbiased estimate of $P\left(S\right)=t_{\mathrm{ind}}\left(S,h\right)$,
the ``population'' induced subgraph density.

Consider the two-star ($\vcenter{\hbox{\includegraphics[scale=0.125]{twostar}}}$)
and triangle ($\vcenter{\hbox{\includegraphics[scale=0.125]{triangle}}}$)
triad configurations. Applying (\ref{eq: P(S)_hat}) gives the estimates
\begin{align}
P_{N}\left(\vcenter{\hbox{\includegraphics[scale=0.125]{twostar}}}\right)= & \tbinom{N}{3}^{-1}\frac{1}{3}\sum_{\mathbf{i}_{3}\in\mathcal{C}_{3,N}}\left[D_{i_{1}i_{2}}D_{i_{1}i_{3}}\left(1-D_{i_{2}i_{3}}\right)+D_{i_{1}i_{2}}\left(1-D_{i_{1}i_{3}}\right)D_{i_{2}i_{3}}\right.\label{eq: P(S)_hat_twostar}\\
 & \left.+\left(1-D_{i_{1}i_{2}}\right)D_{i_{1}i_{3}}D_{i_{2}i_{3}}\right]\nonumber \\
P_{N}\left(\vcenter{\hbox{\includegraphics[scale=0.125]{triangle}}}\right)= & \tbinom{N}{3}^{-1}\sum_{\mathbf{i}_{3}\in\mathcal{C}_{3,N}}D_{i_{1}i_{2}}D_{i_{1}i_{3}}D_{i_{2}i_{3}}.\label{eq: P(S)_hat_triangle}
\end{align}
From (\ref{eq: P(S)_hat_twostar}) and (\ref{eq: P(S)_hat_triangle})
we can construct an estimate of the \emph{transitivity index} or global
clustering coefficient:
\begin{equation}
\mathrm{TI_{N}}=\frac{3\times\left(\text{\# of triangles}\right)}{\left(\text{\# of two-stars}\right)+3\times\left(\text{\# of triangles}\right)}=\frac{P_{N}\left(\vcenter{\hbox{\includegraphics[scale=0.125]{triangle}}}\right)}{P_{N}\left(\vcenter{\hbox{\includegraphics[scale=0.125]{twostar}}}\right)+P_{N}\left(\vcenter{\hbox{\includegraphics[scale=0.125]{triangle}}}\right)}=\frac{Q_{N}\left(\vcenter{\hbox{\includegraphics[scale=0.125]{triangle}}}\right)}{Q_{N}\left(\vcenter{\hbox{\includegraphics[scale=0.125]{twostar}}}\right)}.\label{eq: transitivity_index}
\end{equation}
Under an Erdös-Renyi graph generation process it is easy to show that
(\ref{eq: transitivity_index}) should be close to the density of
the network \citep[e.g., ][]{Graham_AR15}. \citet{Gao_Lafferty_arXiv17}
develop a test based on this idea. If, suitably normalized, the limit
distribution of the vector $\left(P_{N}\left(\vcenter{\hbox{\includegraphics[scale=0.125]{twostar}}}\right),P_{N}\left(\vcenter{\hbox{\includegraphics[scale=0.125]{triangle}}}\right)\right)'$
can be characterized, then delta methods can be used to conduct large
network inference on transitivity. This idea is developed in detail
below.

Distribution theory for induced subgraph counts may also be useful
for structural model estimation via the method of (simulated) minimum
distance. In this approach model parameters are estimated by matching
model-implied values of subgraph counts with their empirical counterparts.
Sampling uncertainty in such estimates, stems from the corresponding
uncertainty about the reduced form subgraph counts being matched.
This idea is developed more completely in Section \ref{sec: Strategic-models}.

\subsubsection*{Graphlet Stitchings}

In developing an interpretable expression for the variance of graphlet
counts, it is helpful to introduce something I will call a \emph{graphlet
stitching}.\footnote{After completing the initial draft of this Chapter I discovered independent
work by \citet{Green_Shalizi_arXiv17} that develops a closely related
concept which they call ``merged copy sets''. Graphlet stitchings,
as I define them, are more suited to my specific needs; although both
approaches lead to the same answer in the end. The basic idea is already
implicit in \citet{Bhattacharya_Bickel_AS15} (and really even \citet{Holland_Leinhardt_SM76}).
Essentially the same idea is also used in \citet{Graham_EM17} to
derive large network theory for Tetrad Logit.}

Let $R$ and $S$ be two $p^{th}$ order subgraphs of interest to
the econometrician. Furthermore, let $\mathbf{i}_{p}$ and $\mathbf{j}_{p}$
be two p-tuples drawn independently at random from $\mathcal{C}_{p,N}$
(as defined in Section \ref{subsec: Network-moments} above). The
(scaled) covariance of the events ``$G_{N}\left[\mathbf{i}_{p}\right]$
is isomorphic to $R$'' and ``$G_{N}\left[\mathbf{j}_{p}\right]$
is isomorphic to $S$'', when there are $q$ integers/vertices common
to $\mathbf{i}_{p}=\left\{ i_{1},i_{2},\ldots,i_{p}\right\} $ and
$\mathbf{j}_{p}=\left\{ j_{1},j_{2},\ldots,j_{p}\right\} $, is
\begin{align}
\Sigma_{q}\left(R,S\right) & =\Xi\left(\mathcal{W}_{q,R,S}\right)-P\left(R\right)P\left(S\right)\label{eq: SIGMA_q_S_R}
\end{align}
where $P\left(R\right)$ is the induced subgraph density defined in
equation (\ref{eq: induced_subgraph_density}) and
\begin{align}
\Xi\left(\mathcal{W}_{q,R,S}\right)\overset{def}{\equiv} & \frac{\mathbb{E}\left[\mathbf{1}\left(R\cong G_{N}\left[\mathbf{i}_{p}\right]\right)\mathbf{1}\left(S\cong G_{N}\left[\mathbf{j}_{p}\right]\right)\right]}{\left|\mathrm{iso}\left(R\right)\right|\left|\mathrm{iso}\left(S\right)\right|}\label{eq: XI_multiset}
\end{align}
Here $\mathcal{W}_{q,R,S}$ is notation for a \emph{set} of what I
call \emph{graphlet stitchings}. In order to understand the structure
of $\Xi\left(\mathcal{W}_{q,S,R}\right)$ further we need a formal
definition.
\begin{defn}
\emph{\label{def: graphlet_stitching}}\textsc{(Graphlet Stitching)}
Let $W_{q,R,S}$ be the graph union of $R$ and $S$, labelled isomorphisms
of two graphlets of interest, if\\
(i) $\mathcal{V}\left(R\right)\subseteq\mathcal{V}\left(G\right)$
and $\mathcal{V}\left(S\right)\subseteq\mathcal{V}\left(G\right)$;\\
(ii) $\left|\mathcal{V}\left(R\right)\right|=\left|\mathcal{V}\left(S\right)\right|=p$
vertices each;\\
(iii) $\left|\mathcal{V}\left(R\right)\cap\mathcal{V}\left(S\right)\right|=q$
vertices in common;\\
(iv) identical structure across all vertices in common (i.e., $\left(i,j\right)\in\mathcal{E}\left(R\right)\Leftrightarrow\left(i,j\right)\in\mathcal{E}\left(S\right)\thinspace\forall\thinspace i,j\in\mathcal{V}\left(R\right)\cap\mathcal{V}\left(S\right)$),\\
then $W_{q,R,S}$ is a graphlet stitching of $R$ and $S$.
\end{defn}
Next define the \emph{set} of all feasible stitchings of $R$ and
$S$ which satisfy Definition \ref{def: graphlet_stitching} as $\mathcal{W}_{q,S,R}$.
When $R$ and $S$ belong to the same isomorphism class write $\mathcal{W}_{q,S,S}=\mathcal{W}_{q,S}$.

Requirement (iv) of Definition \ref{def: graphlet_stitching} is constraining.
It implies, for example, that some pairs of labelled two-stars cannot
be stitched together. For example $R=\left(\left\{ 1,2,3\right\} ,\left\{ \left(1,2\right),\left(1,3\right)\right\} \right)$
and $S=\left(\left\{ 1,2,4\right\} ,\left\{ \left(1,4\right),\left(2,4\right)\right\} \right)$
cannot be logically stitched together because the $\left(1,2\right)$
edge is present in $R$ but not $S$. This violates requirement (iv)
of Definition \ref{def: graphlet_stitching}. Note also that the set
$\mathcal{W}_{q,S,R}$ may contain elements which are isomorphic to
one another. 

For simplicity consider the vertices $1,2,\ldots,p,p+1,\ldots,2p-q$
in $G_{N}$.\footnote{Since $G_{N}$ is induced by a random sample of vertices, vertices
$1,2,\ldots,p,p+1,\ldots,2p-q$ correspond to a random $2p-q$ tuple.} If $R$, defined on vertices $1,2,\ldots,p$ , is isomorphic to the
subgraph of $G_{N}$ induced by vertices $\left\{ 1,\ldots,p\right\} $
and $S$, defined on vertices vertices $p-q,\ldots,2p-q$, is isomorphic
to the subgraph of $G_{N}$ induced by vertices $\left\{ p-q,\ldots2p-q\right\} $,
then it must be the case that the union of these two induced subgraphs
is an element of $\mathcal{W}_{q,S,R}$. This gives the equality
\begin{equation}
\Xi\left(\mathcal{W}_{q,S,R}\right)=\sum_{W\in\mathcal{W}_{q,S,R}}\frac{\Pr\left(W=G_{N}\left[\left\{ 1,\ldots,p\right\} \right]\cup G_{N}\left[\left\{ p-q,\ldots,2p-q\right\} \right]\right)}{\left|\mathrm{iso}\left(R\right)\right|\left|\mathrm{iso}\left(S\right)\right|}.\label{eq: stitching_probability}
\end{equation}
Note that the graph union of $G_{N}\left[\left\{ 1,\ldots,p\right\} \right]$
and $G_{N}\left[\left\{ p-q,\ldots,2p-q\right\} \right]$ may differ
from the subgraph induced by the union of the two overlapping vertex
sets: 
\[
G_{N}\left[\left\{ 1,\ldots,p\right\} \right]\cup G_{N}\left[\left\{ p-q,\ldots,2p-q\right\} \right]\neq G_{N}\left[\left\{ 1,\ldots,2p-q\right\} \right].
\]
This is because the union of $G_{N}\left[\left\{ 1,\ldots,p\right\} \right]$
and $G_{N}\left[\left\{ p-q,\ldots,2p-q\right\} \right]$ will not
include any edges between $\left\{ 1,\ldots,p-q-1\right\} $, the
vertices in $R$ alone, and $\left\{ p+1,\ldots,2p-q\right\} $, the
vertices in $S$ alone, while $G_{N}\left[\left\{ 1,\ldots,2p-q\right\} \right]$
may. By exchangeability the right-hand-side of (\ref{eq: stitching_probability})
is the same for any vertex sets $\mathbf{i}_{p}=\left\{ i_{1},i_{2},\ldots,i_{p}\right\} $
and $\mathbf{j}_{p}=\left\{ j_{1},j_{2},\ldots,j_{p}\right\} $ sharing,
as is implicitly assumed in what follows, $q$ vertices in common.

To check whether $R\cong G_{N}\left[\mathbf{i}_{p}\right]$ and $S\cong G_{N}\left[\mathbf{j}_{p}\right]$
we therefore check whether $G_{N}\left[\mathbf{i}_{p}\right]\cup G_{N}\left[\mathbf{j}_{p}\right]$
coincides with a particular (labeled) graphlet stitching of $R$ and
$S$. Doing so, in turn, requires us to check for the presence \emph{or}
absence of only $p\left(p-1\right)-\tbinom{q}{2}$ potential edges.
The presence or absence of the $\left(p-q\right)^{2}$ possible edges
from the vertices unique to $R$ to those unique to $S$ is immaterial.
Equation (\ref{eq: XI_multiset}) gives neither an induced or partial
subgraph frequency, but what I will call a \emph{graphlet stitching
frequency}.

\subsubsection*{Calculating graphlet stitching frequencies}

To understand how to calculate graphlet stitching frequencies in practice
it is helpful to work through a few examples. Figure \ref{fig: W_q1_twostar}
shows all the elements of $\mathcal{W}_{1,\vcenter{\hbox{\includegraphics[scale=0.125]{twostar}}}}$
on vertex set $\left\{ 1,2,3,4,5\right\} $, with vertex $1$ being
the vertex in common. The top row shows all isomorphisms of $\vcenter{\hbox{\includegraphics[scale=0.125]{twostar}}}$
on vertices $\left\{ 1,4,5\right\} $, while the left-most column
shows all such isomorphisms on vertices $\left\{ 1,2,3\right\} $.
The nine figures in the corresponding grid show all the associated
graphlet stitchings. 

A more complicated example is provide by $\mathcal{W}_{2,\vcenter{\hbox{\includegraphics[scale=0.125]{twostar}}}}$
, which is shown in Figure \ref{fig: W_q2_twostar}. The format of
the figure is the same as that of Figure \ref{fig: W_q1_twostar}.
The two vertices in common are $1$ and $2$. An interesting feature
of this example is that not all graphlet stitchings are feasible.

In evaluating $\Xi\left(\mathcal{W}_{q,S,R}\right)$ it is helpful
to observe that $\mathcal{W}_{q,S,R}$ may include multiple isomorphisms
of the same graph. Since the probabilities $\Pr\left(W=G_{N}\left[\left\{ 1,\ldots,p\right\} \right]\cup G_{N}\left[\left\{ p-q,\ldots,2p-q\right\} \right]\right)$
and $\Pr\left(W'=G_{N}\left[\left\{ 1,\ldots,p\right\} \right]\cup G_{N}\left[\left\{ p-q,\ldots,2p-q\right\} \right]\right)$
coincide when $W$ and $W'$ are isomorphic to one another, we can
also ``represent'' $\mathcal{W}_{q,S,R}$ as a multi-set, with one
(arbitrary) labelling of each of the non-isomorphic graphlet stitchings
retained as elements, but with multiplicities equal to the number
of isomorphic appearances. For example, the cardinality of $\mathcal{W}_{1,\vcenter{\hbox{\includegraphics[scale=0.125]{twostar}}}}$
is $\left|\mathrm{iso}\left(\vcenter{\hbox{\includegraphics[scale=0.125]{twostar}}}\right)\right|\times\left|\mathrm{iso}\left(\vcenter{\hbox{\includegraphics[scale=0.125]{twostar}}}\right)\right|=9$,
but with only three non-isomorphic elements. Inspecting Figure \ref{fig: W_q1_twostar}
we define the multi-set:
\[
\mathcal{W}_{1,\vcenter{\hbox{\includegraphics[scale=0.125]{twostar}}}}^{\mathrm{m}}=\left(\left\{ \vcenter{\hbox{\includegraphics[scale=0.125]{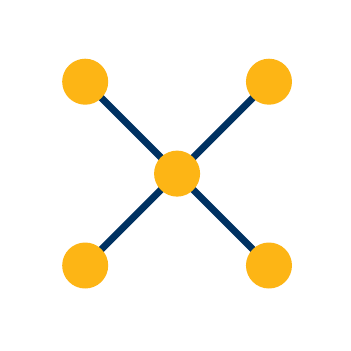}}},\vcenter{\hbox{\includegraphics[scale=0.125]{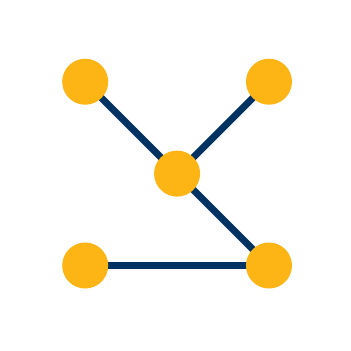}}},\vcenter{\hbox{\includegraphics[scale=0.125]{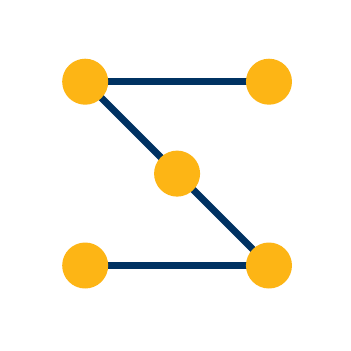}}}\right\} ,\left\{ \left(\vcenter{\hbox{\includegraphics[scale=0.125]{fourstar}}},1\right),\left(\vcenter{\hbox{\includegraphics[scale=0.125]{tailedthreestar}}},4\right),\left(\vcenter{\hbox{\includegraphics[scale=0.125]{fivepath}}},4\right)\right\} \right).
\]
Let $\nu_{q,R,S}\left(W\right)$ denote the multiplicity of $W$ in
$\mathcal{W}_{q,R,S}^{\mathrm{m}}$; for example the multiplicity
of $\vcenter{\hbox{\includegraphics[scale=0.125]{tailedthreestar}}}$
in $\mathcal{W}_{1,\vcenter{\hbox{\includegraphics[scale=0.125]{twostar}}}}^{\mathrm{m}}$
is $\nu_{1,\vcenter{\hbox{\includegraphics[scale=0.125]{twostar}}}}\left(\vcenter{\hbox{\includegraphics[scale=0.125]{tailedthreestar}}}\right)=4$.

We then have that, using equation (\ref{eq: stitching_probability}),
the equality $\Xi\left(\mathcal{W}_{1,\vcenter{\hbox{\includegraphics[scale=0.125]{twostar}}}}^{\mathrm{m}}\right)=\Xi\left(\mathcal{W}_{1,\vcenter{\hbox{\includegraphics[scale=0.125]{twostar}}}}\right)$.
Similarly, inspecting $\mathcal{W}_{2,\vcenter{\hbox{\includegraphics[scale=0.125]{twostar}}}}$
(see Figure \ref{fig: W_q2_twostar}) , we see that it also contains
three non-isomorphic elements, yielding
\[
\mathcal{W}_{2,\vcenter{\hbox{\includegraphics[scale=0.125]{twostar}}}}^{\mathrm{m}}=\left(\left\{ \vcenter{\hbox{\includegraphics[scale=0.125]{onethreewheel}}},\vcenter{\hbox{\includegraphics[scale=0.125]{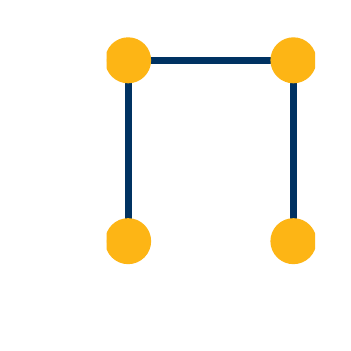}}},\vcenter{\hbox{\includegraphics[scale=0.125]{fourcycle}}}\right\} ,\left\{ \left(\vcenter{\hbox{\includegraphics[scale=0.125]{onethreewheel}}},2\right),\left(\vcenter{\hbox{\includegraphics[scale=0.125]{fourpath}}},2\right),\left(\vcenter{\hbox{\includegraphics[scale=0.125]{fourcycle}}},1\right)\right\} \right).
\]
Finally, it is easy to see that $\mathcal{W}_{3,\vcenter{\hbox{\includegraphics[scale=0.125]{twostar}}}}^{\mathrm{m}}=\left(\left\{ \vcenter{\hbox{\includegraphics[scale=0.125]{twostar}}}\right\} ,\left\{ \left(\vcenter{\hbox{\includegraphics[scale=0.125]{twostar}}},3\right)\right\} \right)$.
The reader may verify that
\[
\mathcal{W}_{1,\vcenter{\hbox{\includegraphics[scale=0.125]{triangle}}}}^{\mathrm{m}}=\left(\left\{ \vcenter{\hbox{\includegraphics[scale=0.125]{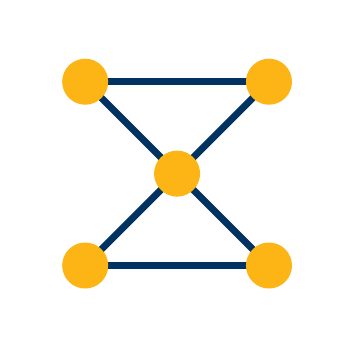}}}\right\} ,\left\{ \left(\vcenter{\hbox{\includegraphics[scale=0.125]{twotriangle}}},1\right)\right\} \right),\thinspace\mathcal{W}_{2,\vcenter{\hbox{\includegraphics[scale=0.125]{triangle}}}}^{\mathrm{m}}=\left(\left\{ \vcenter{\hbox{\includegraphics[scale=0.125]{chordalcycle}}}\right\} ,\left\{ \left(\vcenter{\hbox{\includegraphics[scale=0.125]{chordalcycle}}},1\right)\right\} \right),\thinspace\mathcal{W}_{3,\vcenter{\hbox{\includegraphics[scale=0.125]{triangle}}}}^{\mathrm{m}}=\left(\left\{ \vcenter{\hbox{\includegraphics[scale=0.125]{triangle}}}\right\} ,\left\{ \left(\vcenter{\hbox{\includegraphics[scale=0.125]{triangle}}},1\right)\right\} \right)
\]
as well as that
\[
\mathcal{W}_{1,\vcenter{\hbox{\includegraphics[scale=0.125]{twostar}}},\vcenter{\hbox{\includegraphics[scale=0.125]{triangle}}}}^{\mathrm{m}}=\left(\left\{ \vcenter{\hbox{\includegraphics[scale=0.125]{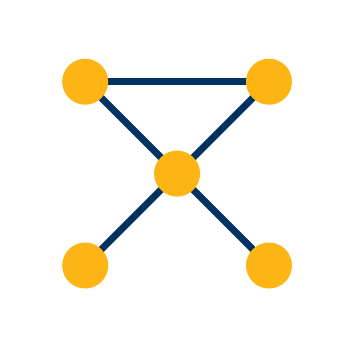}}},\vcenter{\hbox{\includegraphics[scale=0.125]{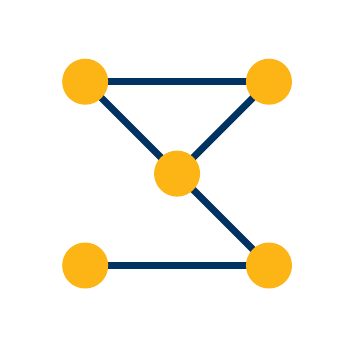}}}\right\} ,\left\{ \left(\vcenter{\hbox{\includegraphics[scale=0.125]{twotailedtriangle}}},1\right),\left(\vcenter{\hbox{\includegraphics[scale=0.125]{rattailedtriangle}}},2\right)\right\} \right),\thinspace\mathcal{W}_{2,\vcenter{\hbox{\includegraphics[scale=0.125]{twostar}}},\vcenter{\hbox{\includegraphics[scale=0.125]{triangle}}}}^{\mathrm{m}}=\left(\left\{ \vcenter{\hbox{\includegraphics[scale=0.125]{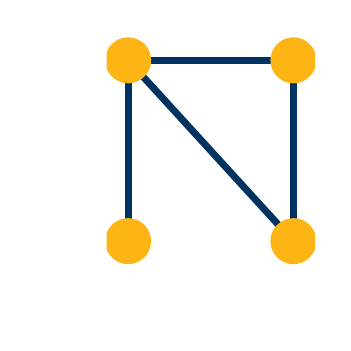}}}\right\} ,\left\{ \left(\vcenter{\hbox{\includegraphics[scale=0.125]{tailedtriangle}}},2\right)\right\} \right),\thinspace\mathcal{W}_{3,\vcenter{\hbox{\includegraphics[scale=0.125]{twostar}}},\vcenter{\hbox{\includegraphics[scale=0.125]{triangle}}}}^{\mathrm{m}}=\emptyset.
\]
These multi-sets will be used to study the covariance of $\left(P_{N}\left(\vcenter{\hbox{\includegraphics[scale=0.125]{twostar}}}\right),P_{N}\left(\vcenter{\hbox{\includegraphics[scale=0.125]{triangle}}}\right)\right)'$
as well as the variance of the transitivity index.

\begin{figure}
\caption{\label{fig: W_q1_twostar} Stitchings of two-star graphlets with one
common node}

\begin{centering}
\includegraphics{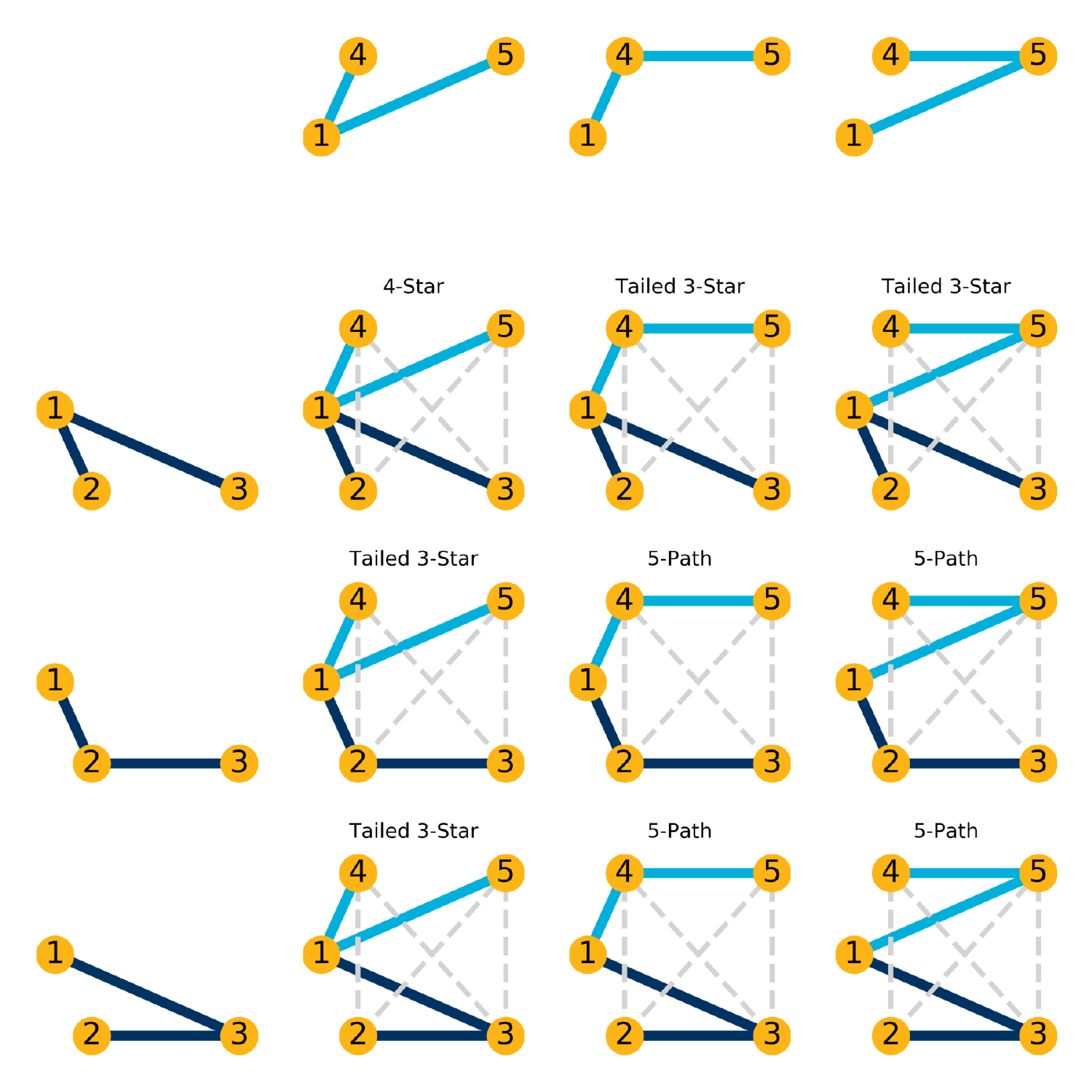}
\par\end{centering}
\textsc{\small{}\uline{Notes:}}{\small{} Depiction of all possible
ways to join (or ``stitch'') a pair of two-star ($\vcenter{\hbox{\includegraphics[scale=0.125]{twostar}}}$)
subgraphs together with one node in common. Each of the resulting
subgraphs is a pentad wiring. The dashed gray edges involve pairs
of nodes that are not common across the pair of two-stars. Hence the
subgraph induced by the five nodes in the pentad may or may not include
these edges. The set $\mathcal{W}_{1,\vcenter{\hbox{\includegraphics[scale=0.125]{twostar}}}}$
has $\left|\mathrm{iso}\left(\vcenter{\hbox{\includegraphics[scale=0.125]{twostar}}}\right)\right|\times\left|\mathrm{iso}\left(\vcenter{\hbox{\includegraphics[scale=0.125]{twostar}}}\right)\right|=9$
elements.}{\small\par}

\textsc{\small{}\uline{Source:}}{\small{} Author's calculations.}{\small\par}
\end{figure}

At the risk of overkill, the following calculations illustrate how
the two stitching probability definitions, equations (\ref{eq: XI_multiset})
and (\ref{eq: stitching_probability}), coincide. For the two-star
example, starting with equation (\ref{eq: XI_multiset}), I get
\begin{align}
\Xi\left(\mathcal{W}_{1,\vcenter{\hbox{\includegraphics[scale=0.125]{twostar}}}}\right)= & \frac{\Pr\left(\vcenter{\hbox{\includegraphics[scale=0.125]{twostar}}}\cong G_{N}\left[\left\{ 1,2,3\right\} \right]\thinspace\&\thinspace\vcenter{\hbox{\includegraphics[scale=0.125]{twostar}}}\cong G_{N}\left[\left\{ 1,4,5\right\} \right]\right)}{\left|\mathrm{iso}\left(\vcenter{\hbox{\includegraphics[scale=0.125]{twostar}}}\right)\right|^{2}}\nonumber \\
= & \frac{1}{\left|\mathrm{iso}\left(\vcenter{\hbox{\includegraphics[scale=0.125]{twostar}}}\right)\right|^{2}}\mathbb{E}\left[\left\{ D_{12}D_{13}\left(1-D_{23}\right)+D_{12}\left(1-D_{13}\right)D_{23}+\left(1-D_{12}\right)D_{13}D_{23}\right\} \right.\nonumber \\
 & \left.\times\left\{ D_{14}D_{15}\left(1-D_{45}\right)+D_{14}\left(1-D_{15}\right)D_{45}+\left(1-D_{14}\right)D_{15}D_{45}\right\} \right]\nonumber \\
= & \frac{1}{\left|\mathrm{iso}\left(\vcenter{\hbox{\includegraphics[scale=0.125]{twostar}}}\right)\right|^{2}}\left\{ \mathbb{E}\left[D_{12}D_{13}\left(1-D_{23}\right)D_{14}D_{15}\left(1-D_{45}\right)\right]\right|\nonumber \\
 & +4\mathbb{E}\left[D_{12}D_{13}\left(1-D_{23}\right)D_{14}\left(1-D_{15}\right)D_{45}\right]\nonumber \\
 & \left.+4\mathbb{E}\left[D_{12}\left(1-D_{13}\right)D_{23}D_{14}\left(1-D_{15}\right)D_{45}\right]\right\} \nonumber \\
= & \frac{1}{\left|\mathrm{iso}\left(\vcenter{\hbox{\includegraphics[scale=0.125]{twostar}}}\right)\right|^{2}}\left[\nu_{1,\vcenter{\hbox{\includegraphics[scale=0.125]{twostar}}}}\left(\vcenter{\hbox{\includegraphics[scale=0.125]{fourstar}}}\right)\Pr\left(\vcenter{\hbox{\includegraphics[scale=0.125]{fourstar}}}=G_{N}\left[\left\{ 1,2,3\right\} \right]\cup G_{N}\left[\left\{ 1,4,5\right\} \right]\right)\right.\nonumber \\
 & +\nu_{1,\vcenter{\hbox{\includegraphics[scale=0.125]{twostar}}}}\left(\vcenter{\hbox{\includegraphics[scale=0.125]{tailedthreestar}}}\right)\Pr\left(\vcenter{\hbox{\includegraphics[scale=0.125]{tailedthreestar}}}=G_{N}\left[\left\{ 1,2,3\right\} \right]\cup G_{N}\left[\left\{ 1,4,5\right\} \right]\right)\nonumber \\
 & \left.+\nu_{1,\vcenter{\hbox{\includegraphics[scale=0.125]{twostar}}}}\left(\vcenter{\hbox{\includegraphics[scale=0.125]{fivepath}}}\right)\Pr\left(\vcenter{\hbox{\includegraphics[scale=0.125]{fivepath}}}=G_{N}\left[\left\{ 1,2,3\right\} \right]\cup G_{N}\left[\left\{ 1,4,5\right\} \right]\right)\right]\nonumber \\
= & \sum_{W\in\mathcal{W}_{1,\vcenter{\hbox{\includegraphics[scale=0.125]{twostar}}}}^{\mathrm{m}}}\nu_{1,\vcenter{\hbox{\includegraphics[scale=0.125]{twostar}}}}\left(W\right)\Pr\left(W=G_{N}\left[\left\{ 1,2,3\right\} \right]\cup G_{N}\left[\left\{ 1,4,5\right\} \right]\right)\nonumber \\
= & \Xi\left(\mathcal{W}_{1,\vcenter{\hbox{\includegraphics[scale=0.125]{twostar}}}}^{\mathrm{m}}\right).\label{eq: XI_W_q1_twostar}
\end{align}

The third equality follows from relationships like $\mathbb{E}\left[D_{12}D_{13}\left(1-D_{23}\right)D_{14}\left(1-D_{15}\right)D_{45}\right]=\mathbb{E}\left[D_{12}D_{13}\left(1-D_{23}\right)\left(1-D_{14}\right)D_{15}D_{45}\right]$,
which allow for the grouping together of terms. The balance of the
equalities are consequences of the definitions given above. 

\begin{figure}
\caption{\label{fig: W_q2_twostar} Stitchings of two-star graphlets with two
common nodes}

\begin{centering}
\includegraphics{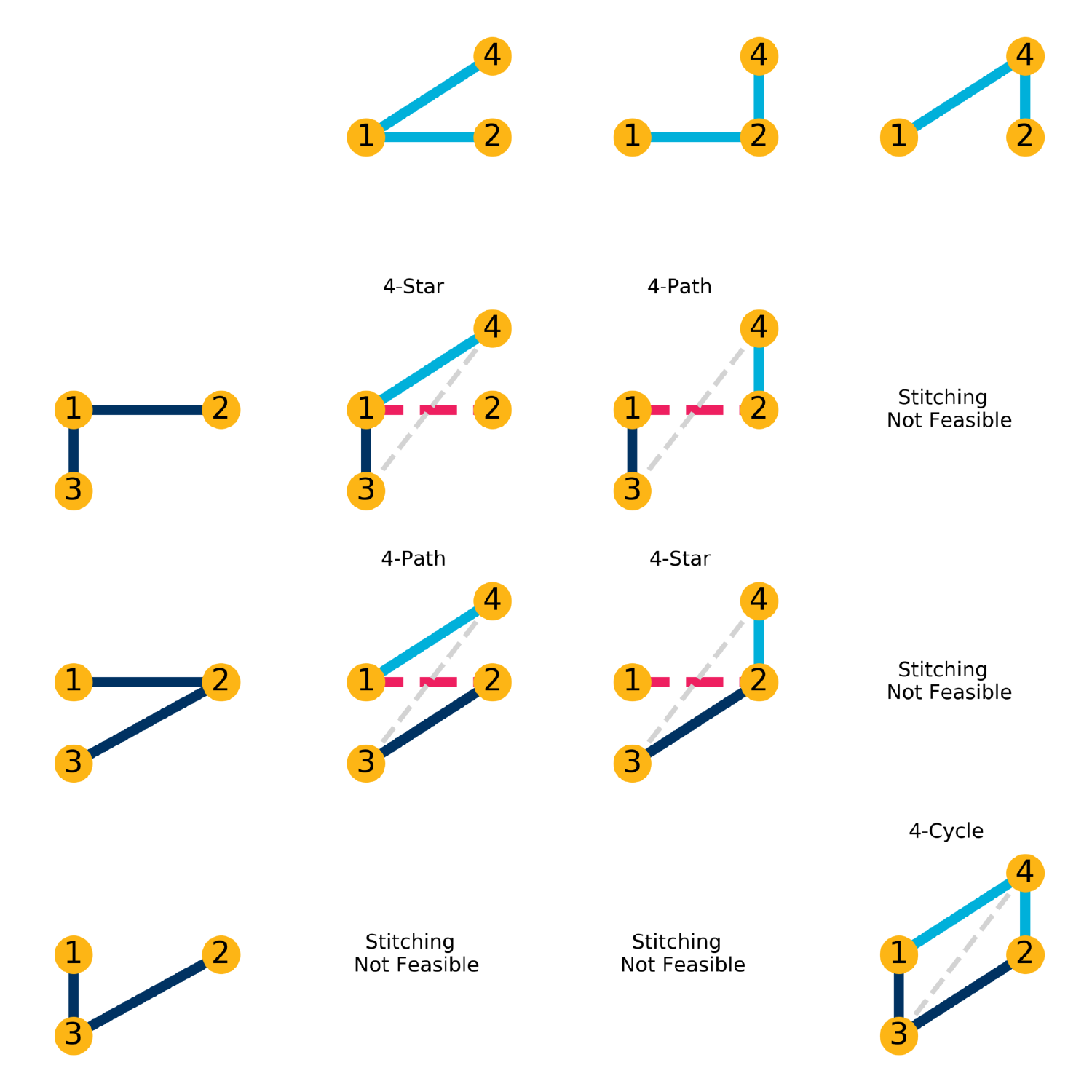}
\par\end{centering}
\textsc{\small{}\uline{Notes:}}{\small{} Depiction of all possible
ways to join (or ``stitch'') a pair of two-star ($\vcenter{\hbox{\includegraphics[scale=0.125]{twostar}}}$)
subgraphs together with two nodes in common. Each of the resulting
subgraphs is a tetrad wiring. The dashed gray edges involve the pair
of nodes that is not common across the pair of two-stars. Hence the
subgraph induced by the four nodes in the tetrad may or may not include
this edge.}{\small\par}

\textsc{\small{}\uline{Source:}}{\small{} Author's calculations.}{\small\par}
\end{figure}

\subsubsection*{Sampling variances}

With the above notations in hand, I now calculate the sampling variances
of $\hat{P}(\vcenter{\hbox{\includegraphics[scale=0.125]{twostar}}})$
and $P_{N}(\vcenter{\hbox{\includegraphics[scale=0.125]{triangle}}})$
as well as their covariance. \citet{Holland_Leinhardt_AJS70,Holland_Leinhardt_SM76}
were the first to derive variance expressions for subgraph counts.
The specific development presented here follows \citet{Bhattacharya_Bickel_AS15}.
A \citet{Hoeffding_AMS48} variance-composition gives
\[
\mathbb{V}\left(\left(\begin{array}{c}
P_{N}\left(\vcenter{\hbox{\includegraphics[scale=0.125]{triangle}}}\right)\\
P_{N}\left(\vcenter{\hbox{\includegraphics[scale=0.125]{twostar}}}\right)
\end{array}\right)\right)=\binom{N}{3}^{-2}\sum_{q=0}^{3}\binom{N}{3}\binom{3}{q}\binom{N-3}{3-q}\left(\begin{array}{cc}
\Sigma_{q}\left(\vcenter{\hbox{\includegraphics[scale=0.125]{triangle}}}\right) & \Sigma_{q}\left(\vcenter{\hbox{\includegraphics[scale=0.125]{triangle}}},\vcenter{\hbox{\includegraphics[scale=0.125]{twostar}}}\right)\\
\Sigma_{q}\left(\vcenter{\hbox{\includegraphics[scale=0.125]{triangle}}},\vcenter{\hbox{\includegraphics[scale=0.125]{twostar}}}\right) & \Sigma_{q}\left(\vcenter{\hbox{\includegraphics[scale=0.125]{twostar}}}\right)
\end{array}\right)
\]
with $\Sigma_{q}\left(\vcenter{\hbox{\includegraphics[scale=0.125]{triangle}}}\right)$,
$\Sigma_{q}\left(\vcenter{\hbox{\includegraphics[scale=0.125]{twostar}}}\right)$
and $\Sigma_{q}\left(\vcenter{\hbox{\includegraphics[scale=0.125]{triangle}}},\vcenter{\hbox{\includegraphics[scale=0.125]{twostar}}}\right)$
as defined by (\ref{eq: SIGMA_q_S_R}) above (using the shorthand
$\Sigma_{q}\left(S,S\right)=\Sigma_{q}\left(S\right)$ etc). Using
the fact that each of these variances and covariances is zero when
$q=0$ and reorganizing terms gives 
\begin{align*}
\mathbb{V}\left(\left(\begin{array}{c}
P_{N}\left(\vcenter{\hbox{\includegraphics[scale=0.125]{triangle}}}\right)\\
P_{N}\left(\vcenter{\hbox{\includegraphics[scale=0.125]{twostar}}}\right)
\end{array}\right)\right)= & \binom{N}{3}^{-2}\sum_{q=1}^{3}\binom{N}{3}\binom{3}{q}\binom{N-3}{3-q}\left[\begin{array}{cc}
\Xi\left(\mathcal{W}_{q,\vcenter{\hbox{\includegraphics[scale=0.125]{triangle}}}}\right) & \Xi\left(\mathcal{W}_{q,\vcenter{\hbox{\includegraphics[scale=0.125]{triangle}}},\vcenter{\hbox{\includegraphics[scale=0.125]{twostar}}}}\right)\\
\Xi\left(\mathcal{W}_{q,\vcenter{\hbox{\includegraphics[scale=0.125]{triangle}}},\vcenter{\hbox{\includegraphics[scale=0.125]{twostar}}}}\right) & \Xi\left(\mathcal{W}_{q,\vcenter{\hbox{\includegraphics[scale=0.125]{twostar}}}}\right)
\end{array}\right]\\
 & -\left[1-\frac{\left(N-3\right)!^{2}}{N!\left(N-6\right)!}\right]\left[\begin{array}{cc}
P\left(\vcenter{\hbox{\includegraphics[scale=0.125]{triangle}}}\right)^{2} & P\left(\vcenter{\hbox{\includegraphics[scale=0.125]{triangle}}}\right)P\left(\vcenter{\hbox{\includegraphics[scale=0.125]{twostar}}}\right)\\
P\left(\vcenter{\hbox{\includegraphics[scale=0.125]{triangle}}}\right)P\left(\vcenter{\hbox{\includegraphics[scale=0.125]{twostar}}}\right) & P\left(\vcenter{\hbox{\includegraphics[scale=0.125]{twostar}}}\right)^{2}
\end{array}\right].
\end{align*}

In what follows I assume that the network generating process is such
that, for each $N$, $\Sigma_{q}\left(\vcenter{\hbox{\includegraphics[scale=0.125]{triangle}}}\right)$
and $\Sigma_{q}\left(\vcenter{\hbox{\includegraphics[scale=0.125]{twostar}}}\right)$
are not identically equal to zero for $q\geq1$. This prevents $P_{N}\left(\vcenter{\hbox{\includegraphics[scale=0.125]{triangle}}}\right)$
and $P_{N}\left(\vcenter{\hbox{\includegraphics[scale=0.125]{twostar}}}\right)$
from exhibiting degenerate U-Statistic-like attributes \citep[c.f., ][Theorem 1]{Graham_EM17}.
The restriction is a real one, ruling out the Erdös-Renyi case. Separate
results for this special case are presented below.

As introduced earlier, in order to accommodate sequences of networks
with varying degrees of sparsity, we can index the underlying population
graphon by $N$, setting $h_{N}\left(u,v\right)=\rho_{N}w\left(u,v\right)$
with $w\left(u,v\right)=f_{\left.U_{i},U_{j}\right|D_{ij}}\left(\left.u,v\right|D_{ij}=1\right)$
and allowing $\rho_{N}\rightarrow0$ as $N\rightarrow\infty$. Under
such a sequence of GGPs $P\left(\vcenter{\hbox{\includegraphics[scale=0.125]{triangle}}}\right)$
and $P\left(\vcenter{\hbox{\includegraphics[scale=0.125]{twostar}}}\right)$
will tend to zero. In order to understand the properties of $P_{N}\left(\vcenter{\hbox{\includegraphics[scale=0.125]{triangle}}}\right)$
vis-a-vis $P\left(\vcenter{\hbox{\includegraphics[scale=0.125]{triangle}}}\right)$
we must normalize. It is natural to normalize according to the number
edges in the subgraph under consideration. 

Let $\tilde{P}\left(\vcenter{\hbox{\includegraphics[scale=0.125]{triangle}}}\right)=P\left(\vcenter{\hbox{\includegraphics[scale=0.125]{triangle}}}\right)/\rho_{N}^{3}$,
$\tilde{P}\left(\vcenter{\hbox{\includegraphics[scale=0.125]{twostar}}}\right)=P\left(\vcenter{\hbox{\includegraphics[scale=0.125]{twostar}}}\right)/\rho_{N}^{2}$,
$\tilde{P}_{N}\left(\vcenter{\hbox{\includegraphics[scale=0.125]{triangle}}}\right)=P_{N}\left(\vcenter{\hbox{\includegraphics[scale=0.125]{triangle}}}\right)/\rho_{N}^{3}$
and so on. So normalizing I get
\begin{align}
\mathbb{V}\left(\left(\begin{array}{c}
\tilde{P}_{N}\left(\vcenter{\hbox{\includegraphics[scale=0.125]{triangle}}}\right)\\
\tilde{P}_{N}\left(\vcenter{\hbox{\includegraphics[scale=0.125]{twostar}}}\right)
\end{array}\right)\right)= & \binom{N}{3}^{-2}\sum_{q=1}^{3}\binom{N}{3}\binom{3}{q}\binom{N-3}{3-q}\left[\begin{array}{cc}
\rho_{N}^{-6}\Xi\left(\mathcal{W}_{q,\vcenter{\hbox{\includegraphics[scale=0.125]{triangle}}}}\right) & \rho_{N}^{-5}\Xi\left(\mathcal{W}_{q,\vcenter{\hbox{\includegraphics[scale=0.125]{triangle}}},\vcenter{\hbox{\includegraphics[scale=0.125]{twostar}}}}\right)\\
\rho_{N}^{-5}\Xi\left(\mathcal{W}_{q,\vcenter{\hbox{\includegraphics[scale=0.125]{triangle}}},\vcenter{\hbox{\includegraphics[scale=0.125]{twostar}}}}\right) & \rho_{N}^{-4}\Xi\left(\mathcal{W}_{q,\vcenter{\hbox{\includegraphics[scale=0.125]{twostar}}}}\right)
\end{array}\right]\nonumber \\
 & -\left[1-\frac{\left(N-3\right)!^{2}}{N!\left(N-6\right)!}\right]\left[\begin{array}{cc}
\tilde{P}\left(\vcenter{\hbox{\includegraphics[scale=0.125]{triangle}}}\right)^{2} & \tilde{P}\left(\vcenter{\hbox{\includegraphics[scale=0.125]{triangle}}}\right)\tilde{P}\left(\vcenter{\hbox{\includegraphics[scale=0.125]{twostar}}}\right)\\
\tilde{P}\left(\vcenter{\hbox{\includegraphics[scale=0.125]{triangle}}}\right)\tilde{P}\left(\vcenter{\hbox{\includegraphics[scale=0.125]{twostar}}}\right) & \tilde{P}\left(\vcenter{\hbox{\includegraphics[scale=0.125]{twostar}}}\right)^{2}
\end{array}\right].\label{eq: Var_P_Triangle_Twostar}
\end{align}

Expression (\ref{eq: Var_P_Triangle_Twostar}) agrees with the corresponding
expression of \citet{Bhattacharya_Bickel_AS15} for injective homomorphism
frequencies (Equation (3.8), p. 2395).\footnote{See also \citet[Lemma 1]{Green_Shalizi_arXiv17}.}
The main difference is the analog of $\Xi\left(\mathcal{W}_{q,\vcenter{\hbox{\includegraphics[scale=0.125]{triangle}}}}\right)$
in their expression is itself an injective homomorphism density, whereas
here $\Xi\left(\mathcal{W}_{q,\vcenter{\hbox{\includegraphics[scale=0.125]{triangle}}}}\right)$
is neither an injective homomorphism nor an induced subgraph density
and instead involves checking for particular patterns of \emph{both}
adjacency and non-adjacency as described above.

\subsubsection*{Rates of convergence}

To understand the rate of convergence in mean square of, for example,
$\tilde{P}_{N}\left(\vcenter{\hbox{\includegraphics[scale=0.125]{triangle}}}\right)$
toward $\tilde{P}\left(\vcenter{\hbox{\includegraphics[scale=0.125]{triangle}}}\right)$,
we need to determine the order of each of the terms in (\ref{eq: Var_P_Triangle_Twostar}).
Let $e\left(R\right)\overset{def}{\equiv}\left|\mathcal{E}\left(R\right)\right|$
and $e\left(S\right)\overset{def}{\equiv}\left|\mathcal{E}\left(S\right)\right|$
denote the number of edges in graphlets $R$ and $S$. Next observe
that $\binom{N}{p}^{-1}\binom{p}{q}\binom{N-p}{p-q}=O\left(N^{-q}\right)$.
We therefore have that the terms in the summation indexed by $q$
in (\ref{eq: Var_P_Triangle_Twostar}) are $O\left(N^{-q}\rho_{N}^{-e\left(R\right)}\rho_{N}^{-e\left(S\right)}\right)O\left(\Xi\left(\mathcal{W}_{q,R,S}\right)\right)$
for $q=1,\ldots,p.$ I divide these terms, closely following \citet{Bhattacharya_Bickel_AS15},
into three cases:

\textbf{\uline{Case 1 (\mbox{$q=1$}):}}\textbf{ }when $q=1$ the
number of edges in all elements of $\mathcal{W}_{q,R,S}$ equals $e\left(R\right)+e\left(S\right)$
for any subgraphs $R$ and $S$. Hence $O\left(\Xi\left(\mathcal{W}_{1,R,S}\right)\right)=O\left(\rho_{N}^{e\left(R\right)}\rho_{N}^{e\left(S\right)}\right)$,
yielding
\[
O\left(N^{-1}\rho_{N}^{-e\left(R\right)}\rho_{N}^{-e\left(S\right)}\right)O\left(\Xi\left(\mathcal{W}_{1,R,S}\right)\right)=O\left(N^{-1}\right).
\]
The $q=1$ summand in (\ref{eq: Var_P_Triangle_Twostar}) is of order
$N^{-1}$. In general, from the theory of U-statistics, one would
expect this to be the leading variance term; however, the present
situation is more complicated.

\textbf{\uline{Case 2 (\mbox{$q=3$} or \mbox{$q=p$}):}} In this
case the order of $\Xi\left(\mathcal{W}_{3,\vcenter{\hbox{\includegraphics[scale=0.125]{triangle}}}}\right)$
is $O\left(\rho_{N}^{3}\right)$, $\Xi\left(\mathcal{W}_{3,\vcenter{\hbox{\includegraphics[scale=0.125]{twostar}}}}\right)$
is $O\left(\rho_{N}^{2}\right)$ and $\text{\ensuremath{\mathcal{W}_{3,\vcenter{\hbox{\includegraphics[scale=0.125]{triangle}}},\vcenter{\hbox{\includegraphics[scale=0.125]{twostar}}}}}}$
is empty so that $\Xi\left(\text{\ensuremath{\mathcal{W}_{3,\vcenter{\hbox{\includegraphics[scale=0.125]{triangle}}},\vcenter{\hbox{\includegraphics[scale=0.125]{twostar}}}}}}\right)=0$.
Therefore, recalling that $\lambda_{N}=\left(N-1\right)\rho_{N}$
equals average degree,
\begin{align*}
O\left(N^{-3}\rho_{N}^{-2e\left(\vcenter{\hbox{\includegraphics[scale=0.125]{triangle}}}\right)}\right)O\left(\Xi\left(\mathcal{W}_{3,\vcenter{\hbox{\includegraphics[scale=0.125]{triangle}}}}\right)\right) & =O\left(\lambda_{N}^{-3}\right)\\
O\left(N^{-3}\rho_{N}^{-2e\left(\vcenter{\hbox{\includegraphics[scale=0.125]{twostar}}}\right)}\right)O\left(\Xi\left(\mathcal{W}_{3,\vcenter{\hbox{\includegraphics[scale=0.125]{twostar}}}}\right)\right) & =O\left(N^{-1}\lambda_{N}^{-2}\right)\\
O\left(N^{-3}\rho_{N}^{-e\left(\vcenter{\hbox{\includegraphics[scale=0.125]{triangle}}}\right)}\rho_{N}^{-e\left(\vcenter{\hbox{\includegraphics[scale=0.125]{twostar}}}\right)}\right)O\left(\Xi\left(\mathcal{W}_{3,\vcenter{\hbox{\includegraphics[scale=0.125]{triangle}}},\vcenter{\hbox{\includegraphics[scale=0.125]{twostar}}}}\right)\right) & =o\left(1\right).
\end{align*}

\textbf{\uline{Case 3 (\mbox{$q=2$} or \mbox{$\left(2\leq q\leq p-1\right)$}):}}
Here the order of $\Xi\left(\mathcal{W}_{2,\vcenter{\hbox{\includegraphics[scale=0.125]{triangle}}}}\right)$
equals $O\left(\rho_{N}^{2e\left(\vcenter{\hbox{\includegraphics[scale=0.125]{triangle}}}\right)-\left(q-1\right)}\right)=O\left(\rho_{N}^{5}\right)$,
$\Xi\left(\mathcal{W}_{2,\vcenter{\hbox{\includegraphics[scale=0.125]{twostar}}}}\right)$
equals $O\left(\rho_{N}^{2e\left(\vcenter{\hbox{\includegraphics[scale=0.125]{twostar}}}\right)-\left(q-1\right)}\right)=O\left(\rho_{N}^{3}\right)$
and that of $\Xi\left(\text{\ensuremath{\mathcal{W}_{2,\vcenter{\hbox{\includegraphics[scale=0.125]{triangle}}},\vcenter{\hbox{\includegraphics[scale=0.125]{twostar}}}}}}\right)$
equals $O\left(\rho_{N}^{e\left(\vcenter{\hbox{\includegraphics[scale=0.125]{triangle}}}\right)+e\left(\vcenter{\hbox{\includegraphics[scale=0.125]{twostar}}}\right)-\left(q-1\right)}\right)=O\left(\rho_{N}^{4}\right)$
. Therefore
\begin{align*}
O\left(N^{-2}\rho_{N}^{-2e\left(\vcenter{\hbox{\includegraphics[scale=0.125]{triangle}}}\right)}\right)O\left(\Xi\left(\mathcal{W}_{2,\vcenter{\hbox{\includegraphics[scale=0.125]{triangle}}}}\right)\right) & =O\left(N^{-1}\lambda_{N}^{-1}\right)\\
O\left(N^{-2}\rho_{N}^{-2e\left(\vcenter{\hbox{\includegraphics[scale=0.125]{twostar}}}\right)}\right)O\left(\Xi\left(\mathcal{W}_{2,\vcenter{\hbox{\includegraphics[scale=0.125]{twostar}}}}\right)\right) & =O\left(N^{-1}\lambda_{N}^{-1}\right)\\
O\left(N^{-2}\rho_{N}^{-e\left(\vcenter{\hbox{\includegraphics[scale=0.125]{triangle}}}\right)}\rho_{N}^{-e\left(\vcenter{\hbox{\includegraphics[scale=0.125]{twostar}}}\right)}\right)O\left(\Xi\left(\mathcal{W}_{2,\vcenter{\hbox{\includegraphics[scale=0.125]{triangle}}},\vcenter{\hbox{\includegraphics[scale=0.125]{twostar}}}}\right)\right) & =O\left(N^{-1}\lambda_{N}^{-1}\right).
\end{align*}
For the two variance terms we have
\begin{align*}
\mathbb{V}\left(\tilde{P}_{N}\left(\vcenter{\hbox{\includegraphics[scale=0.125]{triangle}}}\right)\right) & =O\left(\frac{1}{N}\right)+O\left(\frac{1}{N\lambda_{N}}\right)+O\left(\frac{1}{\lambda_{N}^{3}}\right)\\
\mathbb{V}\left(\tilde{P}_{N}\left(\vcenter{\hbox{\includegraphics[scale=0.125]{twostar}}}\right)\right) & =O\left(\frac{1}{N}\right)+O\left(\frac{1}{N\lambda_{N}}\right)+O\left(\frac{1}{N\lambda_{N}^{2}}\right)
\end{align*}
indicating that the rate at which, for example, $\tilde{P}_{N}\left(\vcenter{\hbox{\includegraphics[scale=0.125]{triangle}}}\right)$
converges in mean square toward $\tilde{P}\left(\vcenter{\hbox{\includegraphics[scale=0.125]{triangle}}}\right)$,
depends on the behavior of average degree as the network grows large.
This reflects the fact that, depending on a combination of the nature
of the graphlet of interest and the rate at which $\lambda_{N}$ does,
or does not, grows with $N$, several of the terms in (\ref{eq: Var_P_Triangle_Twostar})
may be of equal order.

For any increasing sequence of average degree we have 
\begin{align*}
\mathbb{V}\left(\tilde{P}_{N}\left(\vcenter{\hbox{\includegraphics[scale=0.125]{triangle}}}\right)\right)= & O\left(\max\left(\frac{1}{N},\frac{1}{\lambda_{N}^{3}}\right)\right)\\
\mathbb{V}\left(\tilde{P}_{N}\left(\vcenter{\hbox{\includegraphics[scale=0.125]{twostar}}}\right)\right)= & O\left(\frac{1}{N}\right)\\
\mathbb{C}\left(\tilde{P}_{N}\left(\vcenter{\hbox{\includegraphics[scale=0.125]{triangle}}}\right),\tilde{P}_{N}\left(\vcenter{\hbox{\includegraphics[scale=0.125]{twostar}}}\right)\right)= & O\left(\frac{1}{N}\right).
\end{align*}
If $\lambda_{N}\geq CN^{1/3}$, then the rate of convergence is $\sqrt{N}$
for both $\tilde{P}_{N}\left(\vcenter{\hbox{\includegraphics[scale=0.125]{twostar}}}\right)$
and $\tilde{P}_{N}\left(\vcenter{\hbox{\includegraphics[scale=0.125]{triangle}}}\right)$.
In the sparse case, with $\lambda_{N}\rightarrow\lambda$, $\tilde{P}\left(\vcenter{\hbox{\includegraphics[scale=0.125]{twostar}}}\right)$,
due to the acyclic structure of the two-star graphlet, remains estimable
at the $\sqrt{N}$ rate. However in this case all three of its variance
terms are of equal order. In contrast $\tilde{P}\left(\vcenter{\hbox{\includegraphics[scale=0.125]{triangle}}}\right)$
is (evidently) not consistently estimable in the sparse case.

\subsubsection*{Asymptotic normality}

When average degree is $\lambda_{N}>CN^{1/3}$, such that both $\tilde{P}_{N}\left(\vcenter{\hbox{\includegraphics[scale=0.125]{twostar}}}\right)$
and $\tilde{P}_{N}\left(\vcenter{\hbox{\includegraphics[scale=0.125]{triangle}}}\right)$
converge at the $\sqrt{N}$ rate, an application of Theorem 1.c of
\citet{Bickel_et_al_AS11} establishes that, under some regularity
conditions,
\begin{equation}
\sqrt{N}\left(\begin{array}{c}
\tilde{P}_{N}\left(\vcenter{\hbox{\includegraphics[scale=0.125]{triangle}}}\right)-\tilde{P}\left(\vcenter{\hbox{\includegraphics[scale=0.125]{triangle}}}\right)\\
\tilde{P}_{N}\left(\vcenter{\hbox{\includegraphics[scale=0.125]{twostar}}}\right)-\tilde{P}\left(\vcenter{\hbox{\includegraphics[scale=0.125]{twostar}}}\right)
\end{array}\right)\overset{D}{\rightarrow}\mathcal{N}\left(\left(\begin{array}{c}
0\\
0
\end{array}\right),9\left(\begin{array}{cc}
\tilde{\Sigma}_{1}\left(\vcenter{\hbox{\includegraphics[scale=0.125]{triangle}}}\right) & \tilde{\Sigma}_{1}\left(\vcenter{\hbox{\includegraphics[scale=0.125]{triangle}}},\vcenter{\hbox{\includegraphics[scale=0.125]{twostar}}}\right)\\
\tilde{\Sigma}_{1}\left(\vcenter{\hbox{\includegraphics[scale=0.125]{triangle}}},\vcenter{\hbox{\includegraphics[scale=0.125]{twostar}}}\right) & \tilde{\Sigma}_{1}\left(\vcenter{\hbox{\includegraphics[scale=0.125]{twostar}}}\right)
\end{array}\right)\right),\label{eq: normality_triangle_twostar}
\end{equation}
where $\tilde{\Sigma}_{1}\left(\vcenter{\hbox{\includegraphics[scale=0.125]{triangle}}}\right)=\rho_{N}^{-6}\Sigma_{1}\left(\vcenter{\hbox{\includegraphics[scale=0.125]{triangle}}}\right)$,
$\tilde{\Sigma}_{1}\left(\vcenter{\hbox{\includegraphics[scale=0.125]{triangle}}},\vcenter{\hbox{\includegraphics[scale=0.125]{twostar}}}\right)=\rho_{N}^{-5}\Sigma_{1}\left(\vcenter{\hbox{\includegraphics[scale=0.125]{triangle}}},\vcenter{\hbox{\includegraphics[scale=0.125]{twostar}}}\right)$
etc. Proving (\ref{eq: normality_triangle_twostar}) is relatively
straightforward. I do not sketch the argument here, but note that
the main tools needed were already introduced in the analysis of dyadic
regression appearing in Section \ref{sec: dyadic_regression} above.

As noted previously, if $\lambda_{N}\rightarrow\lambda$ as $N\rightarrow\infty$
such that the network is sparse in the limit, then a general result
on $\sqrt{N}\left(\tilde{P}_{N}\left(\vcenter{\hbox{\includegraphics[scale=0.125]{triangle}}}\right)-\tilde{P}\left(\vcenter{\hbox{\includegraphics[scale=0.125]{triangle}}}\right)\right)$
is unavailable. In contrast, part (b) of Theorem 1 in \citet{Bickel_et_al_AS11}
implies that not only does $\tilde{P}_{N}\left(\vcenter{\hbox{\includegraphics[scale=0.125]{twostar}}}\right)$
remain $\sqrt{N}$ consistent for $\tilde{P}\left(\vcenter{\hbox{\includegraphics[scale=0.125]{twostar}}}\right)$
in this case, but also that $\sqrt{N}\left(\tilde{P}_{N}\left(\vcenter{\hbox{\includegraphics[scale=0.125]{twostar}}}\right)-\tilde{P}\left(\vcenter{\hbox{\includegraphics[scale=0.125]{twostar}}}\right)\right)$
remains asymptotically normal. The limiting variance in this case
differs from the one given in (\ref{eq: normality_triangle_twostar});
all terms in $\mathbb{V}\left(\tilde{P}_{N}\left(\vcenter{\hbox{\includegraphics[scale=0.125]{twostar}}}\right)\right)$
are of equal order (and hence should be retained). 

More generally the sampling properties of induced subgraph frequencies
under sparse graph limits remains relatively unexplored. The sensitivity
of rates of convergence and distributional properties to assumptions
about $\lambda_{N}$ raises concerns about uniformity of inference
procedures. A similar concern is suggested by the properties of these
statistics when the graphon is constant. This last case is considered
next.

\subsubsection*{Two-star and triangle counts in Erdös-Renyi networks}

The analysis above assumes that the graphon is such that $\mathbb{C}\left(\mathbf{1}\left(R\cong G_{N}\left[\mathbf{i}_{p}\right]\right),\mathbf{1}\left(S\cong G_{N}\left[\mathbf{j}_{p}\right]\right)\right)\neq0$
when $\mathbf{i}_{p}$ and $\mathbf{j}_{p}$ share exactly one index
in common (such that $\Xi\left(\mathcal{W}_{1,S}\right)-P\left(S\right)^{2}>0$).
This condition will generally hold for graphons which vary in $u$
and $v$ (such that the events $D_{12}=1$ and $D_{13}=1$ are not
independent), but it does rule out the Erdös-Renyi case (where links
form independently with constant probability $\rho$).\footnote{See \citet{Menzel_arXiv17} for more examples of degenerate graphons.}
This graph generation process has been extensively studied by probabilists
for over sixty years \citep[e.g., ][]{Janson_et_al_RG00}. 

In statistics, \citet{Janson_Nowicki_PTRF91} and \citet{Nowicki_SN91}
studied the sampling properties of induced and partial subgraph frequencies
when the network is an Erdös-Renyi one. They demonstrated asymptotic
normality of such frequencies with a $\sqrt{\tbinom{N}{2}}$ rate
of convergence. These earlier results, at first glance, appear to
be in tension with the more general results of \citet{Bickel_et_al_AS11},
who showed asymptotic normality with a $\sqrt{N}$ rate of convergence
under general graphons. It turns out, however, that the leading (i.e.,
$q=1$) term in (\ref{eq: Var_P_Triangle_Twostar}) is identically
equal to zero under the Erdös-Renyi GPP. The Erdös-Renyi GPP is a
``degenerate'' special case.

To see this, evaluate the stitching probabilities (\ref{eq: stitching_probability})
under the Erdös-Renyi GPP to get
\[
\Xi\left(\mathcal{W}_{1,\vcenter{\hbox{\includegraphics[scale=0.125]{twostar}}}}\right)=\rho^{4}\left(1-\rho\right)^{2},\thinspace\Xi\left(\mathcal{W}_{2,\vcenter{\hbox{\includegraphics[scale=0.125]{twostar}}}}\right)=\frac{4}{9}\rho^{3}\left(1-\rho\right)^{2}+\frac{1}{9}\rho^{4}\left(1-\rho\right),\thinspace\Xi\left(\mathcal{W}_{3,\vcenter{\hbox{\includegraphics[scale=0.125]{twostar}}}}\right)=\frac{1}{3}\rho^{2}\left(1-\rho\right)
\]
and

\[
\Xi\left(\mathcal{W}_{1,\vcenter{\hbox{\includegraphics[scale=0.125]{triangle}}}}\right)=\rho^{6},\thinspace\Xi\left(\mathcal{W}_{2,\vcenter{\hbox{\includegraphics[scale=0.125]{triangle}}}}\right)=\rho^{5},\thinspace\Xi\left(\mathcal{W}_{3,\vcenter{\hbox{\includegraphics[scale=0.125]{triangle}}}}\right)=\rho^{3}
\]
and
\[
\Xi\left(\mathcal{W}_{1,\vcenter{\hbox{\includegraphics[scale=0.125]{twostar}}},\vcenter{\hbox{\includegraphics[scale=0.125]{triangle}}}}\right)=\rho^{5}\left(1-\rho\right),\thinspace\Xi\left(\mathcal{W}_{2,\vcenter{\hbox{\includegraphics[scale=0.125]{twostar}}},\vcenter{\hbox{\includegraphics[scale=0.125]{triangle}}}}\right)=\frac{2}{3}\rho^{4}\left(1-\rho\right),\thinspace\Xi\left(\mathcal{W}_{3,\vcenter{\hbox{\includegraphics[scale=0.125]{twostar}}},\vcenter{\hbox{\includegraphics[scale=0.125]{triangle}}}}\right)=0.
\]

Under these graphlet stitching probabilities the $q=1$ variance term,
which is generally the leading variance term in \citet[Theorem 1]{Bickel_et_al_AS11},
instead equals
\[
\left(\begin{array}{cc}
\Sigma_{1}\left(\vcenter{\hbox{\includegraphics[scale=0.125]{triangle}}}\right) & \Sigma_{1}\left(\vcenter{\hbox{\includegraphics[scale=0.125]{triangle}}},\vcenter{\hbox{\includegraphics[scale=0.125]{twostar}}}\right)\\
\Sigma_{1}\left(\vcenter{\hbox{\includegraphics[scale=0.125]{triangle}}},\vcenter{\hbox{\includegraphics[scale=0.125]{twostar}}}\right) & \Sigma_{1}\left(\vcenter{\hbox{\includegraphics[scale=0.125]{twostar}}}\right)
\end{array}\right)=\mathbf{0}_{2}\mathbf{0}_{2}'.
\]
Hence, under the (dense) Erdös-Renyi GPP, the leading variance term
is instead the $q=2$ one, yielding for $0<\rho<1$ but $\rho\neq2/3$,
\begin{align}
\sqrt{\tbinom{N}{2}}\left(\begin{array}{c}
P_{N}\left(\vcenter{\hbox{\includegraphics[scale=0.125]{triangle}}}\right)-P\left(\vcenter{\hbox{\includegraphics[scale=0.125]{triangle}}}\right)\\
P_{N}\left(\vcenter{\hbox{\includegraphics[scale=0.125]{twostar}}}\right)-P\left(\vcenter{\hbox{\includegraphics[scale=0.125]{twostar}}}\right)
\end{array}\right) & \overset{D}{\rightarrow}\mathcal{N}\left(\left(\begin{array}{c}
0\\
0
\end{array}\right),9\left(\begin{array}{cc}
\Sigma_{2}\left(\vcenter{\hbox{\includegraphics[scale=0.125]{triangle}}}\right) & \Sigma_{2}\left(\vcenter{\hbox{\includegraphics[scale=0.125]{triangle}}},\vcenter{\hbox{\includegraphics[scale=0.125]{twostar}}}\right)\\
\Sigma_{2}\left(\vcenter{\hbox{\includegraphics[scale=0.125]{triangle}}},\vcenter{\hbox{\includegraphics[scale=0.125]{twostar}}}\right) & \Sigma_{2}\left(\vcenter{\hbox{\includegraphics[scale=0.125]{twostar}}}\right)
\end{array}\right)\right)\label{eq: jason_nowicki}
\end{align}
where
\[
\left(\begin{array}{cc}
\Sigma_{2}\left(\vcenter{\hbox{\includegraphics[scale=0.125]{triangle}}}\right) & \Sigma_{2}\left(\vcenter{\hbox{\includegraphics[scale=0.125]{triangle}}},\vcenter{\hbox{\includegraphics[scale=0.125]{twostar}}}\right)\\
\Sigma_{2}\left(\vcenter{\hbox{\includegraphics[scale=0.125]{triangle}}},\vcenter{\hbox{\includegraphics[scale=0.125]{twostar}}}\right) & \Sigma_{2}\left(\vcenter{\hbox{\includegraphics[scale=0.125]{twostar}}}\right)
\end{array}\right)=\rho^{3}\left(1-\rho\right)\left(\begin{array}{cc}
\rho^{2} & \frac{1}{3}\rho\left(2-3\rho\right)\\
\frac{1}{3}\rho\left(2-3\rho\right) & \frac{1}{9}\left(2-3\rho\right)^{2}
\end{array}\right).
\]
See Corollaries 2 and 4 of \citet{Nowicki_SN91} for additional context
and references to the primary literature.

\subsubsection*{Uniform Inference}

The analysis of the previous two subsections showed how the limiting
distributions of two-star and triangle frequencies are sensitive to
the form of the graphon, $h_{N}\left(u,v\right)=\rho_{N}w\left(u,v\right)$.
If $\rho_{N}$ approaches zero too quickly, or $w\left(u,v\right)$
is a constant, the rate of convergence of the estimator changes. This
raises concerns about how to conduct inference in settings where the
limiting graph is `close to sparse' and/or the graphon is `nearly'
constant, or equivalently, dependence across dyads sharing agents
in common is weak. In such settings an approach to inference based
on (\ref{eq: normality_triangle_twostar}), may have poor properties
when $N$ is finite. This is because the $q=1$, $q=2$ and $q=3$
terms in the variance expression (\ref{eq: Var_P_Triangle_Twostar})
may all be of similar order. For this reason, it seems advisable to
keep all terms when calculating variances for test statistics. Clearly,
there are open questions on how best to undertake testing in this
setting.

\subsubsection*{Application of results to inference on transitivity in Nyakatoke}

\begin{figure}
\caption{\label{fig: nyakatoke_risk_sharing_network}Nyakatoke risk-sharing
network}

\begin{centering}
\includegraphics[scale=0.8]{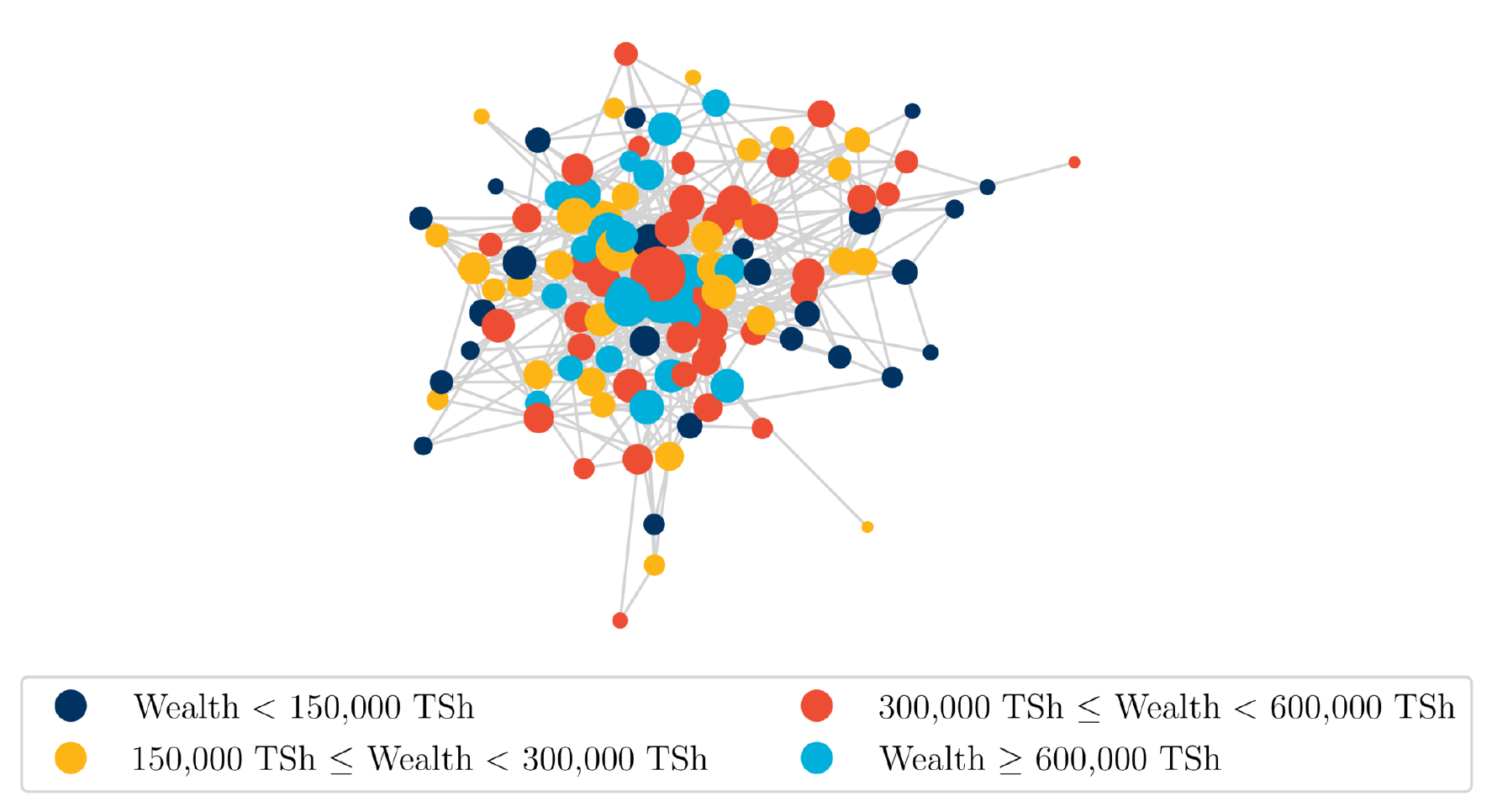}
\par\end{centering}
\uline{Sources:} \citet{deWeerdt_IAP04} and authors' calculations.
`TSh' is an abbreviation for Tanzanian Shillings.
\end{figure}

\citet{deWeerdt_IAP04} collected information of risk-sharing relationships
across 119 houses in Nyakatoke, a small village in Tanzania (see Figure
\ref{sec: Strategic-models}). The density of this network is 0.0698,
while its transitivity index is 0.1884, nearly three times a large.
A natural question is whether the high transitivity index simple reflects
``chance'' or is a real feature of Nyakatoke. To assess this I construct
a confidence interval for the transitivity index using the delta method
and the results outlined above. Other than the empirical illustration
included in \citet{Bhattacharya_Bickel_AS15}, I am aware of no other
published examples of large network inference on the transitivity
index.

The natural analog estimates of $\Xi\left(\mathcal{W}_{1,\vcenter{\hbox{\includegraphics[scale=0.125]{triangle}}}}\right)$,
$\Xi\left(\mathcal{W}_{1,\vcenter{\hbox{\includegraphics[scale=0.125]{twostar}}}}\right)$,
and $\Xi\left(\mathcal{W}_{1,\vcenter{\hbox{\includegraphics[scale=0.125]{triangle}}},\vcenter{\hbox{\includegraphics[scale=0.125]{twostar}}}}\right)$
involve summations over all $\tbinom{N}{3}\tbinom{3}{1}\tbinom{N-3}{3-1}=30\times\tbinom{N}{5}$
\emph{pairs} of triads sharing exactly one common agent. This requires
evaluating the configuration of all $\tbinom{N}{5}$ \emph{pentads}
in the network; a computationally non-trivial task even for medium-sized
networks.\footnote{For each pentad we look at the thirty pairs of triads that can be
constructed from it, such that the two triads share exactly one agent
in common.} It is for this reason that \citet{Bhattacharya_Bickel_AS15} suggest
a subsampling approach to variance estimation. 

For the Nyakatoke network we have a total of $\binom{119}{3}=273,819$
triad configurations to count and a total of $\binom{119}{5}=182,637,273$
pentads that need to be inspected in order to calculate variances.
These are large numbers, but nevertheless small enough enough for
a desktop computer to handle in a few minutes. Direct calculation
gives 
\[
P_{N}\left(\vcenter{\hbox{\includegraphics[scale=0.125]{triangle}}}\right)=\begin{array}{c}
0.00115\\
(0.00030)
\end{array},\thinspace\thinspace P_{N}\left(\vcenter{\hbox{\includegraphics[scale=0.125]{twostar}}}\right)=\begin{array}{c}
0.00496\\
(0.00100)
\end{array}
\]
These standard errors include estimates of both the first \emph{and}
last terms in (\ref{eq: Var_P_Triangle_Twostar}) above, although
the second of these is asymptotically negligible as long as average
degree grows fast enough (which is assumed for the asymptotic normality
result).

Applying the delta method I get an estimated standard error for the
transitivity index of 0.011; this suggests that transitivity is significantly
greater than what we would expect to observe under the Erdös-Renyi
random graph null.

\subsection{\label{subsec: the-degree-sequences}Moments of the degree distribution}

Networks are complex objects, making their analysis both conceptually
and technically challenging. One approach to simplification involves
looking only at the number of links each agent has, that is their
degree, $D_{i+}=\sum_{j\neq i}D_{ij}$, ignoring all other architectural
features of the network. Indeed, a substantial empirical literature
focuses on the degree sequence of a network as its primary object
of interest \citep{Barabasi_Albert_Sc99,Barabasi_Book16}.

Most real world networks exhibit substantial degree heterogeneity,
making the degree sequence an interesting statistic to study and model.
A network's degree sequence is also straightforward to measure. A
researcher need only ask about the number of friends, suppliers, or
partners each agent has, not their identity. Many general purpose
datasets collect such information. For example General Social Survey
(GSS) sometimes collects information on the number of close confidants
\citep[cf., ][]{Marsden_ASR87,McPherson_et_al_ASR06}, while demographers
routinely collect information on the number of lifetime and/or concurrent
sexual partners. Simplicity and data availability both drive the substantial
focus on degree distributions in empirical work.

It is possible for two graphs with the same degree sequence to be
topologically different; their diameters and transitivity indices,
for example, may differ substantially. At the same time a network's
degree sequence is an important summary statistic, constraining other
features of it, even local ones. This is shown by Theorem \ref{thm: degree_sequence_moments},
which I believe is an original result (albeit perhaps folk wisdom;
\citealp[cf., ][]{Snijders_SM06}).
\begin{thm}
\textsc{\label{thm: degree_sequence_moments}(Degree Sequence Moments)}
Let $G$ be an exchangeable random graph of order $N$. The $m^{th}$
moment of $D_{i+}=\sum_{j\neq i}D_{ij}$ equals
\begin{align*}
\mathbb{E}\left[D_{i+}^{m}\right] & =\sum_{k=1}^{m}C_{k,m}\times\mathbb{E}\left[D_{ij_{1}}\times\cdots\times D_{ij_{k}}\right]
\end{align*}
for $m=1,2,\ldots,N-1$ and $C_{k,m}=\tbinom{N-1}{k}\left(\sum_{\mathbf{p}\in\mathcal{P}_{k,m}}\frac{m!}{p_{1}!\times\cdots\times p_{k}!}\right)$
with 
\[
\mathcal{P}_{k,m}=\left\{ \left(p_{1},\ldots,p_{k}\right)\thinspace:\thinspace\sum_{j=1}^{k}p_{j}=m,\thinspace p_{j}\in\mathbb{N}\text{ for \ensuremath{j=1,\ldots,k}}\right\} 
\]
 and $\mathbb{N}$ the set of positive integers.
\end{thm}
\begin{proof}
See Appendix \ref{app: proofs_and_derivations}.
\end{proof}
Theorem \ref{thm: degree_sequence_moments} implies that the first
four uncentered moments of $D_{i+}$ equal
\begin{align}
\mathbb{E}\left[D_{i+}\right]= & \left(N-1\right)P\left(\vcenter{\hbox{\includegraphics[scale=0.125]{edge}}}\right)\label{eq: degree_sequence_m1}\\
\mathbb{E}\left[D_{i+}^{2}\right]= & \left(N-1\right)P\left(\vcenter{\hbox{\includegraphics[scale=0.125]{edge}}}\right)+\left(N-1\right)\left(N-2\right)Q\left(\vcenter{\hbox{\includegraphics[scale=0.125]{twostar}}}\right)\label{eq: degree_sequence_m2}\\
\mathbb{E}\left[D_{i+}^{3}\right]= & \left(N-1\right)P\left(\vcenter{\hbox{\includegraphics[scale=0.125]{edge}}}\right)+3\left(N-1\right)\left(N-2\right)Q\left(\vcenter{\hbox{\includegraphics[scale=0.125]{twostar}}}\right)\label{eq: degree_sequence_m3}\\
 & +\left(N-1\right)\left(N-2\right)\left(N-3\right)Q\left(\vcenter{\hbox{\includegraphics[scale=0.125]{onethreewheel}}}\right)\nonumber \\
\mathbb{E}\left[D_{i+}^{4}\right]= & \left(N-1\right)P\left(\vcenter{\hbox{\includegraphics[scale=0.125]{edge}}}\right)+7\left(N-1\right)\left(N-2\right)Q\left(\vcenter{\hbox{\includegraphics[scale=0.125]{twostar}}}\right)\label{eq: degree_sequence_m4}\\
 & +6\left(N-1\right)\left(N-2\right)\left(N-3\right)Q\left(\vcenter{\hbox{\includegraphics[scale=0.125]{onethreewheel}}}\right)\nonumber \\
 & +\left(N-1\right)\left(N-2\right)\left(N-3\right)\left(N-4\right)Q\left(\vcenter{\hbox{\includegraphics[scale=0.125]{fourstar}}}\right).\nonumber 
\end{align}

In dense networks it is natural to divide $D_{i+}$ by $N-1$. With
degrees so normalized, all terms in equations (\ref{eq: degree_sequence_m1})
to (\ref{eq: degree_sequence_m4}) are asymptotically dominated by
the last one as $N\rightarrow\infty$. Hence, in the limit, the $k^{th}$
moment of normalized degree equals the injective homomorphism density
of $k$-stars in the limiting graphon \citep[cf., ][Lemma 4.1]{Diaconis_Holmes_Janson_IM08}.
In the dense case we have, for example, that
\begin{equation}
\underset{N\rightarrow\infty}{\lim}\mathbb{V}\left(\frac{D_{i+}}{N-1}\right)=Q\left(\vcenter{\hbox{\includegraphics[scale=0.125]{twostar}}}\right)-P\left(\vcenter{\hbox{\includegraphics[scale=0.125]{edge}}}\right)P\left(\vcenter{\hbox{\includegraphics[scale=0.125]{edge}}}\right).\label{eq: degree_variance_dense}
\end{equation}

When the network is not dense, the natural normalization is instead
by average degree, $\lambda_{N}=\left(N-1\right)\rho_{N}$, which
may no longer be proportional to $N$ in the limit. In the sparse
case, $\lambda_{N}\rightarrow\lambda$, and all terms in equations
(\ref{eq: degree_sequence_m1}) to (\ref{eq: degree_sequence_m4})
are of equal order. For example, the fourth uncentered moment of $D_{i+}/\lambda_{N}$
equals, in a sparse limit,
\[
\underset{N\rightarrow\infty}{\lim}\mathbb{E}\left[\left(\frac{D_{i+}}{\lambda_{N}}\right)^{4}\right]=\frac{\tilde{P}\left(\vcenter{\hbox{\includegraphics[scale=0.125]{edge}}}\right)}{\lambda^{3}}+\frac{7\tilde{Q}\left(\vcenter{\hbox{\includegraphics[scale=0.125]{twostar}}}\right)}{\lambda^{2}}+\frac{6\tilde{Q}\left(\vcenter{\hbox{\includegraphics[scale=0.125]{onethreewheel}}}\right)}{\lambda}+\tilde{Q}\left(\vcenter{\hbox{\includegraphics[scale=0.125]{fourstar}}}\right),
\]
where a tilde above a subgraph/homomorphism density, as earlier, denotes
the density divided by $\rho_{N}$ raised to the power of the number
of edges in the subgraph under consideration (e.g., $\tilde{Q}\left(\vcenter{\hbox{\includegraphics[scale=0.125]{twostar}}}\right)=Q\left(\vcenter{\hbox{\includegraphics[scale=0.125]{twostar}}}\right)/\rho_{N}^{2}$).

From (\ref{eq: degree_sequence_m1}) and (\ref{eq: degree_sequence_m2})
have that, in the sparse case,
\begin{equation}
\underset{N\rightarrow\infty}{\lim}\mathbb{V}\left(\frac{D_{i+}}{\lambda_{N}}\right)=\left[\tilde{Q}\left(\vcenter{\hbox{\includegraphics[scale=0.125]{twostar}}}\right)-\tilde{P}\left(\vcenter{\hbox{\includegraphics[scale=0.125]{edge}}}\right)\tilde{P}\left(\vcenter{\hbox{\includegraphics[scale=0.125]{edge}}}\right)\right]+\frac{\tilde{P}\left(\vcenter{\hbox{\includegraphics[scale=0.125]{edge}}}\right)}{\lambda}.\label{eq: variance_degree_sparse}
\end{equation}

There are several peculiarities in these expressions. Returning to
the dense case, when the graph is an Erdös-Renyi homogenous random
graph, $Q\left(\vcenter{\hbox{\includegraphics[scale=0.125]{twostar}}}\right)=P\left(\vcenter{\hbox{\includegraphics[scale=0.125]{edge}}}\right)P\left(\vcenter{\hbox{\includegraphics[scale=0.125]{edge}}}\right)$
, and (\ref{eq: degree_variance_dense}) indicates that the distribution
of $D_{i+}/\left(N-1\right)$ is degenerate in the limit. In that
case normalizing $D_{i+}$ by $\sqrt{N-1}$ results in a random variable
with a non-degenerate variance in the limit since
\[
\mathbb{V}\left(D_{i+}\right)=\left(N-1\right)\left(N-2\right)\left[Q\left(\vcenter{\hbox{\includegraphics[scale=0.125]{twostar}}}\right)-P\left(\vcenter{\hbox{\includegraphics[scale=0.125]{edge}}}\right)P\left(\vcenter{\hbox{\includegraphics[scale=0.125]{edge}}}\right)\right]+\left(N-1\right)P\left(\vcenter{\hbox{\includegraphics[scale=0.125]{edge}}}\right)\left(1-P\left(\vcenter{\hbox{\includegraphics[scale=0.125]{edge}}}\right)\right).
\]
Observations such as these suggest that, as with subgraph frequencies,
it may be desirable to retain all terms -- including nominally asymptotically
dominated ones -- when calculating the variance and other moments
of the degree distribution.

\citet{Atalay_et_al_PNAS11} construct a theoretical model of supply
chain formation. They informally assess the plausibility of a calibrated
version of their model by comparing their model-predicted degree sequence
with the one observed in the US Buyer-Supplier network (see their
Figure 1). A formal minimum $\chi^{2}$ type specification test of
their model could be constructed on the basis of Theorem \ref{thm: degree_sequence_moments}.

\subsection{\label{subsec: open-questions}Further reading and open questions}

Subgraph frequencies are, in many ways, analogous to moments of a
distribution. Relatedly methods of estimation and inference for subgraph
frequencies have many applications, from attaching a measure of uncertainty
to statistics like the transitivity index, to facilitating specification
testing and model estimation. As the discussion here shows, the large
network properties of empirical subgraph frequencies depend on the
nature and magnitude of dependence across links induced by the graphon
as well as properties like sparsity. Formulating methods of inference
for subgraph frequencies that are adaptive to these features of the
GGP would be useful. \citet{Menzel_arXiv17} makes some progress in
this direction, but substantial work remains.

\citet{Volfosky_Airoldi_SPL16}, extending results due to \citet{Diaconis_Freedman_AP80},
present results relating finitely and infinitely exchangeable arrays.
Results of this type could be useful for understanding how best to
proceed when the network in hand corresponds to the equilibrium of
an $N$-player game where the conditional independence structure associated
Aldous-Hoover type GGPs may not formally hold, but where -- for $N$
large enough -- it should hold approximately.

\citet{Bhattacharya_Bickel_AS15}, \citet{Green_Shalizi_arXiv17}
and \citet{Menzel_arXiv17} discuss subsampling and bootstrap methods
for exchangeable random arrays. Adapting ideas introduced by, for
example, \citet{Menzel_arXiv17} to accommodate sparse networks would
be theoretically interesting and practically useful.

I have emphasized more recent work on subgraph frequencies, but the
earlier papers, beginning with the seminal one by \citet{Holland_Leinhardt_SM76}
are rewarding to read (or re-read) in the light of contemporary developments.
The survey paper by \citet{Jackson_et_al_JEL17} presents many real
work examples of degree distributions and other network statistics.
This paper also relates these measures to theoretical ideas in the
economic literature on network formation and network games.

A rather different approach to asymptotic analysis of network statistics
builds off the probability literature on random geometric graphs \citep{Penrose_Bk2003}.
These models posit a strong form of homophily such that agents which
are far apart from one another (in some, perhaps latent, space) link
infrequently (or not at all). The (latent) spatial structure renders
agents non-exchangeable. This mechanism generates sufficient independence
among distant units such that LLNs and CLTs can be proven. \citet{Leung_JOE19},
\citet{Kuersteiner_arXiv19}, and \citet{Leung_Moon_arXiv19} develop
these ideas to prove LLNs and CLTs for network statistics where the
observed network is assumed to be a strategic network equilibrium
configuration. A challenge of this approach is that valid inference
appears to require information on agents' positions (so that HAC type
variance estimators can be used). Unfortunately it is often most natural
to view such positions as latent \citep[e.g.,][]{Hoff_etal_JASA2002,Krivitsky_et_al_SN09}.
Nevertheless, this approach, by building on insights from the literature
of random geometric graphs, as well as spatial statistics, seems well
calibrated to some network applications. For example, sparseness seems
to be easier to handle in this framework \citep[cf.,][]{Graham_NBER16}.

Understanding the connections between approaches to large network
inference based upon random geometric graphs versus exchangeable random
graphs remains, to my knowledge, largely unexplored.

\section{\label{sec: Strategic-models}Strategic models of network formation}

The models of network formation introduced in Sections \ref{sec: dyadic_regression},
\ref{sec: policy_analysis} and \ref{sec: heterogeneity} are externality
free: the utility two agents create by forming a link is invariant
to the presence or absence of links elsewhere in the network. In contrast,
the theoretical literature on network formation, beginning with the
seminal paper by \citet{Jackson_Wolinsky_JET96}, is decidedly focused
on the study of models where agents' preferences are interdependent.
That is, the utility dyad $\left\{ i,j\right\} $ generates by forming
an edge may vary with the presence or absence of additional edges
elsewhere in the network. For example, if $i$ and $j$ share many
neighbors (``friends'') in common, they may reap utility gains from
`triadic closure' when linking; incidentally also forming many triangles
\citep{Simmel_Book1908,Coleman_AJS88,Jackson_et_al_AER12}. 

Models with interdependencies in preferences are typically called
\emph{strategic} models of network formation. The use of the word
strategic here stems from connections, both historical and substantive,
between recent theoretical research on networks and game theory. I
will comment on this nomenclature after first introducing the standard
notion of equilibrium used by theoretical network researchers in this
area: \emph{pairwise stability}. Pairwise stability is the equilibrium
concept introduced by \citet{Jackson_Wolinsky_JET96}. Here I introduce
the definition which excludes the possibility of transfers between
agents; the transferable utility case was introduced in Section \ref{sec: Basic-probability-tools}.

Let $\nu_{i}\thinspace:\thinspace\mathbb{D}_{N}\rightarrow\mathbb{R}$
be a utility function for agent $i$, which maps adjacency matrices
into utils. In order to define pairwise stability I need a definition
of \emph{marginal utility}. As earlier, the marginal utility for agent
$i$ associated with (possible) edge $\left(i,j\right)$ is
\begin{equation}
MU_{ij}\left(\mathbf{D}\right)=\left\{ \begin{array}{cc}
\nu_{i}\left(\mathbf{D}\right)-\nu_{i}\left(\mathbf{D}-ij\right) & \text{if}\thinspace D_{ij}=1\\
\nu_{i}\left(\mathbf{D}+ij\right)-\nu_{i}\left(\mathbf{D}\right) & \text{if}\thinspace D_{ij}=0
\end{array}\right.\label{eq: marginal_utility-1}
\end{equation}
recalling that $\mathbf{D}-ij$ is the adjacency matrix associated
with the network obtained after deleting edge $\left(i,j\right)$
and $\mathbf{D}+ij$ the one obtained via link addition. 
\begin{defn}
\label{def: Pairwise-stability-1}(\textsc{Pairwise stability without
Transfers}) The network $G$ is pairwise stable if (i) no agent wishes
to dissolve a link
\begin{equation}
\forall\left(i,j\right)\in\mathcal{E}\left(G\right),\thinspace MU_{ij}\left(\mathbf{D}\right)\geq0\thinspace\text{and}\thinspace MU_{ji}\left(\mathbf{D}\right)\geq0\label{eq: pairwise_stable_1}
\end{equation}
and (ii) no pair of agents wishes to form a link
\begin{equation}
\forall\left(i,j\right)\notin\mathcal{E}\left(G\right),\thinspace MU_{ij}\left(\mathbf{D}\right)>0\thinspace\Rightarrow\thinspace MU_{ji}\left(\mathbf{D}\right)<0.\label{eq: pairwise_stable_2}
\end{equation}
\end{defn}
Two features of Definition \ref{def: Pairwise-stability-1} merit
emphasis. First, an implication of the definition is that utility
is \emph{nontransferable }across agents. This differs from some of
the models introduced earlier. Second, the strategic moniker aside,
pairwise stability is a really non-strategic/myopic notion of equilibrium.
This point is elegantly made by \citet{Ostrovsky_AER08} in a related
context. Pairwise stability does not require agents to engage in any
``what if'' or forward-looking introspection. Specifically it does
not require agents to imagine what might happen to the rest of the
network were they to add or delete a link, rather it simply requires
them to behave optimally given the actions of all other agents in
the network. The key feature of so-called strategic models relative
to those in Sections \ref{sec: dyadic_regression}, \ref{sec: policy_analysis}
and \ref{sec: heterogeneity} is not behavioral, but in their different
assumptions about the nature of preferences. Here utility is interdependent;
this is the interesting complication.\footnote{It might also be interesting to consider estimation and inference
under different refinements of the pairwise stability concept; such
refinements might posit more sophisticated play by agents.}

For any profile of preferences $\left\{ \nu_{i}\right\} _{i=1}^{M}$
there many be many network configurations satisfying Definition \ref{def: Pairwise-stability-1}.\footnote{For results on the existence and uniqueness of pairwise stable networks
see \citet{Jackson_Watts_SJE01} and \citet{Hellmann_IJGT13}.} The potentially high cardinality of the set of pairwise stable network
configurations makes the direct application of econometric methods
designed for the analysis of games computationally prohibitive \citep[c.f., ][]{Bajari_et_al_EM10,Bajari_et_al_WC13}.
Nevertheless insights from research in this area is valuable for analyzing
empirical models of network formation.

\subsection{A fixed point approach with increasing preferences}

One example of this claim is provided by an elegant and interesting
paper by \citet{Miyauchi_JOE16}. This paper draws on insights from
the theory of supermodular games \citep[e.g., ][]{Topkis_Book98}
and their empirical analysis \citep{Jia_EM08,Uetake_Watanabe_AE13}
to formulate a tractable estimation strategy for a class of strategic
network formation models. Following \citet{Miyauchi_JOE16} consider
the mapping $\varphi\left(\mathbf{D}\right)\thinspace:\thinspace\mathbb{D}_{N}\rightarrow\mathbb{I}_{\tbinom{N}{2}}$:
\begin{equation}
\varphi\left(\mathbf{D}\right)\equiv\left[\begin{array}{c}
\mathbf{1}\left(MU_{12}\left(\mathbf{D}\right)\geq0\right)\mathbf{1}\left(MU_{21}\left(\mathbf{D}\right)\geq0\right)\\
\mathbf{1}\left(MU_{13}\left(\mathbf{D}\right)\geq0\right)\mathbf{1}\left(MU_{31}\left(\mathbf{D}\right)\geq0\right)\\
\vdots\\
\mathbf{1}\left(MU_{N-1N}\left(\mathbf{D}\right)\geq0\right)\mathbf{1}\left(MU_{NN-1}\left(\mathbf{D}\right)\geq0\right)
\end{array}\right].\label{eq: miyauchi_fixed_point}
\end{equation}
Observe that $\mathbf{1}\left(MU_{ij}\left(\mathbf{D}\right)\geq0\right)\mathbf{1}\left(MU_{ji}\left(\mathbf{D}\right)\geq0\right)$
equals $1$ if condition (\ref{eq: pairwise_stable_1}) of pairwise
stability holds (which implies edge $\left(i,j\right)$ is present)
and zero otherwise (which implies condition (\ref{eq: pairwise_stable_2})
and hence the absence of edge $\left(i,j\right)$). Under the maintained
assumption that the observed network is pairwise stable, its adjacency
matrix is therefore the fixed point 
\begin{equation}
\mathbf{D}=\mathrm{vech^{-1}}\left[\varphi\left(\mathbf{D}\right)\right].\label{eq: fixed_point}
\end{equation}
Here $\mathrm{vech}(\cdot)$ vectorizes the $\tbinom{N}{2}$ elements
in the lower triangle of an $N\times N$ matrix and I define its inverse
operator as creating a symmetric matrix with a zero diagonal. There
may, of course, be many $\mathbf{d}\in\mathbb{D}_{N}$ such that $\mathbf{d}=\mathrm{vech^{-1}}\left[\varphi\left(\mathbf{d}\right)\right]$.
\citet{Miyauchi_JOE16} notes, however, that if the preference profile
$\left\{ \nu_{i}\right\} _{i=1}^{N}$ satisfies what he calls a \emph{non-negative
externality} condition, namely that the marginal utilities $MU_{ij}\left(\mathbf{d}\right)$
are weakly increasing in $\mathbf{d}$ for all $i$ and $j$, then
one can characterize the set of pairwise stable networks with Tarski's
\citeyearpar{Tarski_PJM55} fixed point theorem \citep[Proposition 1]{Miyauchi_JOE16}.
The invocation of \citet{Tarski_PJM55} implies that the set of pairwise
stable networks corresponds to a complete lattice with a \emph{maximum}
and \emph{minimum equilibrium}. Furthermore any pairwise stable network
is a partial subgraph, defined on nodes $\left\{ 1,\ldots,N\right\} $
of the maximum equilibrium. And the minimum equilibrium is always
a partial subgraph, again defined on nodes $\left\{ 1,\ldots,N\right\} $
of any pairwise stable equilibrium. This has many useful implications.
Trivially, the set of equilibrium networks can be sorted according
to density; less trivially their degree sequences can also be ordered. 

Of course, the non-negative externality requirement is restrictive;
there are many settings where diminishing marginal utility in links
might be plausible (e.g., capacity constraints). At the same time,
many extant empirical models of network formation do satisfy the restriction,
so exploring estimation maintaining it is reasonable. \citet[Section 3.3]{Miyauchi_JOE16}
provides additional discussion. 

Again borrowing results from the theory of supermodular games, \citet{Miyauchi_JOE16}
shows that the minimum equilibrium, say $\mathbf{\underline{d}}$,
can be computed by fixed point iteration of (\ref{eq: miyauchi_fixed_point})
starting from the empty adjacency matrix, while the maximum equilibrium,
say $\overline{\mathbf{d}}$, may be computed by fixed point iteration
starting from the adjacency matrix associated with the complete graph
$K_{N}$. A similar computational insight, albeit in non-network settings,
features in \citet{Jia_EM08} and \citet{Uetake_Watanabe_AE13}.

At this stage, to show how the above insights can be used concretely,
it is helpful to parameterize the utility function, introducing both
explicit heterogeneity and a parameter vector. Adopting the random
utility approach pioneered by \citet{McFadden_FinE74}, assume, for
example, that
\begin{equation}
\nu_{i}\left(\mathbf{d},\mathbf{U};\theta_{0}\right)=\sum_{j}d_{ij}\left[\alpha_{0}+\beta_{0}\left[\sum_{k}d_{ik}d_{jk}\right]-U_{ij}\right],\label{eq: transitivity utility}
\end{equation}
with $\mathbf{U}=\left[U_{ij}\right]_{i,j\in\left\{ 1,\ldots,N\right\} ,i\neq j}$,
$\theta=\left(\alpha,\beta\right)'$ and the change in notation for
the utility function emphasizing that the econometrician does not
observe the matrix of random utility shifters $\mathbf{U}$. In practice
the elements of $\mathbf{U}$, as is common in discrete choice analysis,
are assumed to be i.i.d random draws from some known distribution
(e.g, the standard Normal or Logistic distribution).

Equation (\ref{eq: transitivity utility}) implies that the marginal
utility agent $i$ gets from a link with $j$ is
\begin{equation}
MU_{ij}\left(\mathbf{d},\mathbf{U};\theta_{0}\right)=\alpha_{0}+\beta_{0}\left[\sum_{k}d_{ik}d_{jk}\right]-U_{ij}\label{eq: transitivity_marginal_utility}
\end{equation}

This marginal utility is increasing in the number of links $i$ and
$j$ have in common, embodying a structural taste for transitive closure
(here I assume that $\beta_{0}>0$). Clearly (\ref{eq: transitivity_marginal_utility})
is weakly increasing in $\mathbf{d}\in\mathbb{D}_{N}$ and hence Tarski's
\citeyearpar{Tarski_PJM55} theorem applies. For a given draw of $\mathbf{U}$
and value of $\theta$ we can compute minimum and maximum equilibria,
respectively $\mathbf{\underline{d}}\left(\mathbf{U};\theta\right)$
and $\overline{\mathbf{d}}\left(\mathbf{U};\theta\right)$, by fixed
point iteration. Let $\underline{G}_{N}\left(\mathbf{U};\theta\right)$
and $\overline{G}_{N}\left(\mathbf{U};\theta\right)$ be the graphs
corresponding to these adjacency matrices. Using these graphs we can
compute, for example, the injective homomorphism frequencies $t_{\mathrm{inj}}\left(S,\underline{G}_{N}\left(\mathbf{U};\theta\right)\right)$
and $t_{\mathrm{inj}}\left(S,\overline{G}_{N}\left(\mathbf{U};\theta\right)\right)$
for $S=\vcenter{\hbox{\includegraphics[scale=0.125]{twostar}}},\vcenter{\hbox{\includegraphics[scale=0.125]{triangle}}}$
etc. These homomorphism frequencies correspond to model predictions
associated with specific draws of $\mathbf{U}$ and values of $\theta$.
Using simulation to integrate out the former, yields the vectors
\[
\underline{\pi}\left(\theta\right)=\frac{1}{B}\sum_{b=1}^{N}\left(\begin{array}{c}
t_{\mathrm{inj}}\left(S_{1},\underline{G}_{N}\left(\mathbf{U}^{\left(b\right)};\theta\right)\right)\\
\vdots\\
t_{\mathrm{inj}}\left(S_{J},\underline{G}_{N}\left(\mathbf{U}^{\left(b\right)};\theta\right)\right)
\end{array}\right),\thinspace\thinspace\thinspace\overline{\pi}\left(\theta\right)=\frac{1}{B}\sum_{b=1}^{N}\left(\begin{array}{c}
t_{\mathrm{inj}}\left(S_{1},\overline{G}_{N}\left(\mathbf{U}^{\left(b\right)};\theta\right)\right)\\
\vdots\\
t_{\mathrm{inj}}\left(S_{J},\overline{G}_{N}\left(\mathbf{U}^{\left(b\right)};\theta\right)\right)
\end{array}\right)
\]
for $\mathbf{U}^{\left(1\right)},\mathbf{U}^{\left(2\right)}\ldots,\mathbf{U}^{\left(B\right)}$
a sequence of independent simulated random utility shifter profiles
and $S_{1},\ldots,S_{J}$ a set of $J$ identifying motifs of interest.

\citet{Miyauchi_JOE16} works under the assumption that the econometrician
observes of $c=1,\ldots,C$ independent networks, with, in a slight
change relative to earlier notation, $G_{c}$ denoting the $c^{th}$
network/graph. Let $\pi\left(G_{c}\right)$ be the vector of $S_{1},\ldots,S_{J}$
injective homomorphism frequencies as observed in the $c^{th}$ network
and let $\underline{\pi}_{c}\left(\theta\right)$ and $\overline{\pi}^{c}\left(\theta\right)$
be the corresponding expected frequencies at the minimum and maximum
pairwise stable equilibria for that network at parameter $\theta$.
These frequencies are computed using simulation as described above.
Under preferences (\ref{eq: transitivity utility}) the only reason
these frequencies might vary with $c$ is if the networks observed
by the econometrician vary in the number of agents within them.\footnote{In more complicated models, with covariates, these minimums and maximums
will vary with $c$ due to differences in the distribution of covariates
across networks.}

\citet{Miyauchi_JOE16} focuses on assumptions which may only partially
identify $\theta$, but to begin with consider adding to the set-up
the assumption that agents select the maximum equilibrium \citep[cf., ][]{Jia_EM08}.
In that case
\begin{equation}
\mathbb{E}\left[\bar{\pi}_{c}\left(\theta_{0}\right)-\pi\left(G_{c}\right)\right]=0,\label{eq: miyauchi_moment_equality}
\end{equation}
is a valid moment condition. If the set of chosen motifs is sufficiently
rich so as to point identify $\theta$, then consistent estimation
of $\theta_{0}$ by the method of simulated moments is straightforward
\citep{McFadden_EM89,Pakes_Pollard_EM89,Gourieroux_et_al_JAE93}.

Because the asymptotic approximation involves $C\rightarrow\infty$,
this approach hinges upon the availability of a large number of independent
networks (each described by the same $\theta_{0}$). If, instead,
only a single large network is observed, then econometrician might
use the methods outlined in Section \ref{sec: Statistics} to form
an estimate of the variance of $\pi\left(G_{N}\right)$, say $\hat{\Omega}_{\pi}$.
An estimate of $\theta$ could then be formed by minimizing the simulated
minimum distance (SMD) criterion:
\[
\hat{\theta}_{\mathrm{SMD}}=\left(\bar{\pi}\left(\theta\right)-\pi\left(G_{N}\right)\right)'\hat{\Omega}_{\pi}^{-1}\left(\bar{\pi}\left(\theta\right)-\pi\left(G_{N}\right)\right).
\]
Note that $\hat{\Omega}_{\pi}$ is constructed under an Aldous-Hoover
dependence/independence structure; such a structure may not characterize
the finite $N$ structural model. An additional approximation argument
is involved; understanding and formalizing this argument is required
to (rigorously) derive the law of $\hat{\theta}_{\mathrm{SMD}}$ (suitably
scaled and centered).

When analyzing incomplete models, researchers are often reluctant
to make assumptions about equilibrium selection (which complete the
model). \citet{Miyauchi_JOE16} shows that if the chosen vector of
moments satisfies a certain monotonicity property (see his Property
1), then inference can be based upon the pair of moment inequalities
\begin{align}
\mathbb{E}\left[\bar{\pi}_{c}\left(\theta_{0}\right)-\pi\left(G_{c}\right)\right] & \geq0\label{eq: miyauchi_moment_inequalities}\\
\mathbb{E}\left[\underline{\pi}_{c}\left(\theta_{0}\right)-\pi\left(G_{c}\right)\right] & \leq0.\nonumber 
\end{align}
Confidence intervals which asymptotically cover $\theta_{0}$ with
probability at least $1-\alpha$ can be constructed using the approach
outlined by, for example, \citet{Andrews_Soares_EM10}. Injective
homomorphism frequencies appear to satisfy the needed property, although
some induced subgraph frequencies may not.

\subsection{Directed links with private information}

\citet{Leung_JOE15} studies a model of simultaneous \emph{directed
}network formation where agents have private information. In a directed
network agent $i$ may send a link to agent $j$ such that $D_{ij}=1$;
agent $j$ may or may not decide to reciprocate and send a link back
to $i$. The adjacency matrix is no longer symmetric, although it
retains a diagonal of structural zeros. The $i^{th}$ row of the adjacency
matrix records the set of links agent $i$ chooses to send to other
agents, while the $i^{th}$ column records the set of links that other
agents choose to send to $i$.

To describe Leung's \citeyearpar{Leung_JOE15} approach let $\mathbf{D}_{\left[-i,\cdot\right]}$
be the sub-adjacency matrix constructed by deleting the $i^{th}$
row from $\mathbf{D}$. The marginal utility agent $i$ receives when
she directs a link to agent $j$ is 
\begin{equation}
\mathrm{MU}_{ij}\left(\mathbf{D}_{\left[-i,\cdot\right]},\mathbf{X};\theta_{0}\right)-U_{ij}.\label{eq: Lueng_MU}
\end{equation}
An important implication of (\ref{eq: Lueng_MU}) is that while $i$'s
gain from sending a link to $j$ may vary with the presence or absence
of links elsewhere in the network, it \emph{does not }vary with the
presence or absence of other links which $i$ herself may or may not
direct. This restriction rules out interesting preference structures
(see below), but simplifies the analysis substantially. \citet{Ridder_Sheng_WP17}
develop an approach to relaxing this feature of Leung's \citeyearpar{Leung_JOE15}
setup. To describe the main ideas I work with a special case of (\ref{eq: Lueng_MU})
:
\begin{equation}
\mathrm{MU}_{ij}\left(\mathbf{D}_{\left[-i,\cdot\right]},\mathbf{X};\theta_{0}\right)=\alpha_{0}+\beta_{0}D_{ji}+\gamma_{0}\sum_{k\neq i,j}D_{ki}D_{kj}+t\left(X_{i},X_{j}\right)'\delta_{0},\label{eq: Leung_MU_ex}
\end{equation}
for $\theta_{0}=\left(\alpha_{0},\beta_{0},\gamma_{0},\delta_{0}'\right)'$
and $t\left(X_{i},X_{j}\right)$ a vector of possibly non-symmetric
functions of exogenous agent attributes. The parameter $\beta_{0}$
indexes the utility gain associated with \emph{reciprocity} in links,
while $\gamma_{0}$ captures the utility gain arising when a link
is \emph{supported}. A directed edge from $i$ to $j$ is supported
by agent $k$, if $k$ directs links to both $i$ and $j$ (this allows,
for example, $k$ to ``referee'' transactions between $i$ and $j$). 

Support, although related, differs from transitivity \citep[cf., ][]{Jackson_et_al_AER12};
replacing $\sum_{k}D_{ki}D_{kj}$ with $\sum_{k}D_{ik}D_{jk}$ in
(\ref{eq: Leung_MU_ex}) means that $\gamma_{0}$ would instead index
a structural taste for \emph{transitivity} (i.e., that a link to a
``friend of one of my friends'' generates more utility). However
a transitivity term of this type is ruled out by the restriction that
the marginal utility of an $i$ to $j$ link does not vary with the
presence or absence of other links which $i$ may or may not send
\citep[cf., ][]{Ridder_Sheng_WP17}.

\citet{Leung_JOE15} assumes that $\mathbf{U}_{i}=\left(U_{i1},\ldots,U_{ii-1},U_{ii+1},\ldots,U_{iN}\right)'$,
the idiosyncratic components of link utilities, are private information
to agent $i$, while all other features of the game are common knowledge
to all agents. Let $P_{ij}$ denote the common prior held by all players
other than $i$ regarding the probability that she directs a link
to $j$. Let $\mathbf{P}$ denote the $N\left(N-1\right)\times1$
vector of such common priors. In a Bayes-Nash equilibrium agent $i$
will best respond to the common prior by choosing to direct an edge
toward $j$ according to

\[
D_{ij}=\mathbf{1}\left(\alpha_{0}+\beta_{0}P_{ji}+\gamma_{0}\sum_{k\neq i,j}P_{ki}P_{kj}+t\left(X_{i},X_{j}\right)'\delta_{0}-U_{ij}\geq0\right);
\]
that is $i$ forms the directed edge only if the expected marginal
utility from doing so is positive. Assuming, for example, that the
$\left\{ U_{ij}\right\} _{i,j=1}^{N}$ are i.i.d. standard normals.
Let
\[
\varphi_{ij}\left(\mathbf{P},\mathbf{X};\theta_{0}\right)=\Phi\left(\alpha_{0}+\beta_{0}P_{ji}+\gamma_{0}\sum_{k\neq i,j}P_{ki}P_{kj}+t\left(X_{i},X_{j}\right)'\delta_{0}\right)
\]
with $\Phi\left(\cdot\right)$ the standard normal CDF. A Bayesian-Nash
equilibrium requires self-consistency of beliefs such that $\mathbf{P}$
corresponds to a fixed point of the mapping.

\begin{equation}
\varphi\left(\mathbf{P},\mathbf{X};\theta_{0}\right)=\left[\begin{array}{c}
\varphi_{12}\left(\mathbf{P},\mathbf{X};\theta_{0}\right)\\
\vdots\\
\varphi_{1N}\left(\mathbf{P},\mathbf{X};\theta\right)\\
\vdots\\
\varphi_{N1}\left(\mathbf{P},\mathbf{X};\theta\right)\\
\vdots\\
\varphi_{NN-1}\left(\mathbf{P},\mathbf{X};\theta_{0}\right)
\end{array}\right].\label{eq: Lueng_Fixed_Point}
\end{equation}
One approach would be to apply ideas analogous to those developed
in \citet{Miyauchi_JOE16} using (\ref{eq: Lueng_Fixed_Point}). \citet{Leung_JOE15},
instead creatively adapts the two-step approach familiar from the
wider econometrics literature on incomplete information games \citep{Bajari_et_al_JBES10,Bajari_et_al_WC13}.
Let $\hat{\mathbf{P}}$ be a nonparametric estimate of the belief
vector $\mathbf{P}$. With this estimate in hand, $\theta_{0}$, may
be estimated by finding the maximum of the criterion
\begin{equation}
\underset{\theta}{\max\thinspace\thinspace}\frac{1}{N\left(N-1\right)}\sum_{i=1}^{N}\sum_{j\neq i}D_{ij}\ln\varphi_{ij}\left(\mathbf{\hat{P}},\mathbf{X};\theta\right)+\left(1-D_{ij}\right)\ln\left[1-\varphi_{ij}\left(\mathbf{\hat{P}},\mathbf{X};\theta\right)\right],\label{eq: Leung_feasible_criterion}
\end{equation}
using a standard Probit MLE program. 

The challenge with this approach is that it is not obvious how one
can consistently estimate $\mathbf{P}$. Unlike in, for example, the
literature on entry games, where the same player is observed playing
independent replications of a game across different markets, in the
present set-up there is only a single game. Leung's \citeyearpar{Leung_JOE15}
key insight is to note that under an exchangeability assumption and
a focus on symmetric equilibria, estimation of $\mathbf{P}$ is possible
because it implies that (ordered dyads) with identical covariate configurations
have identical \emph{ex ante} linking probabilities. For the case
of discretely-valued $X_{i}$ \citet[Proposition 1]{Leung_JOE15}
shows, under assumptions, that for $\hat{P}_{ij}=\left[\sum_{k,l}\mathbf{1}\left(t\left(X_{k},X_{l}\right)=t\left(X_{i},X_{j}\right)\right)\right]^{-1}\times\left[\sum_{k,l}D_{kl}\mathbf{1}\left(t\left(X_{k},X_{l}\right)=t\left(X_{i},X_{j}\right)\right)\right]$
\[
\underset{i,j\in\mathcal{V}\left(G_{N}\right)}{\sup}\left|\hat{P}_{ij}-P_{ij}\right|\overset{p}{\rightarrow}0
\]
at rate $N^{1/2}$. Using these estimates in (\ref{eq: Leung_feasible_criterion})
results in a consistent and asymptotically normal estimate of $\theta_{0}$
under regularity conditions. The interesting features of these results
involve the need to account for first step estimation error as well
as for dependencies across dyads sharing one agent in common. \citet{Leung_JOE15}
also presents a variance estimator and an empirical illustration based
on the network data collected by \citet{Banerjee_et_al_Sci13}.

\subsection{Bounded degree and restricted heterogeneity}

Like \citet{Miyauchi_JOE16}, \citet{dePaula_et_al_EM18} study a
simultaneous-move complete information model of network formation.
They place three key restrictions on the graph generating process.
First, they assume that agents only wish to maintain a small number
of links. Second, that utility only varies with the addition or deletion
of links within a finite radius. For example an agent may care about
the friends of her friends, but not the friends of the friends of
her friends. Third, there are only a finite number of agent types
and, crucially, agents are indifferent among links of the same type.
There is some nuance to the last restriction since indirect connections
may matter. Consider two Black individuals, each with a Black and
White friend, the restriction is that any third agent is indifferent
between forming a link with either of these two individuals. Similar
restrictions feature prominently in one-to-one transferable utility
matching models \citep[e.g., ][]{Choo_Siow_JPE06,Graham_AinE13,Galichon_Salanie_AERPP17}.
The first and last of these assumptions are ``non-standard'', but
\citet{dePaula_et_al_EM18} show how they make identification analysis
tractable.

They begin by noting that, under their assumptions, any \emph{rooted
network} -- a configuration of links within a fixed distance about
a focal ``root'' node -- will take one of a finite number of configurations.
Identification of preference parameters comes from comparing model
predictions about the frequency of these configurations with their
empirical counterparts.

The operationalization of this intuition into a workable method of
inference is the main contribution of their paper. To describe this
contribution assume that the utility agent $i$ gets from network
configuration $\mathbf{D}=\mathbf{d}$ is, for example,
\[
\nu_{i}\left(\mathbf{d},\mathbf{U};\theta_{0}\right)=\sum_{j}d_{ij}\left[\alpha_{0}'R_{ij}+\beta_{0}\left[\sum_{k}d_{ik}d_{jk}\right]+U_{i}\left(X_{j}\right)\right]-\infty\cdot\mathbf{1}\left(\sum_{j}d_{ij}>L\right).
\]
Here $R_{ij}=r\left(X_{i},X_{j}\right)$ is a vector of known symmetric
functions of $X_{i}$ and $X_{j}$, $L$ is the maximum number links
an agent might desire (known by the econometrician), and $U_{i}\left(x\right)$
is an unobserved utility-shifter with known distribution. This utility
shifter varies with $i$, but only depends on $j$ via the covariate
$X_{j}$. The expression above suggests that associated with each
agent are just $\left|\mathbb{X}\right|$ shocks, one for each type
of agent. \citet{dePaula_et_al_EM18} actually attach $L\times\left|\mathbb{X}\right|$
shocks to each agent, but their main ideas can be conveyed under the
more restrictive set-up.

Let $U_{i}=\left(U_{i}\left(x_{1}\right),\ldots,U_{i}\left(x_{\left|\mathbb{X}\right|}\right)\right)'$
denote the vector of taste shocks associated with agent $i$. Since
agents maintain no more than $L$ links, and preferences are only
affected by network structure within a certain radius, the number
of logically observable rooted network configurations is finite. For
each of these configurations we can ask, for a given value of $\theta$
and draw of $U_{i}$, whether an agent will unilaterally reject it
(e.g., given the configuration's structure she may prefer to unilaterally
dissolve some links). \citet{dePaula_et_al_EM18} call the set of
acceptable rooted networks a \emph{preference class}. Since the distribution
of $U_{i}$ is known, the ex ante probability that any individual
falls into a particular preference class when $\theta$ takes a particular
value can be computed (typically via simulation).

A network can be generated by choosing the frequency with which agents
of (i) a particular type, and (ii) belonging to a particular preference
class, are assigned to specific rooted network configurations. Theorem
1 of \citet{dePaula_et_al_EM18} shows how to construct these frequencies
in a way which satisfies pairwise stability. A parameter value belongs
to the identified set, if there exists a feasible vector of allocation
probabilities such that the predicted frequencies with which the various
rooted network configurations occur match their corresponding empirical
frequencies. Theorem 2 of \citet{dePaula_et_al_EM18} shows how this
question may be answered by solving a particular quadratic program.

The identification analysis assumes there are continuum of agents.
Since their graph is sparse, the object they work with is not a graphon,
but its sparse graph analog, called a \emph{graphing} in the literature
\citep{Lovasz_AMS12}. For inference they assume that the econometrician
observes a random sample of rooted networks, perhaps collected via
snowball sampling.

\subsection{Many agent approximations}

\citet{Menzel_WP16} studies a class of large network formation models
with exchangeable agents. He characterizes the limiting network as
$N\rightarrow\infty$ and investigates how to use this limit to approximate
the finite network in hand. One example of the family of preference
structures accommodated by his set-up is
\begin{equation}
MU_{ij}\left(\mathbf{D},\mathbf{U};\theta_{0}\right)=\underset{\text{marginal benefit}}{\underbrace{\alpha_{0}'R_{ij}+\beta_{0}\left[\min\left(\sum_{k}D_{ik}D_{jk},1\right)\right]+\sigma U_{ij}}}-\underset{\text{marginal cost}}{\underbrace{\left(\ln J+\sigma U_{i0}\right)}}.\label{eq: konrad_preferences}
\end{equation}
The scale and location parameters, $\sigma$ and $J$, vary with $N$
in a particular way in order to obtain useful limits.

\citet{Menzel_WP16} observes that a network is pairwise stable if
and only if $D_{ij}=1$ when $MU_{ij}\geq0$ and $D_{ij}=0$ when
$MU_{ij}<0$ for all $j\in\mathcal{W}_{i}$. The set $\mathcal{W}_{i}$
includes all agents $j$ who are willing to form a link with agent
$i$ or, equivalently, who would not veto such a link:
\[
\mathcal{W}_{i}=\left\{ j\in\left.\mathcal{V}\left(G_{N}\right)\right\backslash \left\{ i\right\} \thinspace:\thinspace MU_{ji}\geq0\right\} .
\]
When $J$ grows with $N$ at the appropriate rate, the number of links
accepted by agent $i$, among those available in $\mathcal{W}_{i}$,
is stochastically bounded (ensuring that the limiting network is sparse).
Furthermore, using extreme value theory, \citet{Menzel_WP16} shows
that the effect of the endogenous choice set, $\mathcal{W}_{i}$,
on the probability of forming a particular link is completely summarized
by a conditional logit type inclusive value. 

Let
\[
Z_{ij}=\left(X_{i},X_{j},T_{ij}\right)'
\]
for $T_{ij}=\min\left(\sum_{k}D_{ik}D_{jk},1\right).$ Further, for
purposes of illustration, let $X_{i}\in\left\{ 0,1\right\} $ be a
binary indicator for, say, gender. In this case $Z_{ij}$ takes values
within the finite set $\mathbb{Z}$. For example woman $i$ may link
with woman $j$, with whom she shares at least one friend in common,
such that $Z_{ij}=\left(1,1,1\right)'$ . \citet{Menzel_WP16} demonstrates
that the probability that agent $i$'s highest utility link is of
type $Z_{ij}=z$, among all those available to her, takes a logit
form. Furthermore, when $N$ is large, the inclusive value in this
probability depends only on agent $i$'s exogenous attribute $X_{i}$
(for preferences structure different than (\ref{eq: konrad_preferences})
the argument may be a bit more complicated).

In setting up the sequence of network formation games appropriately,
and also in characterizing the resulting limit, \citet{Menzel_WP16}
demonstrates considerable ingenuity and technical skill. Stepping
back, the underlying intuition is quite simple. Under exchangeability
of agents, the link formation process for observationally identically
agents should be similar when $N$ is large.

Next consider the link frequency ``distribution''

\[
F_{N}\left(z\right)=\frac{1}{N}\sum_{i=1}^{N}\sum_{j\neq i}\mathbf{1}\left(D_{ij}=1,X_{i}\leq x_{1},X_{j}\leq x_{2},T_{ij}\leq t_{12}\right)
\]
for $z=\left(x_{1},x_{2},t_{12}\right)'$. This is not a proper measure,
it integrates to average degree, not one. Nevertheless these frequencies
have well-defined limits which \citet{Menzel_WP16} is able to relate
to the limiting choice probabilities associated with the infinite
agent network formation game. Note that $F_{N}\left(z\right)$ is
closely related to a network moment, as introduced in (\ref{subsec: Network-moments})
earlier. As in the other papers surveyed in this section, the identification/estimation
approach relates empirical frequencies with model-implied counterparts.
Characterizing these model-implied counterparts (in the limit) is
non-trivial. \citet{Menzel_WP16} shows that depending on the preference
structure considered, as well as researcher assumptions about equilibrium
selection, preference parameters may be point or set identified. For
point identified models \citet{Menzel_WP16} suggests a constrained
maximum likelihood estimator based on the form of the limiting model.

\subsection{Models with (unobserved) sequential meeting processes}

\citet{Miyauchi_JOE16}, \citet{Leung_JOE15}, \citet{dePaula_et_al_EM18}
and \citet{Menzel_WP16} all model network formation as a static game.
Any underlying dynamics governing link formation are left unmodeled.
This is in keeping with the agnosticism regarding equilibrium selection
maintained by these researchers. \citet{Mele_EM17} and \citet{Mele_Zhu_WP17},
in contrast, present models of network formation which make explicit
assumptions about how agents meet, form, dissolve and maintain links. 

In their model pairs of agents meet sequentially. Upon meeting a dyad
decides to either form, maintain, or dissolve a link. Although the
utility attached to any given link may depend on current network structure,
agents are not forward looking. Rather agents myopically add, maintain,
or subtract links in order to raise current utility without anticipating
the effects of their actions on the future decisions of other agents
in the network.

To discuss their results I work with the preference specification
featured in \citet{Mele_Zhu_WP17}. Let $\mathbf{d}_{t}$ be a particular
undirected network configuration in period $t$. The utility agent
$i$ gets from such a configuration is given by
\begin{equation}
\nu_{i}\left(\mathbf{d}_{t},\mathbf{X},\mathbf{U}_{t};\theta_{0}\right)=\sum_{j}d_{ijt}\left[\alpha_{0}'R_{ij}^{*}+\frac{\beta_{0}}{N}\left[\sum_{k}d_{jkt}\right]-U_{ijt}\right]\label{eq: Mele_Zhu_Preferences}
\end{equation}
with $\theta=\left(\alpha',\beta\right)'$. Here $R_{ij}^{*}=r^{*}\left(X_{i},X_{j}\right)$
is a vector of known functions of $X_{i}$ and $X_{j}$. This term
indexes, for example, the utility gains from homophilous sorting.
The second term in (\ref{eq: Mele_Zhu_Preferences}) captures the
benefits associated with indirect connections; that is, the return
agent $i$ receives from linking with $j$ may, in part, depend on
the number of links $j$ already has. If $\beta_{0}>0$ ($\beta_{0}<0$),
then there exist utility gains from linking with more (less) popular
agents. It is also possible to incorporate a transitivity, or mutual
friends, term into (\ref{eq: Mele_Zhu_Preferences}). 

The preference shock $U_{ijt}$ is a Type I extreme value random variable;
independently distributed across dyads and over time. I will return
to the implications of these assumptions for the interpretation and
identification of the model shortly.

Under (\ref{eq: Mele_Zhu_Preferences}) the marginal utility agent
$i$ gets from a link with $j$ is

\begin{equation}
MU_{ij}\left(\mathbf{d}_{t},\mathbf{X},\mathbf{U}_{t};\theta_{0}\right)=\alpha_{0}'R_{ij}^{*}+\frac{\beta_{0}}{N}\left[\sum_{k}d_{jkt}\right]-U_{ijt}.\label{eq: Mele_Zhu_MU}
\end{equation}
\citet{Mele_Zhu_WP17} assume utility is transferable. This implies
that if $i$ and $j$ meet in period $t$ they will form (or maintain)
a link if the net surplus from doing so is positive \citep[cf., ][]{Bloch_Jackson_JET07}:
\begin{align}
MU_{ij}\left(\mathbf{d}_{t},\mathbf{X},\mathbf{U}_{t};\theta_{0}\right)+MU_{ji}\left(\mathbf{d}_{t},\mathbf{X},\mathbf{U}_{t};\theta_{0}\right)\geq0\Longleftrightarrow & R_{ij}'\alpha_{0}+\frac{\beta_{0}}{N}\left[\sum_{k}\left(d_{ikt}+d_{jkt}\right)\right]\label{eq: Mele_Zhu_NetSurplus}\\
 & \geq U_{ijt}+U_{jit},\nonumber 
\end{align}
where $R_{ij}=R_{ij}^{*}+R_{ji}^{*}$. The $R_{ij}$ term is analogous
to the vector of regressors appearing in the dyadic regression model
discussed in Section \ref{sec: dyadic_regression}. Observe, in keeping
with the undirected nature of the network, that (\ref{eq: Mele_Zhu_NetSurplus})
is invariant to permutations of the agents' indices. 

Dyads meet one at a time (i.e., sequentially). In each period the
probability that a particular $ij$ dyad is chosen, say $\rho_{ij}$,
is greater than zero. Let $Z_{t}=ij$ if dyad $\left\{ i,j\right\} $
is chosen to meet in period $t$. This meeting variable equals one
of the $\binom{N}{2}$ possible dyad index pairs each period. Conditional
on $i$ and $j$ meeting, as well as the beginning-of-period-$t$
network structure, the probability that they form (or maintain) a
link is logistic:
\[
\Pr\left(\left.\mathbf{D}_{t+1}=\mathbf{D}_{t}+ij\right|\mathbf{D}_{t},\mathbf{X},Z_{t}=ij;\theta_{0}\right)=\frac{\exp\left(R_{ij}'\alpha_{0}+\frac{\beta_{0}}{N}\left[\sum_{k}\left(D_{ikt}+D_{jkt}\right)\right]\right)}{1+\exp\left(R_{ij}'\alpha_{0}+\frac{\beta_{0}}{N}\left[\sum_{k}\left(D_{ikt}+D_{jkt}\right)\right]\right)}.
\]
This link probability function augments the simple dyadic logistic
regression model introduced earlier with terms, in this case a popularity
effect, which arise due to interdependencies in preferences.

Under these assumptions the sequence of adjacency matrices $\mathbf{D}_{0},\mathbf{D}_{1},\ldots$
is a Markov chain with transition probabilities depending on the exact
specification of the meeting process and the logistic probabilities
specified above. This chain is irreducible and aperiodic. Therefore
the ergodic theorem implies that, in the limit, realized networks
will correspond to draws from a unique stationary distribution. \citet[Theorem 2.1]{Mele_Zhu_WP17}
show that this stationary distribution equals \citep[cf., ][]{Blume_GEB93}
\begin{equation}
\pi_{N}\left(\mathbf{d};\mathbf{X},\theta_{0}\right)=\frac{\exp\left(Q_{N}\left(\mathbf{d};\mathbf{X},\theta_{0}\right)\right)}{\sum_{\mathbf{v}\in\mathbb{D}_{N}}\exp\left(Q_{N}\left(\mathbf{v};\mathbf{X},\theta_{0}\right)\right)}\label{eq: ERGM}
\end{equation}
for $Q_{N}\left(\mathbf{d};\mathbf{X},\theta\right)=\sum_{i=1}^{N}\nu_{i}\left(\mathbf{d},\mathbf{X},\mathbf{0};\theta_{0}\right)$.
See also \citet[Theorem 1]{Mele_EM17}. Equation (\ref{eq: ERGM})
corresponds to what network researchers call an \emph{exponential
random graph model }(ERGM). \citet{Robins_et_al_SN07a} and \citet{Robins_et_al_SN07b}
provide an overview of ERGMs for social network analysis. \citet[Section 4.1]{dePaula_WC17}
provides an interesting overview from the vantage point of an econometrician.
The results of \citet{Mele_EM17} and \citet{Mele_Zhu_WP17} provide
a microeconomic foundation for (certain forms of) ERGMs. This is interesting,
especially in light of peculiarities of the ERGM modeling framework
emphasized by others \citep[e.g., ][]{Shalizi_Rinaldo_AS13}.

It turns out that $Q_{N}\left(\mathbf{d};\mathbf{X},\theta\right)$
is also the potential function, in the sense of \citet{Monderer_Shapley_GEB96},
associated with a particular network formation game. Consider preference
structure (\ref{eq: Mele_Zhu_Preferences}), but with all the pair-specific
preference shocks set identically equal to zero. The set of networks
which (locally) maximize $Q_{N}\left(\mathbf{d};\mathbf{X},\theta\right)$
correspond to the set of Nash equilibrium networks associated with
the simultaneous move network formation game under these zero heterogeneity
preferences. The stationary distribution (\ref{eq: ERGM}) clearly
has modes at these equilibria. \citet{Mele_Zhu_WP17} assume that
the econometrician observes a single draw from this stationary distribution.
This draw, loosely, can be viewed as a random perturbation of an equilibrium
network in the associated ``heterogeneity free'' simultaneous move
static game.

Unfortunately computing the maximum likelihood estimate of $\theta_{0}$
is not straightforward. This is because the denominator in (\ref{eq: ERGM})
involves a summation over all undirected networks of order $N$. It
is impossible to evaluate this summation directly except for trivially
small networks. Furthermore approximate computation of the MLE via,
for example, Markov Chain Monte Carlo (MCMC) methods, is also difficult
\citep[e.g., ][]{Bhamidi_et_al_AAP11}.

\citet{Mele_Zhu_WP17}, building on ideas in \citet{Chatterjee_Diaconis_AS13}
and \citet{Chatterjee_Dembo_AM16}, propose an approximate variational
estimate of $\theta_{0}$ \citep[cf., ][]{Daudin_et_al_SC08,Bickel_et_al_AS13}.
While their approximation does not generally coincide with the MLE,
they show that the difference between the two shrinks to zero as $N$
grows large.

At a high level they proceed as follows. First, consider the conditional
edge independence model introduced in (\ref{sec: Basic-probability-tools}):
\begin{equation}
\Pr\left(\mathbf{D}=\mathbf{d};q\right)=\prod_{i<j}q_{ij}^{d_{ij}}\left(1-q_{ij}\right)^{1-d_{ij}},\label{eq: ced_approx}
\end{equation}
with $q_{ij}$ equaling the probability that $i$ and $j$ link. In
this context the conditional edge independence model is sometimes
called the mean-field approximation. Exploiting ideas in \citet{Wainwright_Jordan_FnTML08}
and \citet{He_Zheng_IEEE13}, they observe that the (log of the) constant
of integration in (\ref{eq: ERGM}) is bounded below by
\begin{align*}
\frac{1}{N^{2}}\ln\left[\sum_{\mathbf{v}\in\mathbb{D}_{N}}\exp\left(Q_{N}\left(\mathbf{v};\theta\right)\right)\right] & \geq\mathbb{E}_{q}\left[Q_{N}\left(\mathbf{D};\theta\right)\right]+\frac{1}{N^{2}}\mathbb{S}\left(q\right)
\end{align*}
with $\mathbb{S}\left(q\right)$ denoting Shannon's Entropy and the
expectation with respect to the approximating mean field model (\ref{eq: ced_approx}).
Next choose the probabilities $q=\left(q_{12},q_{13},\ldots,q_{N-1N}\right)'$
to maximize the above lower bound. This is the variational problem.
Clearly, the optimal approximation will vary with $\theta$, the structural
parameter of interest. The approximation will also not be exact, since
conditional edge independence models represent only a restricted set
of all the possible probability distributions on $\mathbb{D}_{N}$.
The variational estimate of $\theta_{0}$, say $\hat{\theta}_{\mathrm{VE}}$,
is chosen to maximize (\ref{eq: ERGM}) after replacing its denominator
with the lower bound described above. 

\citet{Mele_Zhu_WP17}, using a result in \citet{Chatterjee_Dembo_AM16},
show that the lower bound approximation becomes tight as $N\rightarrow\infty$.
Furthermore the limit of the variational problem corresponds to finding
a graphon. More precisely, they find that as $N\rightarrow\infty$,
\begin{enumerate}
\item the stationary distribution associated with their strategic network
formation model is arbitrarily well approximated by a conditional
edge independence model with some graphon $h\left(u,v\right)$, or
a mixture of such models; 
\item these graphons correspond to local maximizers of a limiting version
of the variational problem; and
\item $\hat{\theta}_{\mathrm{VE}}$ coincides with a local maximizer of
(\ref{eq: ERGM}).
\end{enumerate}
The first finding is to be expected given the Aldous-Hoover Theorem
and associated discussion in Section \ref{sec: Basic-probability-tools}.
The second result is related to work by \citet{Chatterjee_Diaconis_AS13}.
It is of interest here since it provides a connection between a structural
model of strategic network formation and the exchangeable random graph
theory reviewed earlier.

While \citet{Mele_EM17} provides a nice microeconomic potential game
interpretation of ERGMs, and \citet{Mele_Zhu_WP17} make important
progress on methods of estimation, major challenges in the areas of
identification, estimation and inference in this class of models nevertheless
remain.

\citet{Christakis_et_al_NBER10} also model link formation as a sequential
process. Their approach differs from that of \citet{Mele_EM17}. They
assume the initial network is empty and that all $\tbinom{N}{2}$
dyads meet in a specific (unobserved) order. Upon meeting they myopically
decide whether to form a link or not. After all pairs of agents meet
once, further link revisions do not occur. In order to construct a
likelihood \citet{Christakis_et_al_NBER10} assigned a distribution
to the unobserved meeting sequence and integrate it out. For computation
they develop a Bayesian approach based on MCMC methods. One feature
of their set-up is that the model may place positive probability on
network configurations that are not pairwise stable. In contrast the
ergodic distribution associated with Mele's \citeyearpar{Mele_EM17}
model places most of its mass in the neighborhood of equilibrium network
configurations. While this may be viewed as undesirable, from a computational
standpoint the \citet{Christakis_et_al_NBER10} method appears attractive.
In principal their model could be extended to allow each pair of agents
to meet multiple (but still a finite number of) times.

\subsection{Further reading and open questions}

With exception of the paper by \citet{Miyauchi_JOE16}, all of the
papers surveyed above base estimation and inference on a single network.
To get workable LLNs and CLTs each of these authors deal with the
dependence across links induced by strategic interaction in interesting
ways. \citet{Leung_JOE15} introduces private information; in a resulting
Bayes-Nash equilibrium links are conditionally independent given common
information. The reduced form probability of a directed link from
$i$ to $j$ implied by his model is quite similar to the representation
result associated with $X$-exchangeability introduced in the context
of dyadic regression in Section \ref{sec: dyadic_regression}. To
a first approximation this probability depends only upon $X_{i}$
and $X_{j}$ (since the other sources of variation in $\sum_{k\neq i,j}P_{ki}P_{kj}$
should be rather modest when $N$ is large enough). Therefore, relative
to a simple dyadic probit model, the \citet{Leung_JOE15} model adds
an equilibrium constraint.

In \citet{dePaula_et_al_EM18} the key assumption appears to be that
preference heterogeneity is over \emph{types} of links alone, with
no dyad-specific component. As mentioned earlier, similar assumptions
have proved to be very powerful in the literature on matching. Although
\citet{Menzel_WP16} works with a model which generates a sparse graph
in the limit (with dependence across links vanishing), his use of
exchangeability arguments does suggest connections to the Aldous-Hoover
type representation results introduced earlier. \citet{Mele_EM17}
and \citet{Christakis_et_al_NBER10} posit sequential meeting processes
that effectively ``complete'' what would otherwise be an incomplete
simultaneous move $N$-player game. Each of these approaches have
pros and cons; a variety of computational and inference issues remain
unsolved. At the same time the creativity and diversity of them suggests
that forward progress on these types of models is possible. Better
understanding the connections between different modeling assumptions
would be useful.

Another approach, not surveyed here, but nevertheless promising, involves
working with subnetworks. A focus on subnetworks sidesteps some of
the computational challenges that arise when trying to apply methods
from the econometrics of games to network formation problems (where
there are typically many agents). \citet{Sheng_WP14} pioneered this
approach. \citet{Gualdani_JOE19} develops additional (related) results.

\section{The bright and happy future of network econometrics}

This chapter has surveyed a burgeoning literature on the econometrics
of networks. This literature -- combining insights from econometric
research on panel data and games, new tools in applied probability
and statistics, and original thinking -- now provides a basic set
of tools for the analysis of networks. Nevertheless substantial work
remains unfinished. As noted at the start of this chapter, datasets
with natural graph theoretic structure abound in economics, and increasingly
feature in published research. Each dataset exhibits its own peculiarities:
in some links are undirected, in others directed. The network may
be bipartite or even multi-partite\citep[e.g.,][]{Min_WP2019}. The
size and order of available network datasets vary immensely. In some
cases a network may be observed over multiple periods, in others just
once. For many of these settings there exist no extant econometric
modeling strategies, in all of them existing work could be improved
in a number of ways.

A defining feature of the econometric approach to modeling network
formation is its random utility foundation. When preferences are interdependent
-- where the utility two agents attach to a candidate link may vary
with the presence or absence of links elsewhere in the network --
multiple equilibrium network configurations are likely. The analysis
of incomplete models is an important recent accomplishment of econometrics.
The combinatoric complexity of large networks will require new developments
in this area. The set of papers surveyed in Section \ref{sec: Strategic-models}
gives some flavor of the key issues and possible solutions. 

Another defining feature of modern microeconometric research is the
incorporation of unobserved heterogeneity; heterogeneity that agents
observe and act upon, but which is unobserved by the researcher. In
the single agent setting panel data facilitates the identification
and estimation of models with rich heterogeneity structures. Networks
have natural panel-like aspects. In a dense network each agent decides
whether to (attempt to) form a link with all other agents. Multiple
decisions per agent are observed. Leveraging this panel-like structure
has been a key feature of some of the contributions surveyed in Sections
\ref{sec: dyadic_regression}, \ref{sec: policy_analysis} and \ref{sec: heterogeneity}
above.

Understanding the properties of the different methods surveyed above
under sequences of networks which are dense, sparse or somewhere in
between, remains incomplete. Uniformity of testing procedures across
these various cases would be desirable. Some preliminary work on bootstrapping
methods in the networks setting now exists \citep[e.g.,][]{Green_Shalizi_arXiv17,Menzel_arXiv17,Davezies_et_al_arXiv2019},
but this remains relatively unexplored. Semiparametric efficiency
bounds are yet to be characterized, let alone the development of estimators
attaining them. Computational advances will be important for spurring
real world application.

This chapter has focused on network formation. While the question
of how networks form is scientifically interesting, so is that of
what they do? This latter question was a key driver of the peer effects
literature which emerged after \citet{Manski_ReStud93}. Developing
methods for simultaneously modeling the formation and consequences
of social and economic networks remains an important open area \citep{Auerbach_JMP16,Johnson_Moon_INET17}.
Finally, although more and more empirical work with a network dimensions
appears each year, application of the methods outlined above in substantive
empirical work is a high priority. In addition to whatever subject
area insights such applications may produce, they will no doubt spur
further methodological innovations.

\appendix

\section{\label{app: proofs_and_derivations}Appendix}
\begin{lem}
\label{lem: u-statistic_with_estimated_parameter}\textsc{(U-Statistic
with Estimated Parameter) }Let $\left\{ Z_{i}\right\} _{i=1}^{N}$
be a simple random sample drawn from some population $F_{Z}$ and
$\phi\left(Z_{i},Z_{j};\beta,\gamma\right)$ be a function from $\mathbb{Z}\times\mathbb{Z}$
to $\mathbb{R}^{J}$ indexed by $\beta\in\mathbb{B}$ and $\gamma\in\mathbb{C}$
(with $\mathbb{B}$ and $\mathbb{C}$ compact subsets of $\mathbb{R}^{\dim\left(\beta\right)}$
and $\mathbb{R}^{\dim\left(\gamma\right)}$ respectively). Suppose
that $\phi\left(z_{1},z_{2};\beta,\gamma\right)$ is twice continuously
differentiable in $\gamma$ for all $z_{1},z_{2}\in\mathbb{Z}\times\mathbb{Z}$
with
\begin{align}
\mathbb{E}\left[\left\Vert \phi\left(Z_{1},Z_{2};\beta,\gamma\right)\right\Vert _{2}\right] & <\infty\label{eq: 2s_lem_bounded}\\
\mathbb{E}\left[\left\Vert \frac{\partial\phi\left(Z_{1},Z_{2};\beta,\gamma\right)}{\partial\gamma'}\right\Vert _{F}\right] & <\infty\label{eq: 2s_lem_bounded_1st_der}\\
\mathbb{E}\left[\left\Vert \frac{\partial}{\partial\gamma'}\left\{ \frac{\partial\phi\left(Z_{1},Z_{2};\beta,\gamma\right)}{\partial\gamma_{p}}\right\} \right\Vert _{F}\right] & <\infty,\ \ p=1,\ldots,\dim\left(\gamma\right).\label{eq: 2s_lem_bounded_2nd_der}
\end{align}
Then, for $\hat{\gamma}$ a $\sqrt{N}$-consistent estimate of $\gamma_{0}$,
and defining $\bar{\phi}_{N}\left(\beta,\gamma\right)\overset{def}{\equiv}\tbinom{N}{2}^{-1}\sum_{i=1}^{N}\sum_{j=i+1}^{N-1}\phi\left(Z_{i},Z_{j};\beta,\gamma\right)$
and $\Phi\left(\beta,\gamma\right)\overset{def}{\equiv}\mathbb{E}\left[\phi\left(Z_{1},Z_{2};\beta,\gamma\right)\right]$,
we have
\begin{align}
\sqrt{N}\left[\bar{\phi}_{N}\left(\beta,\hat{\gamma}\right)-\Phi\left(\beta,\gamma_{0}\right)\right]= & \frac{2}{\sqrt{N}}\sum_{i=1}^{N}\psi_{0}\left(Z_{1};\beta,\gamma_{0}\right)+\Gamma_{0,\beta\gamma}\left(\beta\right)\sqrt{N}\left(\hat{\gamma}-\gamma_{0}\right)+o_{p}\left(1\right)\label{eq: 2s_lem_u_statistic_expansion}
\end{align}
where $\phi_{1}\left(z;\beta,\gamma\right)=\mathbb{E}\left[\phi\left(z,Z_{1};\beta,\gamma\right)\right]$
and
\begin{align*}
\psi_{0}\left(Z_{1};\beta,\gamma\right)= & \phi_{1}\left(Z_{1};\beta,\gamma\right)-\Phi\left(\beta,\gamma\right)\\
\Gamma_{0,\beta\gamma}\left(\beta\right)= & \mathbb{E}\left[\frac{\partial\phi\left(Z_{1},Z_{2};\beta,\gamma_{0}\right)}{\partial\gamma'}\right].
\end{align*}
\end{lem}

\subsubsection*{Proof of Lemma \ref{lem: u-statistic_with_estimated_parameter}}

A Taylor expansion of $\bar{\phi}_{N}\left(\beta,\hat{\gamma}\right)$
in $\hat{\gamma}$ about $\gamma_{0}$ yields, after some re-arrangement
and centering,
\begin{equation}
\sqrt{N}\left[\bar{\phi}_{N}\left(\beta,\hat{\gamma}\right)-\Phi\left(\beta,\gamma_{0}\right)\right]=\sqrt{N}\left[\bar{\phi}_{N}\left(\beta,\gamma_{0}\right)-\Phi\left(\beta,\gamma_{0}\right)\right]+\Gamma_{N,\beta\gamma}\left(\beta,\bar{\gamma}\right)\sqrt{N}\left(\hat{\gamma}-\gamma_{0}\right),\label{eq: 2s_lem_taylor_expansion}
\end{equation}
with $\bar{\gamma}$ a mean value between $\hat{\gamma}$ and $\gamma_{0}$
which may vary across the rows of the Hessian $\Gamma_{N,\beta\gamma}\left(\beta,\gamma\right)\overset{def}{\equiv}\frac{\partial\bar{\phi}_{N}\left(\beta,\gamma\right)}{\partial\gamma'}.$
Next recall the definition of the $L_{2,1}$ norm:

\begin{equation}
\left\Vert \mathbf{A}\right\Vert _{2,1}=\sum_{j=1}^{n}\left[\sum_{i=1}^{m}\left|a_{ij}\right|^{2}\right]^{1/2}.\label{eq: L_21_norm}
\end{equation}
The mean value theorem, as well as compatibility of the Frobenius
matrix norm with the Euclidean vector norm, gives for any $\gamma$
and $\gamma^{*}$ both in $\mathbb{C}$,
\begin{align}
\left\Vert \frac{\partial\bar{\phi}_{N}\left(\beta,\gamma\right)}{\partial\gamma'}-\frac{\partial\bar{\phi}_{N}\left(\beta,\gamma^{*}\right)}{\partial\gamma'}\right\Vert _{2,1} & \leq\sum_{p=1}^{\dim\left(\gamma\right)}\left\Vert \frac{\partial}{\partial\gamma'}\left\{ \frac{\partial\bar{\phi}_{N}\left(\beta,\gamma\right)}{\partial\gamma_{p}}\right\} \right\Vert _{F}\left\Vert \gamma-\gamma_{*}\right\Vert _{2}.\label{eq: 2s_lem_mean_value_inequality}
\end{align}
Observe that $\frac{\partial}{\partial\gamma'}\left\{ \frac{\partial\bar{\phi}_{N}\left(\beta,\gamma\right)}{\partial\gamma_{p}}\right\} $
is a matrix of U-statistics with kernels whose first moments are finite
(by condition \pageref{eq: 2s_lem_bounded_2nd_der} above). By \citet[Theorem 5.4A]{Serfling_ATMS80}
these U-statistics converge in probability and hence, from (\ref{eq: 2s_lem_mean_value_inequality})
\[
\left\Vert \frac{\partial\bar{\phi}_{N}\left(\beta,\gamma\right)}{\partial\gamma'}-\frac{\partial\bar{\phi}_{N}\left(\beta,\gamma^{*}\right)}{\partial\gamma'}\right\Vert _{2,1}\leq O_{p}\left(1\right)\cdot\left\Vert \gamma-\gamma_{*}\right\Vert _{2}.
\]
This condition, as well compactness of $\mathbb{C}$, continuity of
$\frac{\partial\bar{\phi}_{N}\left(\beta,\gamma\right)}{\partial\gamma}$
in $\gamma$, and condition (\ref{eq: 2s_lem_bounded_1st_der}), allow
for an application of Lemma 2.9 in \citet{Newey_McFadden_HBE94} such
that $\underset{\gamma\in\mathbb{C}}{\sup}\left\Vert \frac{\partial\bar{\phi}_{N}\left(\beta,\gamma\right)}{\partial\gamma'}-\Gamma_{\beta\gamma}\left(\beta,\gamma\right)\right\Vert _{F}\overset{p}{\rightarrow}0$
with $\Gamma_{\beta\gamma}\left(\beta,\gamma\right)=\mathbb{E}\left[\frac{\partial\phi\left(Z_{1},Z_{2};\beta,\gamma\right)}{\partial\gamma'}\right].$
This, along with consistency of $\hat{\gamma}$ for $\gamma_{0}$,
is enough to ensure that $\frac{\partial\bar{\phi}_{N}\left(\beta,\bar{\gamma}\right)}{\partial\gamma'}\overset{p}{\rightarrow}\Gamma_{0,\beta\gamma}\left(\beta\right)$.
Equation (\ref{eq: 2s_lem_u_statistic_expansion}) then follows by
observing that $\bar{\phi}_{N}\left(\beta,\gamma_{0}\right)-\Phi\left(\beta,\gamma_{0}\right)$
is a vector of mean zero U-Statistics with Hájek projections equal
to the corresponding components of the first term to the right of
the equality in (\ref{eq: 2s_lem_u_statistic_expansion}) (see, for
example, Theorem 5.3.3. of \citet{Serfling_ATMS80} and invoke condition
(\ref{eq: 2s_lem_bounded}) above). See \citet[Lemma S1]{Mao_BM18}
for a related Lemma.

\subsubsection*{Order of variances and covariances for $p^{th}$ order induced subgraph
frequencies}

Here I present the order of the covariance between empirical subgraph
frequencies, where the subgraph is of arbitrary order. For general
$p^{th}$-order graphlets $R$ and $S$ we have that
\begin{align}
\mathbb{C}\left(P_{N}\left(R\right),P_{N}\left(S\right)\right)= & \binom{N}{p}^{-2}\sum_{q=1}^{p}\binom{N}{p}\binom{p}{q}\binom{N-p}{p-q}\Sigma_{q}\left(R,S\right)\nonumber \\
= & \binom{N}{p}^{-2}\sum_{q=1}^{p}\binom{N}{p}\binom{p}{q}\binom{N-p}{p-q}\Xi\left(\mathcal{W}_{q,R,S}\right)\nonumber \\
 & -\left[1-\frac{\left(N-p\right)!^{2}}{N!\left(N-2p\right)!}\right]P\left(R\right)P\left(S\right).\label{eq: covariance_P(R)_P(S)}
\end{align}
Normalizing by $\rho_{N}$ raised to the number of edges in $R$ and
$S$, respectively $\rho_{N}^{e\left(R\right)}$ and $\rho_{N}^{e\left(S\right)}$,
yields
\begin{align}
\mathbb{C}\left(\tilde{P}_{N}\left(R\right),\tilde{P}_{N}\left(S\right)\right)= & \left\{ \underset{O\left(N^{-q}\rho_{N}^{-e\left(R\right)}\rho_{N}^{-e\left(S\right)}\right)O\left(\Xi\left(\mathcal{W}_{q,R,S}\right)\right)}{\underbrace{\binom{N}{p}^{-2}\sum_{q=1}^{p-1}\binom{N}{p}\binom{p}{q}\binom{N-p}{p-q}\left[\frac{\Xi\left(\mathcal{W}_{q,R,S}\right)}{\rho_{N}^{e\left(R\right)}\rho_{N}^{e\left(S\right)}}\right]}}\right.\nonumber \\
 & \left.-\left[1-\frac{\left(N-p\right)!^{2}}{N!\left(N-2p\right)!}\right]\tilde{P}\left(R\right)\tilde{P}\left(S\right)\right\} .\label{eq: covariance_P(R)_P(S)_normalized}
\end{align}

There are $2p-q$ vertices in each element of $\mathcal{W}_{q,R,S}$.

\subsubsection*{\uline{Case 1 (\mbox{$q=1$}):}}

If $q=1$, then $e\left(W\right)=e\left(R\right)+e\left(S\right)$
for all $W\in\mathcal{W}_{q,R,S}.$ This gives
\begin{align*}
O\left(N^{-q}\rho_{N}^{-e\left(R\right)}\rho_{N}^{-e\left(S\right)}\right)O\left(\Xi\left(\mathcal{W}_{1,R,S}\right)\right) & =O\left(N^{-1}\rho_{N}^{-e\left(R\right)}\rho_{N}^{-e\left(S\right)}\right)O\left(\rho_{N}^{e\left(R\right)}\rho_{N}^{e\left(S\right)}\right)\\
 & =O\left(N^{-1}\right).
\end{align*}

\subsubsection*{\uline{Case (\mbox{$q=p$}):}}

If $q=p$, then $\Xi\left(\mathcal{W}_{q,R,S}\right)=0$ unless $R=S$.
In that case, the ``variance case'', we have that $e\left(W\right)=p$
since $W=R=S.$ This gives
\begin{align*}
O\left(N^{-q}\rho_{N}^{-2e\left(R\right)}\right)O\left(\Xi\left(\mathcal{W}_{p,R}\right)\right) & =O\left(N^{-p}\rho_{N}^{-2e\left(R\right)}\right)O\left(\rho_{N}^{e\left(R\right)}\right)\\
 & =O\left(N^{-p}\rho_{N}^{-e\left(R\right)}\right).
\end{align*}

If $R$ is a p-cycle, then $p=e\left(R\right)$, yielding the simplification
$O\left(N^{-p}\rho_{N}^{-e\left(R\right)}\right)=O\left(\lambda_{N}^{-p}\right)$.

If $R$ is a tree, then $e\left(R\right)=p-1$, yielding the simplification
$O\left(N^{-p}\rho_{N}^{-e\left(R\right)}\right)=O\left(N^{-1}\lambda_{N}^{-\left(p-1\right)}\right)$.

\subsubsection*{\uline{Case (\mbox{$1<q<p$}):} }

For $q=2,\ldots,p-1$ we have that $e\left(W\right)=e\left(R\right)+e\left(S\right)-\left(q-1\right)$
if $R$ and $S$ are both p-cycles so that
\begin{align*}
O\left(N^{-q}\rho_{N}^{-e\left(R\right)}\rho_{N}^{-e\left(S\right)}\right)O\left(Q\left(\mathcal{W}_{q,R,S}\right)\right) & =O\left(N^{-q}\rho_{N}^{-e\left(R\right)}\rho_{N}^{-e\left(S\right)}\right)O\left(\rho_{N}^{e\left(R\right)+e\left(S\right)-\left(q-1\right)}\right)\\
 & =O\left(N^{-q}\rho_{N}^{-\left(q-1\right)}\right)\\
 & =O\left(N^{-1}\lambda_{N}^{-\left(q-1\right)}\right).
\end{align*}
Whereas we have that $e\left(W\right)\geq e\left(R\right)+e\left(S\right)-\left(q-1\right)$
if $R$ and $S$ are both trees, or one is a tree and the other a
p-cycle, so that
\begin{align*}
O\left(N^{-q}\rho_{N}^{-e\left(R\right)}\rho_{N}^{-e\left(S\right)}\right)O\left(\Xi\left(\mathcal{W}_{q,R,S}\right)\right) & \leq O\left(N^{-q}\rho_{N}^{-e\left(R\right)}\rho_{N}^{-e\left(S\right)}\right)O\left(\rho_{N}^{e\left(R\right)+e\left(S\right)-\left(q-1\right)}\right)\\
 & =O\left(N^{-1}\lambda_{N}^{-\left(q-1\right)}\right).
\end{align*}

\subsubsection*{Proof of Theorem \ref{thm: degree_sequence_moments}}

Without loss of generality set $i=1$. By the definition of degree
we have that 
\[
\mathbb{E}\left[D_{1+}^{m}\right]=\mathbb{E}\left[\left(\sum_{j=2}^{N}D_{1j}\right)^{m}\right],
\]
the multinomial theorem allows us to write the term inside the expectation
above as
\begin{align}
D_{1+}^{m}=\left(\sum_{j=2}^{N}D_{1j}\right)^{m} & =\sum_{q_{2}+\cdots+q_{N}=m}\binom{m}{q_{2},q_{3},\cdots,q_{N}}\prod_{j=2}^{N}D_{1j}^{q_{j}}\label{eq: degree_moments_multi_thm}
\end{align}
where $\binom{m}{q_{2},q_{3},\cdots,q_{N}}=\frac{m!}{q_{2}!q_{3}!\cdots q_{N}!}.$
Since $D_{1j}$ is binary $D_{1j}^{q_{j}}=D_{1j}$ for all $q_{j}=1,2,\ldots,m$
and zero when $q_{j}=0$. This implies that $\prod_{j=2}^{N}D_{1j}^{q_{j}}=D_{1j_{1}}\times\cdots\times D_{1j_{k}}$
for $D_{1j_{1}}$, $D_{1j_{2}},\ldots,D_{1j_{k}}$ the set of $1\leq k\leq m$
link indicators with $q_{j}\geq1$. Consider agents $j_{1},j_{2},\ldots,j_{k}$,
with, say, $q_{j_{1}}=p_{1}$, $q_{j_{2}}=p_{2},\ldots,q_{j_{k}}=p_{k}$
such that $\mathbf{p}\in\mathcal{P}_{k,m}$, it follows that
\begin{equation}
\prod_{j=2}^{N}D_{1j}^{q_{j}}=D_{1j_{1}}^{p_{1}}\times\cdots\times D_{1j_{k}}^{p_{k}}.\label{eq: degree_moments_ex_term}
\end{equation}
By the multinomial theorem the coefficient on (\ref{eq: degree_moments_ex_term})
equals $\frac{m!}{p_{1}!\times\cdots\times p_{k}!}$, but since 
\[
D_{1j_{1}}^{p_{1}}\times\cdots\times D_{1j_{k}}^{p_{k}}=D_{1j_{1}}^{p_{1}^{*}}\times\cdots\times D_{1j_{k}}^{p_{k}^{*}}=D_{1j_{1}}\times\cdots\times D_{1j_{k}}
\]
for any $\mathbf{p},\mathbf{p}^{*}\in\mathcal{P}_{k,m}$, the coefficient
on $D_{1j_{1}}\times\cdots\times D_{1j_{k}}$ after combining identical
terms in (\ref{eq: degree_moments_multi_thm}) equals $\sum_{\mathbf{p}\in\mathcal{P}_{k,m}}\frac{m!}{p_{1}!\times\cdots\times p_{k}!}.$
Putting these pieces together yields
\[
\mathbb{E}\left[D_{i+}^{m}\right]=\sum_{k=1}^{m}\left(\sum_{\mathbf{p}\in\mathcal{P}_{k,m}}\frac{m!}{p_{1}!\times\cdots\times p_{k}!}\right)\mathbb{E}\left[\sum_{j_{1}<\cdots<j_{k}}D_{ij_{1}}\times\cdots\times D_{ij_{k}}\right]
\]
The expectations of the summands in $\sum_{j_{1}<\cdots<j_{k}}D_{ij_{1}}\times\cdots\times D_{ij_{k}}$
are all identical with cardinality $\tbinom{N-1}{k}$. The assertion
follows.

\bibliographystyle{apalike2}
\bibliography{../../Networks_Book/Finished/Reference_BibTex/Networks_References}

\begin{thebibliography}{}

\bibitem[Acemoglu et~al., 2016a]{Acemoglu_et_al_NBERMacro16}
Acemoglu, D., Akcigit, U., \& Kerr, W. (2016a).
\newblock Networks and the macroeconomy: an empirical exploration.
\newblock {\em NBER Macroeconomics Annual}, 31(1), 273 -- 335.

\bibitem[Acemoglu et~al., 2016b]{Acemoglu_Akcigit_Kerr_PNAS16}
Acemoglu, D., Akcigit, U., \& Kerr, W.~R. (2016b).
\newblock Innovation network.
\newblock {\em Proceedings of the National Academy of Sciences}, 113(41), 11482
  -- 11488.

\bibitem[Acemoglu et~al., 2012]{Acemoglu_etal_EM12}
Acemoglu, D., Carvalho, V., Ozdaglar, A., \& Tahbaz-Salehi, A. (2012).
\newblock The network origins of aggregate fluctuations.
\newblock {\em Econometrica}, 80(5), 1977 -- 2016.

\bibitem[Ahern \& Harford, 2014]{Ahern_Harford_JoF14}
Ahern, K.~R. \& Harford, J. (2014).
\newblock The importance of industry links in merger waves.
\newblock {\em Journal of Finance}, 69(2), 527 -- 576.

\bibitem[Airoldi et~al., 2008]{Airoldi_et_al_JMLR08}
Airoldi, E.~M., Blei, D.~M., Fienberg, S.~E., \& Xing, E.~P. (2008).
\newblock Mixed membership stochastic blockmodels.
\newblock {\em Journal of Machine Learning Research}, 9, 1981 -- 2014.

\bibitem[Alatas et~al., 2016]{Alatas_et_al_AER16}
Alatas, V., Banerjee, A., Chandrasekhar, A.~G., Hanna, R., \& Olken, B.~A.
  (2016).
\newblock Network structure and the aggregation of information: theory and
  evidence from indonesia.
\newblock {\em American Economic Review}, 106(7), 1663 -- 1704.

\bibitem[Aldous, 1981]{Aldous_JMA81}
Aldous, D.~J. (1981).
\newblock Representations for partially exchangeable arrays of random
  variables.
\newblock {\em Journal of Multivariate Analysis}, 11(4), 581 -- 598.

\bibitem[Ambrus et~al., 2014]{Ambrus_et_al_AER14}
Ambrus, A., Mobius, M., \& Szeidl, A. (2014).
\newblock Consumption risk-sharing in social networks.
\newblock {\em American Economic Review}, 104(1), 149 -- 182.

\bibitem[An et~al., 2018]{An_et_al_SIM18}
An, C., O'Malley, A.~J., Rockmore, D.~N., \& Stock, C.~D. (2018).
\newblock Analysis of the u.s. patient referral network.
\newblock {\em Statistics in Medicine}, 37(5), 847 -- 866.

\bibitem[Anderson, 2011]{Anderson_AR11}
Anderson, J.~E. (2011).
\newblock The gravity model.
\newblock {\em Annual Review in Economics}, 3(1), 133 -- 160.

\bibitem[Andrews \& Soares, 2010]{Andrews_Soares_EM10}
Andrews, D.~W. \& Soares, G. (2010).
\newblock Inference for parameters defined by moment inequalities using
  generalized moment selection.
\newblock {\em Econometrica}, 78(1), 119 -- 157.

\bibitem[Angrist, 2014]{Angrist_LE14}
Angrist, J. (2014).
\newblock The perils of peer effects.
\newblock {\em Labour Economics}, 30, 98--108.

\bibitem[Apicella et~al., 2012]{Apicella_et_al_Nat12}
Apicella, C.~L., Marlowe, F.~W., Fowler, J.~H., \& Christakis, N.~A. (2012).
\newblock Social networks and cooperation in hunter-gatherers.
\newblock {\em Nature}, 481(7382), 497 -- 501.

\bibitem[Arellano \& Bover, 1995]{Arellano_Bover_JE95}
Arellano, M. \& Bover, O. (1995).
\newblock Another look at the instrumental variables estimation of
  error-component models.
\newblock {\em Journal of Econometrics}, 68(1), 29 -- 51.

\bibitem[Arellano \& Hahn, 2007]{Arellano_Hahn_WC07}
Arellano, M. \& Hahn, J. (2007).
\newblock {\em Advances in Economics and Econometrics: Theory and
  Applications}, volume~3, chapter Understanding bias in nonlinear panel
  models: some recent developments, (pp.\ 381 -- 309).
\newblock Cambridge University Press.: Cambridge.

\bibitem[Arellano \& Honor{\'e}, 2001]{Arellano_Honore_HBE01}
Arellano, M. \& Honor{\'e}, B. (2001).
\newblock {\em Handbook of Econometrics}, volume~5, chapter Panel data models:
  some recent developments, (pp.\ 3229 -- 3296).
\newblock North-Holland: Amsterdam.

\bibitem[Aronow et~al., 2017]{Aronow_et_al_PA17}
Aronow, P.~M., Samii, C., \& Assenova, V.~A. (2017).
\newblock Cluster-robust variance estimation for dyadic data.
\newblock {\em Political Analysis}, 23(4), 564 -- 577.

\bibitem[Atalay et~al., 2011]{Atalay_et_al_PNAS11}
Atalay, E., Horta{\c c}su, A., Roberts, J., \& Syverson, C. (2011).
\newblock Network structure of production.
\newblock {\em Proceedings of the National Academy of Sciences}, 108(13), 5199
  -- 5202.

\bibitem[Atalay et~al., 2014]{Atalay_et_Al_AER14}
Atalay, E., Horta{\c c}su, A., \& Syverson, C. (2014).
\newblock Vertical integration and input flows.
\newblock {\em American Economic Review}, 104(4), 1120 -- 1148.

\bibitem[Athey et~al., 2018]{Athey_Eckles_Imbens_JASA18}
Athey, S., Eckles, D., \& Imbens, G.~W. (2018).
\newblock Exact p-values for network interference.
\newblock {\em Journal of American Statistical Association}, forthcoming.

\bibitem[Attanasio et~al., 2012]{Attanasio_AEJ12}
Attanasio, O., Barr, A., Cardenas, J.~C., Genicot, G., \& Meghir, C. (2012).
\newblock Risk pooling, risk preferences, and social networks.
\newblock {\em American Economic Journal: Applied Economics}, 4(2), 134 -- 167.

\bibitem[Auerbach, 2016]{Auerbach_JMP16}
Auerbach, E. (2016).
\newblock {\em A matching estimator for models with endogenous network
  formation}.
\newblock Technical report, University of California - Berkeley.

\bibitem[Badev, 2017]{Badev_arXiv17}
Badev, A. (2017).
\newblock {\em Discrete games in endogenous networks: equilibria and policy}.
\newblock Technical Report 1705.03137, arXiv.

\bibitem[Bajari et~al., 2010a]{Bajari_et_al_JBES10}
Bajari, P., Hong, H., Krainer, J., \& Nekipelov, D. (2010a).
\newblock Estimating static models of strategic interactions.
\newblock {\em Journal of Business and Economic Statistics}, 28(4), 469 -- 482.

\bibitem[Bajari et~al., 2013]{Bajari_et_al_WC13}
Bajari, P., Hong, H., \& Nekipelov, D. (2013).
\newblock {\em Advances in Economics and Econometrics, Tenth World Congress},
  volume~3, chapter Game theory and econometrics: a survey of some recent
  research, (pp.\ 3 -- 52).
\newblock Cambridge University Press: Cambridge.

\bibitem[Bajari et~al., 2010b]{Bajari_et_al_EM10}
Bajari, P., Hong, H., \& Ryan, S.~P. (2010b).
\newblock Identification and estimation of a discrete game of complete
  information.
\newblock {\em Econometrica}, 78(5), 1529 -- 1568.

\bibitem[Baldwin \& Taglioni, 2007]{Baldwin_Taglioni_JEI07}
Baldwin, R. \& Taglioni, D. (2007).
\newblock Trade effects of the euro: a comparison of estimators.
\newblock {\em Journal of Economic Integration}, 22(4), 780 -- 818.

\bibitem[Banerjee et~al., 2013]{Banerjee_et_al_Sci13}
Banerjee, A., Chandrasekhar, A.~G., Dulfo, E., \& Jackson, M.~O. (2013).
\newblock The diffusion of microfinance.
\newblock {\em Science}, 341(6144), 363 -- 370.

\bibitem[Barab{\'a}si, 2016]{Barabasi_Book16}
Barab{\'a}si, A.-L. (2016).
\newblock {\em Network Science}.
\newblock Cambridge: Cambridge University Press.

\bibitem[Barab{\'a}si \& Albert, 1999]{Barabasi_Albert_Sc99}
Barab{\'a}si, A.-L. \& Albert, R. (1999).
\newblock Emergence of scaling in random networks.
\newblock {\em Science}, 286(5439), 509 -- 512.

\bibitem[Barab{\'a}si \& Bonabau, 2003]{Barabasi_Bonabau_SciAmer03}
Barab{\'a}si, A.-L. \& Bonabau, E. (2003).
\newblock Scale-free networks.
\newblock {\em Scientific American}, (pp.\ 50 -- 59).

\bibitem[Barnett et~al., 2011]{Barnett_et_al_HSR11}
Barnett, M., Landon, B., O'Malley, A., Keating, N., \& Christakis, N. (2011).
\newblock Mapping physician networks with self-reported and administrative
  data.
\newblock {\em Health Services Research}, 46(5), 1592 -- 1609.

\bibitem[Barrot \& Sauvagnat, 2016]{Barrot_Sauvagnat_QJE16}
Barrot, J.-N. \& Sauvagnat, J. (2016).
\newblock Input specificity and the propagation of idiosyncratic shocks in
  production networks.
\newblock {\em Quarterly Journal of Economics}, 131(3), 1543 -- 1592.

\bibitem[Beaman, 2011]{Beaman_ReStud11}
Beaman, L.~A. (2011).
\newblock Social networks and the dynamics of labour market outcomes: Evidence
  from refugees resettled in the us.
\newblock {\em Review of Economic Studies}, 79(1), 128 -- 161.

\bibitem[Bearman et~al., 2004]{Bearman_et_al_AJS04}
Bearman, P.~S., Moody, J., \& Stovel, K. (2004).
\newblock Chains of affection: the structure of adolescent romantic and sexual
  networks.
\newblock {\em American Journal of Sociology}, 110(1), 44 -- 91.

\bibitem[Bech \& Atalay, 2010]{Bech_Atalay_PhyA10}
Bech, M.~L. \& Atalay, E. (2010).
\newblock The topology of the federal funds market.
\newblock {\em Physica A: Statistical Mechanics and its Applications}, 389(22),
  5223 -- 5246.

\bibitem[Bellio \& Varin, 2005]{Belio_Varin_SM05}
Bellio, R. \& Varin, C. (2005).
\newblock A pairwise likelihood approach to generalized linear models with
  crossed random effects.
\newblock {\em Statistical Modelling}, 5(3), 217 -- 227.

\bibitem[Bernard et~al., 2018]{Bernard_et_al_JPE18}
Bernard, A.~B., Moxnes, A., \& Saito, Y. (2018).
\newblock Geography and firm performance in the japanese production network.
\newblock {\em Journal of Political Economy}.
\newblock RIETI Discussion Paper 14-E-034.

\bibitem[Bhamidi et~al., 2011]{Bhamidi_et_al_AAP11}
Bhamidi, S., Bresler, G., \& Sly, A. (2011).
\newblock Mixing time of exponential random graphs.
\newblock {\em The Annals of Applied Probability}, 21(6), 2146 -- 2170.

\bibitem[Bhattacharya \& Bickel, 2015]{Bhattacharya_Bickel_AS15}
Bhattacharya, S. \& Bickel, P.~J. (2015).
\newblock Subsampling bootstrap of count features of networks.
\newblock {\em Annals of Statistics}, 43(6), 2384 -- 2411.

\bibitem[Bhattacharya \& Nain, 2011]{Bhattacharya_Nain_JFE11}
Bhattacharya, S. \& Nain, A. (2011).
\newblock Horizontal acquisitions and buying power: A product market analysis.
\newblock {\em Journal of Financial Economics}, 99(1), 97 -- 115.

\bibitem[Bickel et~al., 2013]{Bickel_et_al_AS13}
Bickel, P., Choi, D., Chang, X., \& Zhang, H. (2013).
\newblock Asymptotic normality of maximum likelihood and its variational
  approximation for stochastic blockmodels.
\newblock {\em Annals of Statistics}, 41(4), 1922 -- 1943.

\bibitem[Bickel \& Chen, 2009]{Bickel_Chen_PNAS09}
Bickel, P.~J. \& Chen, A. (2009).
\newblock A nonparametric view of network models and newman-girvan and other
  modularities.
\newblock {\em Proceedings of the National Academy of Sciences}, 106(50), 21068
  -- 21073.

\bibitem[Bickel et~al., 2011]{Bickel_et_al_AS11}
Bickel, P.~J., Chen, A., \& Levina, E. (2011).
\newblock The method of moments and degree distributions for network models.
\newblock {\em Annals of Statistics}, 39(5), 2280 -- 2301.

\bibitem[Blitzstein \& Diaconis, 2011]{Blitzstein_Diaconis_IM11}
Blitzstein, J. \& Diaconis, P. (2011).
\newblock A sequential importance sampling algorithm for generating random
  graphs with prescribed degrees.
\newblock {\em Internet Mathematics}, 6(4), 489 -- 522.

\bibitem[Bloch \& Jackson, 2006]{Bloch_Jackson_IJGT06}
Bloch, F. \& Jackson, M.~O. (2006).
\newblock Definitions of equilibrium in network formation games.
\newblock {\em International Journal of Game Theory}, 34(3), 305 -- 318.

\bibitem[Bloch \& Jackson, 2007]{Bloch_Jackson_JET07}
Bloch, F. \& Jackson, M.~O. (2007).
\newblock The formation of networks with transfers among players.
\newblock {\em Journal of Economic Theory}, 113(1), 83 -- 110.

\bibitem[Bloom et~al., 2013]{Bloom_et_al_EM13}
Bloom, N., Schankerman, M., \& van Reenen, J. (2013).
\newblock Identifying technology spillovers and product market rivalry.
\newblock {\em Econometrica}, 81(4), 1347 -- 1393.

\bibitem[Blume, 1993]{Blume_GEB93}
Blume, L. (1993).
\newblock The statistical mechanics of strategic interaction.
\newblock {\em Games and Economic Behavior}, 5, 387--424.

\bibitem[Blume et~al., 2011]{Blume_et_al_HSE2011}
Blume, L.~E., Brock, W.~A., Durlauf, S.~N., \& Ioannides, Y.~M. (2011).
\newblock Identification of social interactions.
\newblock In J. Benhabib, A. Bisin, \& M.~O. Jackson (Eds.), {\em Handbook of
  Social Economics}, volume~1B  (pp.\ 853 -- 964). Amsterdam: North-Holland.

\bibitem[Blume et~al., 2015]{Blume_et_al_JPE15}
Blume, L.~E., Brock, W.~A., Durlauf, S.~N., \& Jayaraman, R. (2015).
\newblock Linear social interaction models.
\newblock {\em Journal of Political Economy}, 123(2), 444 -- 496.

\bibitem[Blundell \& Powell, 2003]{Blundel_Powell_WC03}
Blundell, R. \& Powell, J.~L. (2003).
\newblock {\em Advances in Economics and Econometrics: Theory and Applications,
  Eighth World Congress}, volume~2, chapter Endogeneity in nonparametric and
  semiparametric regression models, (pp.\ 312 -- 357).
\newblock Cambridge University Press.

\bibitem[Bobollas et~al., 2007]{Bobollas_et_al_RSA07}
Bobollas, B., Janson, S., \& Riordan, O. (2007).
\newblock The phase transition in inhomogenous random graphs.
\newblock {\em Random Structures and Algorithms}, 31(1), 3 -- 122.

\bibitem[Bonhomme \& Manresa, 2015]{Bonhomme_Manresa_EM15}
Bonhomme, S. \& Manresa, E. (2015).
\newblock Grouped patterns of heterogeneity in panel data.
\newblock {\em Econometrica}, 83(3), 1147 -- 1184.

\bibitem[Bramoull{\'e} et~al., 2009]{Bramoulle_et_al_JOE09}
Bramoull{\'e}, Y., Djebbari, H., \& Fortin, B. (2009).
\newblock Identification of peer effects through social networks.
\newblock {\em Journal of Econometrics}, 150(1), 41 -- 55.

\bibitem[Brock \& Durlauf, 2001]{Brock_Durlauf_RES01}
Brock, W. \& Durlauf, S. (2001).
\newblock Discrete choice with social interactions.
\newblock {\em Review of Economic Studies}, 68, 235--261.

\bibitem[Brown \& Newey, 1998]{Brown_Newey_EM98}
Brown, B.~W. \& Newey, W.~K. (1998).
\newblock Efficient semiparametric estimation of expectations.
\newblock {\em Econometrica}, 66(2), 453 -- 464.

\bibitem[Callaert \& Veraverbeke, 1981]{Callaert_Veeraverbeke_AS81}
Callaert, H. \& Veraverbeke, N. (1981).
\newblock The order of the normal approximation for a studentized u-statistic.
\newblock {\em Annals of Statistics}, 9(1), 194 -- 200.

\bibitem[Cameron \& Golotvina, 2005]{Cameron_Golotvina_WP05}
Cameron, A.~C. \& Golotvina, N. (2005).
\newblock {\em Estimation of country-pair data models controlling for clustered
  errors: with international trade applications}.
\newblock Technical Report 06-13, University of California - Davis.

\bibitem[Cameron \& Miller, 2014]{Cameron_Miller_WP14}
Cameron, A.~C. \& Miller, D.~L. (2014).
\newblock {\em Robust inference for dyadic data}.
\newblock Technical report, University of California - Davis.

\bibitem[Cattaneo et~al., 2014]{Cattaneo_et_al_ET14}
Cattaneo, M., Crump, R., \& Jansson, M. (2014).
\newblock Small bandwidth asymptotics for density-weighted average derivatives.
\newblock {\em Econometric Theory}, 30(1), 176 -- 200.

\bibitem[Cattelan \& Varin, 2013]{Cattelan_Varin_ATAS13}
Cattelan, M. \& Varin, C. (2013).
\newblock {\em Advances in Theoretical and Applied Statistics}, chapter A model
  for correlated paired comparison data, (pp.\ 167 -- 176).
\newblock Springer-Verlag: Heidelberg.

\bibitem[Chamberlain, 1980]{Chamberlain_ReStud80}
Chamberlain, G. (1980).
\newblock Analysis of covariance with qualitative data.
\newblock {\em Review of Economic Studies}, 47(1), 225 -- 238.

\bibitem[Chamberlain, 1984]{Chamberlain_HBE84}
Chamberlain, G. (1984).
\newblock {\em Handbook of Econometrics}, volume~2, chapter Panel Data, (pp.\
  1247 -- 1318).
\newblock North-Holland: Amsterdam.

\bibitem[Chamberlain, 1985]{Chamberlain_LALMD85}
Chamberlain, G. (1985).
\newblock {\em Longitudinal Analysis of Labor Market Data}, chapter
  Heterogeneity, omitted variable bias, and duration dependence, (pp.\ 3 --
  38).
\newblock Cambridge University Press: Cambridge.

\bibitem[Chamberlain, 1992]{Chamberlain_EM92}
Chamberlain, G. (1992).
\newblock Efficiency bounds for semiparametric regression.
\newblock {\em Econometrica}, 60(3), 567 -- 596.

\bibitem[Chandrasekhar, 2015]{Chandrasekhar_Book15}
Chandrasekhar, A. (2015).
\newblock Econometrics of network formation.
\newblock In Y. Bramoull\'{e}, A. Galeotti, \& B. Rogers (Eds.), {\em Oxford
  Handbook on the Economics of Networks}. Oxford University Press.

\bibitem[Chatterjee, 2017]{Chatterjee_LNM17}
Chatterjee, S. (2017).
\newblock {\em Large Deviations for Random Graphs}.
\newblock Cham, Switzerland: Springer.

\bibitem[Chatterjee \& Dembo, 2016]{Chatterjee_Dembo_AM16}
Chatterjee, S. \& Dembo, A. (2016).
\newblock Nonlinear large deviations.
\newblock {\em Advances in Mathematics}, 299(20), 396 -- 450.

\bibitem[Chatterjee \& Diaconis, 2013]{Chatterjee_Diaconis_AS13}
Chatterjee, S. \& Diaconis, P. (2013).
\newblock Estimating and understanding exponential random graph models.
\newblock {\em Annals of Statistics}, 41(5), 2428 -- 2461.

\bibitem[Chatterjee et~al., 2011]{Chatterjee_et_al_AAP11}
Chatterjee, S., Diaconis, P., \& Sly, A. (2011).
\newblock Random graphs with a given degree sequence.
\newblock {\em Annals of Applied Probability}, 21(4), 1400 -- 1435.

\bibitem[Chib \& Greenberg, 1998]{Chib_Greenberg_BM98}
Chib, S. \& Greenberg, E. (1998).
\newblock Analysis of multivariate probit models.
\newblock {\em Biometrika}, 82(2), 347 -- 361.

\bibitem[Choi \& Wu, 2009]{Choi_Wu_JSCM09}
Choi, T.~Y. \& Wu, Z. (2009).
\newblock Triads in supply networks: theorizing buyer-supplier-supplier
  relationships.
\newblock {\em Journal of Supply Chain Management}, 45(1), 8 -- 25.

\bibitem[Choo \& Siow, 2006]{Choo_Siow_JPE06}
Choo, E. \& Siow, A. (2006).
\newblock Who marries whom and why?
\newblock {\em Journal of Political Economy}, 114(1), 175 -- 201.

\bibitem[Christakis et~al., 2010]{Christakis_et_al_NBER10}
Christakis, N.~A., Fowler, J.~H., Imbens, G.~W., \& Kalyanaraman, K. (2010).
\newblock {\em An empirical model for strategic network formation}.
\newblock NBER Working Paper 16039, National Bureau of Economic Research.

\bibitem[Chu \& Davis, 2016]{Chu_Davis_AJS16}
Chu, J. S.~G. \& Davis, G.~F. (2016).
\newblock Who killed the inner circle? the decline of the american corporate
  interlock network.
\newblock {\em American Journal of Sociology}, 122(3), 714 -- 754.

\bibitem[Coleman, 1988]{Coleman_AJS88}
Coleman, J.~S. (1988).
\newblock Social capital in the creation of human capital.
\newblock {\em American Journal of Sociology}, 94(S), S95 -- S120.

\bibitem[Conley \& Udry, 2010]{Conley_Udry_AER10}
Conley, T.~G. \& Udry, C.~R. (2010).
\newblock Learning about a new technology: pineapple in ghana.
\newblock {\em American Economic Review}, 100(1), 35 -- 69.

\bibitem[Cox \& Reid, 2004]{Cox_Reid_BM04}
Cox, D.~R. \& Reid, N. (2004).
\newblock A note on pseudolikelihood constructed from marginal densities.
\newblock {\em Biometrika}, 91(3), 729 -- 737.

\bibitem[Crane, 2018]{Crane_PFSNA18}
Crane, H. (2018).
\newblock {\em Probabilistic Foundations of Statistical Network Analysis}.
\newblock Boca Raton: CRC Press.

\bibitem[Crane \& Towsner, 2018]{Crane_Towsner_JSL18}
Crane, H. \& Towsner, H. (2018).
\newblock Relatively exchangeable structures.
\newblock {\em Journal of Symbolic Logic}, 83(2), 416 -- 442.

\bibitem[Currarini et~al., 2009]{Currarini_et_al_EM09}
Currarini, S., Jackson, M., \& Pin, P. (2009).
\newblock An economic model of friendship: homophily, minorities and
  segregation.
\newblock {\em Econometrica}, 77(4), 1003 -- 1045.

\bibitem[Daudin et~al., 2008]{Daudin_et_al_SC08}
Daudin, J.-J., Picard, F., \& Robin, S. (2008).
\newblock A mixture model for random graphs.
\newblock {\em Statistics and Computing}, 18(2), 173 -- 183.

\bibitem[Davezies et~al., 2019]{Davezies_et_al_arXiv2019}
Davezies, L., d'Haultfoeuille, X., \& Guyonvarch, Y. (2019).
\newblock {\em Empirical process results for exchangeable arrayes}.
\newblock Technical report, CREST-ENSAE.

\bibitem[Davis, 1991]{Davis_ASM91}
Davis, G.~F. (1991).
\newblock Agents without principles? the spread of the poison pill through the
  intercorporate network.
\newblock {\em Administrative Sciences Quarterly}, 36(4), 583 -- 613.

\bibitem[Davis, 1996]{Davis_CG96}
Davis, G.~F. (1996).
\newblock The significance of board interlocks for corporate governance.
\newblock {\em Corporate Governance: An International Review}, 4(3), 154 --
  159.

\bibitem[De~Benedictis et~al., 2014]{DeBenedictis_GE14}
De~Benedictis, L., Nenci, S., Santoni, G., Tajoli, L., \& Vicarelli, C. (2014).
\newblock Network analysis of world trade using the baci-cepii dataset.
\newblock {\em Global Economy Journal}, 14(3-4), 287--343.

\bibitem[de~Finetti, 1931]{deFinetti_AN1931}
de~Finetti, B. (1931).
\newblock Funzione caratteristica di un fenomeno aleatorio.
\newblock {\em Atti della R. Academia Nazionale dei Lincei, Serie 6. Memorie,
  Classe di Scienze Fisiche, Mathematice e Naturale}, 4, 251 -- 299.

\bibitem[de~Paula, 2013]{dePaula_ARE13}
de~Paula, A. (2013).
\newblock Econometric analysis of games with multiple equilibria.
\newblock {\em Annual Review of Economics}, 5, 107--131.

\bibitem[de~Paula, 2017]{dePaula_WC17}
de~Paula, {\'A}. (2017).
\newblock {\em Advances in Economics and Econometrics, Eleventh World
  Congress}, volume~1, chapter Econometrics of network models, (pp.\ 268 --
  323).
\newblock Cambridge University Press: Cambridge.

\bibitem[de~Paula et~al., 2018]{dePaula_et_al_EM18}
de~Paula, {\'A}., Richards-Shubik, S., \& Tamer, E. (2018).
\newblock Identifying preferences in networks with bounded degree.
\newblock {\em Econometrica}, 86(1), 263 -- 288.

\bibitem[De~Weerdt, 2004]{deWeerdt_IAP04}
De~Weerdt, J. (2004).
\newblock {\em Insurance Against Poverty}, chapter Risk-sharing and endogenous
  network formation, (pp.\ 197 -- 216).
\newblock Oxford University Press: Oxford.

\bibitem[Dhyne et~al., 2015]{Dhyne_et_al_WP15}
Dhyne, E., Magerman, G., \& Rub{\'\i}nova, S. (2015).
\newblock {\em The Belgian production network 2002-2012}.
\newblock Working Paper 288, National Bank of Belgium.

\bibitem[Diaconis, 1977]{Diaconis_Syn77}
Diaconis, P. (1977).
\newblock Finite forms of de finetti's theorem on exchangeability.
\newblock {\em Synthese}, 36(2), 271 -- 281.

\bibitem[Diaconis \& Freedman, 1980]{Diaconis_Freedman_AP80}
Diaconis, P. \& Freedman, D. (1980).
\newblock Finite exchangeable sequences.
\newblock {\em Annals of Probability}, 8(4), 745 -- 764.

\bibitem[Diaconis et~al., 2008]{Diaconis_Holmes_Janson_IM08}
Diaconis, P., Holmes, S., \& Janson, S. (2008).
\newblock Threshold graph limits and random threshold graphs.
\newblock {\em Internet Mathematics}, 5(3), 267 -- 320.

\bibitem[Diaconis \& Janson, 2008]{Diaconis_Janson_RM08}
Diaconis, P. \& Janson, S. (2008).
\newblock Graph limits and exchangeable random graphs.
\newblock {\em Rendiconti di Matematica}, 28(1), 33 -- 61.

\bibitem[Dooley, 1969]{Dooley_AER69}
Dooley, P.~C. (1969).
\newblock The interlocking directorate.
\newblock {\em American Economic Review}, 59(3), 314 -- 323.

\bibitem[Ductor et~al., 2014]{Ductor_et_al_RESTAT14}
Ductor, L., Fafchamps, M., \& Goyal, S. (2014).
\newblock Social networks and research output.
\newblock {\em Review of Economics and Statistics}, 96(5), 936 -- 948.

\bibitem[Dzemski, 2018]{Dzemski_RESTAT18}
Dzemski, A. (2018).
\newblock An empirical model of dyadic link formation in a network with
  unobserved heterogeneity.
\newblock {\em Review of Economics and Statistics}.
\newblock University of Mannheim.

\bibitem[Efron \& Stein, 1981]{Efron_Stein_AS81}
Efron, B. \& Stein, C. (1981).
\newblock The jackknife estimate of variance.
\newblock {\em Annals of Statistics}, 9(3), 586 -- 596.

\bibitem[Erikson et~al., 2014]{Erikson_Pinto_Rader_PA14}
Erikson, R.~S., Pinto, P.~M., \& Rader, K.~T. (2014).
\newblock Dyadic analysis in international relations: a cautionary tale.
\newblock {\em Political Analysis}, 22(4), 457 -- 463.

\bibitem[Fafchamps \& Gubert, 2007]{Fafchamp_Gubert_JDE07}
Fafchamps, M. \& Gubert, F. (2007).
\newblock The formation of risk sharing networks.
\newblock {\em Journal of Development Economics}, 83(2), 326 -- 350.

\bibitem[Fafchamps \& Lund, 2003]{Fafchamps_Lund_JDE03}
Fafchamps, M. \& Lund, S. (2003).
\newblock Risk sharing networks in rural philippines.
\newblock {\em Journal of Development Economics}, 71(2), 261 -- 287.

\bibitem[Fafchamps \& Minten, 2002]{Fafchamps_Minten_OEP02}
Fafchamps, M. \& Minten, B. (2002).
\newblock Returns to social network capital among traders.
\newblock {\em Oxford Economic Papers}, 54(2), 173 -- 206.

\bibitem[Fee \& Thomas, 2004]{Fee_Thomas_JFE04}
Fee, C.~E. \& Thomas, S. (2004).
\newblock Sources of gains in horizontal mergers: evidence from customer,
  supplier, and rival firms.
\newblock {\em Journal of Financial Economics}, 74(3), 423 -- 460.

\bibitem[Fern{\'a}ndez-Val \& Weidner, 2016]{FernandezVal_Weidner_JOE16}
Fern{\'a}ndez-Val, I. \& Weidner, M. (2016).
\newblock Individual and time effects in nonlinear panel data models with large
  $n$, $t$.
\newblock {\em Journal of Econometrics}, 192(1), 291 -- 312.

\bibitem[Frank, 1979]{Frank_ANYAS79}
Frank, O. (1979).
\newblock Moment properties of subgraph counts in stochastic graphs.
\newblock {\em Annals of the New York Academy of Sciences}, 319(1), 207 -- 218.

\bibitem[Frank, 1980]{Frank_JMS80}
Frank, O. (1980).
\newblock Transitivity in stochastic graphs and digraphs.
\newblock {\em Journal of Mathematical Sociology}, 7(2), 199 -- 213.

\bibitem[Frank, 1988]{Frank_DM88}
Frank, O. (1988).
\newblock Triad count statistics.
\newblock {\em Discrete Mathematics}, 72(1-3), 141 -- 149.

\bibitem[Frank, 1997]{Frank_MSH97}
Frank, O. (1997).
\newblock Composition and structure of social networks.
\newblock {\em Mathematiques, Informatique, et Sciences Humaines}, 137(1), 11
  -- 24.

\bibitem[Frank \& Snijders, 1994]{Frank_Snijder_JOS94}
Frank, O. \& Snijders, T. A.~B. (1994).
\newblock Estimating the size of hidden populations using snowball sampling.
\newblock {\em Journal of Official Statistics}, 10(1), 53 -- 67.

\bibitem[Galichon \& Salanie, 2017]{Galichon_Salanie_AERPP17}
Galichon, A. \& Salanie, B. (2017).
\newblock The econometrics and some properties of separable matching models.
\newblock {\em American Economic Review}, 107(5), 251 -- 55.

\bibitem[Gao \& Lafferty, 2017]{Gao_Lafferty_arXiv17}
Gao, C. \& Lafferty, J. (2017).
\newblock {\em Testing network structuring using relations between small
  subgraph probabilities}.
\newblock arxiv 1704.06742, arXiv.

\bibitem[Gao et~al., 2015]{Gao_et_al_AS15}
Gao, C., Lu, Y., \& Zhou, H.~H. (2015).
\newblock Rate-optimal graphon estimation.
\newblock {\em Annals of Statistics}, 43(6), 2624 -- 2652.

\bibitem[Gaulier \& Zignago, 2010]{Gaulier_et_al_CEPII10}
Gaulier, G. \& Zignago, S. (2010).
\newblock {\em BACI: International trade database at the product-level}.
\newblock Working Paper 2010-23, CEPII.

\bibitem[Gofman, 2017]{Gofman_JFE17}
Gofman, M. (2017).
\newblock Efficiency and stability of a financial architecture with
  too-interconnected-to-fail institutions.
\newblock {\em Journal of Financial Economics}, 124(1), 113 -- 146.

\bibitem[Goldenberg et~al., 2009]{Goldenberg_etal_FTML09}
Goldenberg, A., Zheng, A., Fienberg, S.~E., \& Airoldi, E.~M. (2009).
\newblock A survey of statistical network models.
\newblock {\em Foundations and Trends in Machine Learning}, 2(2), 129--333.

\bibitem[Goldsmith-Pinkham \& Imbens, 2013]{Goldsmith-Pinkham_Imbens_JBES13}
Goldsmith-Pinkham, P. \& Imbens, G.~W. (2013).
\newblock Social networks and the identification of peer effects.
\newblock {\em Journal of Business and Economic Statistics}, 31(3), 253 -- 264.

\bibitem[Goodman, 2017]{Goodman_JPAM17}
Goodman, L. (2017).
\newblock The effect of the affordable care act medicaid expansion on
  migration.
\newblock {\em Journal of Policy Analysis and Management}, 36(1), 211 -- 238.

\bibitem[Gould \& Fernandez, 1989]{Gould_Fernandez_SM89}
Gould, R.~V. \& Fernandez, R.~M. (1989).
\newblock Structures of mediation: a formal approach to brokerage in
  transaction networks.
\newblock {\em Sociological Methodology}, 19, 89 -- 126.

\bibitem[Gourieroux et~al., 1993]{Gourieroux_et_al_JAE93}
Gourieroux, C., Monfort, A., \& Renault, E. (1993).
\newblock Indirect inference.
\newblock {\em Journal of Applied Econometrics}, 8(S), S85 -- S118.

\bibitem[Graham, 2008]{Graham_EM08}
Graham, B. (2008).
\newblock Identifying social interactions through conditional variance
  restrictions.
\newblock {\em Econometrica}, 76(3), 643--660.

\bibitem[Graham, 2011]{Graham_EM11}
Graham, B.~S. (2011).
\newblock Efficiency bounds for missing data models with semiparametric
  restrictions.
\newblock {\em Econometrica}, 79(2), 437 -- 452.

\bibitem[Graham, 2013]{Graham_AinE13}
Graham, B.~S. (2013).
\newblock Comparative static and computational methods for an empirical
  one-to-one transferable utility matching model.
\newblock {\em Advances in Econometrics}, 31, 153 -- 181.

\bibitem[Graham, 2015]{Graham_AR15}
Graham, B.~S. (2015).
\newblock Methods of identification in social networks.
\newblock {\em Annual Review of Economics}, 7(1), 465--485.

\bibitem[Graham, 2016]{Graham_NBER16}
Graham, B.~S. (2016).
\newblock {\em Homophily and transitivity in dynamic network formation}.
\newblock NBER Working Paper 22186, National Bureau of Economic Research.

\bibitem[Graham, 2017]{Graham_EM17}
Graham, B.~S. (2017).
\newblock An econometric model of network formation with degree heterogeneity.
\newblock {\em Econometrica}, 85(4), 1033 -- 1063.

\bibitem[Graham, 2018a]{Graham_DyadLectureNotes18}
Graham, B.~S. (2018a).
\newblock Dyadic regression.
\newblock Lecture Notes.

\bibitem[Graham, 2018b]{Graham_JEL18}
Graham, B.~S. (2018b).
\newblock Identifying and estimating neighborhood effects.
\newblock {\em Journal of Economic Literature}, 56(2), 450 -- 500.

\bibitem[Graham \& de~Paula, 2020]{Graham_dePaula_Bk20}
Graham, B.~S. \& de~Paula, A., Eds. (2020).
\newblock {\em The Econometric Analysis of Network Data}.
\newblock Academic Press.

\bibitem[Graham et~al., 2010]{Graham_Imbens_Ridder_NBER10}
Graham, B.~S., Imbens, G.~W., \& Ridder, G. (2010).
\newblock {\em Measuring the effects of segregation in the presence of social
  spillovers: a nonparametric approach}.
\newblock Working Paper 16499, NBER.

\bibitem[Graham et~al., 2014]{Graham_Imbens_Ridder_QE14}
Graham, B.~S., Imbens, G.~W., \& Ridder, G. (2014).
\newblock Complementarity and aggregate implications of assortative matching: a
  nonparametric analysis.
\newblock {\em Quantitative Economics}, 5(1), 29 -- 66.

\bibitem[Graham et~al., 2018]{Graham_Imbens_Ridder_JBES18}
Graham, B.~S., Imbens, G.~W., \& Ridder, G. (2018).
\newblock Identification and efficiency bounds for the average match function
  under conditionally exogenous matching.
\newblock {\em Journal of Business and Economic Statistics}.

\bibitem[Graham et~al., 2019]{Graham_Niu_Powell_WP2019}
Graham, B.~S., Niu, F., \& Powell, J.~L. (2019).
\newblock {\em Kernel density estimation for undirected dyadic data}.
\newblock Technical report, University of California - Berkeley.

\bibitem[Graham \& Pelican, 2020]{Graham_Pelican_BookCh2020}
Graham, B.~S. \& Pelican, A. (2020).
\newblock {\em The Econometrics of Social and Economic Networks}, chapter
  Testing for externalities in network formation using simulation'',.
\newblock Elsevier: Amsterdam.

\bibitem[Granovetter, 1973]{Granovetter_AJS73}
Granovetter, M.~S. (1973).
\newblock The strength of weak ties.
\newblock {\em American Journal of Sociology}, 78(6), 1360 -- 1380.

\bibitem[Granovetter, 1985]{Granovetter_AJS85}
Granovetter, M.~S. (1985).
\newblock Economic action and social structure: the problem of embeddedness.
\newblock {\em American Journal of Sociology}, 91(3), 481 -- 510.

\bibitem[Green \& Shalizi, 2017]{Green_Shalizi_arXiv17}
Green, A. \& Shalizi, C.~R. (2017).
\newblock {\em Bootstrapping exchangeable random graphs}.
\newblock Technical Report 1711.00813v1, arXiv.

\bibitem[Gualdani, ming]{Gualdani_JOE19}
Gualdani, C. (forthcoming).
\newblock An econometric model of network formation with an application to
  board interlocks between firms.
\newblock {\em Journal of Econometrics}.

\bibitem[Hagedoorn, 2002]{Hagedoorn_RP02}
Hagedoorn, J. (2002).
\newblock Inter-firm r \& d partnerships: an overview of major trends and
  patterns since 1960.
\newblock {\em Research Policy}, 31(4), 477 -- 492.

\bibitem[Hahn \& Newey, 2004]{Hahn_Newey_EM04}
Hahn, J. \& Newey, W.~K. (2004).
\newblock Jackknife and analytical bias reduction for nonlinear panel data
  models.
\newblock {\em Econometrica}, 72(4), 1295 -- 1319.

\bibitem[Hausman \& Taylor, 1981]{Hausman_Taylor_EM81}
Hausman, J. \& Taylor, W. (1981).
\newblock Panel data and unobservable individual effects.
\newblock {\em Journal of Econometrics}, 49(6), 1377 -- 1398.

\bibitem[He \& Zheng, 2013]{He_Zheng_IEEE13}
He, R. \& Zheng, T. (2013).
\newblock Estimation of exponential random graph models for large social
  networks via graph limits.
\newblock In {\em Advances in Social Networks Analysis and Mining (ASONAM)}.

\bibitem[Heckman et~al., 1997]{Heckman_Smith_Clements_ReStud97}
Heckman, J.~J., Smith, J., \& Clements, N. (1997).
\newblock Making the most out of programme evaluations and social experiments:
  accounting for heterogeneity in programme impacts.
\newblock {\em Review of Economic Studies}, 64(4), 487 -- 535.

\bibitem[Heckman \& Vytlacil, 2007]{Heckman_Vytlacil_HBE07}
Heckman, J.~J. \& Vytlacil, E.~J. (2007).
\newblock {\em Handbook of Econometrics}, volume~6B, chapter Econometric
  evaluation of social programs, part I: causal models, structural models and
  econometric policy evaluation, (pp.\ 4779 -- 4874).
\newblock North-Holland: Amsterdam.

\bibitem[Hellmann, 2013]{Hellmann_IJGT13}
Hellmann, T. (2013).
\newblock On the existence and uniqueness of pairwise stable networks.
\newblock {\em International Journal of Game Theory}, 42(1), 211 -- 237.

\bibitem[Helpman et~al., 2008]{Helpman_et_al_QJE08}
Helpman, E., Melitz, M., \& Rubinstein, Y. (2008).
\newblock Estimating trade flows: trading partners and trading volumes.
\newblock {\em Quarterly Journal of Economics}, 123(2), 441 -- 487.

\bibitem[Hensvik \& Skans, 2016]{Hensvik_Skans_JOLE16}
Hensvik, L. \& Skans, O.~N. (2016).
\newblock Social networks, employee selection, and labor market outcomes.
\newblock {\em Journal of Labor Economics}, 34(4), 825 -- 867.

\bibitem[Hilgerdt, 1943]{Hilgerdt_AER43}
Hilgerdt, F. (1943).
\newblock The case for multilateral trade.
\newblock {\em American Economic Review}, 33(1 (Part 2)), 393 -- 407.

\bibitem[Ho, 2009]{Ho_AER09}
Ho, K. (2009).
\newblock Insurer-provider networks in the medical care market.
\newblock {\em American Economic Review}, 99(1), 393 -- 430.

\bibitem[Hoeffding, 1948]{Hoeffding_AMS48}
Hoeffding, W. (1948).
\newblock A class of statistics with asymptotically normal distribution.
\newblock {\em Annals of Mathematical Statistics}, 19(3), 293 -- 325.

\bibitem[Hoff et~al., 2002]{Hoff_etal_JASA2002}
Hoff, P., Raftery, A., \& Handcock, M. (2002).
\newblock Latent space approaches to social network analysis.
\newblock {\em Journal of the American Statistical Association}, 97,
  1090--1098.

\bibitem[Holland \& Leinhardt, 1970]{Holland_Leinhardt_AJS70}
Holland, P.~W. \& Leinhardt, S. (1970).
\newblock A method for detecting structure in sociometric data.
\newblock {\em American Journal of Sociology}, 76(3), 492 -- 513.

\bibitem[Holland \& Leinhardt, 1976]{Holland_Leinhardt_SM76}
Holland, P.~W. \& Leinhardt, S. (1976).
\newblock Local structure in social networks.
\newblock {\em Sociological Methodology}, 7, 1 -- 45.

\bibitem[Holland \& Leinhardt, 1981]{Holland_Leinhardt_JASA81}
Holland, P.~W. \& Leinhardt, S. (1981).
\newblock An exponential family of probability distributions for directed
  graphs.
\newblock {\em Journal of the American Statistical Association}, 76(373), 33 --
  50.

\bibitem[Honor{\'e} \& Powell, 1994]{Honore_Powell_JOE94}
Honor{\'e}, B.~E. \& Powell, J.~L. (1994).
\newblock Pairwise difference estimators of censored and truncated regression
  models.
\newblock {\em Journal of Econometrics}, 64(1-2), 241 -- 278.

\bibitem[Hoover, 1979]{Hoover_WP79}
Hoover, D.~N. (1979).
\newblock {\em Relations on probability spaces and arrays of random variables}.
\newblock Technical report, Institute for Advanced Study, Princeton, NJ.

\bibitem[Imbens \& Rubin, 2015]{Imbens_Rubin_CIBook15}
Imbens, G.~W. \& Rubin, D.~B. (2015).
\newblock {\em Causal Inference for Statistics, Social, and Biomedical
  Sciences: An Introduction}.
\newblock Cambridge: Cambridge University Press.

\bibitem[Imbens \& Wooldridge, 2009]{Imbens_Wooldridge_JEL09}
Imbens, G.~W. \& Wooldridge, J.~M. (2009).
\newblock Recent developments in the econometrics of program evaluation.
\newblock {\em Journal of Economic Literature}, 47(1), 5 -- 86.

\bibitem[Ioannides \& Loury, 2004]{Loury_Ioannides_JEL04}
Ioannides, Yannis, M. \& Loury, L.~D. (2004).
\newblock Job information networks, neighborhood effects, and inequality.
\newblock {\em Journal of Economic Literature}, 42(4), 1056 -- 1093.

\bibitem[Isakov et~al., 2019]{Isakov_et_al_SS19}
Isakov, A., Fowler, J.~H., Airoldi, E.~M., \& Christakis, N.~A. (2019).
\newblock The structure of negative social ties in rural village networks.
\newblock {\em Sociological Science}, 6(8), 197 -- 217.

\bibitem[Jackson \& Watts, 2001]{Jackson_Watts_SJE01}
Jackson, M. \& Watts, A. (2001).
\newblock The existence of pairwise stable networks.
\newblock {\em Seoul Journal of Economics}, 14(3), 299 -- 321.

\bibitem[Jackson, 2008]{Jackson_NetBook08}
Jackson, M.~O. (2008).
\newblock {\em Social and Economic Networks}.
\newblock Princeton: Princeton University Press.

\bibitem[Jackson et~al., 2012]{Jackson_et_al_AER12}
Jackson, M.~O., Rodriguez-Barraquer, T., \& Tan, X. (2012).
\newblock Social capital and social quilts: network patterns of favor exchange.
\newblock {\em American Economic Review}, 102(5), 1857--1897.

\bibitem[Jackson \& Rogers, 2007]{Jackson_Rogers_BE07}
Jackson, M.~O. \& Rogers, B.~W. (2007).
\newblock Relating network structure to diffusion properties through stochastic
  dominance.
\newblock {\em B.E. Journal of Theoretical Economics}, 7(1), (Advances) Article
  6.

\bibitem[Jackson et~al., 2017]{Jackson_et_al_JEL17}
Jackson, M.~O., Rogers, B.~W., \& Zenou, Y. (2017).
\newblock The economic consequences of social-network structure.
\newblock {\em Journal of Economic Literature}, 55(1), 49 -- 95.

\bibitem[Jackson \& Wolinsky, 1996]{Jackson_Wolinsky_JET96}
Jackson, M.~O. \& Wolinsky, A. (1996).
\newblock A strategic model of social and economic networks.
\newblock {\em Journal of Economic Theory}, 71(1), 44 -- 74.

\bibitem[Jackson \& Yariv, 2011]{Jackson_Yariv_HSE11}
Jackson, M.~O. \& Yariv, L. (2011).
\newblock {\em Handbook of Social Economics}, volume~1A, chapter Diffusion,
  strategic interaction, and social structure, (pp.\ 645 -- 678).
\newblock North-Holland: Amsterdam.

\bibitem[Jackson \& Zenou, 2015]{Jackson_Zenou_HBGT15}
Jackson, M.~O. \& Zenou, Y. (2015).
\newblock {\em Handbook of Game Theory}, volume~4, chapter Games on networks,
  (pp.\ 95 -- 163).
\newblock North-Holland: Amsterdam.

\bibitem[Jaffe, 1986]{Jaffe_AER86}
Jaffe, A. (1986).
\newblock Technological opportunity and spillovers of r{\&}d: evidence from
  firms' patents, profits, and market value.
\newblock {\em American Economic Review}, 76(5), 984 -- 1001.

\bibitem[Janson et~al., 2000]{Janson_et_al_RG00}
Janson, S., Luczak, T., \& Rucinski, A. (2000).
\newblock {\em Random Graphs}.
\newblock New York: John Wiley \& Sons, Inc.

\bibitem[Janson \& Nowicki, 1991]{Janson_Nowicki_PTRF91}
Janson, S. \& Nowicki, K. (1991).
\newblock The asymptotic distributions of generalized u-statistics with
  applications to random graphs.
\newblock {\em Probability Theory and Related Fields}, 90(3), 341 -- 375.

\bibitem[Janssen, 1994]{Janssen_JSPI94}
Janssen, P. (1994).
\newblock Weighted bootstrapping of u-statistics.
\newblock {\em Journal of Statistical Planning and Inference}, 38(1), 31--41.

\bibitem[Jia, 2008]{Jia_EM08}
Jia, P. (2008).
\newblock What happens when wal-mart comes to town: an empirical analysis of
  the discount retailing industry.
\newblock {\em Econometrica}, 76(6), 1263 -- 1316.

\bibitem[Jochmans, 2018]{Jochmans_JBES18}
Jochmans, K. (2018).
\newblock Semiparametric analysis of network formation.
\newblock {\em Journal of Business and Economic Statistics}, 36(4), 705 -- 713.

\bibitem[Johnsson \& Moon, 2017]{Johnson_Moon_INET17}
Johnsson, I. \& Moon, H.~R. (2017).
\newblock {\em Estimation of peer effects in endogenous social networks:
  control function approach}.
\newblock USC-INET Research Paper No. 17-25, University of Souther California.

\bibitem[Kallenberg, 2005]{Kallenberg_PSIP05}
Kallenberg, O. (2005).
\newblock {\em Probabilistic Symmetries and Invariance Principles}.
\newblock New York: Springer-Verlag.

\bibitem[Keane, 1994]{Keane_EM94}
Keane, M.~P. (1994).
\newblock A computationally practical simulation estimator for panel data.
\newblock {\em Econometrica}, 62(1), 95 -- 116.

\bibitem[Kim et~al., 2015]{Kim_et_al_Lancet2015}
Kim, D.~A., Hwong, A.~R., Stafford, D., Hughes, D.~A., O'Malley, A.~J., Fowler,
  J.~H., \& Christakis, N.~A. (2015).
\newblock Social network targeting to maximise population behaviour change: a
  cluster randomised controlled trial.
\newblock {\em Lancet}, 386(9989), 145 -- 153.

\bibitem[Kolaczyk, 2009]{Kolaczyk_NetBook09}
Kolaczyk, E.~D. (2009).
\newblock {\em Statistical Analysis of Network Data}.
\newblock New York: Springer.

\bibitem[K{\"o}nig et~al., 2019]{Konig_et_al_RESTAT19}
K{\"o}nig, M.~D., Liu, X., \& Zenou, Y. (2019).
\newblock R{\&}d networks: theory, empirics and policy implications.
\newblock {\em Review of Economics and Statistics}, 101(3), 476 -- 491.

\bibitem[Kranton \& Minehart, 2001]{Kranton_Minehart_AER01}
Kranton, R. \& Minehart, D. (2001).
\newblock A theory of buyer-seller networks.
\newblock {\em American Economic Review}, 91(3), 485--508.

\bibitem[Krivitsky et~al., 2009]{Krivitsky_et_al_SN09}
Krivitsky, P.~N., Handcock, M.~S., Raftery, A.~E., \& Hoff, P.~D. (2009).
\newblock Representing degree distributions, clustering, and homophily in
  social networks with latent cluster random effects models.
\newblock {\em Social Networks}, 31(2), 204 -- 213.

\bibitem[Kuersteiner, 2019]{Kuersteiner_arXiv19}
Kuersteiner, G.~M. (2019).
\newblock {\em Limit theorems for data with network structure}.
\newblock arXiv arXiv:1908.02375, University of Maryland.

\bibitem[Lehmann, 1999]{Lehmann_LSTBook99}
Lehmann, E.~L. (1999).
\newblock {\em Elements of Large-Sample Theory}.
\newblock New York: Springer.

\bibitem[Leung, 2015]{Leung_JOE15}
Leung, M. (2015).
\newblock Two-step estimation of network-formation models with incomplete
  information.
\newblock {\em Journal of Econometrics}, 188(1), 182 -- 195.

\bibitem[Leung, 2019]{Leung_JOE19}
Leung, M. (2019).
\newblock A weak law for moments of pairwise stable networks.
\newblock {\em Journal of Econometrics}, 210(2), 310 -- 326.

\bibitem[Leung \& Moon, 2019]{Leung_Moon_arXiv19}
Leung, M. \& Moon, H.~R. (2019).
\newblock {\em Normal approximation in large network models}.
\newblock arXiv arXiv:1904.11060, University of Southern California.

\bibitem[Lindsey, 1988]{Lindsay_CM88}
Lindsey, B.~G. (1988).
\newblock Composite likelihood.
\newblock {\em Contemporary Mathematics}, 80, 221 -- 239.

\bibitem[Loury, 2002]{Loury_ARI02}
Loury, G.~C. (2002).
\newblock {\em The Anatomy of Racial Inequality}.
\newblock Cambridge, MA: Harvard University Press.

\bibitem[Lov{\'a}sz, 2012]{Lovasz_AMS12}
Lov{\'a}sz, L. (2012).
\newblock {\em Large Networks and Graph Limits}, volume~60 of {\em American
  Mathematical Society Colloquium Publications}.
\newblock American Mathematical Society.

\bibitem[Lov{\'a}sz \& Szegedy, 2006]{Lovasz_Szegedy_JCT06}
Lov{\'a}sz, L. \& Szegedy, B. (2006).
\newblock Limits of dense graph sequences.
\newblock {\em Journal of Combinatorial Theory, Series B}, 96, 933--957.

\bibitem[Manski, 1993]{Manski_ReStud93}
Manski, C.~F. (1993).
\newblock Identification of endogenous social effects: the reflection problem,.
\newblock {\em Review of Economic Studies}, 60(3), 531 -- 542.

\bibitem[Mao, 2018]{Mao_BM18}
Mao, L. (2018).
\newblock On causal estimation using u-statistics.
\newblock {\em Biometrika}, 105(1), 215 -- 220.

\bibitem[Marsden, 1987]{Marsden_ASR87}
Marsden, P.~V. (1987).
\newblock Core discussion networks of americans.
\newblock {\em American Sociological Review}, 52(1), 122 -- 131.

\bibitem[Mayda, 2010]{Mayda_JPopE10}
Mayda, A.~M. (2010).
\newblock International migration: a panel data analysis of the determinants of
  bilateral flows.
\newblock {\em Journal of Population Economics}, 23(4), 1249 -- 1274.

\bibitem[McFadden, 1974]{McFadden_FinE74}
McFadden, D. (1974).
\newblock {\em Frontiers in Econometrics}, chapter Conditional logit analysis
  of qualitative choice behavior, (pp.\ 105 -- 142).
\newblock Academic Press: New York.

\bibitem[McFadden, 1989]{McFadden_EM89}
McFadden, D. (1989).
\newblock A method of simulated moments for estimation of discrete response
  models without numerical integration.
\newblock {\em Econometrica}, 67(5), 995 -- 1026.

\bibitem[McPherson et~al., 2006]{McPherson_et_al_ASR06}
McPherson, M., Smith-Lovin, L., \& Brashears, M.~E. (2006).
\newblock Social isolation in america: changes in core discussion networks over
  two decades.
\newblock {\em American Sociological Review}, 71(3), 353 -- 375.

\bibitem[Mele, 2017]{Mele_EM17}
Mele, A. (2017).
\newblock A structural model of dense network formation.
\newblock {\em Econometrica}, 85(3), 825 -- 850.
\newblock John Hopkins University.

\bibitem[Mele \& Zhu, 2017]{Mele_Zhu_WP17}
Mele, A. \& Zhu, L. (2017).
\newblock {\em Approximate variational estimation for a model of network
  formation}.
\newblock Technical report, John Hopkins University.

\bibitem[Menzel, 2016]{Menzel_WP16}
Menzel, K. (2016).
\newblock Strategic network formation with many agents.
\newblock New York University.

\bibitem[Menzel, 2017]{Menzel_arXiv17}
Menzel, K. (2017).
\newblock {\em Bootstrap with clustering in two or more dimensions}.
\newblock Technical Report 1703.03043v2, arXiv.

\bibitem[Milo et~al., 2002]{Milo_el_al_Sci02}
Milo, R., Shen-Orr, S., Itzkovitz, S., Kashtan, N., Chklovskii, D., \& Alon, U.
  (2002).
\newblock Network motifs: simple building blocks of complex networks.
\newblock {\em Science}, 298(5594), 824 -- 827.

\bibitem[Min, 2019]{Min_WP2019}
Min, S. (2019).
\newblock {\em Network of loyalty programs: a sequential formation}.
\newblock Technical report, University of California - Berkeley.

\bibitem[Miyauchi, 2016]{Miyauchi_JOE16}
Miyauchi, Y. (2016).
\newblock Structural estimation of a pairwise stable network with nonnegative
  externality.
\newblock {\em Journal of Econometrics}, 195(2), 224 -- 235.

\bibitem[Monderer \& Shapley, 1996]{Monderer_Shapley_GEB96}
Monderer, D. \& Shapley, L.~S. (1996).
\newblock Potential games.
\newblock {\em Games and Economic Behavior}, 14(1), 124 -- 143.

\bibitem[Moreno, 1934]{Moreno_Book34}
Moreno, J.~L. (1934).
\newblock {\em Who Shall Survive?}
\newblock New York: Beacon House.

\bibitem[Munski, 2003]{Munshi_QJE03}
Munski, K. (2003).
\newblock Networks in the modern economy: Mexican migrants in the us labor
  market.
\newblock {\em Quarterly Journal of Economics}, 118(2), 549 -- 599.

\bibitem[Nadler, 2015]{Nadler_WP15}
Nadler, C. (2015).
\newblock Networked inequality: evidence from freelancers.
\newblock University of California - Berkeley.

\bibitem[Newey \& McFadden, 1994]{Newey_McFadden_HBE94}
Newey, W.~K. \& McFadden, D. (1994).
\newblock {\em Handbook of Econometrics}, volume~4, chapter Large sample
  estimation and hypothesis testing, (pp.\ 2111 -- 2245).
\newblock North-Holland: Amsterdam.

\bibitem[Newman, 2001]{Newman_PNAS01}
Newman, M. E.~J. (2001).
\newblock The structure of scientific collaboration networks.
\newblock {\em Proceedings of the National Academy of Sciences}, 98(2), 404 --
  409.

\bibitem[Newman, 2010]{Newman_NetBook10}
Newman, M. E.~J. (2010).
\newblock {\em Networks: An Introduction}.
\newblock Oxford: Oxford University Press.

\bibitem[Nowicki, 1991]{Nowicki_SN91}
Nowicki, K. (1991).
\newblock Asymptotic distributions in random graphs with applications to social
  networks.
\newblock {\em Statistica Neerlandica}, 45(3), 295 -- 325.

\bibitem[Nowicki \& Wierman, 1988]{Nowicki_Wierman_DM88}
Nowicki, K. \& Wierman, J.~C. (1988).
\newblock Subgraph counts in random graphs using incomplete u-statistics
  methods.
\newblock {\em Discrete Mathematics}, 72(1-3), 299 -- 310.

\bibitem[Olhede \& Wolfe, 2014]{Olhede_Wolfe_PNAS14}
Olhede, S.~C. \& Wolfe, P.~J. (2014).
\newblock Network histograms and universality of blockmodel approximation.
\newblock {\em Proceedings of the National Academy of Sciences}, 11(41), 14722
  -- 14727.

\bibitem[Orbanz \& Roy, 2015]{Orbanz_Roy_IEEE15}
Orbanz, P. \& Roy, D.~M. (2015).
\newblock Bayesian models of graphs, arrays and other exchangeable randoms
  structures.
\newblock {\em IEEE Transactions on Pattern Analysis and Machine Intelligence},
  37(2), 437 -- 461.

\bibitem[Ortega \& Peri, 2013]{Ortega_Peri_MS13}
Ortega, F. \& Peri, G. (2013).
\newblock The effect of income and immigration policies on international
  migration.
\newblock {\em Migration Studies}, 1(1), 47 -- 74.

\bibitem[Ostrovsky, 2008]{Ostrovsky_AER08}
Ostrovsky, M. (2008).
\newblock Stability in supply chain networks.
\newblock {\em American Economic Review}, 98(3), 897 -- 923.

\bibitem[Owen, 2007]{Owen_AAS07}
Owen, A.~B. (2007).
\newblock The pigeonhole bootstrap.
\newblock {\em Annals of Applied Statistics}, 1(2), 386 -- 411.

\bibitem[Pakes \& Pollard, 1989]{Pakes_Pollard_EM89}
Pakes, A. \& Pollard, D. (1989).
\newblock Simulation and the asymptotics of optimization estimators.
\newblock {\em Econometrica}, 57(5), 1027 -- 1057.

\bibitem[Pelican \& Graham, 2019]{Pelican_Graham_WP2019}
Pelican, A. \& Graham, B.~S. (2019).
\newblock {\em Testing for strategic interaction in social and economic network
  formation}.
\newblock Technical report, University of California - Berkeley.

\bibitem[Penrose, 2003]{Penrose_Bk2003}
Penrose, M. (2003).
\newblock {\em Random Geometric Graphs}.
\newblock Oxford: Oxford University Press.

\bibitem[Picard et~al., 2008]{Picard_et_al_JCB08}
Picard, F., Daudin, J.~J., Koskas, M., Schbath, S., \& Robin, S. (2008).
\newblock Assessing the exceptionality of network motifs.
\newblock {\em Journal of Computational Biology}, 15(1), 1 -- 20.

\bibitem[Pr{\v z}ulj et~al., 2004]{Przulj_et_al_BI04}
Pr{\v z}ulj, N., Corneil, D.~G., \& Jurisica, I. (2004).
\newblock Modeling interactome: scale-free or geometric?
\newblock {\em Bioinformatics}, 20(18), 3508--3515.

\bibitem[Qu \& Lee, 2015]{Qu_Lee_JOE15}
Qu, X. \& Lee, L.-F. (2015).
\newblock Estimating a spatial autoregressive model with an endogenous spatial
  weight matrix.
\newblock {\em Journal of Econometrics}, 184(2), 209 -- 232.

\bibitem[Ridder \& Sheng, 2017]{Ridder_Sheng_WP17}
Ridder, G. \& Sheng, S. (2017).
\newblock {\em Estimation of large network formation games}.
\newblock Technical report, University of California - Los Angeles.

\bibitem[Robins et~al., 2007a]{Robins_et_al_SN07a}
Robins, G., Pattison, P., Kalish, Y., \& Lusher, D. (2007a).
\newblock An introduction to exponential random graph ($p^*$) models for social
  networks.
\newblock {\em Social Networks}, 29(2), 173 -- 191.

\bibitem[Robins et~al., 2007b]{Robins_et_al_SN07b}
Robins, G., Snijders, T., Wang, P., Handcock, M., \& Pattison, P. (2007b).
\newblock Recent developments in exponential random graph ($p^*$) models for
  social networks.
\newblock {\em Social Networks}, 29(2), 192 -- 215.

\bibitem[Rose, 2005]{Rose_RIE05}
Rose, A. (2005).
\newblock Which international institutions promote international trade?
\newblock {\em Review of International Economics}, 13(4), 682 -- 698.

\bibitem[Rose, 2004]{Rose_AER04}
Rose, A.~K. (2004).
\newblock Do we really know that the wto increases trade?
\newblock {\em American Economic Review}, 94(1), 98 -- 114.

\bibitem[Rosenbaum, 2007]{Rosenbaum_JASA07}
Rosenbaum, P.~R. (2007).
\newblock Interference between units in randomized experiments.
\newblock {\em Journal of American Statistical Association}, 102(477), 191 --
  200.

\bibitem[Rubin, 1981]{Rubin_JES81}
Rubin, D.~B. (1981).
\newblock Estimation in parallel randomized experiments.
\newblock {\em Journal of Educational Statistics}, 6(4), 377 -- 400.

\bibitem[Russett \& Oneal, 2001]{Russett_Oneal_Book01}
Russett, B.~M. \& Oneal, J.~R. (2001).
\newblock {\em Triangulating Peace: Democracy, Interdependence and
  International Organizations}.
\newblock New York: Norton.

\bibitem[Santos~Silva \& Tenreyro, 2006]{SantosSilva_Tenreyro_RESTAT06}
Santos~Silva, J. \& Tenreyro, S. (2006).
\newblock The log of gravity.
\newblock {\em Review of Economics and Statistics}, 88(4), 641 -- 658.

\bibitem[Santos~Silva \& Tenreyro, 2010]{SantosSilva_Tenreyro_AR10}
Santos~Silva, J. \& Tenreyro, S. (2010).
\newblock Currency unions in prospect and retrospect.
\newblock {\em Annual Review of Economics}, 2, 51 -- 74.

\bibitem[Saygin et~al., 2014]{Saygin_et_al_IZA14}
Saygin, P.~O., Weber, A., \& Weynandt, M.~A. (2014).
\newblock {\em Coworkers, networks and job search outcomes}.
\newblock Iza dp no. 8174, IZA Institute of Labor Economics.

\bibitem[Schwartz \& Sommers, 2014]{Schwartz_Sommers_HA14}
Schwartz, A.~L. \& Sommers, B.~D. (2014).
\newblock Moving for medicaid? recent eligibility expansions did not induce
  migration from other states.
\newblock {\em Health Affairs}, 33(1), 88 -- 94.

\bibitem[Serfling, 1980]{Serfling_ATMS80}
Serfling, R.~J. (1980).
\newblock {\em Approximation Theorems of Mathematical Statistics}.
\newblock Wiley Series in Probability and Statistics. New York: John Wiley \&
  Sons, Inc.

\bibitem[Serpa \& Krishnan, 2017]{Serpa_Krishnan_MS17}
Serpa, J.~C. \& Krishnan, H. (2017).
\newblock The impact of supply chains on firm-level productivity.
\newblock {\em Management Science}.

\bibitem[Shalizi, 2016]{Shalizi_LN16}
Shalizi, C.~R. (2016).
\newblock {\em Lecture 1: Conditionally-independent dyad models}.
\newblock Lecture note, Carnegie Mellon University.

\bibitem[Shalizi \& Rinaldo, 2013]{Shalizi_Rinaldo_AS13}
Shalizi, C.~R. \& Rinaldo, A. (2013).
\newblock Consistency under sampling of exponential random graph models.
\newblock {\em Annals of Statistics}, 41(2), 508 -- 535.

\bibitem[Sheng, 2014]{Sheng_WP14}
Sheng, S. (2014).
\newblock A structural econometric analysis of network formation games.
\newblock Mimeo, University of California - Los Angeles.

\bibitem[Simmel, 1908]{Simmel_Book1908}
Simmel, G. (1908).
\newblock {\em Soziologie. Untersuchungen {\"u}ber die formen der
  vergesellschaftung}.
\newblock Leipzig: Duncker \& Humblot.

\bibitem[Snijders, 2002]{Snijders_JSS02}
Snijders, T. (2002).
\newblock Markov chain monte carlo estimation of exponential random graph
  models.
\newblock {\em Journal of Social Structure}, 3(2), 1--40.

\bibitem[Snijders \& Borgatti, 1999]{Snijders_Borgatti_C99}
Snijders, T. A.~B. \& Borgatti, S.~P. (1999).
\newblock Non-parametric standard errors and tests for network statistics.
\newblock {\em Connections}, 22(2), 61 -- 70.

\bibitem[Snijders et~al., 2006]{Snijders_SM06}
Snijders, T. A.~B., Pattison, P.~E., Robins, G.~L., \& Handcock, M.~S. (2006).
\newblock New specifications for exponential random graph models.
\newblock {\em Sociological Methodology}, 36(1), 99 -- 153.

\bibitem[Tabord-Meehan, 2018]{Tabord-Meehan_JBES18}
Tabord-Meehan, M. (2018).
\newblock Inference with dyadic data: Asymptotic behavior of the dyadic-robust
  t-statistic.
\newblock {\em Journal of Business and Economic Statistics}.

\bibitem[Tamer, 2003]{Tamer_ReStud03}
Tamer, E. (2003).
\newblock Incomplete simultaneous discrete response model with multiple
  equilibria.
\newblock {\em Review of Economic Studies}, 70(1), 147 -- 165.

\bibitem[Tarski, 1955]{Tarski_PJM55}
Tarski, A. (1955).
\newblock A lattice-theoretical fixpoint theorem and its applications.
\newblock {\em Pacific Journal of Mathematics}, 5(2), 285 -- 309.

\bibitem[Tinbergen, 1962]{Tinbergen_SWE62}
Tinbergen, J. (1962).
\newblock {\em Shaping the World Economy: Suggestions for an International
  Economic Policy}.
\newblock New York: Twentieth Century Fund.

\bibitem[Tomasello et~al., 2017]{Tomasello_et_al_ICC17}
Tomasello, M.~V., Napoletano, M., Garas, A., \& Schweitzer, F. (2017).
\newblock The rise and fall of r\&d networks.
\newblock {\em Industrial and Corporate Change}, 26(4), 617 -- 646.

\bibitem[Topkis, 1998]{Topkis_Book98}
Topkis, D.~M. (1998).
\newblock {\em Supermodularity and Complementarity}.
\newblock Princeton, NJ: Princeton University Press.

\bibitem[Townsend, 1994]{Townsend_EM94}
Townsend, R.~M. (1994).
\newblock Risk and insurance in village india.
\newblock {\em Econometrica}, 62(3), 539 -- 591.

\bibitem[Udry, 1994]{Udry_ReStud94}
Udry, C.~R. (1994).
\newblock Risk and insurance in a rural credit market: an empirical
  investigation in northern nigeria.
\newblock {\em Review of Economic Studies}, 61(3), 495 -- 526.

\bibitem[Uetake \& Watanabe, 2013]{Uetake_Watanabe_AE13}
Uetake, K. \& Watanabe, Y. (2013).
\newblock Estimating supermodular games using rationalizable strategies.
\newblock {\em Advances in Econometrics}, 31(1), 233 -- 247.

\bibitem[van~der Vaart, 2000]{vanderVaart_ASBook00}
van~der Vaart, A.~W. (2000).
\newblock {\em Asymptotic Statistics}.
\newblock Cambridge: Cambridge University Press.

\bibitem[van Duijn et~al., 2004]{vanDuijn_et_al_SN04}
van Duijn, M. A.~J., Snijders, T. A.~B., \& Zijlstra, B. J.~H. (2004).
\newblock p2: a random effects model with covariates for directed graphs.
\newblock {\em Statistica Neerlandica}, 58(2), 234 -- 254.

\bibitem[VanderWeele \& An, 2013]{VanderWeele_An_HCASR13}
VanderWeele, T.~J. \& An, W. (2013).
\newblock {\em Handbook of Causal Analysis for Social Research}, chapter Social
  networks and causal inference, (pp.\ 353 -- 374).
\newblock Springer: Dordrecht.

\bibitem[Varin et~al., 2011]{Varin_et_al_SS11}
Varin, C., Reid, N., \& Firth, D. (2011).
\newblock An overview of composite likelihood methods.
\newblock {\em Statistica Sinica}, 21(1), 5 -- 42.

\bibitem[Volfosky \& Airoldi, 2016]{Volfosky_Airoldi_SPL16}
Volfosky, A. \& Airoldi, E.~M. (2016).
\newblock Sharp total variational bounds for finitely exchangeable arrays.
\newblock {\em Statistics and Probability Letters}, 114, 54 -- 59.

\bibitem[Wainwright \& Jordan, 2008]{Wainwright_Jordan_FnTML08}
Wainwright, M.~J. \& Jordan, M.~I. (2008).
\newblock Graphical models, exponential families, and variational inference.
\newblock {\em Foundations and Trends in Machine Learning}, 1(1-2), 1 -- 305.

\bibitem[Wasserman \& Faust, 1994]{Wasserman_Faust_Bk94}
Wasserman, S. \& Faust, K. (1994).
\newblock {\em Social Network Analysis: Methods and Applications}.
\newblock Cambridge: Cambridge University Press.

\bibitem[Wasserman, 1977]{Wasserman_JMS77}
Wasserman, S.~S. (1977).
\newblock Random directed graph distributions and the triad census in social
  networks.
\newblock {\em Journal of Mathematical Sociology}, 5(1), 61 -- 86.

\bibitem[Wooldridge, 2005]{Wooldridge_IIEM05}
Wooldridge, J.~M. (2005).
\newblock {\em Identification and inference for econometric models}, chapter
  Unobserved heterogeneity and the estimation of average partial effects, (pp.\
  27 -- 55).
\newblock Number~3. Cambridge University Press: Cambridge.

\bibitem[Yan et~al., 2018]{Yan_et_al_JASA18}
Yan, T., Jiang, B., Fienberg, S.~E., \& Leng, C. (2018).
\newblock Statistical inference in a directed network model with covariates.
\newblock {\em Journal of the American Statistical Association}.

\bibitem[Zijlstra et~al., 2009]{Zijlstra_et_al_BJMSP09}
Zijlstra, B. J.~H., van Duijn, M. A.~J., \& Snijders, T. A.~B. (2009).
\newblock Mcmc estimation for the $p_{2}$ network regression model with crossed
  random effects.
\newblock {\em British Journal of Mathematical and Statistical Psychology},
  62(1), 143 -- 166.

\end{thebibliography}

\end{document}